\documentclass[11pt]{article} 

\usepackage[T1]{fontenc}        
\usepackage{textcomp, amssymb}  
\usepackage[english]{babel}

\usepackage{anysize}            
\usepackage{amsfonts}
\usepackage{amsmath}
\usepackage{dsfont}
\usepackage{MnSymbol}
\usepackage{mathtools}
\usepackage{epsfig}
\usepackage{epstopdf}
\usepackage{lmerge}
\usepackage{tikz}
\usepackage{amsthm}
\usepackage{ifpdf}

\usetikzlibrary{calc,decorations.pathmorphing,shapes,graphs}

\newcounter{sarrow}

\newcounter{sarrow1}
\newcommand\xnrsquigarrow[1]{%
\stepcounter{sarrow1}%
\mathrel{\begin{tikzpicture}[baseline= {( $ (current bounding box.south) + (0,-0.5ex) $ )}]
\node[inner sep=.5ex] (\thesarrow) {$\scriptstyle #1$};
\path[draw,<-,decorate,
  decoration={zigzag,amplitude=0.7pt,segment length=1.2mm,pre=lineto,pre length=4pt}]
    (\thesarrow1.south east) -- (\thesarrow1.south west);
    $\slashedarrowfill@\relbar\relbar/$
    \end{tikzpicture}}%
}

\makeatletter
\def\slashedarrowfill@#1#2#3#4#5{%
  $\m@th\thickmuskip0mu\medmuskip\thickmuskip\thinmuskip\thickmuskip
   \relax#5#1\mkern-7mu%
   \cleaders\hbox{$#5\mkern-2mu#2\mkern-2mu$}\hfill
   \mathclap{#3}\mathclap{#2}%
   \cleaders\hbox{$#5\mkern-2mu#2\mkern-2mu$}\hfill
   \mkern-7mu#4$%
}
\def\rightslashedarrowfillb@{%
  \slashedarrowfill@\relbar\relbar/\rightarrow}
\newcommand\xnrightarrow[2][]{%
  \ext@arrow 0055{\rightslashedarrowfillb@}{#1}{#2}}

\def\rightslashedarrowfille@{%
  \slashedarrowfill@\relbar\relbar/\twoheadrightarrow}
\newcommand\xntworightarrow[2][]{%
  \ext@arrow 0055{\rightslashedarrowfille@}{#1}{#2}}

\def\rightslashedarrowfillg@{%
  \slashedarrowfill@\relbar\relbar{\raisebox{.12em}{}}\twoheadrightarrow}
\newcommand\xtworightarrow[2][]{%
  \ext@arrow 0055{\rightslashedarrowfillg@}{#1}{#2}}

\def\rightslashedarrowfillx@{%
  \slashedarrowfill@\Relbar\Relbar/\rightrightarrows}
\newcommand\xnTworightarrow[2][]{%
  \ext@arrow 0055{\rightslashedarrowfillx@}{#1}{#2}}

\def\rightslashedarrowfilly@{%
  \slashedarrowfill@\Relbar\Relbar{\raisebox{.12em}{}}\rightrightarrows}
\newcommand\xTworightarrow[2][]{%
  \ext@arrow 0055{\rightslashedarrowfilly@}{#1}{#2}}

\pgfdeclareshape{slash underlined}
{
  \inheritsavedanchors[from=rectangle] 
  \inheritanchorborder[from=rectangle]
  \inheritanchor[from=rectangle]{north}
  \inheritanchor[from=rectangle]{north west}
  \inheritanchor[from=rectangle]{north east}
  \inheritanchor[from=rectangle]{center}
  \inheritanchor[from=rectangle]{west}
  \inheritanchor[from=rectangle]{east}
  \inheritanchor[from=rectangle]{mid}
  \inheritanchor[from=rectangle]{mid west}
  \inheritanchor[from=rectangle]{mid east}
  \inheritanchor[from=rectangle]{base}
  \inheritanchor[from=rectangle]{base west}
  \inheritanchor[from=rectangle]{base east}
  \inheritanchor[from=rectangle]{south}
  \inheritanchor[from=rectangle]{south west}
  \inheritanchor[from=rectangle]{south east}
  \inheritanchorborder[from=rectangle]
  \foregroundpath{
    \southwest \pgf@xa=\pgf@x \pgf@ya=\pgf@y
    \northeast \pgf@xb=\pgf@x \pgf@yb=\pgf@y
    \pgf@xc=\pgf@xa
    \advance\pgf@xc by .5\pgf@xb
    \pgf@yc=\pgf@ya
    \advance\pgf@xc by -1.3pt
    \advance\pgf@yc by -1.8pt
    \pgfpathmoveto{\pgfqpoint{\pgf@xc}{\pgf@yc}}
    \advance\pgf@xc by  2.6pt
    \advance\pgf@yc by  3.6pt
    \pgfpathlineto{\pgfqpoint{\pgf@xc}{\pgf@yc}}
    \pgfpathmoveto{\pgfqpoint{\pgf@xa}{\pgf@ya}}
    \pgfpathlineto{\pgfqpoint{\pgf@xb}{\pgf@ya}}
 }
}
\tikzset{nomorepostaction/.code=\let\tikz@postactions\pgfutil@empty}

\newtheorem{theorem}{Theorem}[section]
\newtheorem{definition}[theorem]{Definition}

\begin{document}

\begin{titlepage}
\thispagestyle{empty}

\hrule
\begin{center}
{\bf\LARGE Verification of Patterns}

\vspace{0.7cm}
--- Yong Wang ---

\vspace{2cm}
\begin{figure}[!htbp]
 \centering
 \includegraphics[width=1.0\textwidth]{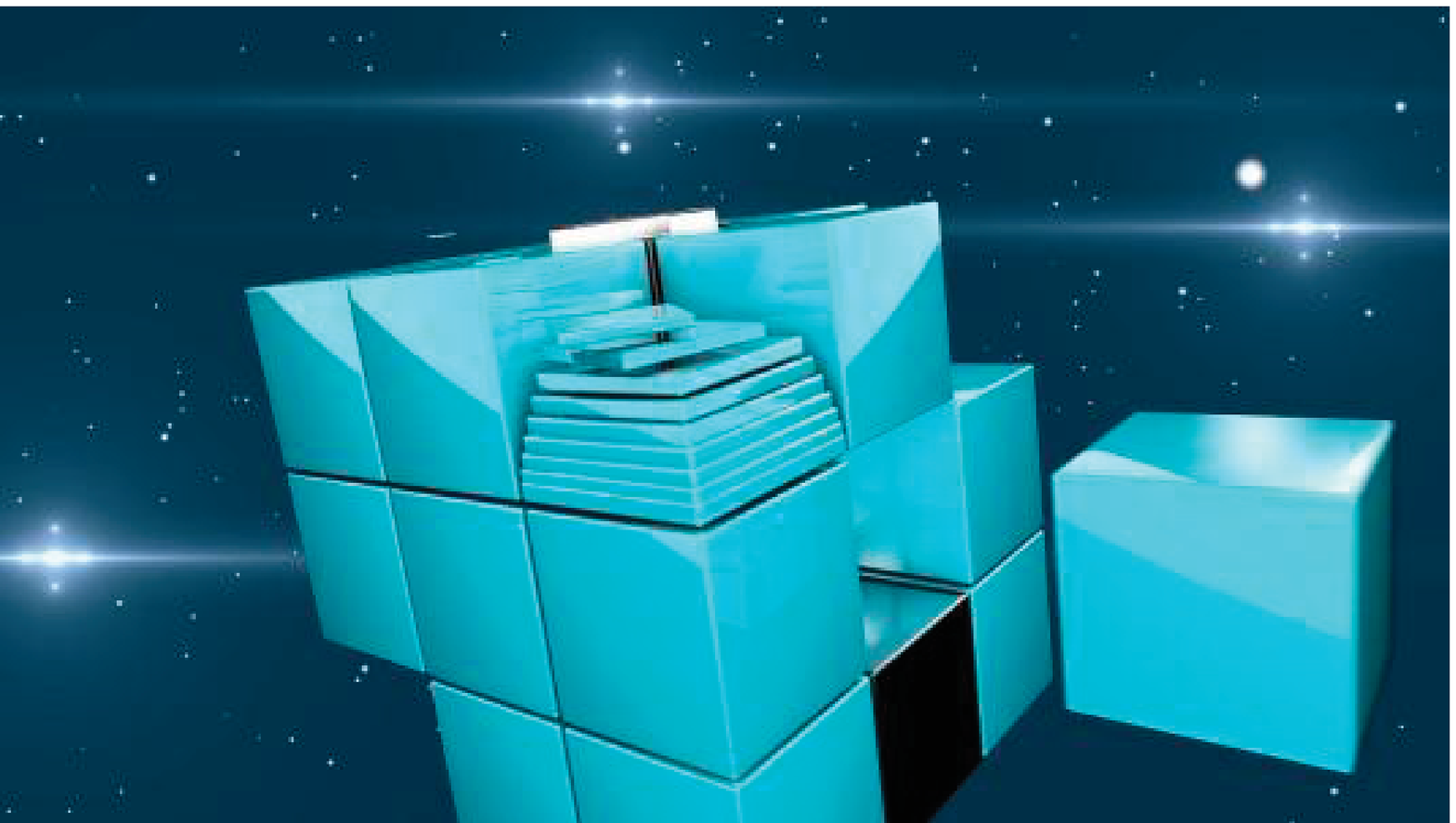}
\end{figure}

\end{center}
\end{titlepage}

\newpage 

\setcounter{page}{1}\pagenumbering{roman}

\tableofcontents

\newpage

\setcounter{page}{1}\pagenumbering{arabic}

        \section{Introduction}

The software patterns provide building blocks to the design and implementation of a software system, and try to make the software engineering to progress from experience to science.
The software patterns were made famous because of the introduction as the design patterns \cite{DP}. After that, patterns have been researched and developed widely and rapidly.

The series of books of pattern-oriented software architecture \cite{P1} \cite{P2} \cite{P3} \cite{P4} \cite{P5} should be marked in the development of software patterns. In these books,
patterns are detailed in the following aspects.

\begin{enumerate}
  \item Patterns are categorized from great granularity to tiny granularity. The greatest granularity is called architecture patterns, the medium granularity is called design patterns, and
  the tiniest granularity is called idioms. In each granularity, patterns are detailed and classified according to their functionalities.
  \item Every pattern is detailed according to a regular format to be understood and utilized easily, which includes introduction to a pattern on example, context, problem, solution, structure, dynamics, implementation,
  example resolved, variants.
  \item Except the general patterns, patterns of the vertical domains are also involved, including the domains of networked objects and resource management.
  \item To make the development and utilization of patterns scientifically, the pattern languages are discussed.
\end{enumerate}

As mentioned in these books, formalization of patterns and an intermediate pattern language are needed and should be developed in the future of patterns. So, in this book,
we formalize software patterns according to the categories of the series of books of pattern-oriented software architecture, and verify the correctness of patterns based on truly concurrent
process algebra \cite{ATC} \cite{CTC} \cite{PITC}. In one aspect, patterns are formalized and verified; in the other aspect, truly concurrent process algebra can play a role of an
intermediate pattern language for its rigorous theory.

This book is organized as follows.

In chapter 2, to make this book be self-satisfied, we introduce the preliminaries of truly concurrent process algebra, including its whole theory, modelling of race condition and
asynchronous communication, and applications.

In chapter 3, we formalize and verify the architectural patterns.

In chapter 4, we formalize and verify the design patterns.

In chapter 5, we formalize and verify the idioms.

In chapter 6, we formalize and verify the patterns for concurrent and networked objects.

In chapter 7, we formalize and verify the patterns for resource management.

In chapter 8, we show the formalization and verification of composition of patterns.

\newpage\section{Truly Concurrent Process Algebra}\label{tcpa}

In this chapter, we introduce the preliminaries on truly concurrent process algebra \cite{ATC} \cite{CTC} \cite{PITC}, which is based on truly concurrent operational semantics.

APTC eliminates the differences of structures of transition system, event structure, etc, and discusses their behavioral equivalences. It considers that there are two kinds of causality
relations: the chronological order modeled by the sequential composition and the causal order between different parallel branches modeled by the communication merge. It also considers
that there exist two kinds of confliction relations: the structural confliction modeled by the alternative composition and the conflictions in different parallel branches which should
be eliminated. Based on conservative extension, there are four modules in APTC: BATC (Basic Algebra for True Concurrency), APTC (Algebra for Parallelism in True Concurrency), recursion
and abstraction.

\subsection{Basic Algebra for True Concurrency}

BATC has sequential composition $\cdot$ and alternative composition $+$ to capture the chronological ordered causality and the structural confliction. The constants are ranged over $A$,
the set of atomic actions. The algebraic laws on $\cdot$ and $+$ are sound and complete modulo truly concurrent bisimulation equivalences (including pomset bisimulation, step
bisimulation, hp-bisimulation and hhp-bisimulation).

\begin{definition}[Prime event structure with silent event]\label{PES}
Let $\Lambda$ be a fixed set of labels, ranged over $a,b,c,\cdots$ and $\tau$. A ($\Lambda$-labelled) prime event structure with silent event $\tau$ is a tuple
$\mathcal{E}=\langle \mathbb{E}, \leq, \sharp, \lambda\rangle$, where $\mathbb{E}$ is a denumerable set of events, including the silent event $\tau$. Let
$\hat{\mathbb{E}}=\mathbb{E}\backslash\{\tau\}$, exactly excluding $\tau$, it is obvious that $\hat{\tau^*}=\epsilon$, where $\epsilon$ is the empty event.
Let $\lambda:\mathbb{E}\rightarrow\Lambda$ be a labelling function and let $\lambda(\tau)=\tau$. And $\leq$, $\sharp$ are binary relations on $\mathbb{E}$,
called causality and conflict respectively, such that:

\begin{enumerate}
  \item $\leq$ is a partial order and $\lceil e \rceil = \{e'\in \mathbb{E}|e'\leq e\}$ is finite for all $e\in \mathbb{E}$. It is easy to see that
  $e\leq\tau^*\leq e'=e\leq\tau\leq\cdots\leq\tau\leq e'$, then $e\leq e'$.
  \item $\sharp$ is irreflexive, symmetric and hereditary with respect to $\leq$, that is, for all $e,e',e''\in \mathbb{E}$, if $e\sharp e'\leq e''$, then $e\sharp e''$.
\end{enumerate}

Then, the concepts of consistency and concurrency can be drawn from the above definition:

\begin{enumerate}
  \item $e,e'\in \mathbb{E}$ are consistent, denoted as $e\frown e'$, if $\neg(e\sharp e')$. A subset $X\subseteq \mathbb{E}$ is called consistent, if $e\frown e'$ for all
  $e,e'\in X$.
  \item $e,e'\in \mathbb{E}$ are concurrent, denoted as $e\parallel e'$, if $\neg(e\leq e')$, $\neg(e'\leq e)$, and $\neg(e\sharp e')$.
\end{enumerate}
\end{definition}

\begin{definition}[Configuration]
Let $\mathcal{E}$ be a PES. A (finite) configuration in $\mathcal{E}$ is a (finite) consistent subset of events $C\subseteq \mathcal{E}$, closed with respect to causality
(i.e. $\lceil C\rceil=C$). The set of finite configurations of $\mathcal{E}$ is denoted by $\mathcal{C}(\mathcal{E})$. We let $\hat{C}=C\backslash\{\tau\}$.
\end{definition}

A consistent subset of $X\subseteq \mathbb{E}$ of events can be seen as a pomset. Given $X, Y\subseteq \mathbb{E}$, $\hat{X}\sim \hat{Y}$ if $\hat{X}$ and $\hat{Y}$ are
isomorphic as pomsets. In the following of the paper, we say $C_1\sim C_2$, we mean $\hat{C_1}\sim\hat{C_2}$.

\begin{definition}[Pomset transitions and step]
Let $\mathcal{E}$ be a PES and let $C\in\mathcal{C}(\mathcal{E})$, and $\emptyset\neq X\subseteq \mathbb{E}$, if $C\cap X=\emptyset$ and $C'=C\cup X\in\mathcal{C}(\mathcal{E})$,
then $C\xrightarrow{X} C'$ is called a pomset transition from $C$ to $C'$. When the events in $X$ are pairwise concurrent, we say that $C\xrightarrow{X}C'$ is a step.
\end{definition}

\begin{definition}[Pomset, step bisimulation]\label{PSB}
Let $\mathcal{E}_1$, $\mathcal{E}_2$ be PESs. A pomset bisimulation is a relation $R\subseteq\mathcal{C}(\mathcal{E}_1)\times\mathcal{C}(\mathcal{E}_2)$, such that if
$(C_1,C_2)\in R$, and $C_1\xrightarrow{X_1}C_1'$ then $C_2\xrightarrow{X_2}C_2'$, with $X_1\subseteq \mathbb{E}_1$, $X_2\subseteq \mathbb{E}_2$, $X_1\sim X_2$ and $(C_1',C_2')\in R$,
and vice-versa. We say that $\mathcal{E}_1$, $\mathcal{E}_2$ are pomset bisimilar, written $\mathcal{E}_1\sim_p\mathcal{E}_2$, if there exists a pomset bisimulation $R$, such that
$(\emptyset,\emptyset)\in R$. By replacing pomset transitions with steps, we can get the definition of step bisimulation. When PESs $\mathcal{E}_1$ and $\mathcal{E}_2$ are step
bisimilar, we write $\mathcal{E}_1\sim_s\mathcal{E}_2$.
\end{definition}

\begin{definition}[Posetal product]
Given two PESs $\mathcal{E}_1$, $\mathcal{E}_2$, the posetal product of their configurations, denoted $\mathcal{C}(\mathcal{E}_1)\overline{\times}\mathcal{C}(\mathcal{E}_2)$,
is defined as

$$\{(C_1,f,C_2)|C_1\in\mathcal{C}(\mathcal{E}_1),C_2\in\mathcal{C}(\mathcal{E}_2),f:C_1\rightarrow C_2 \textrm{ isomorphism}\}.$$

A subset $R\subseteq\mathcal{C}(\mathcal{E}_1)\overline{\times}\mathcal{C}(\mathcal{E}_2)$ is called a posetal relation. We say that $R$ is downward closed when for any
$(C_1,f,C_2),(C_1',f',C_2')\in \mathcal{C}(\mathcal{E}_1)\overline{\times}\mathcal{C}(\mathcal{E}_2)$, if $(C_1,f,C_2)\subseteq (C_1',f',C_2')$ pointwise and $(C_1',f',C_2')\in R$,
then $(C_1,f,C_2)\in R$.

For $f:X_1\rightarrow X_2$, we define $f[x_1\mapsto x_2]:X_1\cup\{x_1\}\rightarrow X_2\cup\{x_2\}$, $z\in X_1\cup\{x_1\}$,(1)$f[x_1\mapsto x_2](z)=
x_2$,if $z=x_1$;(2)$f[x_1\mapsto x_2](z)=f(z)$, otherwise. Where $X_1\subseteq \mathbb{E}_1$, $X_2\subseteq \mathbb{E}_2$, $x_1\in \mathbb{E}_1$, $x_2\in \mathbb{E}_2$.
\end{definition}

\begin{definition}[(Hereditary) history-preserving bisimulation]\label{HHPB}
A history-preserving (hp-) bisimulation is a posetal relation $R\subseteq\mathcal{C}(\mathcal{E}_1)\overline{\times}\mathcal{C}(\mathcal{E}_2)$ such that if $(C_1,f,C_2)\in R$,
and $C_1\xrightarrow{e_1} C_1'$, then $C_2\xrightarrow{e_2} C_2'$, with $(C_1',f[e_1\mapsto e_2],C_2')\in R$, and vice-versa. $\mathcal{E}_1,\mathcal{E}_2$ are history-preserving
(hp-)bisimilar and are written $\mathcal{E}_1\sim_{hp}\mathcal{E}_2$ if there exists a hp-bisimulation $R$ such that $(\emptyset,\emptyset,\emptyset)\in R$.

A hereditary history-preserving (hhp-)bisimulation is a downward closed hp-bisimulation. $\mathcal{E}_1,\mathcal{E}_2$ are hereditary history-preserving (hhp-)bisimilar and are
written $\mathcal{E}_1\sim_{hhp}\mathcal{E}_2$.
\end{definition}

In the following, let $e_1, e_2, e_1', e_2'\in \mathbb{E}$, and let variables $x,y,z$ range over the set of terms for true concurrency, $p,q,s$ range over the set of closed terms.
The set of axioms of BATC consists of the laws given in Table \ref{AxiomsForBATC}.

\begin{center}
    \begin{table}
        \begin{tabular}{@{}ll@{}}
            \hline No. &Axiom\\
            $A1$ & $x+ y = y+ x$\\
            $A2$ & $(x+ y)+ z = x+ (y+ z)$\\
            $A3$ & $x+ x = x$\\
            $A4$ & $(x+ y)\cdot z = x\cdot z + y\cdot z$\\
            $A5$ & $(x\cdot y)\cdot z = x\cdot(y\cdot z)$\\
        \end{tabular}
        \caption{Axioms of BATC}
        \label{AxiomsForBATC}
    \end{table}
\end{center}

We give the operational transition rules of operators $\cdot$ and $+$ as Table \ref{TRForBATC} shows. And the predicate $\xrightarrow{e}\surd$ represents successful termination after
execution of the event $e$.

\begin{center}
    \begin{table}
        $$\frac{}{e\xrightarrow{e}\surd}$$
        $$\frac{x\xrightarrow{e}\surd}{x+ y\xrightarrow{e}\surd} \quad\frac{x\xrightarrow{e}x'}{x+ y\xrightarrow{e}x'} \quad\frac{y\xrightarrow{e}\surd}{x+ y\xrightarrow{e}\surd}
        \quad\frac{y\xrightarrow{e}y'}{x+ y\xrightarrow{e}y'}$$
        $$\frac{x\xrightarrow{e}\surd}{x\cdot y\xrightarrow{e} y} \quad\frac{x\xrightarrow{e}x'}{x\cdot y\xrightarrow{e}x'\cdot y}$$
        \caption{Transition rules of BATC}
        \label{TRForBATC}
    \end{table}
\end{center}

\begin{theorem}[Soundness of BATC modulo truly concurrent bisimulation equivalences]\label{SBATC}
The axiomatization of BATC is sound modulo truly concurrent bisimulation equivalences $\sim_{p}$, $\sim_{s}$, $\sim_{hp}$ and $\sim_{hhp}$. That is,

\begin{enumerate}
  \item let $x$ and $y$ be BATC terms. If BATC $\vdash x=y$, then $x\sim_{p} y$;
  \item let $x$ and $y$ be BATC terms. If BATC $\vdash x=y$, then $x\sim_{s} y$;
  \item let $x$ and $y$ be BATC terms. If BATC $\vdash x=y$, then $x\sim_{hp} y$;
  \item let $x$ and $y$ be BATC terms. If BATC $\vdash x=y$, then $x\sim_{hhp} y$.
\end{enumerate}

\end{theorem}

\begin{theorem}[Completeness of BATC modulo truly concurrent bisimulation equivalences]\label{CBATC}
The axiomatization of BATC is complete modulo truly concurrent bisimulation equivalences $\sim_{p}$, $\sim_{s}$, $\sim_{hp}$ and $\sim_{hhp}$. That is,

\begin{enumerate}
  \item let $p$ and $q$ be closed BATC terms, if $p\sim_{p} q$ then $p=q$;
  \item let $p$ and $q$ be closed BATC terms, if $p\sim_{s} q$ then $p=q$;
  \item let $p$ and $q$ be closed BATC terms, if $p\sim_{hp} q$ then $p=q$;
  \item let $p$ and $q$ be closed BATC terms, if $p\sim_{hhp} q$ then $p=q$.
\end{enumerate}

\end{theorem}

\subsection{Algebra for Parallelism in True Concurrency}

APTC uses the whole parallel operator $\between$, the auxiliary binary parallel $\parallel$ to model parallelism, and the communication merge $\mid$ to model communications among
different parallel branches, and also the unary conflict elimination operator $\Theta$ and the binary unless operator $\triangleleft$ to eliminate conflictions among different parallel
branches. Since a communication may be blocked, a new constant called deadlock $\delta$ is extended to $A$, and also a new unary encapsulation operator $\partial_H$ is introduced to
eliminate $\delta$, which may exist in the processes. The algebraic laws on these operators are also sound and complete modulo truly concurrent bisimulation equivalences (including
pomset bisimulation, step bisimulation, hp-bisimulation, but not hhp-bisimulation). Note that, the parallel operator $\parallel$ in a process cannot be eliminated by deductions on
the process using axioms of APTC, but other operators can eventually be steadied by $\cdot$, $+$ and $\parallel$, this is also why truly concurrent bisimulations are called an
\emph{truly concurrent} semantics.

We design the axioms of APTC in Table \ref{AxiomsForAPTC}, including algebraic laws of parallel operator $\parallel$, communication operator $\mid$, conflict elimination operator
$\Theta$ and unless operator $\triangleleft$, encapsulation operator $\partial_H$, the deadlock constant $\delta$, and also the whole parallel operator $\between$.

\begin{center}
    \begin{table}
        \begin{tabular}{@{}ll@{}}
            \hline No. &Axiom\\
            $A6$ & $x+ \delta = x$\\
            $A7$ & $\delta\cdot x =\delta$\\
            $P1$ & $x\between y = x\parallel y + x\mid y$\\
            $P2$ & $x\parallel y = y \parallel x$\\
            $P3$ & $(x\parallel y)\parallel z = x\parallel (y\parallel z)$\\
            $P4$ & $e_1\parallel (e_2\cdot y) = (e_1\parallel e_2)\cdot y$\\
            $P5$ & $(e_1\cdot x)\parallel e_2 = (e_1\parallel e_2)\cdot x$\\
            $P6$ & $(e_1\cdot x)\parallel (e_2\cdot y) = (e_1\parallel e_2)\cdot (x\between y)$\\
            $P7$ & $(x+ y)\parallel z = (x\parallel z)+ (y\parallel z)$\\
            $P8$ & $x\parallel (y+ z) = (x\parallel y)+ (x\parallel z)$\\
            $P9$ & $\delta\parallel x = \delta$\\
            $P10$ & $x\parallel \delta = \delta$\\
            $C11$ & $e_1\mid e_2 = \gamma(e_1,e_2)$\\
            $C12$ & $e_1\mid (e_2\cdot y) = \gamma(e_1,e_2)\cdot y$\\
            $C13$ & $(e_1\cdot x)\mid e_2 = \gamma(e_1,e_2)\cdot x$\\
            $C14$ & $(e_1\cdot x)\mid (e_2\cdot y) = \gamma(e_1,e_2)\cdot (x\between y)$\\
            $C15$ & $(x+ y)\mid z = (x\mid z) + (y\mid z)$\\
            $C16$ & $x\mid (y+ z) = (x\mid y)+ (x\mid z)$\\
            $C17$ & $\delta\mid x = \delta$\\
            $C18$ & $x\mid\delta = \delta$\\
            $CE19$ & $\Theta(e) = e$\\
            $CE20$ & $\Theta(\delta) = \delta$\\
            $CE21$ & $\Theta(x+ y) = \Theta(x)\triangleleft y + \Theta(y)\triangleleft x$\\
            $CE22$ & $\Theta(x\cdot y)=\Theta(x)\cdot\Theta(y)$\\
            $CE23$ & $\Theta(x\parallel y) = ((\Theta(x)\triangleleft y)\parallel y)+ ((\Theta(y)\triangleleft x)\parallel x)$\\
            $CE24$ & $\Theta(x\mid y) = ((\Theta(x)\triangleleft y)\mid y)+ ((\Theta(y)\triangleleft x)\mid x)$\\
            $U25$ & $(\sharp(e_1,e_2))\quad e_1\triangleleft e_2 = \tau$\\
            $U26$ & $(\sharp(e_1,e_2),e_2\leq e_3)\quad e_1\triangleleft e_3 = e_1$\\
            $U27$ & $(\sharp(e_1,e_2),e_2\leq e_3)\quad e3\triangleleft e_1 = \tau$\\
            $U28$ & $e\triangleleft \delta = e$\\
            $U29$ & $\delta \triangleleft e = \delta$\\
            $U30$ & $(x+ y)\triangleleft z = (x\triangleleft z)+ (y\triangleleft z)$\\
            $U31$ & $(x\cdot y)\triangleleft z = (x\triangleleft z)\cdot (y\triangleleft z)$\\
            $U32$ & $(x\parallel y)\triangleleft z = (x\triangleleft z)\parallel (y\triangleleft z)$\\
            $U33$ & $(x\mid y)\triangleleft z = (x\triangleleft z)\mid (y\triangleleft z)$\\
            $U34$ & $x\triangleleft (y+ z) = (x\triangleleft y)\triangleleft z$\\
            $U35$ & $x\triangleleft (y\cdot z)=(x\triangleleft y)\triangleleft z$\\
            $U36$ & $x\triangleleft (y\parallel z) = (x\triangleleft y)\triangleleft z$\\
            $U37$ & $x\triangleleft (y\mid z) = (x\triangleleft y)\triangleleft z$\\
            $D1$ & $e\notin H\quad\partial_H(e) = e$\\
            $D2$ & $e\in H\quad \partial_H(e) = \delta$\\
            $D3$ & $\partial_H(\delta) = \delta$\\
            $D4$ & $\partial_H(x+ y) = \partial_H(x)+\partial_H(y)$\\
            $D5$ & $\partial_H(x\cdot y) = \partial_H(x)\cdot\partial_H(y)$\\
            $D6$ & $\partial_H(x\parallel y) = \partial_H(x)\parallel\partial_H(y)$\\
        \end{tabular}
        \caption{Axioms of APTC}
        \label{AxiomsForAPTC}
    \end{table}
\end{center}

we give the transition rules of APTC in Table \ref{TRForAPTC}, it is suitable for all truly concurrent behavioral equivalence, including pomset bisimulation, step bisimulation,
hp-bisimulation and hhp-bisimulation.

\begin{center}
    \begin{table}
        $$\frac{x\xrightarrow{e_1}\surd\quad y\xrightarrow{e_2}\surd}{x\parallel y\xrightarrow{\{e_1,e_2\}}\surd} \quad\frac{x\xrightarrow{e_1}x'\quad y\xrightarrow{e_2}\surd}{x\parallel y\xrightarrow{\{e_1,e_2\}}x'}$$
        $$\frac{x\xrightarrow{e_1}\surd\quad y\xrightarrow{e_2}y'}{x\parallel y\xrightarrow{\{e_1,e_2\}}y'} \quad\frac{x\xrightarrow{e_1}x'\quad y\xrightarrow{e_2}y'}{x\parallel y\xrightarrow{\{e_1,e_2\}}x'\between y'}$$
        $$\frac{x\xrightarrow{e_1}\surd\quad y\xrightarrow{e_2}\surd}{x\mid y\xrightarrow{\gamma(e_1,e_2)}\surd} \quad\frac{x\xrightarrow{e_1}x'\quad y\xrightarrow{e_2}\surd}{x\mid y\xrightarrow{\gamma(e_1,e_2)}x'}$$
        $$\frac{x\xrightarrow{e_1}\surd\quad y\xrightarrow{e_2}y'}{x\mid y\xrightarrow{\gamma(e_1,e_2)}y'} \quad\frac{x\xrightarrow{e_1}x'\quad y\xrightarrow{e_2}y'}{x\mid y\xrightarrow{\gamma(e_1,e_2)}x'\between y'}$$
        $$\frac{x\xrightarrow{e_1}\surd\quad (\sharp(e_1,e_2))}{\Theta(x)\xrightarrow{e_1}\surd} \quad\frac{x\xrightarrow{e_2}\surd\quad (\sharp(e_1,e_2))}{\Theta(x)\xrightarrow{e_2}\surd}$$
        $$\frac{x\xrightarrow{e_1}x'\quad (\sharp(e_1,e_2))}{\Theta(x)\xrightarrow{e_1}\Theta(x')} \quad\frac{x\xrightarrow{e_2}x'\quad (\sharp(e_1,e_2))}{\Theta(x)\xrightarrow{e_2}\Theta(x')}$$
        $$\frac{x\xrightarrow{e_1}\surd \quad y\nrightarrow^{e_2}\quad (\sharp(e_1,e_2))}{x\triangleleft y\xrightarrow{\tau}\surd}
        \quad\frac{x\xrightarrow{e_1}x' \quad y\nrightarrow^{e_2}\quad (\sharp(e_1,e_2))}{x\triangleleft y\xrightarrow{\tau}x'}$$
        $$\frac{x\xrightarrow{e_1}\surd \quad y\nrightarrow^{e_3}\quad (\sharp(e_1,e_2),e_2\leq e_3)}{x\triangleleft y\xrightarrow{e_1}\surd}
        \quad\frac{x\xrightarrow{e_1}x' \quad y\nrightarrow^{e_3}\quad (\sharp(e_1,e_2),e_2\leq e_3)}{x\triangleleft y\xrightarrow{e_1}x'}$$
        $$\frac{x\xrightarrow{e_3}\surd \quad y\nrightarrow^{e_2}\quad (\sharp(e_1,e_2),e_1\leq e_3)}{x\triangleleft y\xrightarrow{\tau}\surd}
        \quad\frac{x\xrightarrow{e_3}x' \quad y\nrightarrow^{e_2}\quad (\sharp(e_1,e_2),e_1\leq e_3)}{x\triangleleft y\xrightarrow{\tau}x'}$$
        $$\frac{x\xrightarrow{e}\surd}{\partial_H(x)\xrightarrow{e}\surd}\quad (e\notin H)\quad\quad\frac{x\xrightarrow{e}x'}{\partial_H(x)\xrightarrow{e}\partial_H(x')}\quad(e\notin H)$$
        \caption{Transition rules of APTC}
        \label{TRForAPTC}
    \end{table}
\end{center}

\begin{theorem}[Soundness of APTC modulo truly concurrent bisimulation equivalences]\label{SAPTC}
The axiomatization of APTC is sound modulo truly concurrent bisimulation equivalences $\sim_{p}$, $\sim_{s}$, and $\sim_{hp}$. That is,

\begin{enumerate}
  \item let $x$ and $y$ be APTC terms. If APTC $\vdash x=y$, then $x\sim_{p} y$;
  \item let $x$ and $y$ be APTC terms. If APTC $\vdash x=y$, then $x\sim_{s} y$;
  \item let $x$ and $y$ be APTC terms. If APTC $\vdash x=y$, then $x\sim_{hp} y$.
\end{enumerate}

\end{theorem}

\begin{theorem}[Completeness of APTC modulo truly concurrent bisimulation equivalences]\label{CAPTC}
The axiomatization of APTC is complete modulo truly concurrent bisimulation equivalences $\sim_{p}$, $\sim_{s}$, and $\sim_{hp}$. That is,

\begin{enumerate}
  \item let $p$ and $q$ be closed APTC terms, if $p\sim_{p} q$ then $p=q$;
  \item let $p$ and $q$ be closed APTC terms, if $p\sim_{s} q$ then $p=q$;
  \item let $p$ and $q$ be closed APTC terms, if $p\sim_{hp} q$ then $p=q$.
\end{enumerate}

\end{theorem}

\subsection{Recursion}

To model infinite computation, recursion is introduced into APTC. In order to obtain a sound and complete theory, guarded recursion and linear recursion are needed. The corresponding
axioms are RSP (Recursive Specification Principle) and RDP (Recursive Definition Principle), RDP says the solutions of a recursive specification can represent the behaviors of the
specification, while RSP says that a guarded recursive specification has only one solution, they are sound with respect to APTC with guarded recursion modulo several truly concurrent
bisimulation equivalences (including pomset bisimulation, step bisimulation and hp-bisimulation), and they are complete with respect to APTC with linear recursion modulo several truly
concurrent bisimulation equivalences (including pomset bisimulation, step bisimulation and hp-bisimulation). In the following, $E,F,G$ are recursion specifications, $X,Y,Z$ are
recursive variables.

For a guarded recursive specifications $E$ with the form

$$X_1=t_1(X_1,\cdots,X_n)$$
$$\cdots$$
$$X_n=t_n(X_1,\cdots,X_n)$$

the behavior of the solution $\langle X_i|E\rangle$ for the recursion variable $X_i$ in $E$, where $i\in\{1,\cdots,n\}$, is exactly the behavior of their right-hand sides
$t_i(X_1,\cdots,X_n)$, which is captured by the two transition rules in Table \ref{TRForGR}.

\begin{center}
    \begin{table}
        $$\frac{t_i(\langle X_1|E\rangle,\cdots,\langle X_n|E\rangle)\xrightarrow{\{e_1,\cdots,e_k\}}\surd}{\langle X_i|E\rangle\xrightarrow{\{e_1,\cdots,e_k\}}\surd}$$
        $$\frac{t_i(\langle X_1|E\rangle,\cdots,\langle X_n|E\rangle)\xrightarrow{\{e_1,\cdots,e_k\}} y}{\langle X_i|E\rangle\xrightarrow{\{e_1,\cdots,e_k\}} y}$$
        \caption{Transition rules of guarded recursion}
        \label{TRForGR}
    \end{table}
\end{center}

The $RDP$ (Recursive Definition Principle) and the $RSP$ (Recursive Specification Principle) are shown in Table \ref{RDPRSP}.

\begin{center}
\begin{table}
  \begin{tabular}{@{}ll@{}}
\hline No. &Axiom\\
  $RDP$ & $\langle X_i|E\rangle = t_i(\langle X_1|E,\cdots,X_n|E\rangle)\quad (i\in\{1,\cdots,n\})$\\
  $RSP$ & if $y_i=t_i(y_1,\cdots,y_n)$ for $i\in\{1,\cdots,n\}$, then $y_i=\langle X_i|E\rangle \quad(i\in\{1,\cdots,n\})$\\
\end{tabular}
\caption{Recursive definition and specification principle}
\label{RDPRSP}
\end{table}
\end{center}

\begin{theorem}[Soundness of $APTC$ with guarded recursion]\label{SAPTCR}
Let $x$ and $y$ be $APTC$ with guarded recursion terms. If $APTC\textrm{ with guarded recursion}\vdash x=y$, then
\begin{enumerate}
  \item $x\sim_{s} y$;
  \item $x\sim_{p} y$;
  \item $x\sim_{hp} y$.
\end{enumerate}
\end{theorem}

\begin{theorem}[Completeness of $APTC$ with linear recursion]\label{CAPTCR}
Let $p$ and $q$ be closed $APTC$ with linear recursion terms, then,
\begin{enumerate}
  \item if $p\sim_{s} q$ then $p=q$;
  \item if $p\sim_{p} q$ then $p=q$;
  \item if $p\sim_{hp} q$ then $p=q$.
\end{enumerate}
\end{theorem}

\subsection{Abstraction}

To abstract away internal implementations from the external behaviors, a new constant $\tau$ called silent step is added to $A$, and also a new unary abstraction operator
$\tau_I$ is used to rename actions in $I$ into $\tau$ (the resulted APTC with silent step and abstraction operator is called $\textrm{APTC}_{\tau}$). The recursive specification
is adapted to guarded linear recursion to prevent infinite $\tau$-loops specifically. The axioms of $\tau$ and $\tau_I$ are sound modulo rooted branching truly concurrent bisimulation
 equivalences (several kinds of weakly truly concurrent bisimulation equivalences, including rooted branching pomset bisimulation, rooted branching step bisimulation and rooted branching hp-bisimulation). To eliminate infinite $\tau$-loops caused by $\tau_I$ and obtain the completeness, CFAR (Cluster Fair Abstraction Rule) is used to prevent infinite $\tau$-loops in a constructible way.

\begin{definition}[Weak pomset transitions and weak step]
Let $\mathcal{E}$ be a PES and let $C\in\mathcal{C}(\mathcal{E})$, and $\emptyset\neq X\subseteq \hat{\mathbb{E}}$, if $C\cap X=\emptyset$ and
$\hat{C'}=\hat{C}\cup X\in\mathcal{C}(\mathcal{E})$, then $C\xRightarrow{X} C'$ is called a weak pomset transition from $C$ to $C'$, where we define
$\xRightarrow{e}\triangleq\xrightarrow{\tau^*}\xrightarrow{e}\xrightarrow{\tau^*}$. And $\xRightarrow{X}\triangleq\xrightarrow{\tau^*}\xrightarrow{e}\xrightarrow{\tau^*}$,
for every $e\in X$. When the events in $X$ are pairwise concurrent, we say that $C\xRightarrow{X}C'$ is a weak step.
\end{definition}

\begin{definition}[Branching pomset, step bisimulation]\label{BPSB}
Assume a special termination predicate $\downarrow$, and let $\surd$ represent a state with $\surd\downarrow$. Let $\mathcal{E}_1$, $\mathcal{E}_2$ be PESs. A branching pomset
bisimulation is a relation $R\subseteq\mathcal{C}(\mathcal{E}_1)\times\mathcal{C}(\mathcal{E}_2)$, such that:
 \begin{enumerate}
   \item if $(C_1,C_2)\in R$, and $C_1\xrightarrow{X}C_1'$ then
   \begin{itemize}
     \item either $X\equiv \tau^*$, and $(C_1',C_2)\in R$;
     \item or there is a sequence of (zero or more) $\tau$-transitions $C_2\xrightarrow{\tau^*} C_2^0$, such that $(C_1,C_2^0)\in R$ and $C_2^0\xRightarrow{X}C_2'$ with
     $(C_1',C_2')\in R$;
   \end{itemize}
   \item if $(C_1,C_2)\in R$, and $C_2\xrightarrow{X}C_2'$ then
   \begin{itemize}
     \item either $X\equiv \tau^*$, and $(C_1,C_2')\in R$;
     \item or there is a sequence of (zero or more) $\tau$-transitions $C_1\xrightarrow{\tau^*} C_1^0$, such that $(C_1^0,C_2)\in R$ and $C_1^0\xRightarrow{X}C_1'$ with
     $(C_1',C_2')\in R$;
   \end{itemize}
   \item if $(C_1,C_2)\in R$ and $C_1\downarrow$, then there is a sequence of (zero or more) $\tau$-transitions $C_2\xrightarrow{\tau^*}C_2^0$ such that $(C_1,C_2^0)\in R$
   and $C_2^0\downarrow$;
   \item if $(C_1,C_2)\in R$ and $C_2\downarrow$, then there is a sequence of (zero or more) $\tau$-transitions $C_1\xrightarrow{\tau^*}C_1^0$ such that $(C_1^0,C_2)\in R$
   and $C_1^0\downarrow$.
 \end{enumerate}

We say that $\mathcal{E}_1$, $\mathcal{E}_2$ are branching pomset bisimilar, written $\mathcal{E}_1\approx_{bp}\mathcal{E}_2$, if there exists a branching pomset bisimulation $R$,
such that $(\emptyset,\emptyset)\in R$.

By replacing pomset transitions with steps, we can get the definition of branching step bisimulation. When PESs $\mathcal{E}_1$ and $\mathcal{E}_2$ are branching step bisimilar,
we write $\mathcal{E}_1\approx_{bs}\mathcal{E}_2$.
\end{definition}

\begin{definition}[Rooted branching pomset, step bisimulation]\label{RBPSB}
Assume a special termination predicate $\downarrow$, and let $\surd$ represent a state with $\surd\downarrow$. Let $\mathcal{E}_1$, $\mathcal{E}_2$ be PESs. A branching pomset
bisimulation is a relation $R\subseteq\mathcal{C}(\mathcal{E}_1)\times\mathcal{C}(\mathcal{E}_2)$, such that:
 \begin{enumerate}
   \item if $(C_1,C_2)\in R$, and $C_1\xrightarrow{X}C_1'$ then $C_2\xrightarrow{X}C_2'$ with $C_1'\approx_{bp}C_2'$;
   \item if $(C_1,C_2)\in R$, and $C_2\xrightarrow{X}C_2'$ then $C_1\xrightarrow{X}C_1'$ with $C_1'\approx_{bp}C_2'$;
   \item if $(C_1,C_2)\in R$ and $C_1\downarrow$, then $C_2\downarrow$;
   \item if $(C_1,C_2)\in R$ and $C_2\downarrow$, then $C_1\downarrow$.
 \end{enumerate}

We say that $\mathcal{E}_1$, $\mathcal{E}_2$ are rooted branching pomset bisimilar, written $\mathcal{E}_1\approx_{rbp}\mathcal{E}_2$, if there exists a rooted branching pomset
bisimulation $R$, such that $(\emptyset,\emptyset)\in R$.

By replacing pomset transitions with steps, we can get the definition of rooted branching step bisimulation. When PESs $\mathcal{E}_1$ and $\mathcal{E}_2$ are rooted branching step
bisimilar, we write $\mathcal{E}_1\approx_{rbs}\mathcal{E}_2$.
\end{definition}

\begin{definition}[Branching (hereditary) history-preserving bisimulation]\label{BHHPB}
Assume a special termination predicate $\downarrow$, and let $\surd$ represent a state with $\surd\downarrow$. A branching history-preserving (hp-) bisimulation is a weakly posetal
relation $R\subseteq\mathcal{C}(\mathcal{E}_1)\overline{\times}\mathcal{C}(\mathcal{E}_2)$ such that:

 \begin{enumerate}
   \item if $(C_1,f,C_2)\in R$, and $C_1\xrightarrow{e_1}C_1'$ then
   \begin{itemize}
     \item either $e_1\equiv \tau$, and $(C_1',f[e_1\mapsto \tau],C_2)\in R$;
     \item or there is a sequence of (zero or more) $\tau$-transitions $C_2\xrightarrow{\tau^*} C_2^0$, such that $(C_1,f,C_2^0)\in R$ and $C_2^0\xrightarrow{e_2}C_2'$ with
     $(C_1',f[e_1\mapsto e_2],C_2')\in R$;
   \end{itemize}
   \item if $(C_1,f,C_2)\in R$, and $C_2\xrightarrow{e_2}C_2'$ then
   \begin{itemize}
     \item either $X\equiv \tau$, and $(C_1,f[e_2\mapsto \tau],C_2')\in R$;
     \item or there is a sequence of (zero or more) $\tau$-transitions $C_1\xrightarrow{\tau^*} C_1^0$, such that $(C_1^0,f,C_2)\in R$ and $C_1^0\xrightarrow{e_1}C_1'$ with
     $(C_1',f[e_2\mapsto e_1],C_2')\in R$;
   \end{itemize}
   \item if $(C_1,f,C_2)\in R$ and $C_1\downarrow$, then there is a sequence of (zero or more) $\tau$-transitions $C_2\xrightarrow{\tau^*}C_2^0$ such that $(C_1,f,C_2^0)\in R$
   and $C_2^0\downarrow$;
   \item if $(C_1,f,C_2)\in R$ and $C_2\downarrow$, then there is a sequence of (zero or more) $\tau$-transitions $C_1\xrightarrow{\tau^*}C_1^0$ such that $(C_1^0,f,C_2)\in R$
   and $C_1^0\downarrow$.
 \end{enumerate}

$\mathcal{E}_1,\mathcal{E}_2$ are branching history-preserving (hp-)bisimilar and are written $\mathcal{E}_1\approx_{bhp}\mathcal{E}_2$ if there exists a branching hp-bisimulation
$R$ such that $(\emptyset,\emptyset,\emptyset)\in R$.

A branching hereditary history-preserving (hhp-)bisimulation is a downward closed branching hhp-bisimulation. $\mathcal{E}_1,\mathcal{E}_2$ are branching hereditary history-preserving
(hhp-)bisimilar and are written $\mathcal{E}_1\approx_{bhhp}\mathcal{E}_2$.
\end{definition}

\begin{definition}[Rooted branching (hereditary) history-preserving bisimulation]\label{RBHHPB}
Assume a special termination predicate $\downarrow$, and let $\surd$ represent a state with $\surd\downarrow$. A rooted branching history-preserving (hp-) bisimulation is a weakly
posetal relation $R\subseteq\mathcal{C}(\mathcal{E}_1)\overline{\times}\mathcal{C}(\mathcal{E}_2)$ such that:

 \begin{enumerate}
   \item if $(C_1,f,C_2)\in R$, and $C_1\xrightarrow{e_1}C_1'$, then $C_2\xrightarrow{e_2}C_2'$ with $C_1'\approx_{bhp}C_2'$;
   \item if $(C_1,f,C_2)\in R$, and $C_2\xrightarrow{e_2}C_1'$, then $C_1\xrightarrow{e_1}C_2'$ with $C_1'\approx_{bhp}C_2'$;
   \item if $(C_1,f,C_2)\in R$ and $C_1\downarrow$, then $C_2\downarrow$;
   \item if $(C_1,f,C_2)\in R$ and $C_2\downarrow$, then $C_1\downarrow$.
 \end{enumerate}

$\mathcal{E}_1,\mathcal{E}_2$ are rooted branching history-preserving (hp-)bisimilar and are written $\mathcal{E}_1\approx_{rbhp}\mathcal{E}_2$ if there exists rooted a branching
hp-bisimulation $R$ such that $(\emptyset,\emptyset,\emptyset)\in R$.

A rooted branching hereditary history-preserving (hhp-)bisimulation is a downward closed rooted branching hhp-bisimulation. $\mathcal{E}_1,\mathcal{E}_2$ are rooted branching
hereditary history-preserving (hhp-)bisimilar and are written $\mathcal{E}_1\approx_{rbhhp}\mathcal{E}_2$.
\end{definition}

The axioms and transition rules of $\textrm{APTC}_{\tau}$ are shown in Table \ref{AxiomsForTau} and Table \ref{TRForTau}.

\begin{center}
\begin{table}
  \begin{tabular}{@{}ll@{}}
\hline No. &Axiom\\
  $B1$ & $e\cdot\tau=e$\\
  $B2$ & $e\cdot(\tau\cdot(x+y)+x)=e\cdot(x+y)$\\
  $B3$ & $x\parallel\tau=x$\\
  $TI1$ & $e\notin I\quad \tau_I(e)=e$\\
  $TI2$ & $e\in I\quad \tau_I(e)=\tau$\\
  $TI3$ & $\tau_I(\delta)=\delta$\\
  $TI4$ & $\tau_I(x+y)=\tau_I(x)+\tau_I(y)$\\
  $TI5$ & $\tau_I(x\cdot y)=\tau_I(x)\cdot\tau_I(y)$\\
  $TI6$ & $\tau_I(x\parallel y)=\tau_I(x)\parallel\tau_I(y)$\\
  $CFAR$ & If $X$ is in a cluster for $I$ with exits \\
           & $\{(a_{11}\parallel\cdots\parallel a_{1i})Y_1,\cdots,(a_{m1}\parallel\cdots\parallel a_{mi})Y_m, b_{11}\parallel\cdots\parallel b_{1j},\cdots,b_{n1}\parallel\cdots\parallel b_{nj}\}$, \\
           & then $\tau\cdot\tau_I(\langle X|E\rangle)=$\\
           & $\tau\cdot\tau_I((a_{11}\parallel\cdots\parallel a_{1i})\langle Y_1|E\rangle+\cdots+(a_{m1}\parallel\cdots\parallel a_{mi})\langle Y_m|E\rangle+b_{11}\parallel\cdots\parallel b_{1j}+\cdots+b_{n1}\parallel\cdots\parallel b_{nj})$\\
\end{tabular}
\caption{Axioms of $\textrm{APTC}_{\tau}$}
\label{AxiomsForTau}
\end{table}
\end{center}

\begin{center}
    \begin{table}
        $$\frac{}{\tau\xrightarrow{\tau}\surd}$$
        $$\frac{x\xrightarrow{e}\surd}{\tau_I(x)\xrightarrow{e}\surd}\quad e\notin I
        \quad\quad\frac{x\xrightarrow{e}x'}{\tau_I(x)\xrightarrow{e}\tau_I(x')}\quad e\notin I$$

        $$\frac{x\xrightarrow{e}\surd}{\tau_I(x)\xrightarrow{\tau}\surd}\quad e\in I
        \quad\quad\frac{x\xrightarrow{e}x'}{\tau_I(x)\xrightarrow{\tau}\tau_I(x')}\quad e\in I$$
        \caption{Transition rule of $\textrm{APTC}_{\tau}$}
        \label{TRForTau}
    \end{table}
\end{center}

\begin{theorem}[Soundness of $APTC_{\tau}$ with guarded linear recursion]\label{SAPTCABS}
Let $x$ and $y$ be $APTC_{\tau}$ with guarded linear recursion terms. If $APTC_{\tau}$ with guarded linear recursion $\vdash x=y$, then
\begin{enumerate}
  \item $x\approx_{rbs} y$;
  \item $x\approx_{rbp} y$;
  \item $x\approx_{rbhp} y$.
\end{enumerate}
\end{theorem}

\begin{theorem}[Soundness of $CFAR$]\label{SCFAR}
$CFAR$ is sound modulo rooted branching truly concurrent bisimulation equivalences $\approx_{rbs}$, $\approx_{rbp}$ and $\approx_{rbhp}$.
\end{theorem}

\begin{theorem}[Completeness of $APTC_{\tau}$ with guarded linear recursion and $CFAR$]\label{CCFAR}
Let $p$ and $q$ be closed $APTC_{\tau}$ with guarded linear recursion and $CFAR$ terms, then,
\begin{enumerate}
  \item if $p\approx_{rbs} q$ then $p=q$;
  \item if $p\approx_{rbp} q$ then $p=q$;
  \item if $p\approx_{rbhp} q$ then $p=q$.
\end{enumerate}
\end{theorem}

\subsection{Placeholder}

We introduce a constant called shadow constant $\circledS$ to act for the placeholder that we ever used to deal entanglement in quantum process algebra. The transition rule of the shadow constant $\circledS$ is shown in Table \ref{TRForShadow}. The rule say that $\circledS$ can terminate successfully without executing any action.

\begin{center}
    \begin{table}
        $$\frac{}{\circledS\rightarrow\surd}$$
        \caption{Transition rule of the shadow constant}
        \label{TRForShadow}
    \end{table}
\end{center}

We need to adjust the definition of guarded linear recursive specification
to the following one.

\begin{definition}[Guarded linear recursive specification]\label{GLRSS}
A linear recursive specification $E$ is guarded if there does not exist an infinite sequence of $\tau$-transitions $\langle X|E\rangle\xrightarrow{\tau}\langle X'|E\rangle\xrightarrow{\tau}\langle X''|E\rangle\xrightarrow{\tau}\cdots$, and there does not exist an infinite sequence of $\circledS$-transitions $\langle X|E\rangle\rightarrow\langle X'|E\rangle\rightarrow\langle X''|E\rangle\rightarrow\cdots$.
\end{definition}

\begin{theorem}[Conservativity of $APTC$ with respect to the shadow constant]
$APTC_{\tau}$ with guarded linear recursion and shadow constant is a conservative extension of $APTC_{\tau}$ with guarded linear recursion.
\end{theorem}

We design the axioms for the shadow constant $\circledS$ in Table \ref{AxiomsForShadow}. And for $\circledS^e_i$, we add superscript $e$ to denote $\circledS$ is belonging to $e$ and subscript $i$ to denote that it is the $i$-th shadow of $e$. And we extend the set $\mathbb{E}$ to the set $\mathbb{E}\cup\{\tau\}\cup\{\delta\}\cup\{\circledS^{e}_i\}$.

\begin{center}
\begin{table}
  \begin{tabular}{@{}ll@{}}
\hline No. &Axiom\\
  $SC1$ & $\circledS\cdot x = x$\\
  $SC2$ & $x\cdot \circledS = x$\\
  $SC3$ & $\circledS^{e}\parallel e=e$\\
  $SC4$ & $e\parallel(\circledS^{e}\cdot y) = e\cdot y$\\
  $SC5$ & $\circledS^{e}\parallel(e\cdot y) = e\cdot y$\\
  $SC6$ & $(e\cdot x)\parallel\circledS^{e} = e\cdot x$\\
  $SC7$ & $(\circledS^{e}\cdot x)\parallel e = e\cdot x$\\
  $SC8$ & $(e\cdot x)\parallel(\circledS^{e}\cdot y) = e\cdot (x\between y)$\\
  $SC9$ & $(\circledS^{e}\cdot x)\parallel(e\cdot y) = e\cdot (x\between y)$\\
\end{tabular}
\caption{Axioms of shadow constant}
\label{AxiomsForShadow}
\end{table}
\end{center}

The mismatch of action and its shadows in parallelism will cause deadlock, that is, $e\parallel \circledS^{e'}=\delta$ with $e\neq e'$. We must make all shadows $\circledS^e_i$ are distinct, to ensure $f$ in hp-bisimulation is an isomorphism.

\begin{theorem}[Soundness of the shadow constant]\label{SShadow}
Let $x$ and $y$ be $APTC_{\tau}$ with guarded linear recursion and the shadow constant terms. If $APTC_{\tau}$ with guarded linear recursion and the shadow constant $\vdash x=y$, then
\begin{enumerate}
  \item $x\approx_{rbs} y$;
  \item $x\approx_{rbp} y$;
  \item $x\approx_{rbhp} y$.
\end{enumerate}
\end{theorem}

\begin{theorem}[Completeness of the shadow constant]\label{CRenaming}
Let $p$ and $q$ be closed $APTC_{\tau}$ with guarded linear recursion and $CFAR$ and the shadow constant terms, then,
\begin{enumerate}
  \item if $p\approx_{rbs} q$ then $p=q$;
  \item if $p\approx_{rbp} q$ then $p=q$;
  \item if $p\approx_{rbhp} q$ then $p=q$.
\end{enumerate}
\end{theorem}

With the shadow constant, we have

\begin{eqnarray}
\partial_H((a\cdot r_b)\between w_b)&=&\partial_H((a\cdot r_b) \between (\circledS^a_1\cdot w_b)) \nonumber\\
&=&a\cdot c_b\nonumber
\end{eqnarray}

with $H=\{r_b,w_b\}$ and $\gamma(r_b,w_b)\triangleq c_b$.

And we see the following example:

\begin{eqnarray}
a\between b&=&a\parallel b+a\mid b \nonumber\\
&=&a\parallel b + a\parallel b + a\parallel b +a\mid b \nonumber\\
&=&a\parallel (\circledS^a_1\cdot b) + (\circledS^b_1\cdot a)\parallel b+a\parallel b +a\mid b \nonumber\\
&=&(a\parallel\circledS^a_1)\cdot b + (\circledS^b_1\parallel b)\cdot a+a\parallel b +a\mid b \nonumber\\
&=&a\cdot b+b\cdot a+a\parallel b + a\mid b\nonumber
\end{eqnarray}

What do we see? Yes. The parallelism contains both interleaving and true concurrency. This may be why true concurrency is called \emph{\textbf{true} concurrency}.

\subsection{State and Race Condition}


State operator permits explicitly to describe states, where $S$ denotes a finite set of states, $action(s,e)$ denotes the visible behavior of $e$ in state $s$ with
$action:S\times \mathbb{E}\rightarrow \mathbb{E}$, $effect(s,e)$ represents the state that results if $e$ is executed in $s$ with $effect:S\times \mathbb{E}\rightarrow S$.
State operator $\lambda_s(t)$ which denotes process term $t$ in $s$, is expressed by the following transition rules in Table \ref{TRForState}. Note that $action$ and $effect$
are extended to $\mathbb{E}\cup\{\tau\}$ by defining $action(s,\tau)\triangleq\tau$ and $effect(s,\tau)\triangleq s$. We use $e_1\%e_2$ to denote that $e_1$ and $e_2$ are in
race condition.

\begin{center}
    \begin{table}
        $$\frac{x\xrightarrow{e}\surd}{\lambda_s(x)\xrightarrow{action(s,e)}\surd}
        \quad\frac{x\xrightarrow{e}x'}{\lambda_s(x)\xrightarrow{action(s,e)}\lambda_{effect(s,e)}(x')}$$

        $$\frac{x\xrightarrow{e_1}\surd\quad y\xnrightarrow{e_2}\quad(e_1\%e_2)}{\lambda_s(x\parallel y)\xrightarrow{action(s,e_1)}\lambda_{effect(s,e_1)}(y)} \quad\frac{x\xrightarrow{e_1}x'\quad y\xnrightarrow{e_2}\quad(e_1\%e_2)}{\lambda_s(x\parallel y)\xrightarrow{action(s,e_1)}\lambda_{effect(s,e_1)}(x'\between y)}$$

        $$\frac{x\xnrightarrow{e_1}\quad y\xrightarrow{e_2}\surd\quad(e_1\%e_2)}{\lambda_s(x\parallel y)\xrightarrow{action(s,e_2)}\lambda_{effect(s,e_2)}(x)} \quad\frac{x\xnrightarrow{e_1}\quad y\xrightarrow{e_2}y'\quad(e_1\%e_2)}{\lambda_s(x\parallel y)\xrightarrow{action(s,e_2)}\lambda_{effect(s,e_2)}(x\between y')}$$

        $$\frac{x\xrightarrow{e_1}\surd\quad y\xrightarrow{e_2}\surd}{\lambda_s(x\parallel y)\xrightarrow{\{action(s,e_1),action(s,e_2)\}}\surd}$$

        $$\frac{x\xrightarrow{e_1}x'\quad y\xrightarrow{e_2}\surd}{\lambda_s(x\parallel y)\xrightarrow{\{action(s,e_1),action(s,e_2)\}}\lambda_{effect(s,e_1)\cup effect(s,e_2)}(x')}$$

        $$\frac{x\xrightarrow{e_1}\surd\quad y\xrightarrow{e_2}y'}{\lambda_s(x\parallel y)\xrightarrow{\{action(s,e_1),action(s,e_2)\}}\lambda_{effect(s,e_1)\cup effect(s,e_2)}(y')}$$

        $$\frac{x\xrightarrow{e_1}x'\quad y\xrightarrow{e_2}y'}{\lambda_s(x\parallel y)\xrightarrow{\{action(s,e_1),action(s,e_2)\}}\lambda_{effect(s,e_1)\cup effect(s,e_2)}(x'\between y')}$$
        \caption{Transition rule of the state operator}
        \label{TRForState}
    \end{table}
\end{center}

\begin{theorem}[Conservativity of $APTC$ with respect to the state operator]
$APTC_{\tau}$ with guarded linear recursion and state operator is a conservative extension of $APTC_{\tau}$ with guarded linear recursion.
\end{theorem}

\begin{proof}
It follows from the following two facts.

\begin{enumerate}
  \item The transition rules of $APTC_{\tau}$ with guarded linear recursion are all source-dependent;
  \item The sources of the transition rules for the state operator contain an occurrence of $\lambda_s$.
\end{enumerate}
\end{proof}

\begin{theorem}[Congruence theorem of the state operator]
Rooted branching truly concurrent bisimulation equivalences $\approx_{rbp}$, $\approx_{rbs}$ and $\approx_{rbhp}$ are all congruences with respect to $APTC_{\tau}$ with guarded linear recursion and the state operator.
\end{theorem}

\begin{proof}

(1) Case rooted branching pomset bisimulation equivalence $\approx_{rbp}$.

Let $x$ and $y$ be $APTC_{\tau}$ with guarded linear recursion and the state operator processes, and $x\approx_{rbp} y$, it is sufficient to prove that $\lambda_s(x)\approx_{rbp} \lambda_s(y)$.

By the transition rules for operator $\lambda_s$ in Table \ref{TRForState}, we can get

$$\lambda_s(x)\xrightarrow{action(s,X)} \surd \quad \lambda_s(y)\xrightarrow{action(s,Y)} \surd$$

with $X\subseteq x$, $Y\subseteq y$, and $X\sim Y$.

Or, we can get

$$\lambda_s(x)\xrightarrow{action(s,X)} \lambda_{effect(s,X)}(x') \quad \lambda_s(y)\xrightarrow{action(s,Y)} \lambda_{effect(s,Y)}(y')$$

with $X\subseteq x$, $Y\subseteq y$, and $X\sim Y$ and the hypothesis $\lambda_{effect(s,X)}(x')\approx_{rbp}\lambda_{effect(s,Y)}(y')$.

So, we get $\lambda_s(x)\approx_{rbp} \lambda_s(y)$, as desired

(2) The cases of rooted branching step bisimulation $\approx_{rbs}$, rooted branching hp-bisimulation $\approx_{rbhp}$ can be proven similarly, we omit them.
\end{proof}


We design the axioms for the state operator $\lambda_s$ in Table \ref{AxiomsForState}.

\begin{center}
\begin{table}
  \begin{tabular}{@{}ll@{}}
\hline No. &Axiom\\
  $SO1$ & $\lambda_s(e)=action(s,e)$\\
  $SO2$ & $\lambda_s(\delta)=\delta$\\
  $SO3$ & $\lambda_s(x+y)=\lambda_s(x)+\lambda_s(y)$\\
  $SO4$ & $\lambda_s(e\cdot y)=action(s,e)\cdot\lambda_{effect(s,e)}(y)$\\
  $SO5$ & $\lambda_s(x\parallel y)=\lambda_s(x)\parallel\lambda_s(y)$\\
\end{tabular}
\caption{Axioms of state operator}
\label{AxiomsForState}
\end{table}
\end{center}

\begin{theorem}[Soundness of the state operator]\label{SState}
Let $x$ and $y$ be $APTC_{\tau}$ with guarded linear recursion and the state operator terms. If $APTC_{\tau}$ with guarded linear recursion and the state operator $\vdash x=y$, then
\begin{enumerate}
  \item $x\approx_{rbs} y$;
  \item $x\approx_{rbp} y$;
  \item $x\approx_{rbhp} y$.
\end{enumerate}
\end{theorem}

\begin{proof}
(1) Soundness of $APTC_{\tau}$ with guarded linear recursion and the state operator with respect to rooted branching step bisimulation $\approx_{rbs}$.

Since rooted branching step bisimulation $\approx_{rbs}$ is both an equivalent and a congruent relation with respect to $APTC_{\tau}$ with guarded linear recursion and the state operator, we only need to check if each axiom in Table \ref{AxiomsForState} is sound modulo rooted branching step bisimulation equivalence.

Though transition rules in Table \ref{TRForState} are defined in the flavor of single event, they can be modified into a step (a set of events within which each event is pairwise concurrent), we omit them. If we treat a single event as a step containing just one event, the proof of this soundness theorem does not exist any problem, so we use this way and still use the transition rules in Table \ref{AxiomsForState}.

We only prove soundness of the non-trivial axioms $SO3-SO5$, and omit the defining axioms $SO1-SO2$.

\begin{itemize}
  \item \textbf{Axiom $SO3$}. Let $p,q$ be $APTC_{\tau}$ with guarded linear recursion and the state operator processes, and $\lambda_s(p+ q)=\lambda_s(p)+\lambda_s(q)$, it is sufficient to prove that $\lambda_s(p+ q)\approx_{rbs} \lambda_s(p)+\lambda_s(q)$. By the transition rules for operator $+$ and $\lambda_s$ in Table \ref{TRForState}, we get

      $$\frac{p\xrightarrow{e_1}\surd}{\lambda_s(p+ q)\xrightarrow{action(s,e_1)}\surd}
      \quad\frac{p\xrightarrow{e_1}\surd}{\lambda_s(p)+ \lambda_s(q)\xrightarrow{action(s,e_1)}\surd}$$

      $$\frac{q\xrightarrow{e_2}\surd}{\lambda_s(p+ q)\xrightarrow{action(s,e_2)}\surd}
      \quad\frac{q\xrightarrow{e_2}\surd}{\lambda_s(p)+ \lambda_s(q)\xrightarrow{action(s,e_2)}\surd}$$

      $$\frac{p\xrightarrow{e_1}p'}{\lambda_s(p+ q)\xrightarrow{action(s,e_1)}\lambda_{effect(s,e_1)}(p')}
      \quad\frac{p\xrightarrow{e_1}p'}{\lambda_s(p)+ \lambda_s(q)\xrightarrow{action(s,e_1)}\lambda_{effect(s,e_1)}(p')}$$

      $$\frac{q\xrightarrow{e_2}q'}{\lambda_s(p+ q)\xrightarrow{action(s,e_2)}\lambda_{effect(s,e_2)}(q')}
      \quad\frac{q\xrightarrow{e_2}q'}{\lambda_s(p)+ \lambda_s(q)\xrightarrow{action(s,e_2)}\lambda_{effect(s,e_2)}(q')}$$

      So, $\lambda_s(p+ q)\approx_{rbs} \lambda_s(p)+\lambda_s(q)$, as desired.
  \item \textbf{Axiom $SO4$}. Let $q$ be $APTC_{\tau}$ with guarded linear recursion and the state operator processes, and $\lambda_s(e\cdot q)=action(s,e)\cdot\lambda_{effect(s,e)}(q)$, it is sufficient to prove that $\lambda_s(e\cdot q)\approx_{rbs}action(s,e)\cdot\lambda_{effect(s,e)}(q)$. By the transition rules for operator $\cdot$ and $\lambda_s$ in Table \ref{TRForState}, we get

      $$\frac{e\xrightarrow{e}\surd}{\lambda_s(e\cdot q)\xrightarrow{action(s,e)}\lambda_{effect(s,e)}(q)}
      \quad\frac{action(s,e)\xrightarrow{action(s,e)}\surd}{action(s,e)\cdot \lambda_{effect(s,e)}(q)\xrightarrow{action(s,e)}\lambda_{effect(s,e)}(q)}$$

      So, $\lambda_s(e\cdot q)\approx_{rbs}action(s,e)\cdot\lambda_{effect(s,e)}(q)$, as desired.
  \item \textbf{Axiom $SO5$}. Let $p,q$ be $APTC_{\tau}$ with guarded linear recursion and the state operator processes, and $\lambda_s(p\parallel q)=\lambda_s(p)\parallel\lambda_s(q)$, it is sufficient to prove that $\lambda_s(p\parallel q)\approx_{rbs} \lambda_s(p)\parallel\lambda_s(q)$. By the transition rules for operator $\parallel$ and $\lambda_s$ in Table \ref{TRForState}, we get for the case $\neg(e_1\%e_2)$

      $$\frac{p\xrightarrow{e_1}\surd\quad q\xrightarrow{e_2}\surd}{\lambda_s(p\parallel q)\xrightarrow{\{action(s,e_1),action(s,e_2)\}}\surd}$$
      $$\frac{p\xrightarrow{e_1}\surd\quad q\xrightarrow{e_2}\surd}{\lambda_s(p)\parallel \lambda_s(q)\xrightarrow{\{action(s,e_1),action(s,e_2)\}}\surd}$$

      $$\frac{p\xrightarrow{e_1}p'\quad q\xrightarrow{e_2}\surd}{\lambda_s(p\parallel q)\xrightarrow{\{action(s,e_1),action(s,e_2)\}}\lambda_{effect(s,e_1)\cup effect(s,e_2)}(p')}$$
      $$\frac{p\xrightarrow{e_1}p'\quad q\xrightarrow{e_2}\surd}{\lambda_s(p)\parallel \lambda_s(q)\xrightarrow{\{action(s,e_1),action(s,e_2)\}}\lambda_{effect(s,e_1)\cup effect(s,e_2)}(p')}$$

      $$\frac{p\xrightarrow{e_1}\surd\quad q\xrightarrow{e_2}q'}{\lambda_s(p\parallel q)\xrightarrow{\{action(s,e_1),action(s,e_2)\}}\lambda_{effect(s,e_1)\cup effect(s,e_2)}(q')}$$
      $$\frac{p\xrightarrow{e_1}\surd\quad q\xrightarrow{e_2}q'}{\lambda_s(p)\parallel \lambda_s(q)\xrightarrow{\{action(s,e_1),action(s,e_2)\}}\lambda_{effect(s,e_1)\cup effect(s,e_2)}(q')}$$

      $$\frac{p\xrightarrow{e_1}p'\quad q\xrightarrow{e_2}q'}{\lambda_s(p\parallel q)\xrightarrow{\{action(s,e_1),action(s,e_2)\}}\lambda_{effect(s,e_1)\cup effect(s,e_2)}(p'\between q')}$$
      $$\frac{p\xrightarrow{e_1}p'\quad q\xrightarrow{e_2}q'}{\lambda_s(p)\parallel \lambda_s(q)\xrightarrow{\{action(s,e_1),action(s,e_2)\}}\lambda_{effect(s,e_1)\cup effect(s,e_2)}(p')\between\lambda_{effect(s,e_1)\cup effect(s,e_2)}(q')}$$

      So, with the assumption $\lambda_{effect(s,e_1)\cup effect(s,e_2)}(p'\between q')=\lambda_{effect(s,e_1)\cup effect(s,e_2)}(p')\between\lambda_{effect(s,e_1)\cup effect(s,e_2)}(q')$, $\lambda_s(p\parallel q)\approx_{rbs} \lambda_s(p)\parallel\lambda_s(q)$, as desired.
      For the case $e_1\%e_2$, we get

      $$\frac{p\xrightarrow{e_1}\surd\quad q\xnrightarrow{e_2}}{\lambda_s(p\parallel q)\xrightarrow{action(s,e_1)}\lambda_{effect(s,e_1)}(q)}$$
      $$\frac{p\xrightarrow{e_1}\surd\quad q\xnrightarrow{e_2}}{\lambda_s(p)\parallel \lambda_s(q)\xrightarrow{action(s,e_1)}\lambda_{effect(s,e_1)}(q)}$$

      $$\frac{p\xrightarrow{e_1}p'\quad q\xnrightarrow{e_2}}{\lambda_s(p\parallel q)\xrightarrow{action(s,e_1)}\lambda_{effect(s,e_1)}(p'\between q)}$$
      $$\frac{p\xrightarrow{e_1}p'\quad q\xnrightarrow{e_2}}{\lambda_s(p)\parallel \lambda_s(q)\xrightarrow{action(s,e_1)}\lambda_{effect(s,e_1)}(p')\between\lambda_{effect(s,e_1)}(q)}$$

      $$\frac{p\xnrightarrow{e_1}\quad q\xrightarrow{e_2}\surd}{\lambda_s(p\parallel q)\xrightarrow{action(s,e_2)}\lambda_{effect(s,e_2)}(p)}$$
      $$\frac{p\xnrightarrow{e_1}\quad q\xrightarrow{e_2}\surd}{\lambda_s(p)\parallel \lambda_s(q)\xrightarrow{action(s,e_2)}\lambda_{effect(s,e_2)}(p)}$$

      $$\frac{p\xnrightarrow{e_1}\quad q\xrightarrow{e_2}q'}{\lambda_s(p\parallel q)\xrightarrow{action(s,e_2)}\lambda_{effect(s,e_2)}(p\between q')}$$
      $$\frac{p\xnrightarrow{e_1}\quad q\xrightarrow{e_2}q'}{\lambda_s(p)\parallel \lambda_s(q)\xrightarrow{action(s,e_2)}\lambda_{effect(s,e_2)}(p)\between\lambda_{effect(s,e_2)}(q')}$$

      So, with the assumption $\lambda_{effect(s,e_1)}(p'\between q)=\lambda_{effect(s,e_1)}(p')\between\lambda_{effect(s,e_1)}(q)$ and $\lambda_{effect(s,e_2)}(p\between q')=\lambda_{effect(s,e_2)}(p)\between\lambda_{effect(s,e_2)}(q')$, $\lambda_s(p\parallel q)\approx_{rbs} \lambda_s(p)\parallel\lambda_s(q)$, as desired.
\end{itemize}

(2) Soundness of $APTC_{\tau}$ with guarded linear recursion and the state operator with respect to rooted branching pomset bisimulation $\approx_{rbp}$.

Since rooted branching pomset bisimulation $\approx_{rbp}$ is both an equivalent and a congruent relation with respect to $APTC_{\tau}$ with guarded linear recursion and the state operator, we only need to check if each axiom in Table \ref{AxiomsForState} is sound modulo rooted branching pomset bisimulation $\approx_{rbp}$.

From the definition of rooted branching pomset bisimulation $\approx_{rbp}$ (see Definition \ref{RBPSB}), we know that rooted branching pomset bisimulation $\approx_{rbp}$ is defined by weak pomset transitions, which are labeled by pomsets with $\tau$. In a weak pomset transition, the events in the pomset are either within causality relations (defined by $\cdot$) or in concurrency (implicitly defined by $\cdot$ and $+$, and explicitly defined by $\between$), of course, they are pairwise consistent (without conflicts). In (1), we have already proven the case that all events are pairwise concurrent, so, we only need to prove the case of events in causality. Without loss of generality, we take a pomset of $P=\{e_1,e_2:e_1\cdot e_2\}$. Then the weak pomset transition labeled by the above $P$ is just composed of one single event transition labeled by $e_1$ succeeded by another single event transition labeled by $e_2$, that is, $\xRightarrow{P}=\xRightarrow{e_1}\xRightarrow{e_2}$.

Similarly to the proof of soundness of $APTC_{\tau}$ with guarded linear recursion and the state operator modulo rooted branching step bisimulation $\approx_{rbs}$ (1), we can prove that each axiom in Table \ref{AxiomsForState} is sound modulo rooted branching pomset bisimulation $\approx_{rbp}$, we omit them.

(3) Soundness of $APTC_{\tau}$ with guarded linear recursion and the state operator with respect to rooted branching hp-bisimulation $\approx_{rbhp}$.

Since rooted branching hp-bisimulation $\approx_{rbhp}$ is both an equivalent and a congruent relation with respect to $APTC_{\tau}$ with guarded linear recursion and the state operator, we only need to check if each axiom in Table \ref{AxiomsForState} is sound modulo rooted branching hp-bisimulation $\approx_{rbhp}$.

From the definition of rooted branching hp-bisimulation $\approx_{rbhp}$ (see Definition \ref{RBHHPB}), we know that rooted branching hp-bisimulation $\approx_{rbhp}$ is defined on the weakly posetal product $(C_1,f,C_2),f:\hat{C_1}\rightarrow \hat{C_2}\textrm{ isomorphism}$. Two process terms $s$ related to $C_1$ and $t$ related to $C_2$, and $f:\hat{C_1}\rightarrow \hat{C_2}\textrm{ isomorphism}$. Initially, $(C_1,f,C_2)=(\emptyset,\emptyset,\emptyset)$, and $(\emptyset,\emptyset,\emptyset)\in\approx_{rbhp}$. When $s\xrightarrow{e}s'$ ($C_1\xrightarrow{e}C_1'$), there will be $t\xRightarrow{e}t'$ ($C_2\xRightarrow{e}C_2'$), and we define $f'=f[e\mapsto e]$. Then, if $(C_1,f,C_2)\in\approx_{rbhp}$, then $(C_1',f',C_2')\in\approx_{rbhp}$.

Similarly to the proof of soundness of $APTC_{\tau}$ with guarded linear recursion and the state operator modulo rooted branching pomset bisimulation equivalence (2), we can prove that each axiom in Table \ref{AxiomsForState} is sound modulo rooted branching hp-bisimulation equivalence, we just need additionally to check the above conditions on rooted branching hp-bisimulation, we omit them.
\end{proof}

\begin{theorem}[Completeness of the state operator]\label{CState}
Let $p$ and $q$ be closed $APTC_{\tau}$ with guarded linear recursion and $CFAR$ and the state operator terms, then,
\begin{enumerate}
  \item if $p\approx_{rbs} q$ then $p=q$;
  \item if $p\approx_{rbp} q$ then $p=q$;
  \item if $p\approx_{rbhp} q$ then $p=q$.
\end{enumerate}
\end{theorem}

\begin{proof}
(1) For the case of rooted branching step bisimulation, the proof is following.

Firstly, we know that each process term $p$ in $APTC_{\tau}$ with guarded linear recursion is equal to a process term $\langle X_1|E\rangle$ with $E$ a guarded linear recursive specification. And we prove if $\langle X_1|E_1\rangle\approx_{rbs}\langle Y_1|E_2\rangle$, then $\langle X_1|E_1\rangle=\langle Y_1|E_2\rangle$

Structural induction with respect to process term $p$ can be applied. The only new case (where $SO1-SO5$ are needed) is $p \equiv \lambda_{s_0}(q)$. First assuming $q=\langle X_1|E\rangle$ with a guarded linear recursive specification $E$, we prove the case of $p=\lambda_{s_0}(\langle X_1|E\rangle)$. Let $E$ consist of guarded linear recursive equations

$$X_i=(a_{1i1}\parallel\cdots\parallel a_{k_{i1}i1})X_{i1}+...+(a_{1im_i}\parallel\cdots\parallel a_{k_{im_i}im_i})X_{im_i}+b_{1i1}\parallel\cdots\parallel b_{l_{i1}i1}+...+b_{1im_i}\parallel\cdots\parallel b_{l_{im_i}im_i}$$

for $i\in{1,...,n}$. Let $F$ consist of guarded linear recursive equations

$Y_i(s)=(action(s,a_{1i1})\parallel\cdots\parallel action(s,a_{k_{i1}i1}))Y_{i1}(effect(s,a_{1i1})\cup\cdots\cup effect(s,a_{k_{i1}i1}))$

$+...+(action(s,a_{1im_i})\parallel\cdots\parallel action(s,a_{k_{im_i}im_i}))Y_{im_i}(effect(s,a_{1im_i})\cup\cdots\cup effect(s,a_{k_{im_i}im_i}))$

$+action(s,b_{1i1})\parallel\cdots\parallel action(s,b_{l_{i1}i1})+...+action(s,b_{1im_i})\parallel\cdots\parallel action(s,b_{l_{im_i}im_i})$

for $i\in{1,...,n}$.

\begin{eqnarray}
&&\lambda_s(\langle X_i|E\rangle) \nonumber\\
&\overset{\text{RDP}}{=}&\lambda_s((a_{1i1}\parallel\cdots\parallel a_{k_{i1}i1})X_{i1}+...+(a_{1im_i}\parallel\cdots\parallel a_{k_{im_i}im_i})X_{im_i}\nonumber\\
&&+b_{1i1}\parallel\cdots\parallel b_{l_{i1}i1}+...+b_{1im_i}\parallel\cdots\parallel b_{l_{im_i}im_i}) \nonumber\\
&\overset{\text{SO1-SO5}}{=}&(action(s,a_{1i1})\parallel\cdots\parallel action(s,a_{k_{i1}i1}))\lambda_{effect(s,a_{1i1})\cup\cdots\cup effect(s,a_{k_{i1}i1}))}(X_{i1})\nonumber\\
&&+...+(action(s,a_{1im_i})\parallel\cdots\parallel action(s,a_{k_{im_i}im_i}))\lambda_{effect(s,a_{1im_i})\cup\cdots\cup effect(s,a_{k_{im_i}im_i}))}(X_{im_i})\nonumber\\
&&+action(s,b_{1i1})\parallel\cdots\parallel action(s,b_{l_{i1}i1})+...+action(s,b_{1im_i})\parallel\cdots\parallel action(s,b_{l_{im_i}im_i}) \nonumber
\end{eqnarray}

Replacing $Y_i(s)$ by $\lambda_s(\langle X_i|E\rangle)$ for $i\in\{1,...,n\}$ is a solution for $F$. So by $RSP$, $\lambda_{s_0}(\langle X_1|E\rangle)=\langle Y_1(s_0)|F\rangle$, as desired.

(2) For the case of rooted branching pomset bisimulation, it can be proven similarly to (1), we omit it.

(3) For the case of rooted branching hp-bisimulation, it can be proven similarly to (1), we omit it.
\end{proof}

\subsection{Asynchronous Communication}

The communication in APTC is synchronous, for two atomic actions $a,b\in A$, if there exists a communication between $a$ and $b$, then they merge into a new communication action
$\gamma(a,b)$; otherwise let $\gamma(a,b)=\delta$.

Asynchronous communication between actions $a,b\in A$ does not exist a merge $\gamma(a,b)$, and it is only explicitly defined by the causality relation $a\leq b$ to ensure that the send
action $a$ to be executed before the receive action $b$.

APTC naturally support asynchronous communication to be adapted to the following aspects:

\begin{enumerate}
  \item remove the communication merge operator $\mid$, just because there does not exist a communication merger $\gamma(a,b)$ between two asynchronous communicating action $a,b\in A$;
  \item remove the asynchronous communicating actions $a,b\in A$ from $H$ of the encapsulation operator $\partial_H$;
  \item ensure the send action $a$ to be executed before the receive action $b$, by inserting appropriate numbers of placeholders during modeling time; or by adding a causality constraint
  between the communicating actions $a\leq b$, all process terms violate this constraint will cause deadlocks.
\end{enumerate}

\subsection{Applications}\label{app}

$APTC$ provides a formal framework based on truly concurrent behavioral semantics, which can be used to verify the correctness of system behaviors. In this subsection,
we tend to choose alternating bit protocol (ABP) \cite{ABP}.

The ABP protocol is used to ensure successful transmission of data through a corrupted channel. This success is based on the assumption that data can be resent an unlimited number of times, which is illustrated in Figure \ref{ABP}, we alter it into the true concurrency situation.

\begin{enumerate}
  \item Data elements $d_1,d_2,d_3,\cdots$ from a finite set $\Delta$ are communicated between a Sender and a Receiver.
  \item If the Sender reads a datum from channel $A_1$, then this datum is sent to the Receiver in parallel through channel $A_2$.
  \item The Sender processes the data in $\Delta$, formes new data, and sends them to the Receiver through channel $B$.
  \item And the Receiver sends the datum into channel $C_2$.
  \item If channel $B$ is corrupted, the message communicated through $B$ can be turn into an error message $\bot$.
  \item Every time the Receiver receives a message via channel $B$, it sends an acknowledgement to the Sender via channel $D$, which is also corrupted.
  \item Finally, then Sender and the Receiver send out their outputs in parallel through channels $C_1$ and $C_2$.
\end{enumerate}

\begin{figure}
    \centering
    \includegraphics{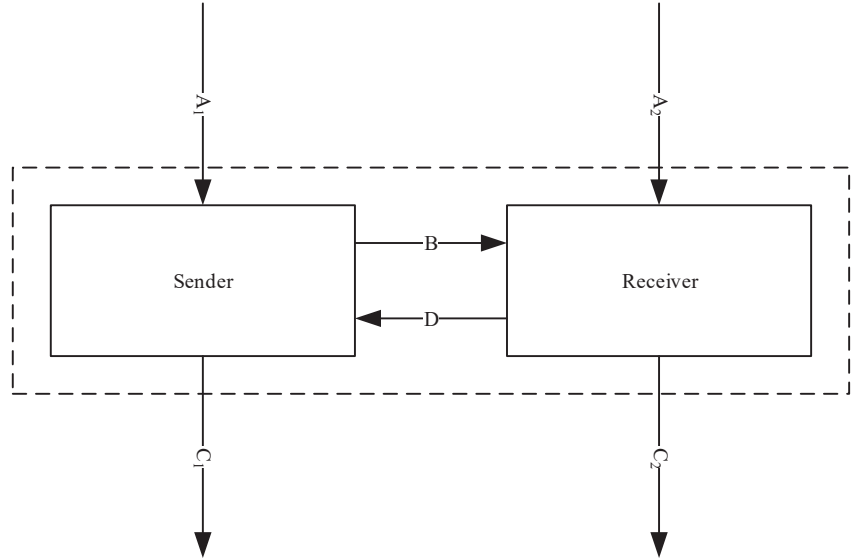}
    \caption{Alternating bit protocol}
    \label{ABP}
\end{figure}

In the truly concurrent ABP, the Sender sends its data to the Receiver; and the Receiver can also send its data to the Sender, for simplicity and without loss of generality, we assume that only the Sender sends its data and the Receiver only receives the data from the Sender. The Sender attaches a bit 0 to data elements $d_{2k-1}$ and a bit 1 to data elements $d_{2k}$, when they are sent into channel $B$. When the Receiver reads a datum, it sends back the attached bit via channel $D$. If the Receiver receives a corrupted message, then it sends back the previous acknowledgement to the Sender.

Then the state transition of the Sender can be described by $APTC$ as follows.

\begin{eqnarray}
&&S_b=\sum_{d\in\Delta}r_{A_1}(d)\cdot T_{db}\nonumber\\
&&T_{db}=(\sum_{d'\in\Delta}(s_B(d',b)\cdot s_{C_1}(d'))+s_B(\bot))\cdot U_{db}\nonumber\\
&&U_{db}=r_D(b)\cdot S_{1-b}+(r_D(1-b)+r_D(\bot))\cdot T_{db}\nonumber
\end{eqnarray}

where $s_B$ denotes sending data through channel $B$, $r_D$ denotes receiving data through channel $D$, similarly, $r_{A_1}$ means receiving data via channel $A_1$, $s_{C_1}$ denotes sending data via channel $C_1$, and $b\in\{0,1\}$.

And the state transition of the Receiver can be described by $APTC$ as follows.

\begin{eqnarray}
&&R_b=\sum_{d\in\Delta}r_{A_2}(d)\cdot R_b'\nonumber\\
&&R_b'=\sum_{d'\in\Delta}\{r_B(d',b)\cdot s_{C_2}(d')\cdot Q_b+r_B(d',1-b)\cdot Q_{1-b}\}+r_B(\bot)\cdot Q_{1-b}\nonumber\\
&&Q_b=(s_D(b)+s_D(\bot))\cdot R_{1-b}\nonumber
\end{eqnarray}

where $r_{A_2}$ denotes receiving data via channel $A_2$, $r_B$ denotes receiving data via channel $B$, $s_{C_2}$ denotes sending data via channel $C_2$, $s_D$ denotes sending data via channel $D$, and $b\in\{0,1\}$.

The send action and receive action of the same data through the same channel can communicate each other, otherwise, a deadlock $\delta$ will be caused. We define the following communication functions.

\begin{eqnarray}
&&\gamma(s_B(d',b),r_B(d',b))\triangleq c_B(d',b)\nonumber\\
&&\gamma(s_B(\bot),r_B(\bot))\triangleq c_B(\bot)\nonumber\\
&&\gamma(s_D(b),r_D(b))\triangleq c_D(b)\nonumber\\
&&\gamma(s_D(\bot),r_D(\bot))\triangleq c_D(\bot)\nonumber
\end{eqnarray}

Let $R_0$ and $S_0$ be in parallel, then the system $R_0S_0$ can be represented by the following process term.

$$\tau_I(\partial_H(\Theta(R_0\between S_0)))=\tau_I(\partial_H(R_0\between S_0))$$

where $H=\{s_B(d',b),r_B(d',b),s_D(b),r_D(b)|d'\in\Delta,b\in\{0,1\}\}\\
\{s_B(\bot),r_B(\bot),s_D(\bot),r_D(\bot)\}$

$I=\{c_B(d',b),c_D(b)|d'\in\Delta,b\in\{0,1\}\}\cup\{c_B(\bot),c_D(\bot)\}$.

Then we get the following conclusion.

\begin{theorem}[Correctness of the ABP protocol]
The ABP protocol $\tau_I(\partial_H(R_0\between S_0))$ exhibits desired external behaviors.
\end{theorem}

\begin{proof}

By use of the algebraic laws of $APTC$, we have the following expansions.

\begin{eqnarray}
R_0\between S_0&\overset{\text{P1}}{=}&R_0\parallel S_0+R_0\mid S_0\nonumber\\
&\overset{\text{RDP}}{=}&(\sum_{d\in\Delta}r_{A_2}(d)\cdot R_0')\parallel(\sum_{d\in\Delta}r_{A_1}(d)T_{d0})\nonumber\\
&&+(\sum_{d\in\Delta}r_{A_2}(d)\cdot R_0')\mid(\sum_{d\in\Delta}r_{A_1}(d)T_{d0})\nonumber\\
&\overset{\text{P6,C14}}{=}&\sum_{d\in\Delta}(r_{A_2}(d)\parallel r_{A_1}(d))R_0'\between T_{d0} + \delta\cdot R_0'\between T_{d0}\nonumber\\
&\overset{\text{A6,A7}}{=}&\sum_{d\in\Delta}(r_{A_2}(d)\parallel r_{A_1}(d))R_0'\between T_{d0}\nonumber
\end{eqnarray}

\begin{eqnarray}
\partial_H(R_0\between S_0)&=&\partial_H(\sum_{d\in\Delta}(r_{A_2}(d)\parallel r_{A_1}(d))R_0'\between T_{d0})\nonumber\\
&&=\sum_{d\in\Delta}(r_{A_2}(d)\parallel r_{A_1}(d))\partial_H(R_0'\between T_{d0})\nonumber
\end{eqnarray}

Similarly, we can get the following equations.

\begin{eqnarray}
\partial_H(R_0\between S_0)&=&\sum_{d\in\Delta}(r_{A_2}(d)\parallel r_{A_1}(d))\cdot\partial_H(T_{d0}\between R_0')\nonumber\\
\partial_H(T_{d0}\between R_0')&=&c_B(d',0)\cdot(s_{C_1}(d')\parallel s_{C_2}(d'))\cdot\partial_H(U_{d0}\between Q_0)+c_B(\bot)\cdot\partial_H(U_{d0}\between Q_1)\nonumber\\
\partial_H(U_{d0}\between Q_1)&=&(c_D(1)+c_D(\bot))\cdot\partial_H(T_{d0}\between R_0')\nonumber\\
\partial_H(Q_0\between U_{d0})&=&c_D(0)\cdot\partial_H(R_1\between S_1)+c_D(\bot)\cdot\partial_H(R_1'\between T_{d0})\nonumber\\
\partial_H(R_1'\between T_{d0})&=&(c_B(d',0)+c_B(\bot))\cdot\partial_H(Q_0\between U_{d0})\nonumber\\
\partial_H(R_1\between S_1)&=&\sum_{d\in\Delta}(r_{A_2}(d)\parallel r_{A_1}(d))\cdot\partial_H(T_{d1}\between R_1')\nonumber\\
\partial_H(T_{d1}\between R_1')&=&c_B(d',1)\cdot(s_{C_1}(d')\parallel s_{C_2}(d'))\cdot\partial_H(U_{d1}\between Q_1)+c_B(\bot)\cdot\partial_H(U_{d1}\between Q_0')\nonumber\\
\partial_H(U_{d1}\between Q_0')&=&(c_D(0)+c_D(\bot))\cdot\partial_H(T_{d1}\between R_1')\nonumber\\
\partial_H(Q_1\between U_{d1})&=&c_D(1)\cdot\partial_H(R_0\between S_0)+c_D(\bot)\cdot\partial_H(R_0'\between T_{d1})\nonumber\\
\partial_H(R_0'\between T_{d1})&=&(c_B(d',1)+c_B(\bot))\cdot\partial_H(Q_1\between U_{d1})\nonumber
\end{eqnarray}

Let $\partial_H(R_0\between S_0)=\langle X_1|E\rangle$, where E is the following guarded linear recursion specification:

\begin{eqnarray}
&&\{X_1=\sum_{d\in \Delta}(r_{A_2}(d)\parallel r_{A_1}(d))\cdot X_{2d},Y_1=\sum_{d\in\Delta}(r_{A_2}(d)\parallel r_{A_1}(d))\cdot Y_{2d},\nonumber\\
&&X_{2d}=c_B(d',0)\cdot X_{4d}+c_B(\bot)\cdot X_{3d}, Y_{2d}=c_B(d',1)\cdot Y_{4d}+c_B(\bot)\cdot Y_{3d},\nonumber\\
&&X_{3d}=(c_D(1)+c_D(\bot))\cdot X_{2d}, Y_{3d}=(c_D(0)+c_D(\bot))\cdot Y_{2d},\nonumber\\
&&X_{4d}=(s_{C_1}(d')\parallel s_{C_2}(d'))\cdot X_{5d}, Y_{4d}=(s_{C_1}(d')\parallel s_{C_2}(d'))\cdot Y_{5d},\nonumber\\
&&X_{5d}=c_D(0)\cdot Y_1+c_D(\bot)\cdot X_{6d}, Y_{5d}=c_D(1)\cdot X_1+c_D(\bot)\cdot Y_{6d},\nonumber\\
&&X_{6d}=(c_B(d,0)+c_B(\bot))\cdot X_{5d}, Y_{6d}=(c_B(d,1)+c_B(\bot))\cdot Y_{5d}\nonumber\\
&&|d,d'\in\Delta\}\nonumber
\end{eqnarray}

Then we apply abstraction operator $\tau_I$ into $\langle X_1|E\rangle$.

\begin{eqnarray}
\tau_I(\langle X_1|E\rangle)
&=&\sum_{d\in\Delta}(r_{A_1}(d)\parallel r_{A_2}(d))\cdot\tau_I(\langle X_{2d}|E\rangle)\nonumber\\
&=&\sum_{d\in\Delta}(r_{A_1}(d)\parallel r_{A_2}(d))\cdot\tau_I(\langle X_{4d}|E\rangle)\nonumber\\
&=&\sum_{d,d'\in\Delta}(r_{A_1}(d)\parallel r_{A_2}(d))\cdot (s_{C_1}(d')\parallel s_{C_2}(d'))\cdot\tau_I(\langle X_{5d}|E\rangle)\nonumber\\
&=&\sum_{d,d'\in\Delta}(r_{A_1}(d)\parallel r_{A_2}(d))\cdot (s_{C_1}(d')\parallel s_{C_2}(d'))\cdot\tau_I(\langle Y_1|E\rangle)\nonumber
\end{eqnarray}

Similarly, we can get $\tau_I(\langle Y_1|E\rangle)=\sum_{d,d'\in\Delta}(r_{A_1}(d)\parallel r_{A_2}(d))\cdot (s_{C_1}(d')\parallel s_{C_2}(d'))\cdot\tau_I(\langle X_1|E\rangle)
$.

We get $\tau_I(\partial_H(R_0\between S_0))=\sum_{d,d'\in \Delta}(r_{A_1}(d)\parallel r_{A_2}(d))\cdot (s_{C_1}(d')\parallel s_{C_2}(d'))\cdot \tau_I(\partial_H(R_0\between S_0))$. So, the ABP protocol $\tau_I(\partial_H(R_0\between S_0))$ exhibits desired external behaviors.
\end{proof}

With the help of shadow constant, now we can verify the traditional alternating bit protocol (ABP) \cite{ABP}.

The ABP protocol is used to ensure successful transmission of data through a corrupted channel. This success is based on the assumption that data can be resent an unlimited number of times, which is illustrated in Figure \ref{ABP2}, we alter it into the true concurrency situation.

\begin{enumerate}
  \item Data elements $d_1,d_2,d_3,\cdots$ from a finite set $\Delta$ are communicated between a Sender and a Receiver.
  \item If the Sender reads a datum from channel $A$.
  \item The Sender processes the data in $\Delta$, formes new data, and sends them to the Receiver through channel $B$.
  \item And the Receiver sends the datum into channel $C$.
  \item If channel $B$ is corrupted, the message communicated through $B$ can be turn into an error message $\bot$.
  \item Every time the Receiver receives a message via channel $B$, it sends an acknowledgement to the Sender via channel $D$, which is also corrupted.
\end{enumerate}

\begin{figure}
    \centering
    \includegraphics{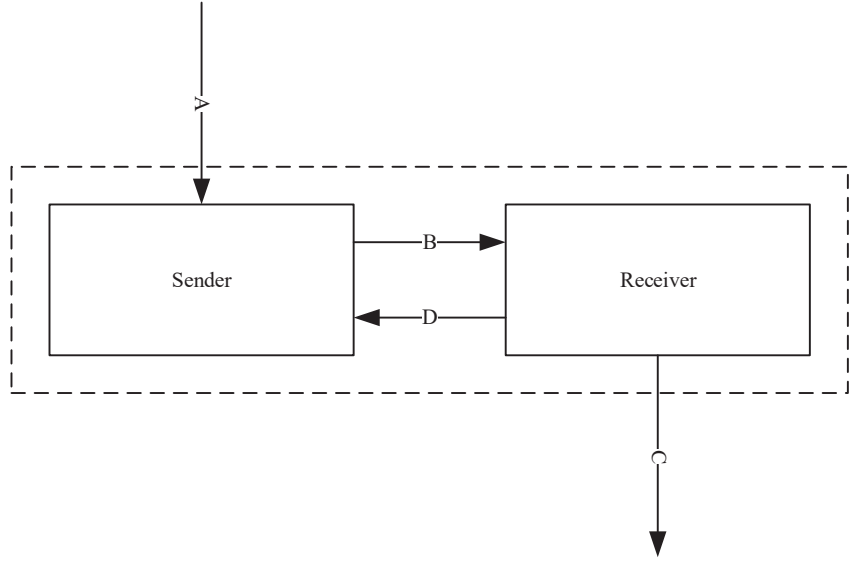}
    \caption{Alternating bit protocol}
    \label{ABP2}
\end{figure}

The Sender attaches a bit 0 to data elements $d_{2k-1}$ and a bit 1 to data elements $d_{2k}$, when they are sent into channel $B$. When the Receiver reads a datum, it sends back the attached bit via channel $D$. If the Receiver receives a corrupted message, then it sends back the previous acknowledgement to the Sender.

Then the state transition of the Sender can be described by $APTC$ as follows.

\begin{eqnarray}
&&S_b=\sum_{d\in\Delta}r_{A}(d)\cdot T_{db}\nonumber\\
&&T_{db}=(\sum_{d'\in\Delta}(s_B(d',b)\cdot \circledS^{s_{C}(d')})+s_B(\bot))\cdot U_{db}\nonumber\\
&&U_{db}=r_D(b)\cdot S_{1-b}+(r_D(1-b)+r_D(\bot))\cdot T_{db}\nonumber
\end{eqnarray}

where $s_B$ denotes sending data through channel $B$, $r_D$ denotes receiving data through channel $D$, similarly, $r_{A}$ means receiving data via channel $A$, $\circledS^{s_{C}(d')}$ denotes the shadow of $s_{C}(d')$.

And the state transition of the Receiver can be described by $APTC$ as follows.

\begin{eqnarray}
&&R_b=\sum_{d\in\Delta}\circledS^{r_{A}(d)}\cdot R_b'\nonumber\\
&&R_b'=\sum_{d'\in\Delta}\{r_B(d',b)\cdot s_{C}(d')\cdot Q_b+r_B(d',1-b)\cdot Q_{1-b}\}+r_B(\bot)\cdot Q_{1-b}\nonumber\\
&&Q_b=(s_D(b)+s_D(\bot))\cdot R_{1-b}\nonumber
\end{eqnarray}

where $\circledS^{r_{A}(d)}$ denotes the shadow of $r_{A}(d)$, $r_B$ denotes receiving data via channel $B$, $s_{C}$ denotes sending data via channel $C$, $s_D$ denotes sending data via channel $D$, and $b\in\{0,1\}$.

The send action and receive action of the same data through the same channel can communicate each other, otherwise, a deadlock $\delta$ will be caused. We define the following communication functions.

\begin{eqnarray}
&&\gamma(s_B(d',b),r_B(d',b))\triangleq c_B(d',b)\nonumber\\
&&\gamma(s_B(\bot),r_B(\bot))\triangleq c_B(\bot)\nonumber\\
&&\gamma(s_D(b),r_D(b))\triangleq c_D(b)\nonumber\\
&&\gamma(s_D(\bot),r_D(\bot))\triangleq c_D(\bot)\nonumber
\end{eqnarray}

Let $R_0$ and $S_0$ be in parallel, then the system $R_0S_0$ can be represented by the following process term.

$$\tau_I(\partial_H(\Theta(R_0\between S_0)))=\tau_I(\partial_H(R_0\between S_0))$$

where $H=\{s_B(d',b),r_B(d',b),s_D(b),r_D(b)|d'\in\Delta,b\in\{0,1\}\}\\
\{s_B(\bot),r_B(\bot),s_D(\bot),r_D(\bot)\}$

$I=\{c_B(d',b),c_D(b)|d'\in\Delta,b\in\{0,1\}\}\cup\{c_B(\bot),c_D(\bot)\}$.

Then we get the following conclusion.

\begin{theorem}[Correctness of the ABP protocol]
The ABP protocol $\tau_I(\partial_H(R_0\between S_0))$ can exhibit desired external behaviors.
\end{theorem}

\begin{proof}

Similarly, we can get $\tau_I(\langle X_1|E\rangle)=\sum_{d,d'\in\Delta}r_{A}(d)\cdot s_{C}(d')\cdot\tau_I(\langle Y_1|E\rangle)$ and $\tau_I(\langle Y_1|E\rangle)=\sum_{d,d'\in\Delta}r_{A}(d)\cdot s_{C}(d')\cdot\tau_I(\langle X_1|E\rangle)$.

So, the ABP protocol $\tau_I(\partial_H(R_0\between S_0))$ can exhibit desired external behaviors.
\end{proof}

\newpage\section{Verification of Architectural Patterns}

Architecture patterns are highest-level patterns which present structural organizations for software systems and contain a set of subsystems and the relationships among them.

In this chapter, we verify four categories of architectural patterns, in subsection \ref{MUD3}, we verify structural patterns including the Layers pattern, the Pipes and 
Filters pattern and the Blackboard pattern. In section \ref{DS3}, we verify patterns considering distribution aspects. We verify patterns that feature human-computer interaction 
in section \ref{IS3}. In section \ref{AS3}, we verify patterns supporting extension of applications.

\subsection{From Mud to Structure}\label{MUD3}

In this subsection, we verify structural patterns including the Layers pattern, the Pipes and Filters pattern and the Blackboard pattern.

\subsubsection{Verification of the Layers Pattern}\label{Layers3}

The Layers pattern contains several layers with each layer being a particular level of abstraction of subtasks. In the Layers pattern, there are only communications between the
adjacent layers. That is, for layer $i$, it receives data (the data denoted $d_{U_i}$) from layer $i+1$, processes the data (the processing function denoted $UF_i$) and sends the
processed data (the processed data denoted $UF_i(d_{U_i})$) to layer $i-1$; in the other direction, it receives data (the data denoted $d_{L_i}$) from layer $i-1$, processes the data
(the processing function denoted $LF_i$) and sends the processed data (the processed data denoted $LF_i(d_{L_i})$)to layer $i+1$, as Figure \ref{L3} illustrated. The four channels are
denoted $UI_i$ (the Upper Input of layer $i$), $LO_i$ (the Lower Output of layer $i$), $LI_i$ (the Lower Input of layer $i$) and $UO_i$ (the Upper Output of layer $i$) respectively.

\begin{figure}
    \centering
    \includegraphics{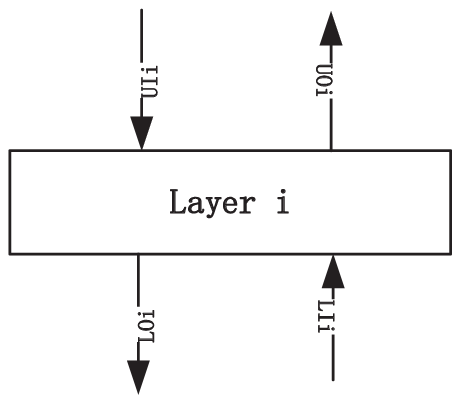}
    \caption{Layer i}
    \label{L3}
\end{figure}

The whole Layers pattern containing $n$ layers are illustrated in Figure \ref{LS3}. Note that, the numbering of layers are in a reverse order, that is, the highest layer is called
layer $n$ and the lowest layer is called layer $1$.

\begin{figure}
    \centering
    \includegraphics{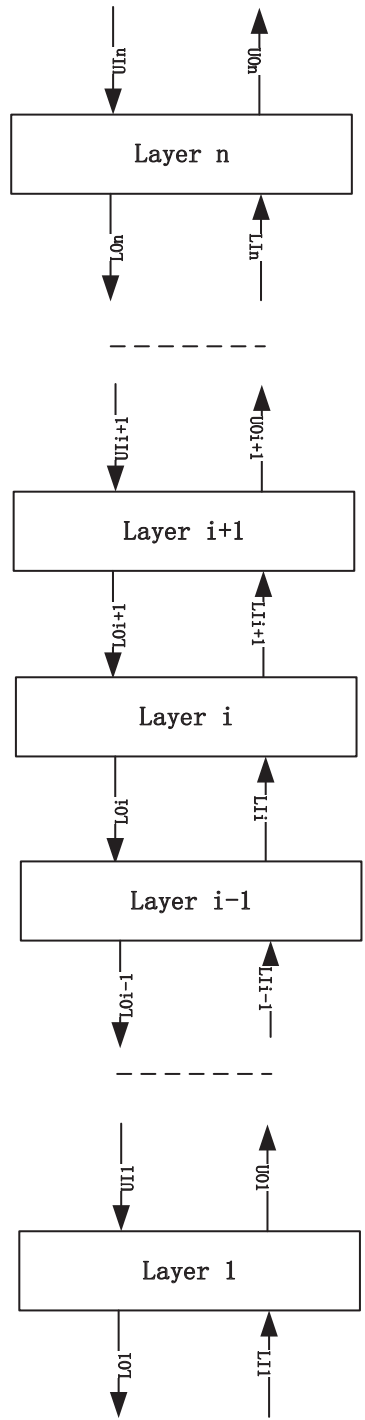}
    \caption{Layers pattern}
    \label{LS3}
\end{figure}

There exist two typical processes in the Layers pattern corresponding to two directions of data processing as Figure \ref{LS3P} illustrated. One process is as follows.

\begin{enumerate}
  \item The highest layer $n$ receives data from the application which is denoted $d_{U_n}$ through channel $UI_n$ (the corresponding reading action is denoted $r_{UI_n}(d_{U_n})$),
  then processes the data, and sends the processed data to layer $n-1$ which is denoted $UF_n(d_{U_n})$ through channel $LO_n$ (the corresponding sending action is denoted
  $s_{LO_n}(UF_n(d_{U_n}))$);
  \item The layer $i$ receives data from the layer $i+1$ which is denoted $d_{U_i}$ through channel $UI_i$ (the corresponding reading action is denoted $r_{UI_i}(d_{U_i})$),
  then processes the data, and sends the processed data to layer $i-1$ which is denoted $UF_i(d_{U_i})$ through channel $LO_i$ (the corresponding sending action is denoted
  $s_{LO_i}(UF_i(d_{U_i}))$);
  \item The lowest layer $1$ receives data from the layer $2$ which is denoted $d_{U_1}$ through channel $UI_1$ (the corresponding reading action is denoted $r_{UI_1}(d_{U_1})$),
  then processes the data, and sends the processed data to another layers peer which is denoted $UF_1(d_{U_1})$ through channel $LO_1$ (the corresponding sending action is denoted
  $s_{LO_1}(UF_1(d_{U_1}))$).
\end{enumerate}

The other process is following.

\begin{enumerate}
  \item The lowest layer $1$ receives data from the another layers peer which is denoted $d_{L_1}$ through channel $LI_1$ (the corresponding reading action is denoted $r_{LI_1}(d_{L_1})$),
  then processes the data, and sends the processed data to layer $2$ which is denoted $LF_1(d_{L_1})$ through channel $UO_1$ (the corresponding sending action is denoted
  $s_{UO_1}(LF_1(d_{L_1}))$);
  \item The layer $i$ receives data from the layer $i-1$ which is denoted $d_{L_i}$ through channel $LI_i$ (the corresponding reading action is denoted $r_{LI_i}(d_{L_i})$),
  then processes the data, and sends the processed data to layer $i+1$ which is denoted $LF_i(d_{L_i})$ through channel $UO_i$ (the corresponding sending action is denoted
  $s_{UO_i}(LF_i(d_{L_i}))$);
  \item The highest layer $n$ receives data from layer $n-1$ which is denoted $d_{L_n}$ through channel $LI_n$ (the corresponding reading action is denoted $r_{LI_n}(d_{L_n})$),
  then processes the data, and sends the processed data to the application which is denoted $LF_n(d_{L_n})$ through channel $UO_n$ (the corresponding sending action is denoted
  $s_{UO_n}(LF_n(d_{L_n}))$).
\end{enumerate}

\begin{figure}
    \centering
    \includegraphics{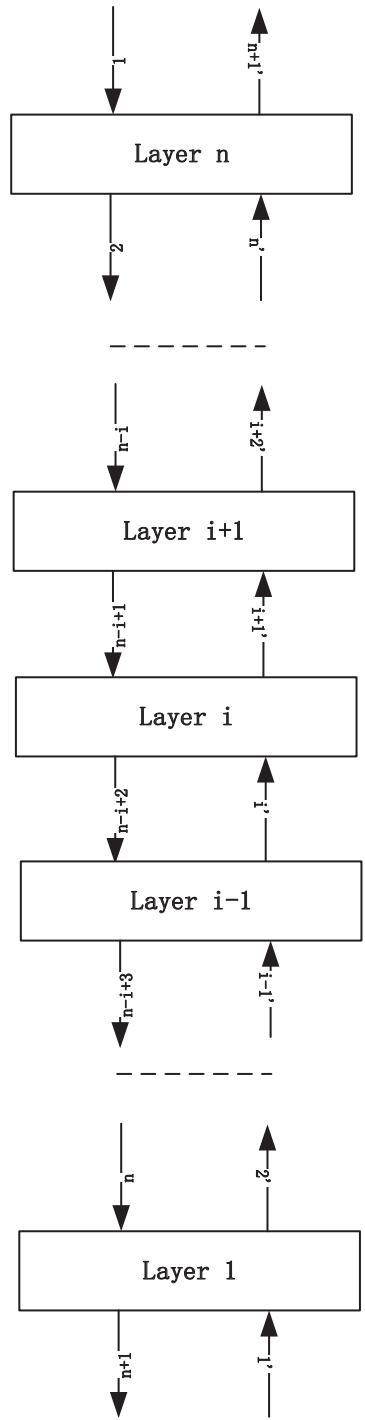}
    \caption{Typical process of Layers pattern}
    \label{LS3P}
\end{figure}

We begin to verify the Layers pattern. We assume all data elements $d_{U_i}$ and $d_{L_i}$ (for $1\leq i\leq n$) are from a finite set $\Delta$. The state transitions of layer $i$ (for $1\leq i\leq n$)
described by APTC are as follows.

$L_i=\sum_{d_{U_i},d_{L_i}\in\Delta}(r_{UI_i}(d_{U_i})\cdot L_{i_2}\between r_{LI_i}(d_{L_i})\cdot L_{i_3})$

$L_{i_2}=UF_i\cdot L_{i_4}$

$L_{i_3}=LF_i\cdot L_{i_5}$

$L_{i_4}=\sum_{d_{U_i}\in\Delta}(s_{LO_i}(UF_i(d_{U_i}))\cdot L_i)$

$L_{i_5}=\sum_{d_{L_i}\in\Delta}(s_{UO_i}(LF_i(d_{L_i}))\cdot L_i)$

The sending action and the reading action of the same data through the same channel can communicate with each other, otherwise, will cause a deadlock $\delta$. We define the following
communication functions for $1\leq i\leq n$. Note that, the channel of $LO_{i+1}$ of layer $i+1$ and the channel $UI_i$ of layer $i$ are the same one channel, and the channel $LI_{i+1}$ of layer $i+1$ and the
channel $UO_i$ of layer $i$ are the same one channel. And also the data $d_{L_{i+1}}$ of layer $i+1$ and the data $LF_i(d_{L_i})$ of layer $i$ are the same data, and the data
$UF_{i+1}(d_{U_{i+1}})$ of layer $i+1$ and the data $d_{U_i}$ of layer $i$ are the same data.

$$\gamma(r_{UI_i}(d_{U_i}),s_{LO_{i+1}}(UF_{i+1}(d_{U_{i+1}})))\triangleq c_{UI_i}(d_{U_i})$$
$$\gamma(r_{LI_i}(d_{L_i}),s_{UO_{i-1}}(LF_{i-1}(d_{L_{i-1}})))\triangleq c_{LI_i}(d_{L_i})$$
$$\gamma(r_{UI_{i-1}}(d_{U_{i-1}}),s_{LO_{i}}(UF_{i}(d_{U_{i}})))\triangleq c_{UI_{i-1}}(d_{U_{i-1}})$$
$$\gamma(r_{LI_{i+1}}(d_{L_{i+1}}),s_{UO_{i}}(LF_{i}(d_{L_{i}})))\triangleq c_{LI_{i+1}}(d_{L_{i+1}})$$

Note that, for the layer $n$, there are only two communication functions as follows.

$$\gamma(r_{LI_n}(d_{L_n}),s_{UO_{n-1}}(LF_{n-1}(d_{L_{n-1}})))\triangleq c_{LI_n}(d_{L_n})$$
$$\gamma(r_{UI_{n-1}}(d_{U_{n-1}}),s_{LO_{n}}(UF_{n}(d_{U_{n}})))\triangleq c_{UI_{n-1}}(d_{U_{n-1}})$$

And for the layer $1$, there are also only two communication functions as follows.

$$\gamma(r_{UI_1}(d_{U_1}),s_{LO_{2}}(UF_{2}(d_{U_{2}})))\triangleq c_{UI_1}(d_{U_1})$$
$$\gamma(r_{LI_{2}}(d_{L_{2}}),s_{UO_{1}}(LF_{1}(d_{L_{1}})))\triangleq c_{LI_{2}}(d_{L_{2}})$$

Let all layers from layer $n$ to layer $1$ be in parallel, then the Layers pattern $L_n\cdots L_i\cdots L_1$ can be presented by the following process term.

$$\tau_I(\partial_H(\Theta(L_n\between\cdots\between L_i\between\cdots\between L_1)))=\tau_I(\partial_H(L_n\between\cdots\between L_i\between\cdots\between L_1))$$

where $H=\{r_{UI_1}(d_{U_1}),s_{UO_{1}}(LF_{1}(d_{L_{1}})),\cdots, r_{UI_i}(d_{U_i}), r_{LI_i}(d_{L_i}), s_{LO_{i}}(UF_{i}(d_{U_{i}})),s_{UO_{i}}(LF_{i}(d_{L_{i}})), \\
\cdots,r_{LI_n}(d_{L_n}),s_{LO_{n}}(UF_{n}(d_{U_{n}}))|d_{U_1},d_{L_1},\cdots,d_{U_i},d_{L_i}\cdots,d_{U_n},d_{L_n}\in\Delta\}$,

$I=\{c_{UI_1}(d_{U_1}),c_{LI_{2}}(d_{L_{2}}),\cdots,c_{UI_i}(d_{U_i}),c_{LI_i}(d_{L_i}),c_{UI_{i-1}}(d_{U_{i-1}}),c_{LI_{i+1}}(d_{L_{i+1}}),
\cdots,c_{LI_n}(d_{L_n}),c_{UI_{n-1}}(d_{U_{n-1}}),\\LF_1,UF_1,\cdots,LF_i,UF_i,\cdots,LF_n,UF_n|d_{U_1},d_{L_1},\cdots,d_{U_i},d_{L_i}\cdots,d_{U_n},d_{L_n}\in\Delta\}$.

Then we get the following conclusion on the Layers pattern.

\begin{theorem}[Correctness of the Layers pattern]
The Layers pattern $\tau_I(\partial_H(L_n\between\cdots\between L_i\between\cdots\between L_1))$ can exhibit desired external behaviors.
\end{theorem}

\begin{proof}
Based on the above state transitions of layer $i$ (for $1\leq i\leq n$), by use of the algebraic laws of APTC, we can prove that

$\tau_I(\partial_H(L_n\between\cdots\between L_i\between\cdots\between L_1))=\sum_{d_{U_1},d_{L_1},d_{U_n},d_{L_n}\in\Delta}((r_{UI_n}(d_{U_n})\parallel r_{LI_1}(d_{L_1}))
\cdot(s_{UO_n}(LF_n(d_{L_n}))\parallel s_{LO_{1}}(UF_{1}(d_{U_{1}}))))\cdot \tau_I(\partial_H(L_n\between\cdots\between L_i\between\cdots\between L_1))$,

that is, the Layers pattern $\tau_I(\partial_H(L_n\between\cdots\between L_i\between\cdots\between L_1))$ can exhibit desired external behaviors.

For the details of proof, please refer to section \ref{app}, and we omit it.
\end{proof}

Two Layers pattern peers can be composed together, just by linking the lower output of layer $1$ of one peer together with the lower input of layer $1$ of the other peer, and vice versa.
As Figure \ref{LS23} illustrated.

\begin{figure}
    \centering
    \includegraphics{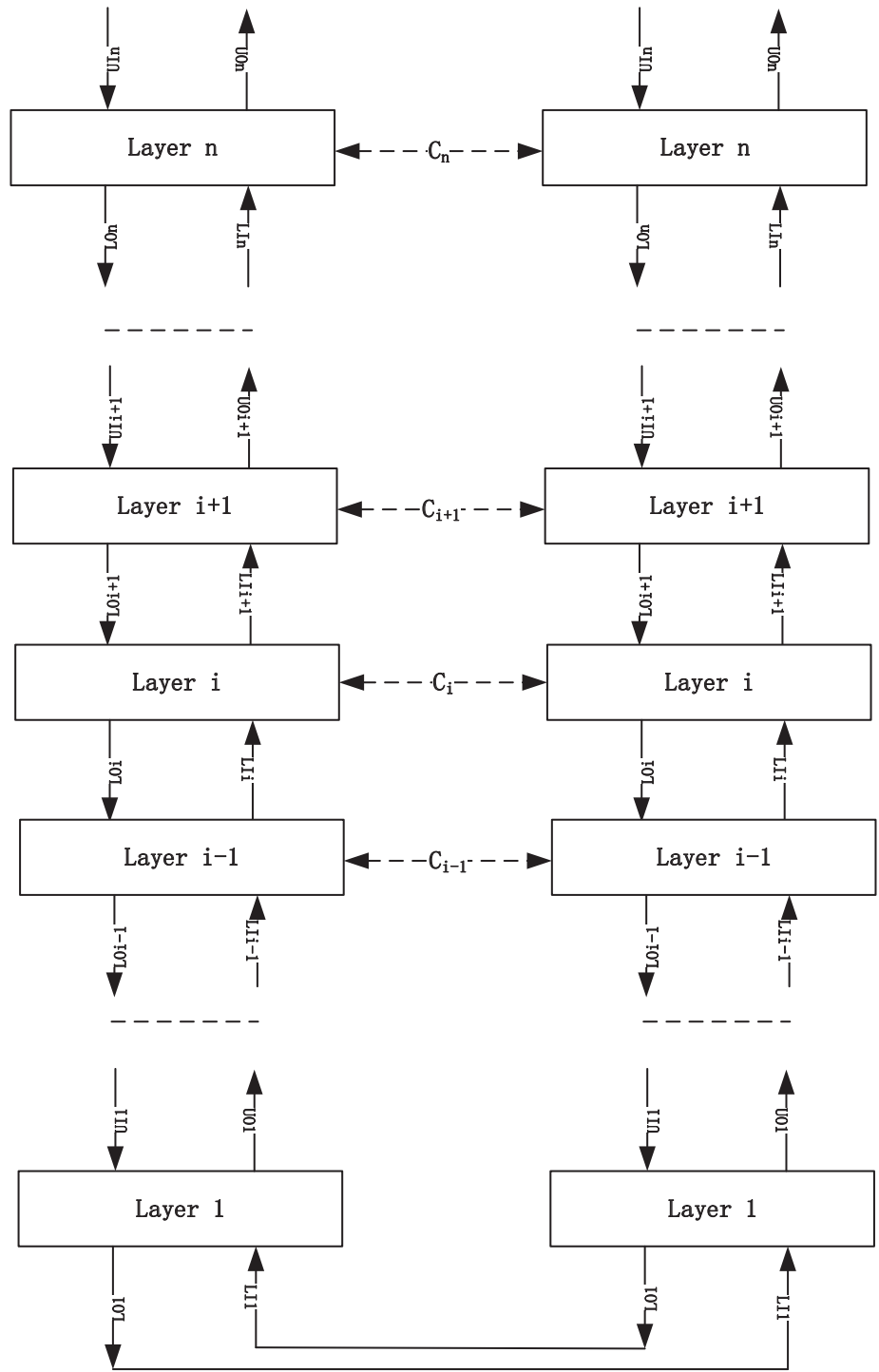}
    \caption{Two layers peers}
    \label{LS23}
\end{figure}

There are also two typical data processing process in the composition of two layers peers, as Figure \ref{LS23P} shows. One process is data transferred from peer $P$ to another peer $P'$ as follows.

\begin{enumerate}
  \item The highest layer $n$ of peer $P$ receives data from the application of peer $P$ which is denoted $d_{U_n}$ through channel $UI_n$ (the corresponding reading action is denoted $r_{UI_n}(d_{U_n})$),
  then processes the data, and sends the processed data to layer $n-1$ of peer $P$ which is denoted $UF_n(d_{U_n})$ through channel $LO_n$ (the corresponding sending action is denoted
  $s_{LO_n}(UF_n(d_{U_n}))$);
  \item The layer $i$ of peer $P$ receives data from the layer $i+1$ of peer $P$ which is denoted $d_{U_i}$ through channel $UI_i$ (the corresponding reading action is denoted $r_{UI_i}(d_{U_i})$),
  then processes the data, and sends the processed data to layer $i-1$ which is denoted $UF_i(d_{U_i})$ through channel $LO_i$ (the corresponding sending action is denoted
  $s_{LO_i}(UF_i(d_{U_i}))$);
  \item The lowest layer $1$ of peer $P$ receives data from the layer $2$ of peer $P$ which is denoted $d_{U_1}$ through channel $UI_1$ (the corresponding reading action is denoted $r_{UI_1}(d_{U_1})$),
  then processes the data, and sends the processed data to another layers peer $P'$ which is denoted $UF_1(d_{U_1})$ through channel $LO_1$ (the corresponding sending action is denoted
  $s_{LO_1}(UF_1(d_{U_1}))$);
  \item The lowest layer $1'$ of $P'$ receives data from the another layers peer $P$ which is denoted $d_{L_{1'}}$ through channel $LI_{1'}$
  (the corresponding reading action is denoted $r_{LI_{1'}}(d_{L_{1'}})$),
  then processes the data, and sends the processed data to layer $2$ of $P'$ which is denoted $LF_{1'}(d_{L_{1'}})$ through channel $UO_{1'}$ (the corresponding sending action is denoted
  $s_{UO_{1'}}(LF_{1'}(d_{L_{1'}}))$);
  \item The layer $i'$ of peer $P'$ receives data from the layer $i'-1$ of peer $P'$ which is denoted $d_{L_{i'}}$ through channel $LI_{i'}$
  (the corresponding reading action is denoted $r_{LI_{i'}}(d_{L_{i'}})$),
  then processes the data, and sends the processed data to layer $i'+1$ of peer $P'$ which is denoted $LF_{i'}(d_{L_{i'}})$ through channel $UO_{i'}$
  (the corresponding sending action is denoted
  $s_{UO_{i'}}(LF_{i'}(d_{L_{i'}}))$);
  \item The highest layer $n'$ of peer $P'$ receives data from layer $n'-1$ of peer $P'$ which is denoted $d_{L_{n'}}$ through channel $LI_{n'}$
  (the corresponding reading action is denoted $r_{LI_{n'}}(d_{L_{n'}})$),
  then processes the data, and sends the processed data to the application of peer $P'$ which is denoted $LF_{n'}(d_{L_{n'}})$ through channel $UO_{n'}$
  (the corresponding sending action is denoted
  $s_{UO_{n'}}(LF_{n'}(d_{L_{n'}}))$).
\end{enumerate}

The other similar process is data transferred from peer $P'$ to peer $P$, we do not repeat again and omit it.

\begin{figure}
    \centering
    \includegraphics{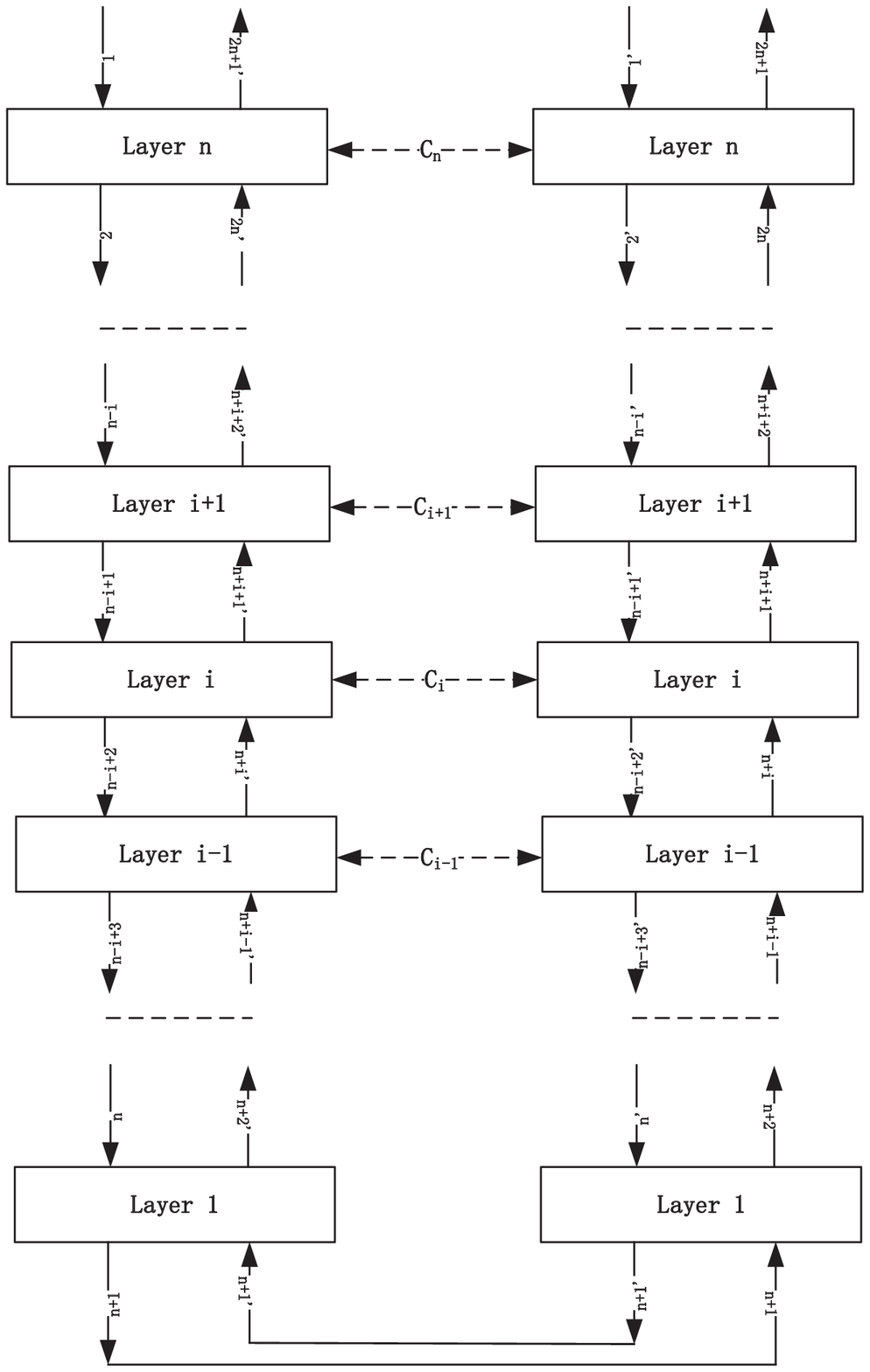}
    \caption{Typical process 1 of two layers peers}
    \label{LS23P}
\end{figure}

The verification of two layers peers is as follows.

We also assume all data elements $d_{U_i}$, $d_{L_i}$, $d_{U_{i'}}$ and $d_{L_{i'}}$ (for $1\leq i,i'\leq n$) are from a finite set $\Delta$. The state transitions of layer $i$ (for $1\leq i\leq n$)
described by APTC are as follows.

$L_i=\sum_{d_{U_i},d_{L_i}\in\Delta}(r_{UI_i}(d_{U_i})\cdot L_{i_2}\between r_{LI_i}(d_{L_i})\cdot L_{i_3})$

$L_{i_2}=UF_i\cdot L_{i_4}$

$L_{i_3}=LF_i\cdot L_{i_5}$

$L_{i_4}=\sum_{d_{U_i}\in\Delta}(s_{LO_i}(UF_i(d_{U_i}))\cdot L_i)$

$L_{i_5}=\sum_{d_{L_i}\in\Delta}(s_{UO_i}(LF_i(d_{L_i}))\cdot L_i)$

The state transitions of layer $i'$ (for $1\leq i'\leq n$)
described by APTC are as follows.

$L_{i'}=\sum_{d_{U_{i'}},d_{L_{i'}}\in\Delta}(r_{UI_{i'}}(d_{U_{i'}})\cdot L_{i'_2}\between r_{LI_{i'}}(d_{L_{i'}})\cdot L_{i'_3})$

$L_{i'_2}=UF_{i'}\cdot L_{i'_4}$

$L_{i'_3}=LF_{i'}\cdot L_{i'_5}$

$L_{i'_4}=\sum_{d_{U_{i'}}\in\Delta}(s_{LO_{i'}}(UF_{i'}(d_{U_{i'}}))\cdot L_{i'})$

$L_{i'_5}=\sum_{d_{L_{i'}}\in\Delta}(s_{UO_{i'}}(LF_{i'}(d_{L_{i'}}))\cdot L_{i'})$

The sending action and the reading action of the same data through the same channel can communicate with each other, otherwise, will cause a deadlock $\delta$. We define the following
communication functions for $1\leq i\leq n$ and $1\leq i'\leq n$. Note that, the channel of $LO_{i+1}$ of layer $i+1$ and the channel $UI_i$ of layer $i$ are the same one channel, and the channel $LI_{i+1}$ of layer $i+1$ and the
channel $UO_i$ of layer $i$ are the same one channel, the channel of $LO_{i'+1}$ of layer $i'+1$ and the channel $UI_{i'}$ of layer $i'$ are the same one channel, and the channel $LI_{i'+1}$ of layer $i'+1$ and the
channel $UO_{i'}$ of layer $i'$ are the same one channel. And also the data $d_{L_{i+1}}$ of layer $i+1$ and the data $LF_i(d_{L_i})$ of layer $i$ are the same data, and the data
$UF_{i+1}(d_{U_{i+1}})$ of layer $i+1$ and the data $d_{U_i}$ of layer $i$ are the same data; the data $d_{L_{i'+1}}$ of layer $i'+1$ and the data $LF_{i'}(d_{L_{i'}})$ of layer $i'$ are the same data, and the data
$UF_{i'+1}(d_{U_{i'+1}})$ of layer $i'+1$ and the data $d_{U_{i'}}$ of layer $i'$ are the same data.

$$\gamma(r_{UI_i}(d_{U_i}),s_{LO_{i+1}}(UF_{i+1}(d_{U_{i+1}})))\triangleq c_{UI_i}(d_{U_i})$$
$$\gamma(r_{LI_i}(d_{L_i}),s_{UO_{i-1}}(LF_{i-1}(d_{L_{i-1}})))\triangleq c_{LI_i}(d_{L_i})$$
$$\gamma(r_{UI_{i-1}}(d_{U_{i-1}}),s_{LO_{i}}(UF_{i}(d_{U_{i}})))\triangleq c_{UI_{i-1}}(d_{U_{i-1}})$$
$$\gamma(r_{LI_{i+1}}(d_{L_{i+1}}),s_{UO_{i}}(LF_{i}(d_{L_{i}})))\triangleq c_{LI_{i+1}}(d_{L_{i+1}})$$

$$\gamma(r_{UI_{i'}}(d_{U_{i'}}),s_{LO_{i'+1}}(UF_{i'+1}(d_{U_{i'+1}})))\triangleq c_{UI_{i'}}(d_{U_{i'}})$$
$$\gamma(r_{LI_{i'}}(d_{L_{i'}}),s_{UO_{i'-1}}(LF_{i'-1}(d_{L_{i'-1}})))\triangleq c_{LI_{i'}}(d_{L_{i'}})$$
$$\gamma(r_{UI_{i'-1}}(d_{U_{i'-1}}),s_{LO_{i'}}(UF_{i'}(d_{U_{i'}})))\triangleq c_{UI_{i'-1}}(d_{U_{i'-1}})$$
$$\gamma(r_{LI_{i'+1}}(d_{L_{i'+1}}),s_{UO_{i'}}(LF_{i}(d_{L_{i'}})))\triangleq c_{LI_{i'+1}}(d_{L_{i'+1}})$$

Note that, for the layer $n$, there are only two communication functions as follows.

$$\gamma(r_{LI_n}(d_{L_n}),s_{UO_{n-1}}(LF_{n-1}(d_{L_{n-1}})))\triangleq c_{LI_n}(d_{L_n})$$
$$\gamma(r_{UI_{n-1}}(d_{U_{n-1}}),s_{LO_{n}}(UF_{n}(d_{U_{n}})))\triangleq c_{UI_{n-1}}(d_{U_{n-1}})$$

For the layer $n'$, there are only two communication functions as follows.

$$\gamma(r_{LI_{n'}}(d_{L_{n'}}),s_{UO_{n'-1}}(LF_{n'-1}(d_{L_{n'-1}})))\triangleq c_{LI_{n'}}(d_{L_{n'}})$$
$$\gamma(r_{UI_{n'-1}}(d_{U_{n'-1}}),s_{LO_{n'}}(UF_{n'}(d_{U_{n'}})))\triangleq c_{UI_{n'-1}}(d_{U_{n'-1}})$$

For the layer $1$, there are four communication functions as follows.

$$\gamma(r_{UI_1}(d_{U_1}),s_{LO_{2}}(UF_{2}(d_{U_{2}})))\triangleq c_{UI_1}(d_{U_1})$$
$$\gamma(r_{LI_{2}}(d_{L_{2}}),s_{UO_{1}}(LF_{1}(d_{L_{1}})))\triangleq c_{LI_{2}}(d_{L_{2}})$$
$$\gamma(r_{LI_1}(d_{L_1}),s_{LO_{1'}}(UF_{1'}(d_{U_{1'}})))\triangleq c_{LI_1}(d_{L_1})$$
$$\gamma(r_{LI_{1'}}(d_{L_{1'}}),s_{LO_{1}}(UF_{1}(d_{U_{1}})))\triangleq c_{LI_{1'}}(d_{L_{1'}})$$

And for the layer $1'$, there are four communication functions as follows.

$$\gamma(r_{UI_{1'}}(d_{U_{1'}}),s_{LO_{2'}}(UF_{2'}(d_{U_{2'}})))\triangleq c_{UI_{1'}}(d_{U_{1'}})$$
$$\gamma(r_{LI_{2'}}(d_{L_{2'}}),s_{UO_{1'}}(LF_{1'}(d_{L_{1'}})))\triangleq c_{LI_{2'}}(d_{L_{2'}})$$
$$\gamma(r_{LI_1}(d_{L_1}),s_{LO_{1'}}(UF_{1'}(d_{U_{1'}})))\triangleq c_{LI_1}(d_{L_1})$$
$$\gamma(r_{LI_{1'}}(d_{L_{1'}}),s_{LO_{1}}(UF_{1}(d_{U_{1}})))\triangleq c_{LI_{1'}}(d_{L_{1'}})$$

Let all layers from layer $n$ to layer $1$ and from layer $1'$ to $n'$ be in parallel, then the Layers pattern $L_n\cdots L_i\cdots L_1L_{1'}\cdots L_{i'}\cdots L_{n'}$ can be presented by the following process term.

$\tau_I(\partial_H(\Theta(L_n\between\cdots\between L_i\between\cdots\between L_1\between L_{1'}\between \cdots\between L_{i'}\between \cdots\between L_{n'})))
=\tau_I(\partial_H(L_n\between\cdots\between L_i\between\cdots\between L_1\between L_{1'}\between \cdots\between L_{i'}\between \cdots\between L_{n'}))$

where $H=\{r_{LI_1}(d_{L_1}),s_{LO_{1}}(UF_{1}(d_{U_{1}})),r_{UI_1}(d_{U_1}),s_{UO_{1}}(LF_{1}(d_{L_{1}})),\cdots, r_{UI_i}(d_{U_i}), r_{LI_i}(d_{L_i}), \\
s_{LO_{i}}(UF_{i}(d_{U_{i}})),s_{UO_{i}}(LF_{i}(d_{L_{i}})),\cdots,r_{LI_n}(d_{L_n}),s_{LO_{n}}(UF_{n}(d_{U_{n}})),\\
r_{LI_{1'}}(d_{L_{1'}}),s_{LO_{1'}}(UF_{1}(d_{U_{1'}})),r_{UI_{1'}}(d_{U_{1'}}),s_{UO_{1'}}(LF_{1'}(d_{L_{1'}})),\cdots, r_{UI_{i'}}(d_{U_{i'}}), r_{LI_{i'}}(d_{L_{i'}}), \\
s_{LO_{i'}}(UF_{i'}(d_{U_{i'}})),s_{UO_{i'}}(LF_{i'}(d_{L_{i'}})),\cdots,r_{LI_{n'}}(d_{L_{n'}}),s_{LO_{n'}}(UF_{n'}(d_{U_{n'}}))\\
|d_{U_1},d_{L_1},\cdots,d_{U_i},d_{L_i}\cdots,d_{U_n},d_{L_n},d_{U_{1'}},d_{L_{1'}},\cdots,d_{U_{i'}},d_{L_{i'}}\cdots,d_{U_{n'}},d_{L_{n'}}\in\Delta\}$,

$I=\{c_{UI_1}(d_{U_1}),c_{LI_1}(d_{L_1}),c_{LI_{2}}(d_{L_{2}}),\cdots,c_{UI_i}(d_{U_i}),c_{LI_i}(d_{L_i}),c_{UI_{i-1}}(d_{U_{i-1}}),c_{LI_{i+1}}(d_{L_{i+1}}),
\cdots,\\c_{LI_n}(d_{L_n}),c_{UI_{n-1}}(d_{U_{n-1}}),LF_1,UF_1,\cdots,LF_i,UF_i,\cdots,LF_n,UF_n,\\
c_{UI_{1'}}(d_{U_{1'}}),c_{LI_{1'}}(d_{L_{1'}}),c_{LI_{2'}}(d_{L_{2'}}),\cdots,c_{UI_{i'}}(d_{U_{i'}}),c_{LI_{i'}}(d_{L_{i'}}),c_{UI_{i'-1}}(d_{U_{i'-1}}),c_{LI_{i'+1}}(d_{L_{i'+1}}),
\cdots,\\c_{LI_{n'}}(d_{L_{n'}}),c_{UI_{n'-1}}(d_{U_{n'-1}}),LF_{1'},UF_{1'},\cdots,LF_{i'},UF_{i'},\cdots,LF_{n'},UF_{n'}\\
|d_{U_1},d_{L_1},\cdots,d_{U_i},d_{L_i}\cdots,d_{U_n},d_{L_n},d_{U_{1'}},d_{L_{1'}},\cdots,d_{U_{i'}},d_{L_{i'}}\cdots,d_{U_{n'}},d_{L_{n'}}\in\Delta\}$.

Then we get the following conclusion on the Layers pattern.

\begin{theorem}[Correctness of two layers peers]
The two layers peers $\tau_I(\partial_H(L_n\between\cdots\between L_i\between\cdots\between L_1\between L_{1'}\between \cdots\between L_{i'}\between \cdots\between L_{n'}))$ can exhibit desired external behaviors.
\end{theorem}

\begin{proof}
Based on the above state transitions of layer $i$ and $i'$ (for $1\leq i,i'\leq n$), by use of the algebraic laws of APTC, we can prove that

$\tau_I(\partial_H(L_n\between\cdots\between L_i\between\cdots\between L_1\between L_{1'}\between \cdots\between L_{i'}\between \cdots\between L_{n'}))=\sum_{d_{U_n},d_{L_n},d_{U_{n'}},d_{L_{n'}}\in\Delta}((r_{UI_n}(d_{U_n})\parallel r_{UI_{n'}}(d_{U_{n'}}))
\cdot(s_{UO_n}(LF_n(d_{L_n}))\parallel s_{UO_{n'}}(LF_{n'}(d_{L_{n'}}))))\cdot \tau_I(\partial_H(L_n\between\cdots\between L_i\between\cdots\between L_1\between L_{1'}\between \cdots\between L_{i'}\between \cdots\between L_{n'}))$,

that is, the two layers peers $\tau_I(\partial_H(L_n\between\cdots\between L_i\between\cdots\between L_1\between L_{1'}\between \cdots\between L_{i'}\between \cdots\between L_{n'}))$ can exhibit desired external behaviors.

For the details of proof, please refer to section \ref{app}, and we omit it.
\end{proof}

There exists another composition of two layers peers. There are communications between two peers's peer layers which are called virtual communication. Virtual communications are specified
by communication protocols, and we assume data transferred between layer $i$ through virtual communications. And the two typical processes are illustrated in Figure \ref{LS23P2}. The
process from peer $P$ to peer $P'$ is as follows.

\begin{enumerate}
  \item The highest layer $n$ of peer $P$ receives data from the application of peer $P$ which is denoted $d_{U_n}$ through channel $UI_n$ (the corresponding reading action is denoted $r_{UI_n}(d_{U_n})$),
  then processes the data, and sends the processed data to layer $n-1$ of peer $P$ which is denoted $UF_n(d_{U_n})$ through channel $LO_n$ (the corresponding sending action is denoted
  $s_{LO_n}(UF_n(d_{U_n}))$);
  \item The layer $i$ of peer $P$ receives data from the layer $i+1$ of peer $P$ which is denoted $d_{U_i}$ through channel $UI_i$ (the corresponding reading action is denoted $r_{UI_i}(d_{U_i})$),
  then processes the data, and sends the processed data to layer $i$ of peer $P'$ which is denoted $UF_i(d_{U_i})$ through channel $LO_i$ (the corresponding sending action is denoted
  $s_{LO_i}(UF_i(d_{U_i}))$);
  \item The layer $i'$ of peer $P'$ receives data from the layer $i$ of peer $P$ which is denoted $d_{L_{i'}}$ through channel $LI_{i'}$
  (the corresponding reading action is denoted $r_{LI_{i'}}(d_{L_{i'}})$),
  then processes the data, and sends the processed data to layer $i'+1$ of peer $P'$ which is denoted $LF_{i'}(d_{L_{i'}})$ through channel $UO_{i'}$
  (the corresponding sending action is denoted $s_{UO_{i'}}(LF_{i'}(d_{L_{i'}}))$);
  \item The highest layer $n'$ of peer $P'$ receives data from layer $n'-1$ of peer $P'$ which is denoted $d_{L_{n'}}$ through channel $LI_{n'}$
  (the corresponding reading action is denoted $r_{LI_{n'}}(d_{L_{n'}})$),
  then processes the data, and sends the processed data to the application of peer $P'$ which is denoted $LF_{n'}(d_{L_{n'}})$ through channel $UO_{n'}$
  (the corresponding sending action is denoted
  $s_{UO_{n'}}(LF_{n'}(d_{L_{n'}}))$).
\end{enumerate}

The other similar process is data transferred from $P'$ to $P$, we do not repeat again and omit it.

\begin{figure}
    \centering
    \includegraphics{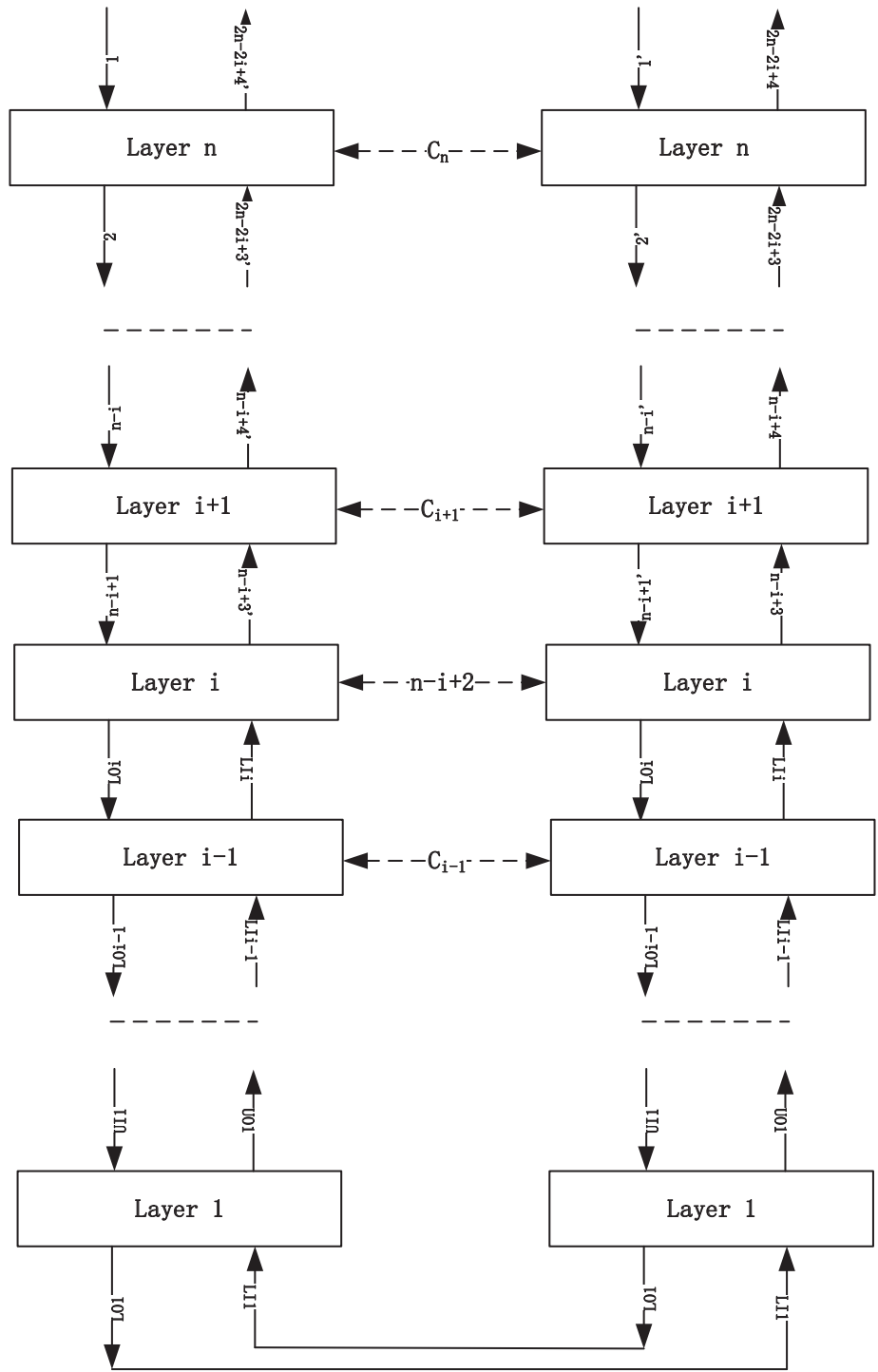}
    \caption{Typical process 2 of two layers peers}
    \label{LS23P2}
\end{figure}

The verification of two layers peers's communication through virtual communication is as follows.

We also assume all data elements $d_{U_i}$, $d_{L_i}$, $d_{U_{i'}}$ and $d_{L_{i'}}$ (for $1\leq i,i'\leq n$) are from a finite set $\Delta$. The state transitions of layer $i$ (for $1\leq i\leq n$)
described by APTC are as follows.

$L_i=\sum_{d_{U_i},d_{L_i}\in\Delta}(r_{UI_i}(d_{U_i})\cdot L_{i_2}\between r_{LI_i}(d_{L_i})\cdot L_{i_3})$

$L_{i_2}=UF_i\cdot L_{i_4}$

$L_{i_3}=LF_i\cdot L_{i_5}$

$L_{i_4}=\sum_{d_{U_i}\in\Delta}(s_{LO_i}(UF_i(d_{U_i}))\cdot L_i)$

$L_{i_5}=\sum_{d_{L_i}\in\Delta}(s_{UO_i}(LF_i(d_{L_i}))\cdot L_i)$

The state transitions of layer $i'$ (for $1\leq i'\leq n$)
described by APTC are as follows.

$L_{i'}=\sum_{d_{U_{i'}},d_{L_{i'}}\in\Delta}(r_{UI_{i'}}(d_{U_{i'}})\cdot L_{i'_2}\between r_{LI_{i'}}(d_{L_{i'}})\cdot L_{i'_3})$

$L_{i'_2}=UF_{i'}\cdot L_{i'_4}$

$L_{i'_3}=LF_{i'}\cdot L_{i'_5}$

$L_{i'_4}=\sum_{d_{U_{i'}}\in\Delta}(s_{LO_{i'}}(UF_{i'}(d_{U_{i'}}))\cdot L_{i'})$

$L_{i'_5}=\sum_{d_{L_{i'}}\in\Delta}(s_{UO_{i'}}(LF_{i'}(d_{L_{i'}}))\cdot L_{i'})$

The sending action and the reading action of the same data through the same channel can communicate with each other, otherwise, will cause a deadlock $\delta$. We define the following
communication functions for $1\leq i\leq n$ and $1\leq i'\leq n$. Note that, the channel of $LO_{i+1}$ of layer $i+1$ and the channel $UI_i$ of layer $i$ are the same one channel, and the channel $LI_{i+1}$ of layer $i+1$ and the
channel $UO_i$ of layer $i$ are the same one channel, the channel of $LO_{i'+1}$ of layer $i'+1$ and the channel $UI_{i'}$ of layer $i'$ are the same one channel, and the channel $LI_{i'+1}$ of layer $i'+1$ and the
channel $UO_{i'}$ of layer $i'$ are the same one channel. And also the data $d_{L_{i+1}}$ of layer $i+1$ and the data $LF_i(d_{L_i})$ of layer $i$ are the same data, and the data
$UF_{i+1}(d_{U_{i+1}})$ of layer $i+1$ and the data $d_{U_i}$ of layer $i$ are the same data; the data $d_{L_{i'+1}}$ of layer $i'+1$ and the data $LF_{i'}(d_{L_{i'}})$ of layer $i'$ are the same data, and the data
$UF_{i'+1}(d_{U_{i'+1}})$ of layer $i'+1$ and the data $d_{U_{i'}}$ of layer $i'$ are the same data.

For the layer $i$, there are four communication functions as follows.

$$\gamma(r_{UI_i}(d_{U_i}),s_{LO_{i+1}}(UF_{i+1}(d_{U_{i+1}})))\triangleq c_{UI_i}(d_{U_i})$$
$$\gamma(r_{LI_{i+1}}(d_{L_{i+1}}),s_{UO_{i}}(LF_{i}(d_{L_{i}})))\triangleq c_{LI_{i+1}}(d_{L_{i+1}})$$
$$\gamma(r_{LI_i}(d_{L_i}),s_{LO_{i'}}(UF_{i'}(d_{U_{i'}})))\triangleq c_{LI_i}(d_{L_i})$$
$$\gamma(r_{LI_{i'}}(d_{L_{i'}}),s_{LO_{i}}(UF_{i}(d_{U_{i}})))\triangleq c_{LI_{i'}}(d_{L_{i'}})$$

For the layer $i'$, there are four communication functions as follows.

$$\gamma(r_{UI_{i'}}(d_{U_{i'}}),s_{LO_{i'+1}}(UF_{i'+1}(d_{U_{i'+1}})))\triangleq c_{UI_{i'}}(d_{U_{i'}})$$
$$\gamma(r_{LI_{i'+1}}(d_{L_{i'+1}}),s_{UO_{i'}}(LF_{i'}(d_{L_{i'}})))\triangleq c_{LI_{i'+1}}(d_{L_{i'+1}})$$
$$\gamma(r_{LI_i}(d_{L_i}),s_{LO_{i'}}(UF_{i'}(d_{U_{i'}})))\triangleq c_{LI_i}(d_{L_i})$$
$$\gamma(r_{LI_{i'}}(d_{L_{i'}}),s_{LO_{i}}(UF_{i}(d_{U_{i}})))\triangleq c_{LI_{i'}}(d_{L_{i'}})$$

Note that, for the layer $n$, there are only two communication functions as follows.

$$\gamma(r_{LI_n}(d_{L_n}),s_{UO_{n-1}}(LF_{n-1}(d_{L_{n-1}})))\triangleq c_{LI_n}(d_{L_n})$$
$$\gamma(r_{UI_{n-1}}(d_{U_{n-1}}),s_{LO_{n}}(UF_{n}(d_{U_{n}})))\triangleq c_{UI_{n-1}}(d_{U_{n-1}})$$

And for the layer $n'$, there are only two communication functions as follows.

$$\gamma(r_{LI_{n'}}(d_{L_{n'}}),s_{UO_{n'-1}}(LF_{n'-1}(d_{L_{n'-1}})))\triangleq c_{LI_{n'}}(d_{L_{n'}})$$
$$\gamma(r_{UI_{n'-1}}(d_{U_{n'-1}}),s_{LO_{n'}}(UF_{n'}(d_{U_{n'}})))\triangleq c_{UI_{n'-1}}(d_{U_{n'-1}})$$

Let all layers from layer $n$ to layer $i$ be in parallel, then the Layers pattern $L_n\cdots L_iL_{i'}\cdots L_{n'}$ can be presented by the following process term.

$\tau_I(\partial_H(\Theta(L_n\between\cdots\between L_i\between L_{1'}\between \cdots\between L_{i'}\between \cdots\between L_{n'})))
=\tau_I(\partial_H(L_n\between\cdots\between L_i\between L_{i'}\between \cdots\between L_{n'}))$

where $H=\{r_{UI_i}(d_{U_i}), r_{LI_i}(d_{L_i}), s_{LO_{i}}(UF_{i}(d_{U_{i}})),s_{UO_{i}}(LF_{i}(d_{L_{i}})),\cdots,r_{LI_n}(d_{L_n}),s_{LO_{n}}(UF_{n}(d_{U_{n}})),\\
r_{UI_{i'}}(d_{U_{i'}}), r_{LI_{i'}}(d_{L_{i'}}), s_{LO_{i'}}(UF_{i'}(d_{U_{i'}})),s_{UO_{i'}}(LF_{i'}(d_{L_{i'}})),\cdots,r_{LI_{n'}}(d_{L_{n'}}),s_{LO_{n'}}(UF_{n'}(d_{U_{n'}}))\\
|d_{U_i},d_{L_i}\cdots,d_{U_n},d_{L_n},d_{U_{i'}},d_{L_{i'}}\cdots,d_{U_{n'}},d_{L_{n'}}\in\Delta\}$,

$I=\{c_{UI_i}(d_{U_i}),c_{LI_i}(d_{L_i}),c_{LI_{i+1}}(d_{L_{i+1}}),\cdots,c_{LI_n}(d_{L_n}),c_{UI_{n-1}}(d_{U_{n-1}}),LF_i,UF_i,\cdots,LF_n,UF_n,\\
c_{UI_{i'}}(d_{U_{i'}}),c_{LI_{i'}}(d_{L_{i'}}),c_{LI_{i'+1}}(d_{L_{i'+1}}),\cdots,c_{LI_{n'}}(d_{L_{n'}}),c_{UI_{n'-1}}(d_{U_{n'-1}}),LF_{i'},UF_{i'},\cdots,LF_{n'},UF_{n'}\\
|d_{U_i},d_{L_i}\cdots,d_{U_n},d_{L_n},d_{U_{i'}},d_{L_{i'}}\cdots,d_{U_{n'}},d_{L_{n'}}\in\Delta\}$.

Then we get the following conclusion on the Layers pattern.

\begin{theorem}[Correctness of two layers peers via virtual communication]
The two layers peers via virtual communication $\tau_I(\partial_H(L_n\between\cdots\between L_i\between L_{i'}\between \cdots\between L_{n'}))$ can exhibit desired external behaviors.
\end{theorem}

\begin{proof}
Based on the above state transitions of layer $i$ and $i'$ (for $1\leq i,i'\leq n$), by use of the algebraic laws of APTC, we can prove that

$\tau_I(\partial_H(L_n\between\cdots\between L_i\between L_{i'}\between \cdots\between L_{n'}))=\sum_{d_{U_n},d_{L_n},d_{U_{n'}},d_{L_{n'}}\in\Delta}((r_{UI_n}(d_{U_n})\parallel r_{UI_{n'}}(d_{U_{n'}}))
\cdot(s_{UO_n}(LF_n(d_{L_n}))\parallel s_{UO_{n'}}(LF_{n'}(d_{L_{n'}}))))\cdot \tau_I(\partial_H(L_n\between\cdots\between L_i\between L_{i'}\between \cdots\between L_{n'}))$,

that is, the two layers peers via virtual communication $\tau_I(\partial_H(L_n\between\cdots\between L_i\between L_{i'}\between \cdots\between L_{n'}))$ can exhibit desired external behaviors.

For the details of proof, please refer to section \ref{app}, and we omit it.
\end{proof}

\subsubsection{Verification of the Pipes and Filters Pattern}

The Pipes and Filters pattern is used to process a stream of data with each processing step being encapsulated in a filter component. The data stream flows out of the data source, and
into the first filter; the first filter processes the data, and sends the processed data to the next filter; eventually, the data stream flows out of the pipes of filters and into the
data sink, as Figure \ref{PS3} illustrated, there are $n$ filters in the pipes. Especially, for filter $i$ ($1\leq i\leq n$), as illustrated in Figure \ref{P3}, it has an input channel
$I_i$ to read the data $d_i$, then processes the data via a processing function $FF_i$, finally send the processed data to the next filter through an output channel $O_i$.

\begin{figure}
    \centering
    \includegraphics{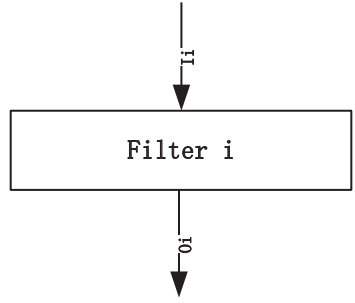}
    \caption{Filter i}
    \label{P3}
\end{figure}

\begin{figure}
    \centering
    \includegraphics{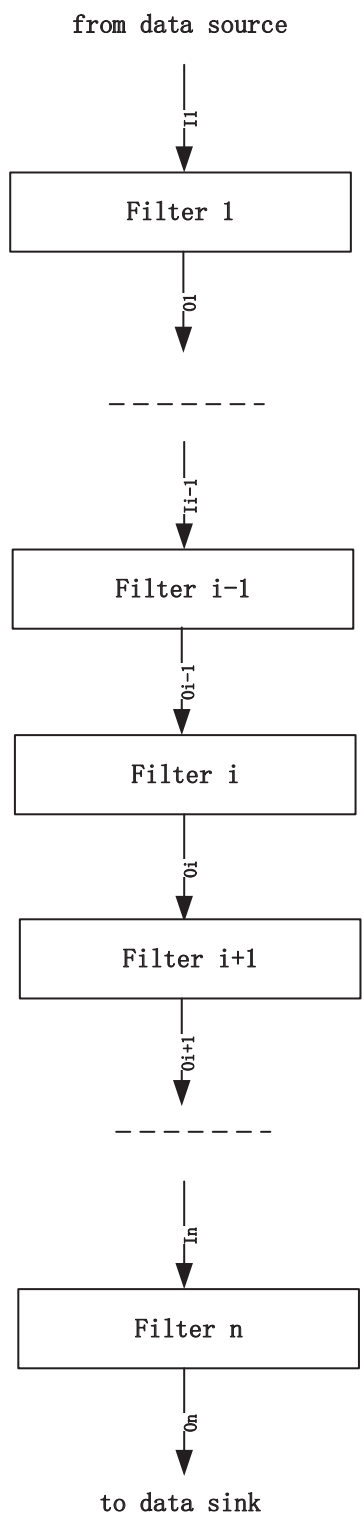}
    \caption{Pipes and Filters pattern}
    \label{PS3}
\end{figure}

There is one typical process in the Pipes and Filters pattern as illustrated in Figure \ref{PS3P} and follows.

\begin{enumerate}
  \item The filter $1$ receives the data from the data source which is denoted $d_1$ through the channel $I_1$ (the corresponding reading action is denoted $r_{I_1}(d_1)$), then processes
  the data through a processing function $FF_1$, and sends the processed data to the filter $2$ which is denoted $FF_1(d_1)$ through the channel $O_1$ (the corresponding sending action is
  denoted $s_{O_1}(FF_1(d_1))$);
  \item The filter $i$ receives the data from filter $i-1$ which is denoted $d_i$ through the channel $I_i$ (the corresponding reading action is denoted $r_{I_i}(d_i)$), then processes
  the data through a processing function $FF_i$, and sends the processed data to the filter $i+1$ which is denoted $FF_i(d_i)$ through the channel $O_i$ (the corresponding sending action is
  denoted $s_{O_i}(FF_i(d_i))$);
  \item The filter $n$ receives the data from filter $n-1$ which is denoted $d_n$ through the channel $I_n$ (the corresponding reading action is denoted $r_{I_n}(d_n)$), then processes
  the data through a processing function $FF_n$, and sends the processed data to the data sink which is denoted $FF_n(d_n)$ through the channel $O_n$ (the corresponding sending action is
  denoted $s_{O_n}(FF_1(d_n))$).
\end{enumerate}

\begin{figure}
    \centering
    \includegraphics{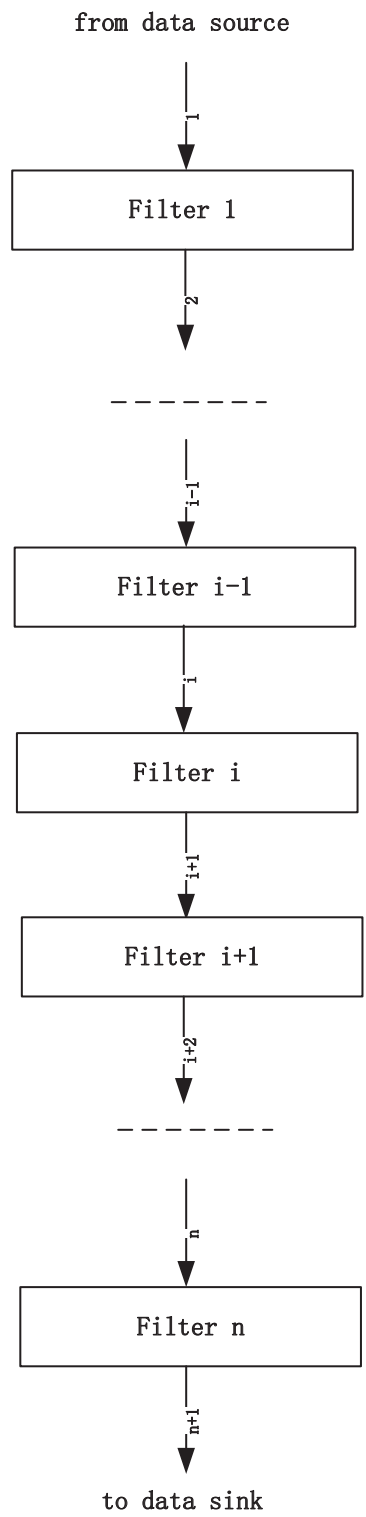}
    \caption{Typical process of Pipes and Filters pattern}
    \label{PS3P}
\end{figure}

In the following, we verify the Pipes and Filters pattern. We assume all data elements $d_{i}$ (for $1\leq i\leq n$) are from a finite set $\Delta$. The state transitions of filter $i$ (for $1\leq i\leq n$)
described by APTC are as follows.

$F_i=\sum_{d_{i},\in\Delta}(r_{I_i}(d_{i})\cdot F_{i_2})$

$F_{i_2}=FF_i\cdot F_{i_3}$

$F_{i_3}=\sum_{d_{i}\in\Delta}(s_{O_i}(FF_i(d_{i}))\cdot F_i)$

The sending action and the reading action of the same data through the same channel can communicate with each other, otherwise, will cause a deadlock $\delta$. We define the following
communication functions for $1\leq i\leq n$. Note that, the channel of $I_{i+1}$ of filter $i+1$ and the channel $O_i$ of filter $i$ are the same one channel. And also the data $d_{i+1}$ of filter $i+1$ and the data $FF_i(d_{i})$ of filter $i$ are the same data.

$$\gamma(r_{I_i}(d_{i}),s_{O_{i-1}}(FF_{i-1}(d_{i-1})))\triangleq c_{I_i}(d_{i})$$
$$\gamma(r_{I_{i+1}}(d_{i+1}),s_{O_{i}}(FF_{i}(d_{i})))\triangleq c_{I_{i+1}}(d_{i+1})$$

Note that, for the filter $n$, there are only one communication functions as follows.

$$\gamma(r_{I_n}(d_{n}),s_{O_{n-1}}(FF_{n-1}(d_{n-1})))\triangleq c_{I_n}(d_{n})$$

And for the filter $1$, there are also only one communication functions as follows.

$$\gamma(r_{I_{2}}(d_{2}),s_{O_{1}}(FF_{1}(d_{1})))\triangleq c_{I_{2}}(d_{2})$$

Let all filters from filter $1$ to filter $n$ be in parallel, then the Pipes and Filters pattern $F_1\cdots F_i\cdots F_n$ can be presented by the following process term.

$$\tau_I(\partial_H(\Theta(F_1\between\cdots\between F_i\between\cdots\between F_n)))=\tau_I(\partial_H(F_1\between\cdots\between F_i\between\cdots\between F_n))$$

where $H=\{s_{O_{1}}(FF_{1}(d_{1})),\cdots, r_{I_i}(d_{i}),s_{O_{i}}(FF_{i}(d_{i})),\cdots,r_{I_n}(d_{n})|d_{1},\cdots,d_{i},\cdots,d_{n}\in\Delta\}$,

$I=\{c_{I_2}(d_{2}),\cdots,c_{I_i}(d_{i}),\cdots,c_{I_n}(d_{n}),FF_1,\cdots,FF_i\cdots,FF_n|d_{1},\cdots,d_{i},\cdots,d_{n}\in\Delta\}$.

Then we get the following conclusion on the Pipes and Filters pattern.

\begin{theorem}[Correctness of the Pipes and Filters pattern]
The Pipes and Filters pattern $\tau_I(\partial_H(F_1\between\cdots\between F_i\between\cdots\between F_n))$ can exhibit desired external behaviors.
\end{theorem}

\begin{proof}
Based on the above state transitions of filter $i$ (for $1\leq i\leq n$), by use of the algebraic laws of APTC, we can prove that

$\tau_I(\partial_H(F_1\between\cdots\between F_i\between\cdots\between F_n))=\sum_{d_{1},d_{n}\in\Delta}(r_{I_1}(d_{1})\cdot s_{O_n}(FF_n(d_{n})))\cdot
\tau_I(\partial_H(F_1\between\cdots\between F_i\between\cdots\between F_n))$,

that is, the Pipes and Filters pattern $\tau_I(\partial_H(F_1\between\cdots\between F_i\between\cdots\between F_n))$ can exhibit desired external behaviors.

For the details of proof, please refer to section \ref{app}, and we omit it.
\end{proof}

\subsubsection{Verification of the Blackboard Pattern}

The Blackboard pattern is used to solve problems with no deterministic solutions. In the Blackboard pattern, there are one Control module, one Blackboard module and several Knowledge
Source modules. When the Control module receives a request, it queries the Blackboard module the involved Knowledge Sources; then the Control module invokes the related Knowledge Sources;
Finally, the related Knowledge Sources update the Blackboard with the invoked results, as illustrated in Figure \ref{BB3}.

\begin{figure}
    \centering
    \includegraphics{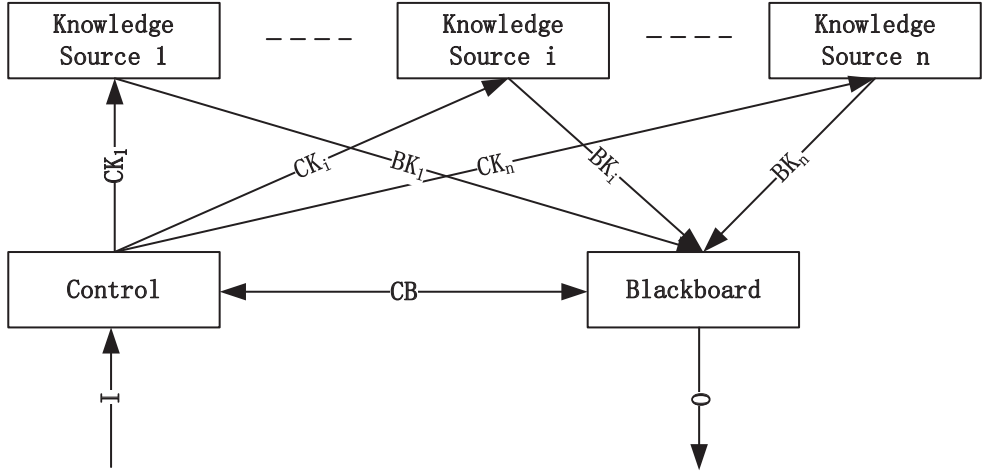}
    \caption{Blackboard pattern}
    \label{BB3}
\end{figure}

The typical process of the Blackboard pattern is illustrated in Figure \ref{BB3P} and as follows.

\begin{enumerate}
  \item The Control module receives the request from outside applications which is denoted $d_I$, through the input channel $I$ (the corresponding reading action is denoted $r_I(d_I)$),
  then processes the request through a processing function $CF_1$, and sends the processed data which is denoted $CF_1(d_I)$ to the Blackboard module through the channel $CB$ (the corresponding
  sending action is denoted $s_{CB}(CF_1(d_I))$);
  \item The Blackboard module receives the request (information of involved Knowledge Sources) from the Control module through the channel $CB$ (the corresponding reading action is
  denoted $r_{CB}(CF_1(d_I))$), then processes the request through a processing function $BF_1$, and generates and sends the response which is denoted $d_B$ to the Control module
  through the channel $CB$ (the corresponding sending action is denoted $s_{CB}(d_B)$);
  \item The Control module receives the data from the Blackboard module through the channel $CB$ (the corresponding reading action is denoted $r_{CB}(d_B)$), then processes the data through
  another processing function $CF_2$, and generates and sends the requests to the related Knowledge Sources which are denoted $d_{C_i}$ through the channels $CK_i$ (the corresponding
  sending action is denoted $s_{CK_i}(d_{C_i})$) with $1\leq i\leq n$;
  \item The Knowledge Source $i$ receives the request from the Control module through the channel $CK_i$ (the corresponding reading action is denoted $r_{CK_i}(d_{C_i})$), then processes
  the request through a processing function $KF_i$, and generates and sends the processed data $d_{K_i}$ to the Blackboard module through the channel $BK_i$ (the corresponding sending
  action is denoted $s_{BK_i}(d_{K_i})$);
  \item The Blackboard module receives the invoked results from Knowledge Source $i$ through the channel $BK_i$ (the corresponding reading action is denoted
  $r_{BK_i}(d_{K_i})$) ($1\leq i\leq n$), then processes the results through another processing function $BF_2$, generates and sends the output $d_O$ through the channel $O$ (the corresponding
  sending action is denoted $s_O(d_O)$).
\end{enumerate}

\begin{figure}
    \centering
    \includegraphics{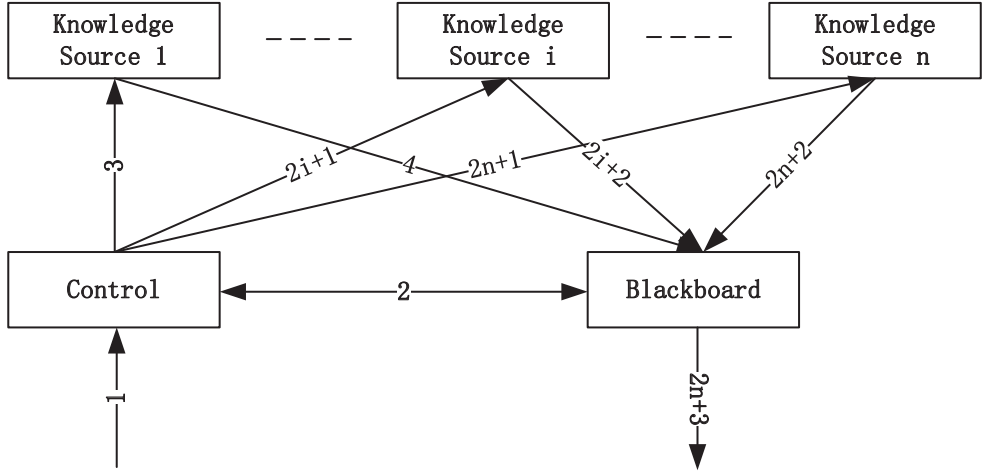}
    \caption{Typical process of Blackboard pattern}
    \label{BB3P}
\end{figure}

In the following, we verify the Blackboard pattern. We assume all data elements $d_{I},d_B,d_{C_i},d_{K_i},d_O$ (for $1\leq i\leq n$) are from a finite set $\Delta$. The state transitions of the Control module
described by APTC are as follows.

$C=\sum_{d_{I}\in\Delta}(r_{I}(d_{I})\cdot C_{2})$

$C_{2}=CF_1\cdot C_{3}$

$C_{3}=\sum_{d_{I}\in\Delta}(s_{CB}(CF_1(d_I))\cdot C_{4})$

$C_{4}=\sum_{d_{B}\in\Delta}(r_{CB}(d_B)\cdot C_{5})$

$C_{5}=CF_2\cdot C_{6}$

$C_{6}=\sum_{d_{C_1},\cdots,d_{C_i},\cdots,d_{C_n}\in\Delta}(s_{CK_1}(d_{C_1})\between \cdots\between s_{CK_i}(d_{C_i})\between\cdots\between s_{CK_n}(d_{C_n})\cdot C)$

The state transitions of the Blackboard module described by APTC are as follows.

$B=\sum_{d_{I}\in\Delta}(r_{CB}(CF_1(d_{I}))\cdot B_{2})$

$B_{2}=BF_1\cdot B_{3}$

$B_{3}=\sum_{d_{B}\in\Delta}(s_{CB}(d_B)\cdot B_{4})$

$B_{4}=\sum_{d_{K_1},\cdots,d_{K_i},\cdots,d_{K_n}\in\Delta}(r_{BK_1}(d_{K_1})\between \cdots\between r_{BK_i}(d_{K_i})\between\cdots\between r_{BK_n}(d_{K_n})\cdot B_5)$

$B_{5}=BF_2\cdot B_{6}$

$B_{6}=\sum_{d_{O}\in\Delta}(s_{O}(d_{O})\cdot B)$

The state transitions of the Knowledge Source $i$ described by APTC are as follows.

$K_{i}=\sum_{d_{C_i}\in\Delta}(r_{CK_i}(d_{C_i})\cdot K_{i_2})$

$K_{i_2}=KF_i\cdot K_{i_3}$

$K_{i_3}=\sum_{d_{K_i}\in\Delta}(s_{BK_i}(d_{K_i})\cdot K_{i})$

The sending action and the reading action of the same data through the same channel can communicate with each other, otherwise, will cause a deadlock $\delta$. We define the following
communication functions for $1\leq i\leq n$.

$$\gamma(r_{CB}(CF_1(d_I)),s_{CB}(CF_1(d_I)))\triangleq c_{CB}(CF_1(d_I))$$
$$\gamma(r_{CB}(d_B),s_{CB}(d_B))\triangleq c_{CB}(d_B)$$
$$\gamma(r_{CK_i}(d_{C_i}),s_{CK_i}(d_{C_i}))\triangleq c_{CK_i}(d_{C_i})$$
$$\gamma(r_{BK_i}(d_{K_i}),s_{BK_i}(d_{K_i}))\triangleq c_{BK_i}(d_{K_i})$$

Let all modules be in parallel, then the Blackboard pattern $C\quad B\quad K_1\cdots K_i\cdots K_n$ can be presented by the following process term.

$$\tau_I(\partial_H(\Theta(C\between B\between K_1\between\cdots\between K_i\between\cdots\between K_n)))=\tau_I(\partial_H(C\between B\between K_1\between\cdots\between K_i\between\cdots\between K_n))$$

where $H=\{r_{CB}(CF_1(d_I)),s_{CB}(CF_1(d_I)),r_{CB}(d_B),s_{CB}(d_B),r_{CK_1}(d_{C_1}),\\
s_{CK_1}(d_{C_1}),\cdots, r_{CK_i}(d_{C_i}),s_{CK_i}(d_{C_i}),\cdots,r_{CK_n}(d_{C_n}),s_{CK_n}(d_{C_n}),\\
r_{BK_1}(d_{K_1}),s_{BK_1}(d_{K_1}),\cdots,r_{BK_i}(d_{K_i}),s_{BK_i}(d_{K_i}),\cdots,r_{BK_n}(d_{K_n}),s_{BK_n}(d_{K_n})\\
|d_I,d_B,d_{C_1},\cdots,d_{C_i},\cdots,d_{C_n},d_{K_1},\cdots,d_{K_i},\cdots,d_{K_n}\in\Delta\}$,

$I=\{c_{CB}(CF_1(d_I)),c_{CB}(d_B),c_{CK_1}(d_{C_1}),\cdots,c_{CK_i}(d_{C_i}),\cdots,c_{CK_n}(d_{C_n}),c_{BK_1}(d_{K_1}),\\
\cdots,c_{BK_i}(d_{K_i}),\cdots,c_{BK_n}(d_{K_n}),CF_1,CF_2,BF_1,BF_2,KF1,\cdots,KF_i,\cdots,KF_n\\
|d_I,d_B,d_{C_1},\cdots,d_{C_i},\cdots,d_{C_n},d_{K_1},\cdots,d_{K_i},\cdots,d_{K_n}\in\Delta\}$.

Then we get the following conclusion on the Blackboard pattern.

\begin{theorem}[Correctness of the Blackboard pattern]
The Blackboard pattern $\tau_I(\partial_H(C\between B\between K_1\between\cdots\between K_i\between\cdots\between K_n))$ can exhibit desired external behaviors.
\end{theorem}

\begin{proof}
Based on the above state transitions of the above modules, by use of the algebraic laws of APTC, we can prove that

$\tau_I(\partial_H(C\between B\between K_1\between\cdots\between K_i\between\cdots\between K_n))=\sum_{d_{I},d_{O}\in\Delta}(r_{I}(d_{I})\cdot s_{O}(d_O))\cdot
\tau_I(\partial_H(C\between B\between K_1\between\cdots\between K_i\between\cdots\between K_n))$,

that is, the Blackboard pattern $\tau_I(\partial_H(C\between B\between K_1\between\cdots\between K_i\between\cdots\between K_n))$ can exhibit desired external behaviors.

For the details of proof, please refer to section \ref{app}, and we omit it.
\end{proof}

\subsection{Distributed Systems}\label{DS3}

In this subsection, we verify the distributed systems oriented patterns, including the Broker pattern, and the Pipes and Filters pattern in subsection \ref{MUD3} and the Microkernel
Pattern in subsection \ref{AS3}.

\subsubsection{Verification of the Broker Pattern}

The Broker pattern decouples the invocation process between the Client and the Server. There are five types of modules in the Broker pattern: the Client, the Client-side Proxy, the Brokers,
the Server-side Proxy and the Server. The Client receives the request from the user and passes it to the Client-side Proxy, then to the first broker and the next one, the last broker passes
it the Server-side Proxy, finally leads to the invocation to the Server; the Server processes the request and generates the response, then the response is returned to the user in a reverse
way, as illustrated in Figure \ref{B3}.

\begin{figure}
    \centering
    \includegraphics{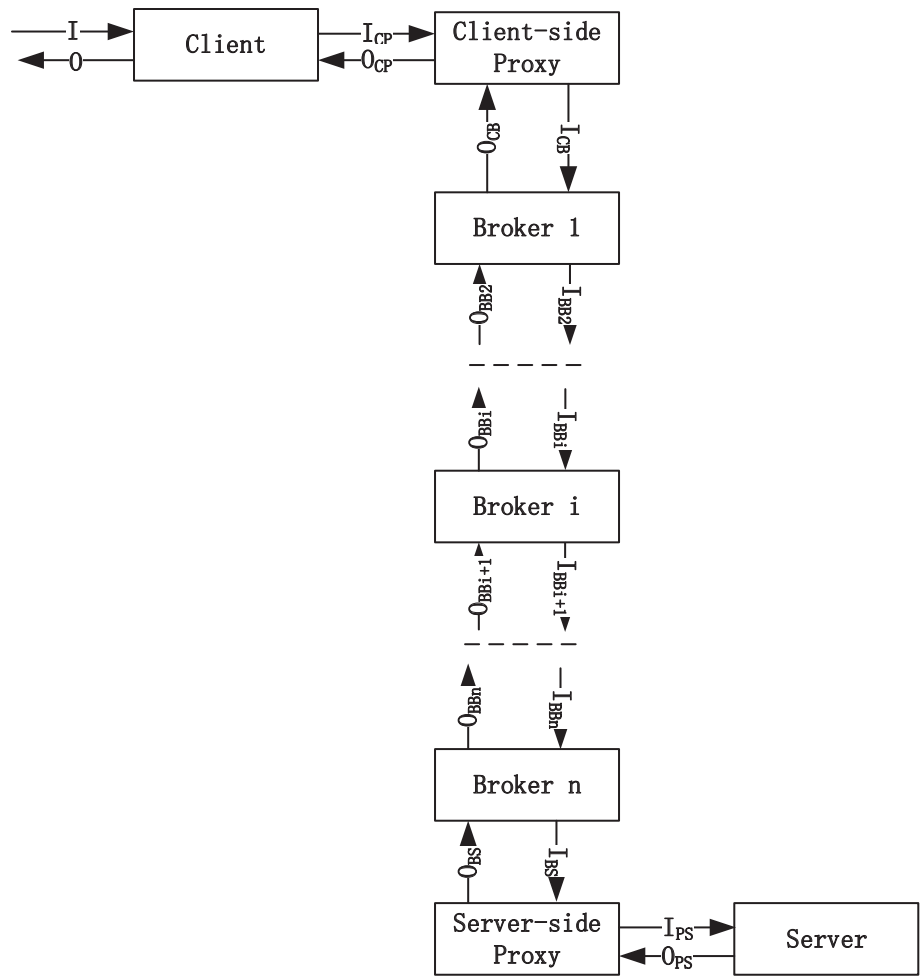}
    \caption{Broker pattern}
    \label{B3}
\end{figure}

The typical process of the Broker pattern is illustrated in Figure \ref{B3P} and as follows.

\begin{enumerate}
  \item The Client receives the request $d_I$ through the channel $I$ (the corresponding reading action is denoted $r_I(d_I)$), then processes the request through a processing function
  which is denoted $CF_1$, then sends the processed request $CF_1(d_I)$ to the Client-side Proxy through the channel $I_{CP}$ (the corresponding sending action is denoted $s_{I_{CP}}(CF_1(d_I))$);
  \item The Client-side Proxy receives the request $d_{I_{CP}}$ from the Client through the channel $I_{CP}$ (the corresponding reading action is denoted $r_{I_{CP}}(d_{I_{CP}})$), then processes
  the request through a processing function $CPF_1$, and then sends the processed request $CPF_1(d_{I_{CP}})$ to the first broker $1$ through the channel $I_{CB}$ (the corresponding sending
  action is denoted $s_{I_{CB}}(CPF_1(d_{I_{CP}}))$);
  \item The broker $i$ (for $1\leq i\leq n$) receives the request $d_{I_{B_i}}$ from the broker $i-1$ through the channel $I_{BB_i}$ (the corresponding reading action is denoted $r_{I_{BB_i}}(d_{I_{B_i}})$),
  then processes the request through a processing function $BF_{i_1}$, and then sends the processes request $BF_{i_1}(d_{I_{B_i}})$ to the broker $i+1$ through the channel $I_{BB_{i+1}}$
  (the corresponding sending action is denoted $s_{I_{BB_{i+1}}}(BF_{i_1}(d_{I_{B_i}}))$);
  \item The Server-side Proxy receives the request $d_{I_{SP}}$ from the last broker $n$ through the channel $I_{BS}$ (the corresponding reading action is denoted $r_{I_{BS}}(d_{I_{SP}})$),
  then processes the request through a processing function $SPF_1$, and then sends the processed request $SPF_1(d_{I_{SP}})$ to the Server through the channel $I_{PS}$ (the corresponding
  sending action is denoted $s_{I_{PS}}(SPF_1(d_{I_{SP}}))$);
  \item The Server receives the request $d_{I_S}$ from the Server-side Proxy through the channel $I_{PS}$ (the corresponding reading action is denoted $r_{I_{PS}}(d_{I_S})$), then processes
  the request and generates the response $d_{O_S}$ through a processing function $SF$, and then sends the response to the Server-side Proxy through the channel $O_{PS}$ (the corresponding
  sending action is denoted $s_{O_{PS}}(d_{O_S})$);
  \item The Server-side Proxy receives the response $d_{O_{SP}}$ from the Server through the channel $O_{PS}$ (the corresponding reading action is denoted $r_{O_{PS}}(d_{O_{SP}})$),
  then processes the response through a processing function $SPF_2$, and sends the processed response $SPF_2(d_{O_{SP}})$ to the last broker $n$ through the channel $O_{BS}$ (the corresponding
  sending action is denoted $s_{O_{BS}}(SPF_2(d_{O_{SP}}))$);
  \item the broker $i$ receives the response $d_{O_{B_i}}$ from the broker $i+1$ through the channel $O_{BB_{i+1}}$ (the corresponding reading action is denoted $r_{O_{BB_{i+1}}}(d_{O_{B_i}})$),
  then processes the response through a processing function $BF_{i_2}$, and then sends the processed response $BF_{i_2}(d_{O_{B_i}})$ to the broker $i-1$ through the channel $O_{BB_i}$
  (the corresponding sending action is denoted $s_{O_{BB_i}}(BF_{i_2}(d_{O_{B_i}}))$);
  \item The Client-side Proxy receives the response $d_{O_{CP}}$ from the first broker $1$ through the channel $O_{CB}$ (the corresponding reading action is denoted $r_{O_{CB}}(d_{O_{CP}})$),
  then processes the response through a processing function $CPF_2$, and sends the processed response $CPF_2(d_{O_{CP}})$ to the Client through the channel $O_{CP}$ (the corresponding
  sending action is denoted $s_{O_{CP}}(CPF_2(d_{O_{CP}}))$);
  \item The Client receives the response $d_{O_{C}}$ from the Client-side Proxy through the channel $O_{CP}$ (the corresponding reading action is denoted $r_{O_{CP}}(d_{O_{C}})$),
  then processes the response through a processing function $CF_2$ and generate the response $d_O$, and then sends the response out through the channel $O$ (the corresponding sending
  action is denoted $s_O(d_O)$).
\end{enumerate}

\begin{figure}
    \centering
    \includegraphics{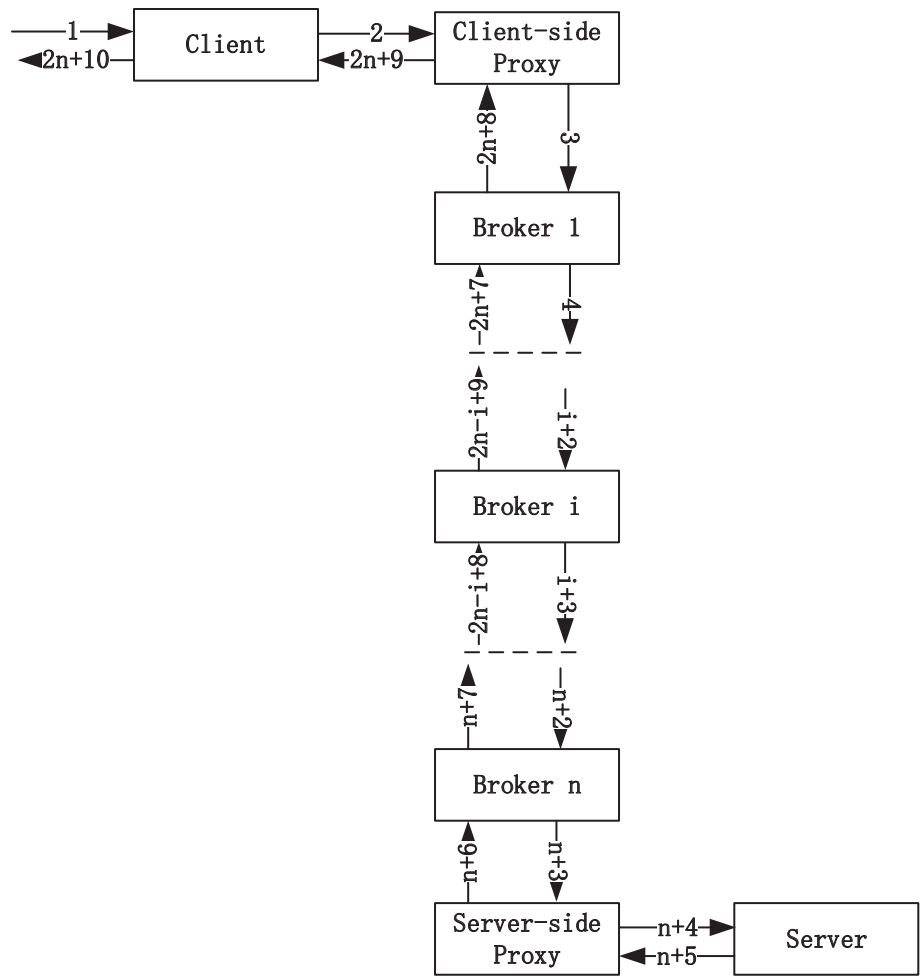}
    \caption{Typical process of Broker pattern}
    \label{B3P}
\end{figure}

In the following, we verify the Broker pattern. We assume all data elements $d_{I}$, $d_{I_{CP}}$, $d_{I_{B_i}}$, $d_{I_{SP}}$, $d_{I_S}$, $d_{O_S}$, $d_{O_{B_i}}$, $d_{O_{CP}}$, $d_{O_C}$, $d_O$
(for $1\leq i\leq n$) are from a finite set $\Delta$. Note that, the channels $I_{BB_1}$ and $I_{CB}$ are the same one channel; the channels $O_{BB_1}$ and $O_{CB}$ are the same one channel;
the channels $I_{BB_{n+1}}$ and $I_{BS}$ are the same one channel; the channels $O_{BB_{n+1}}$ and $O_{BS}$ are the same one channel. And the data $CF_1(d_I)$ and $d_{I_{CP}}$ are the same
data; the data $CPF_1(d_{I_{CP}})$ and $d_{I_{B_1}}$ are the same data; the data $BF_{i_1}(d_{I_{B_i}})$ and $d_{I_{B_{i+1}}}$ are the same data; the data $BF_{n_1}(d_{I_{B_{n}}})$ and
the data $d_{I_{SP}}$ are the same data; the data $SPF_1(d_{I_{SP}})$ and $d_{I_S}$ are the same data; the data $SPF_2(d_{O_{S}})$ and $d_{O_{B_n}}$ are the same data; the data
$BF_{i_2}(d_{O_{B_i}})$ and the data $d_{O_{B_{i-1}}}$ are the same data; the data $BF_{1_2}(d_{O_{B_1}})$ and $d_{O_{CP}}$ are the same data; the data $CPF_2(d_{O_{CP}})$ and $d_{O_{C}}$
are the same data; the data $CF_2(d_{O_{C}})$ and the data $d_O$ are the same data.

The state transitions of the Client module
described by APTC are as follows.

$C=\sum_{d_{I}\in\Delta}(r_{I}(d_{I})\cdot C_{2})$

$C_{2}=CF_1\cdot C_{3}$

$C_{3}=\sum_{d_{I}\in\Delta}(s_{I_{CP}}(CF_1(d_I))\cdot C_{4})$

$C_{4}=\sum_{d_{O_C}\in\Delta}(r_{O_{CP}}(d_{O_{C}})\cdot C_{5})$

$C_{5}=CF_2\cdot C_{6}$

$C_{6}=\sum_{d_{O}\in\Delta}(s_O(d_O)\cdot C)$

The state transitions of the Client-side Proxy module described by APTC are as follows.

$CP=\sum_{d_{I_{CP}}\in\Delta}(r_{I_{CP}}(d_{I_{CP}})\cdot CP_{2})$

$CP_{2}=CPF_1\cdot CP_{3}$

$CP_{3}=\sum_{d_{I_{CP}}\in\Delta}(s_{I_{CB}}(CPF_1(d_{I_{CP}}))\cdot CP_{4})$

$CP_{4}=\sum_{d_{O_{CP}}\in\Delta}(r_{O_{CB}}(d_{O_{CP}})\cdot CP_{5})$

$CP_{5}=CPF_2\cdot CP_{6}$

$CP_{6}=\sum_{d_{O_{CP}}\in\Delta}(s_{O_{CP}}(CPF_2(d_{O_{CP}}))\cdot CP)$

The state transitions of the Broker $i$ described by APTC are as follows.

$B_i=\sum_{d_{I_{B_i}}\in\Delta}(r_{I_{BB_i}}(d_{I_{B_i}})\cdot B_{i_2})$

$B_{i_2}=BF_{i_1}\cdot B_{i_3}$

$B_{i_3}=\sum_{d_{I_{B_i}}\in\Delta}(s_{I_{BB_{i+1}}}(BF_{i_1}(d_{I_{B_i}}))\cdot B_{i_4})$

$B_{i_4}=\sum_{d_{O_{B_i}}\in\Delta}(r_{O_{BB_{i+1}}}(d_{O_{B_i}})\cdot B_{i_5})$

$B_{i_5}=BF_{i_2}\cdot B_{i_6}$

$B_{i_6}=\sum_{d_{O_{B_i}}\in\Delta}(s_{O_{BB_i}}(BF_{i_2}(d_{O_{B_i}}))\cdot B_i)$

The state transitions of the Server-side Proxy described by APTC are as follows.

$SP=\sum_{d_{I_{SP}}\in\Delta}(r_{I_{BS}}(d_{I_{SP}})\cdot SP_{2})$

$SP_{2}=SPF_1\cdot SP_{3}$

$SP_{3}=\sum_{d_{I_{SP}}\in\Delta}(s_{I_{PS}}(SPF_1(d_{I_{SP}}))\cdot SP_{4})$

$SP_{4}=\sum_{d_{O_{SP}}\in\Delta}(r_{O_{PS}}(d_{O_{SP}})\cdot SP_{5})$

$SP_{5}=SPF_2\cdot SP_{6}$

$SP_{6}=\sum_{d_{O_{SP}}\in\Delta}(s_{O_{BS}}(SPF_2(d_{O_{SP}}))\cdot SP)$

The state transitions of the Server described by APTC are as follows.

$S=\sum_{d_{I_S}\in\Delta}(r_{I}(d_{I})\cdot S_{2})$

$S_{2}=SF\cdot S_{3}$

$S_{3}=\sum_{d_{O_S}\in\Delta}(s_{O_{PS}}(d_{O_S})\cdot S)$

The sending action and the reading action of the same data through the same channel can communicate with each other, otherwise, will cause a deadlock $\delta$. We define the following
communication functions of the broker $i$ for $1\leq i\leq n$.

$$\gamma(r_{I_{BB_i}}(d_{I_{B_i}}),s_{I_{BB_{i}}}(BF_{i-1_1}(d_{I_{B_{i-1}}})))\triangleq c_{I_{BB_i}}(d_{I_{B_i}})$$
$$\gamma(r_{O_{BB_{i+1}}}(d_{O_{B_i}}),s_{O_{BB_{i+1}}}(BF_{i+1_2}(d_{O_{B_{i+1}}})))\triangleq c_{O_{BB_{i+1}}}(d_{O_{B_i}})$$

There are two communication functions between the Client and the Client-side Proxy as follows.

$$\gamma(r_{I_{CP}}(d_{I_{CP}}),s_{I_{CP}}(CF_1(d_I)))\triangleq c_{I_{CP}}(d_{I_{CP}})$$
$$\gamma(r_{O_{CP}}(d_{O_{C}}),s_{O_{CP}}(CPF_2(d_{O_{CP}})))\triangleq c_{O_{CP}}(d_{O_{C}})$$

There are two communication functions between the broke $1$ and the Client-side Proxy as follows.

$$\gamma(r_{I_{BB_1}}(d_{I_{B_1}}),s_{I_{CB}}(CPF_1(d_{I_{CP}})))\triangleq c_{I_{BB_1}}(d_{I_{B_1}})$$
$$\gamma(r_{O_{CB}}(d_{O_{CP}}),s_{O_{BB_i}}(BF_{i_2}(d_{O_{B_i}})))\triangleq c_{O_{CB}}(d_{O_{CP}})$$

There are two communication functions between the broker $n$ and the Server-side Proxy as follows.

$$\gamma(r_{I_{BS}}(d_{I_{SP}}),s_{I_{BB_{n+1}}}(BF_{n_1}(d_{I_{B_n}})))\triangleq c_{I_{BS}}(d_{I_{SP}})$$
$$\gamma(r_{O_{BB_{n+1}}}(d_{O_{B_n}}),s_{O_{BS}}(SPF_2(d_{O_{SP}})))\triangleq c_{O_{BB_{n+1}}}(d_{O_{B_n}})$$

There are two communication functions between the Server and the Server-side Proxy as follows.

$$\gamma(r_{I_{PS}}(d_{I_S}),s_{I_{PS}}(SPF_1(d_{I_{SP}})))\triangleq c_{I_{PS}}(d_{I_S})$$
$$\gamma(r_{O_{PS}}(d_{O_{SP}}),s_{O_{PS}}(d_{O_S}))\triangleq c_{O_{PS}}(d_{O_{SP}})$$

Let all modules be in parallel, then the Broker pattern $C\quad CP\quad SP\quad S\quad B_1\cdots B_i\cdots B_n$ can be presented by the following process term.

$\tau_I(\partial_H(\Theta(C\between CP\between SP\between S\between B_1\between\cdots\between B_i\between\cdots\between B_n)))=\tau_I(\partial_H(C\between CP\between SP\between S\between B_1\between\cdots\between B_i\between\cdots\between B_n))$

where $H=\{r_{I_{BB_i}}(d_{I_{B_i}}),s_{I_{BB_{i}}}(BF_{i-1_1}(d_{I_{B_{i-1}}})),r_{O_{BB_{i+1}}}(d_{O_{B_i}}),s_{O_{BB_{i+1}}}(BF_{i+1_2}(d_{O_{B_{i+1}}})),\\
r_{I_{CP}}(d_{I_{CP}}),s_{I_{CP}}(CF_1(d_I)),r_{O_{CP}}(d_{O_{C}}),s_{O_{CP}}(CPF_2(d_{O_{CP}})),r_{I_{BB_1}}(d_{I_{B_1}}),s_{I_{CB}}(CPF_1(d_{I_{CP}})),\\
r_{O_{CB}}(d_{O_{CP}}),s_{O_{BB_i}}(BF_{i_2}(d_{O_{B_i}})),r_{I_{BS}}(d_{I_{SP}}),s_{I_{BB_{n+1}}}(BF_{n_1}(d_{I_{B_n}})),r_{O_{BB_{n+1}}}(d_{O_{B_n}}),s_{O_{BS}}(SPF_2(d_{O_{SP}})),\\
r_{I_{PS}}(d_{I_S}),s_{I_{PS}}(SPF_1(d_{I_{SP}})),r_{O_{PS}}(d_{O_{SP}}),s_{O_{PS}}(d_{O_S})\\
|d_I,d_{I_{CP}},d_{O_{CP}},d_{I_{SP}},d_{O_{SP}},d_{I_S},d_{O_S},d_{O_C},d_{I_{CP}},d_{O_{CP}},d_O,d_{I_{B_1}},\cdots,d_{I_{B_i}},\cdots,d_{I_{B_n}},d_{O_{B_1}},\cdots,d_{O_{B_i}},\cdots,d_{O_{B_n}}\in\Delta\}$,

$I=\{c_{I_{BB_1}}(d_{I_{B_1}}),c_{O_{BB_{2}}}(d_{O_{B_1}}),\cdots,c_{I_{BB_i}}(d_{I_{B_i}}),c_{O_{BB_{i+1}}}(d_{O_{B_i}}),\cdots,c_{I_{BB_n}}(d_{I_{B_n}}),c_{O_{BB_{n+1}}}(d_{O_{B_n}}),\\
c_{I_{CP}}(d_{I_{CP}}),c_{O_{CP}}(d_{O_{C}}),c_{I_{BB_1}}(d_{I_{B_1}}),c_{O_{CB}}(d_{O_{CP}}),c_{I_{BS}}(d_{I_{SP}}),c_{O_{BB_{n+1}}}(d_{O_{B_n}}),c_{I_{PS}}(d_{I_S}),c_{O_{PS}}(d_{O_{SP}}),\\
CF_1,CF_2,CPF_1,CPF_2,BF_{1_1},BF_{1_2},\cdots,BF_{i_1},BF_{i_2},\cdots,BF_{n_1},BF_{n_2},SPF_1,SPF_2,SF\\
|d_I,d_{I_{CP}},d_{O_{CP}},d_{I_{SP}},d_{O_{SP}},d_{I_S},d_{O_S},d_{O_C},d_{I_{CP}},d_{O_{CP}},d_O,d_{I_{B_1}},\cdots,d_{I_{B_i}},\cdots,d_{I_{B_n}},d_{O_{B_1}},\cdots,d_{O_{B_i}},\cdots,d_{O_{B_n}}\in\Delta\}$.

Then we get the following conclusion on the Broker pattern.

\begin{theorem}[Correctness of the Broker pattern]
The Broker pattern $\tau_I(\partial_H(C\between CP\between SP\between S\between B_1\between\cdots\between B_i\between\cdots\between B_n))$ can exhibit desired external behaviors.
\end{theorem}

\begin{proof}
Based on the above state transitions of the above modules, by use of the algebraic laws of APTC, we can prove that

$\tau_I(\partial_H(C\between CP\between SP\between S\between B_1\between\cdots\between B_i\between\cdots\between B_n))=\sum_{d_{I},d_{O}\in\Delta}(r_{I}(d_{I})\cdot s_{O}(d_O))\cdot
\tau_I(\partial_H(C\between CP\between SP\between S\between B_1\between\cdots\between B_i\between\cdots\between B_n))$,

that is, the Broker pattern $\tau_I(\partial_H(C\between CP\between SP\between S\between B_1\between\cdots\between B_i\between\cdots\between B_n))$ can exhibit desired external behaviors.

For the details of proof, please refer to section \ref{app}, and we omit it.
\end{proof}

\subsection{Interactive Systems}\label{IS3}

In this subsection, we verify interactive systems oriented patterns, including the Model-View-Controller (MVC) pattern and the Presentation-Abstraction-Control (PAC) pattern.

\subsubsection{Verification of the MVC Pattern}

The MVC pattern is used to model the interactive systems, which has three components: the Model, the Views and the Controller. The Model is used to contain the data and encapsulate the
core functionalities; the Views is to show the computational results to the user; and the Controller interacts between the system and user, accepts the instructions and controls the Model
and the Views. The Controller receives the instructions from the user through the channel $I$, then it sends the instructions to the Model through the channel $CM$ and to the View $i$
through the channel $CV_i$ for $1\leq i\leq n$; the model receives the instructions from the Controller, updates the data and computes the results, and sends the results to the View $i$
through the channel $MV_i$ for $1\leq i\leq n$; When the View $i$ receives the results from the Model, it generates or updates the view to the user. As illustrates in Figure \ref{MVC3}.

\begin{figure}
    \centering
    \includegraphics{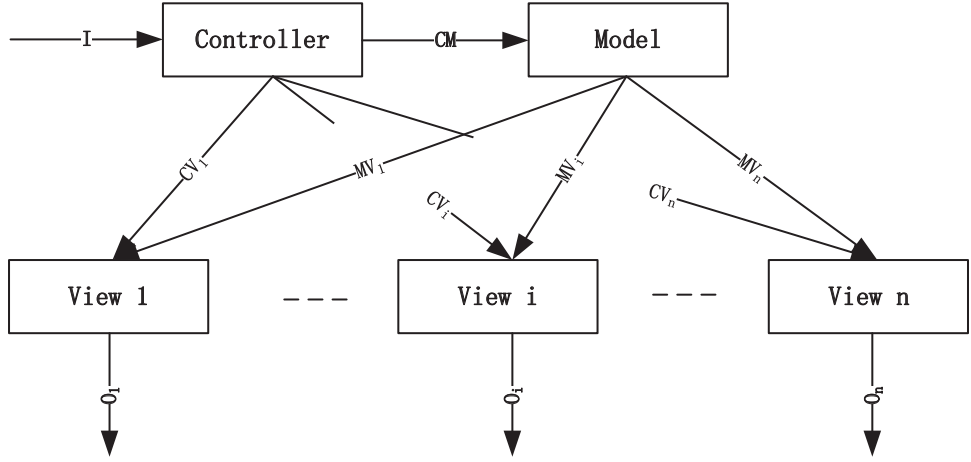}
    \caption{MVC pattern}
    \label{MVC3}
\end{figure}

The typical process of the MVC pattern is shown in Figure \ref{MVC3P} and following.

\begin{enumerate}
  \item The Controller receives the instructions $d_I$ from the user through the channel $I$ (the corresponding reading action is denoted $r_I(D_I)$), processes the instructions through
  a processing function $CF$, and generates the instructions to the Model $d_{I_M}$ and those to the View $i$ $d_{I_{V_i}}$ for $1\leq i\leq n$; it sends $d_{I_M}$ to the Model through the channel $CM$
  (the corresponding sending action is denoted $s_{CM}(d_{I_M})$) and sends $d_{I_{V_i}}$ to the View $i$ through the channel $CV_i$ (the corresponding sending action is denoted
  $s_{CV_i}(d_{I_{V_i}})$);
  \item The Model receives the instructions from the Controller through the channel $CM$ (the corresponding reading action is denoted $r_{CM}(d_{I_M})$), processes the instructions through
  a processing function $MF$, generates the computational results to the View $i$ (for $1\leq i\leq n$) which is denoted $d_{O_{M_i}}$; then sends the results to the View $i$ through the
  channel $MV_i$ (the corresponding sending action is denoted $s_{MV_i}(d_{O_{M_i}})$);
  \item The View $i$ (for $1\leq i\leq n$) receives the instructions from the Controller through the channel $CV_i$ (the corresponding reading action is denoted $r_{CV_i}(d_{I_{V_i}})$),
  processes the instructions through a processing function $VF_{i_1}$ to make ready to receive the computational results from the Model; then it receives the computational results from the Model
  through the channel $MV_i$ (the corresponding reading action is denoted $r_{MV_i}(d_{O_{M_i}})$), processes the results through a processing function $VF_{i_2}$, generates the output
  $d_{O_i}$, then sending the output through the channel $O_i$ (the corresponding sending action is denoted $s_{O_i}(d_{O_i})$).
\end{enumerate}

\begin{figure}
    \centering
    \includegraphics{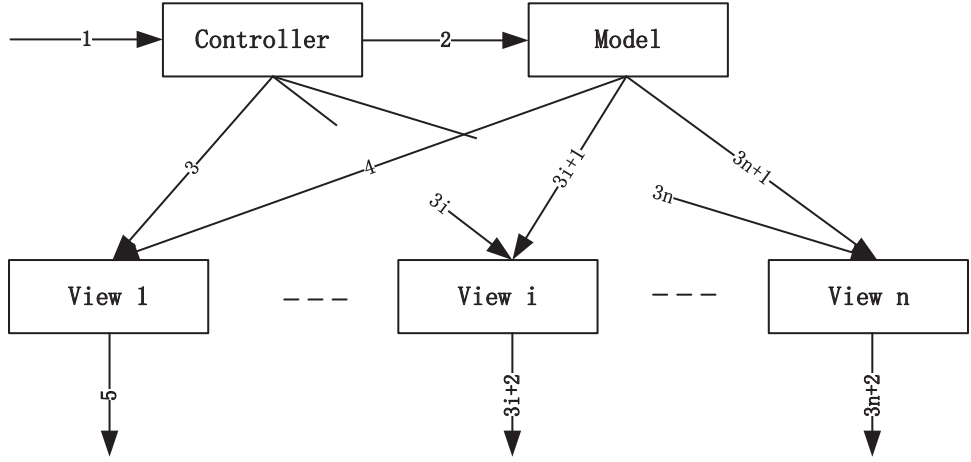}
    \caption{Typical process of MVC pattern}
    \label{MVC3P}
\end{figure}

In the following, we verify the MVC pattern. We assume all data elements $d_{I}$, $d_{I_{M}}$, $d_{I_{V_i}}$, $d_{O_{M_i}}$, $d_{O_{i}}$ (for $1\leq i\leq n$) are from a finite set
$\Delta$.

The state transitions of the Controller module
described by APTC are as follows.

$C=\sum_{d_{I}\in\Delta}(r_{I}(d_{I})\cdot C_{2})$

$C_{2}=CF\cdot C_{3}$

$C_{3}=\sum_{d_{I_M}\in\Delta}(s_{CM}(d_{I_M})\cdot C_{4})$

$C_{4}=\sum_{d_{I_{V_1},\cdots,d_{I_{V_n}}}\in\Delta}(s_{CV_1}(d_{I_{V_1}})\between\cdots\between s_{CV_n}(d_{I_{V_n}})\cdot C)$

The state transitions of the Model described by APTC are as follows.

$M=\sum_{d_{I_{M}}\in\Delta}(r_{CM}(d_{I_{M}})\cdot M_{2})$

$M_{2}=MF\cdot M_{3}$

$M_{3}=\sum_{d_{O_{M_1}},\cdots,d_{O_{M_n}}\in\Delta}(s_{MV_1}(d_{O_{M_1}})\between\cdots\between s_{MV_n}(d_{O_{M_n}})\cdot M)$

The state transitions of the View $i$ described by APTC are as follows.

$V_i=\sum_{d_{I_{V_i}}\in\Delta}(r_{CV_i}(d_{I_{V_i}})\cdot V_{i_2})$

$V_{i_2}=VF_{i_1}\cdot V_{i_3}$

$V_{i_3}=\sum_{d_{O_{M_i}}\in\Delta}(r_{MV_i}(d_{O_{M_i}})\cdot V_{i_4})$

$V_{i_4}=VF_{i_2}\cdot V_{i_5}$

$V_{i_5}=\sum_{d_{O_{i}}\in\Delta}(s_{O_{i}}(d_{O_i})\cdot V_i)$

The sending action and the reading action of the same data through the same channel can communicate with each other, otherwise, will cause a deadlock $\delta$. We define the following
communication functions of the View $i$ for $1\leq i\leq n$.

$$\gamma(r_{CV_i}(d_{I_{V_i}}),s_{CV_i}(d_{I_{V_i}}))\triangleq c_{CV_i}(d_{I_{V_i}})$$
$$\gamma(r_{MV_i}(d_{O_{M_i}}),s_{MV_i}(d_{O_{M_i}}))\triangleq c_{MV_i}(d_{O_{M_i}})$$

There are one communication functions between the Controller and the Model as follows.

$$\gamma(r_{CM}(d_{I_M}),s_{CM}(d_{I_M}))\triangleq c_{CM}(d_{I_M})$$

Let all modules be in parallel, then the MVC pattern $C\quad M\quad V_1\cdots V_i\cdots V_n$ can be presented by the following process term.

$\tau_I(\partial_H(\Theta(C\between M\between V_1\between\cdots\between V_i\between\cdots\between V_n)))=\tau_I(\partial_H(C\between M\between V_1\between\cdots\between V_i\between\cdots\between V_n))$

where $H=\{r_{CV_i}(d_{I_{V_i}}),s_{CV_i}(d_{I_{V_i}}),r_{MV_i}(d_{O_{M_i}}),s_{MV_i}(d_{O_{M_i}}),r_{CM}(d_{I_M}),s_{CM}(d_{I_M})\\
|d_{I}, d_{I_{M}}, d_{I_{V_i}}, d_{O_{M_i}}, d_{O_{i}}\in\Delta\}$ for $1\leq i\leq n$,

$I=\{c_{CV_i}(d_{I_{V_i}}),c_{MV_i}(d_{O_{M_i}}),c_{CM}(d_{I_M}),CF,MF,VF_{1_1},VF_{1_2},\cdots,VF_{n_1},VF_{n_2}\\
|d_{I}, d_{I_{M}}, d_{I_{V_i}}, d_{O_{M_i}}, d_{O_{i}}\in\Delta\}$ for $1\leq i\leq n$.

Then we get the following conclusion on the MVC pattern.

\begin{theorem}[Correctness of the MVC pattern]
The MVC pattern $\tau_I(\partial_H(C\between M\between V_1\between\cdots\between V_i\between\cdots\between V_n))$ can exhibit desired external behaviors.
\end{theorem}

\begin{proof}
Based on the above state transitions of the above modules, by use of the algebraic laws of APTC, we can prove that

$\tau_I(\partial_H(C\between M\between V_1\between\cdots\between V_i\between\cdots\between V_n))=\sum_{d_{I},d_{O_1},\cdots,d_{O_n}\in\Delta}(r_{I}(d_{I})\cdot s_{O_1}(d_{O_1})\parallel\cdots\parallel s_{O_i}(d_{O_i})\parallel\cdots\parallel s_{O_n}(d_{O_n}))\cdot
\tau_I(\partial_H(C\between M\between V_1\between\cdots\between V_i\between\cdots\between V_n))$,

that is, the MVC pattern $\tau_I(\partial_H(C\between M\between V_1\between\cdots\between V_i\between\cdots\between V_n))$ can exhibit desired external behaviors.

For the details of proof, please refer to section \ref{app}, and we omit it.
\end{proof}

\subsubsection{Verification of the PAC Pattern}\label{PAC3}

The PAC pattern is also used to model the interactive systems, which has three components: the Abstraction, the Presentations and the Control. The Abstraction is used to contain the data and encapsulate the
core functionalities; the Presentations is to show the computational results to the user; and the Control interacts between the system and user, accepts the instructions and controls the Abstraction
and the Presentations, and also other PACs. The Control receives the instructions from the user through the channel $I$, then it sends the instructions to the Abstraction through the channel $CA$ and receives
the results through the same channel. Then the Control sends the results to the Presentation $i$
through the channel $CP_i$ for $1\leq i\leq n$, and also sends the unprocessed instructions to other PACs through the channel $O$;
When the Presentation $i$ receives the results from the Control, it generates or updates the presentation to the user. As illustrates in Figure \ref{PAC3}.

\begin{figure}
    \centering
    \includegraphics{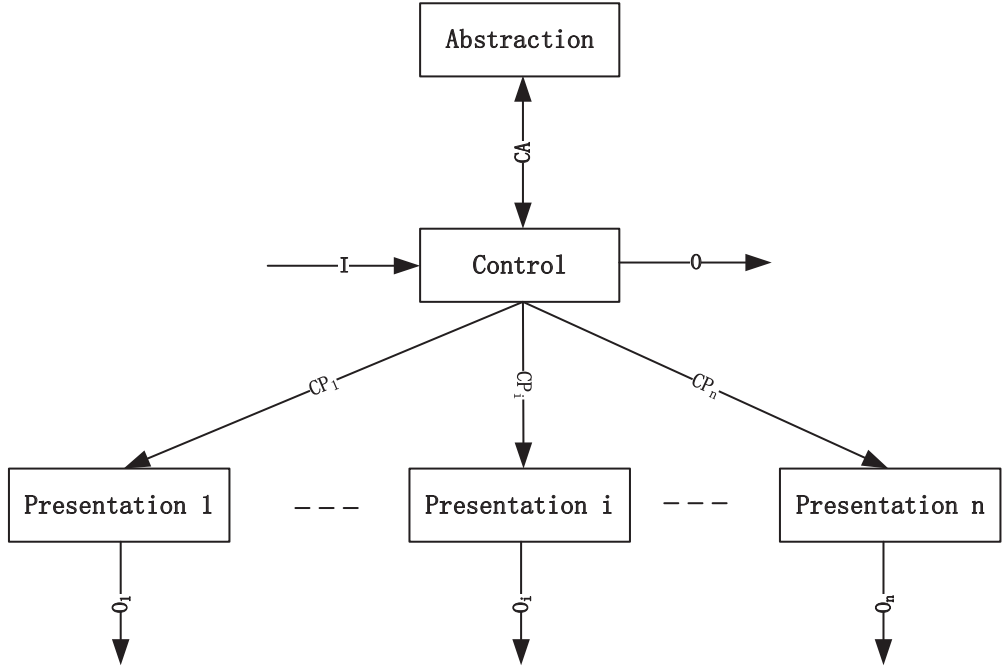}
    \caption{PAC pattern}
    \label{PAC3}
\end{figure}

The typical process of the PAC pattern is shown in Figure \ref{PAC3P} and following.

\begin{enumerate}
  \item The Control receives the instructions $d_I$ from the user through the channel $I$ (the corresponding reading action is denoted $r_I(D_I)$), processes the instructions through
  a processing function $CF_1$, and generates the instructions to the Abstraction $d_{I_A}$ and the remaining instructions $d_O$; it sends $d_{I_A}$ to the Abstraction through the channel $CA$
  (the corresponding sending action is denoted $s_{CA}(d_{I_A})$), sends $d_O$ to the other PAC through the channel $O$ (the corresponding sending action is denoted $s_O(d_O)$);
  \item The Abstraction receives the instructions from the Control through the channel $CA$ (the corresponding reading action is denoted $r_{CA}(d_{I_A})$), processes the instructions through
  a processing function $AF$, generates the computational results to Control which is denoted $d_{O_A}$, and sends the results to the Control through the channel $CA$ (the corresponding
  sending action is denoted $s_{CA}(d_{O_A})$);
  \item The Control receives the computational results from the Abstraction through channel $CA$ (the corresponding reading action is denoted $r_{CA}(d_{O_A})$), processes the results
  through a processing function $CF_2$ to generate the results to the Presentation $i$ (for $1\leq i\leq n$) which is denoted $d_{O_{C_i}}$; then sends the results to the Presentation $i$ through the
  channel $CP_i$ (the corresponding sending action is denoted $s_{CP_i}(d_{O_{C_i}})$);
  \item The Presentation $i$ (for $1\leq i\leq n$) receives the computational results from the Control
  through the channel $CP_i$ (the corresponding reading action is denoted $r_{CP_i}(d_{O_{C_i}})$), processes the results through a processing function $PF_{i}$, generates the output
  $d_{O_i}$, then sending the output through the channel $O_i$ (the corresponding sending action is denoted $s_{O_i}(d_{O_i})$).
\end{enumerate}

\begin{figure}
    \centering
    \includegraphics{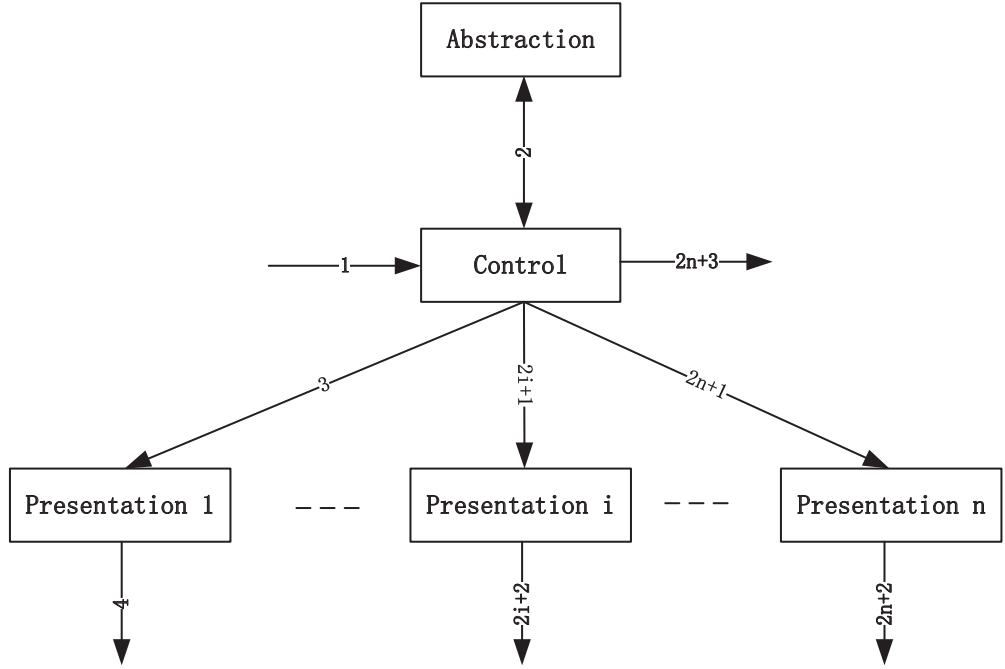}
    \caption{Typical process of PAC pattern}
    \label{PAC3P}
\end{figure}

In the following, we verify the PAC pattern. We assume all data elements $d_{I}$, $d_{I_{A}}$, $d_{O_A}$, $d_O$, $d_{O_{C_i}}$, $d_{O_{i}}$ (for $1\leq i\leq n$) are from a finite set
$\Delta$.

The state transitions of the Control module
described by APTC are as follows.

$C=\sum_{d_{I}\in\Delta}(r_{I}(d_{I})\cdot C_{2})$

$C_{2}=CF_1\cdot C_{3}$

$C_{3}=\sum_{d_{I_A}\in\Delta}(s_{CA}(d_{I_A})\cdot C_{4})$

$C_{4}=\sum_{d_{O}\in\Delta}(s_{O}(d_{O})\cdot C_{5})$

$C_{5}=\sum_{d_{O_A}\in\Delta}(r_{CA}(d_{O_A})\cdot C_{6})$

$C_{6}=CF_2\cdot C_{7}$

$C_{7}=\sum_{d_{O_{C_1}},\cdots,d_{O_{C_n}}\in\Delta}(s_{CP_1}(d_{O_{C_1}})\between\cdots\between s_{CP_n}(d_{O_{C_n}})\cdot C)$

The state transitions of the Abstraction described by APTC are as follows.

$A=\sum_{d_{I_{A}}\in\Delta}(r_{CA}(d_{I_{A}})\cdot A_{2})$

$A_{2}=AF\cdot A_{3}$

$A_{3}=\sum_{d_{O_{A}}\in\Delta}(s_{CA}(d_{O_{A}})\cdot A)$

The state transitions of the Presentation $i$ described by APTC are as follows.

$P_i=\sum_{d_{O_{C_i}}\in\Delta}(r_{CP_i}(d_{O_{C_i}})\cdot P_{i_2})$

$P_{i_2}=PF_i\cdot P_{i_3}$

$P_{i_3}=\sum_{d_{O_{i}}\in\Delta}(s_{O_{i}}(d_{O_i})\cdot P_i)$

The sending action and the reading action of the same data through the same channel can communicate with each other, otherwise, will cause a deadlock $\delta$. We define the following
communication functions of the Presentation $i$ for $1\leq i\leq n$.

$$\gamma(r_{CP_i}(d_{O_{C_i}}),s_{CP_i}(d_{O_{C_i}}))\triangleq c_{CP_i}(d_{O_{C_i}})$$

There are two communication functions between the Control and the Abstraction as follows.

$$\gamma(r_{CA}(d_{I_A}),s_{CA}(d_{I_A}))\triangleq c_{CA}(d_{I_A})$$

$$\gamma(r_{CA}(d_{O_A}),s_{CA}(d_{O_A}))\triangleq c_{CA}(d_{O_A})$$

Let all modules be in parallel, then the PAC pattern $C\quad A \quad P_1\cdots P_i\cdots P_n$ can be presented by the following process term.

$\tau_I(\partial_H(\Theta(C\between A\between P_1\between\cdots\between P_i\between\cdots\between P_n)))=\tau_I(\partial_H(C\between A\between P_1\between\cdots\between P_i\between\cdots\between P_n))$

where $H=\{r_{CP_i}(d_{O_{C_i}}),s_{CP_i}(d_{O_{C_i}}),r_{CA}(d_{I_A}),s_{CA}(d_{I_A}),r_{CA}(d_{O_A}),s_{CA}(d_{O_A})\\
|d_{I}, d_{I_{A}}, d_{O_A}, d_O, d_{O_{C_i}}, d_{O_{i}}\in\Delta\}$ for $1\leq i\leq n$,

$I=\{c_{CA}(d_{I_A}),c_{CA}(d_{O_A}),c_{CP_1}(d_{O_{C_1}}),\cdots,c_{CP_n}(d_{O_{C_n}}),CF_1,CF_2,AF,PF_{1},\cdots,PF_{n}\\
|d_{I}, d_{I_{A}}, d_{O_A}, d_O, d_{O_{C_i}}, d_{O_{i}}\in\Delta\}$ for $1\leq i\leq n$.

Then we get the following conclusion on the PAC pattern.

\begin{theorem}[Correctness of the PAC pattern]
The PAC pattern $\tau_I(\partial_H(C\between A\between P_1\between\cdots\between P_i\between\cdots\between P_n))$ can exhibit desired external behaviors.
\end{theorem}

\begin{proof}
Based on the above state transitions of the above modules, by use of the algebraic laws of APTC, we can prove that

$\tau_I(\partial_H(C\between A\between P_1\between\cdots\between P_i\between\cdots\between P_n))=\sum_{d_{I},d_O,d_{O_1},\cdots,d_{O_n}\in\Delta}(r_{I}(d_{I})\cdot s_O(d_O)\cdot s_{O_1}(d_{O_1})\parallel\cdots\parallel s_{O_i}(d_{O_i})\parallel\cdots\parallel s_{O_n}(d_{O_n}))\cdot
\tau_I(\partial_H(C\between A\between P_1\between\cdots\between P_i\between\cdots\between P_n))$,

that is, the PAC pattern $\tau_I(\partial_H(C\between A\between P_1\between\cdots\between P_i\between\cdots\between P_n))$ can exhibit desired external behaviors.

For the details of proof, please refer to section \ref{app}, and we omit it.
\end{proof}

\subsection{Adaptable Systems}\label{AS3}

In this subsection, we verify adaptive systems oriented patterns, including the Microkernel pattern and the Reflection pattern.

\subsubsection{Verification of the Microkernel Pattern}

The Microkernel pattern adapts the changing of the requirements by implementing the unchangeable requirements as a minimal functional kernel and changeable requirements as external
functionalities. There are five modules in the Microkernel pattern: the Microkernel, the Internal Server, the External Server, the Adapter and the Client. The Client interacts with the user through
the channels $I$ and $O$; The Adapter interacts with the Microkernel through the channels $I_{AM}$ and $O_{AM}$, and with the External Server through the channels $I_{AE}$ and $O_{AE}$;
The Microkernel interacts with the Internal Server through the channels $I_{MI}$ and $O_{MI}$, with the External Server through the channels $I_{EM}$ and $O_{EM}$. As illustrates in
Figure \ref{MK3}.

\begin{figure}
    \centering
    \includegraphics{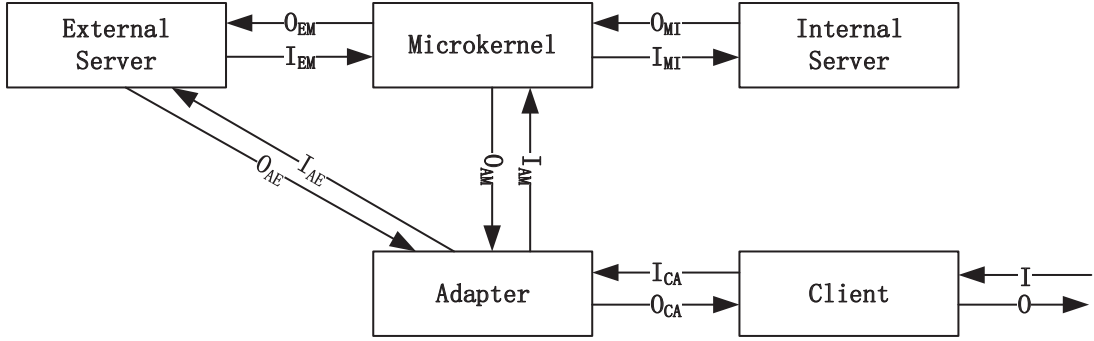}
    \caption{Microkernel pattern}
    \label{MK3}
\end{figure}

The typical process of the Microkernel pattern is shown in Figure \ref{MK3P} and as follows.

\begin{enumerate}
  \item The Client receives the request $d_I$ from the user through the channel $I$ (the corresponding reading action is denoted $r_I(d_I)$), then processes the request $d_I$ through a processing
  function $CF_1$, and sends the processed request $d_{I_C}$ to the Adapter through the channel $I_{CA}$ (the corresponding sending action is denoted $s_{I_{CA}}(d_{I_C})$);
  \item The Adapter receives $d_{I_C}$ from the Client through the channel $I_{CA}$ (the corresponding reading action is denoted $r_{I_{CA}}(d_{I_C})$), then processes the request
  through a processing function $AF_1$, generates and sends the processed request $d_{I_A}$ to the Microkernel through the channel $I_{AM}$ (the corresponding sending action is denoted
  $s_{I_{AM}}(d_{I_A})$);
  \item The Microkernel receives the request $d_{I_A}$ from the Adapter through the channel $I_{AM}$ (the corresponding reading action is denoted $r_{I_{AM}}(d_{I_A})$), then processes
  the request through a processing function $MF_1$, generates and sends the processed request $d_{I_M}$ to the Internal Server through the channel $I_{MI}$ (the corresponding sending action is denoted $s_{I_{MI}}(d_{I_M})$),
  and to the External Server through the channel $I_{EM}$ (the corresponding sending action is denoted $s_{I_{EM}}(d_{I_M})$);
  \item The Internal Server receives the request $d_{I_M}$ from the Microkernel through the channel $I_{MI}$ (the corresponding reading action is denoted $r_{I_{MI}}(d_{I_M})$), then
  processes the request and generates the response $d_{O_I}$ through a processing function $IF$, and sends the response to the Microkernel through the channel $O_{MI}$ (the corresponding sending action is denoted
  $s_{O_{MI}}(d_{O_I})$);
  \item The External Server receives the request $d_{I_M}$ from the Microkernel through the channel $I_{EM}$ (the corresponding reading action is denoted $r_{I_{EM}}(d_{I_M})$), then
  processes the request and generates the response $d_{O_E}$ through a processing function $EF_1$, and sends the response to the Microkernel through the channel $O_{EM}$ (the corresponding sending action is denoted
  $s_{O_{EM}}(d_{O_E})$);
  \item The Microkernel receives the response $d_{O_I}$ from the Internal Server through the channel $O_{MI}$ (the corresponding reading action is denoted $r_{O_{MI}}(d_{O_I})$) and the
  response $d_{O_E}$ from the External Server through the channel $O_{EM}$ (the corresponding reading action is denoted $r_{O_{EM}}(d_{O_E})$), then processes the responses and generate
  the response $d_{O_M}$ through a processing function $MF_2$, and sends $d_{O_M}$ to the Adapter through the channel $O_{AM}$ (the corresponding sending action is denoted $s_{O_{AM}}(d_{O_M})$);
  \item The Adapter receives the response $d_{O_M}$ from the Microkernel through the channel $O_{AM}$ (the corresponding reading action is denoted $r_{O_{AM}}(d_{O_M})$),
  it may send $d_{I_{A'}}$ to the External Server through the channel $I_{AE}$ (the corresponding sending action is denoted $s_{I_{AE}}(d_{I_{A'}})$);
  \item The External Server receives the request $d_{I_{A'}}$ from the Adapter through the channel $I_{AE}$ (the corresponding reading action is denoted $r_{I_{AE}}(d_{I_{A'}})$),
  then processes the request and generate the response $d_{O_{E'}}$ through a processing function $EF_2$, and sends $d_{O_{E'}}$ to the Adapter through the channel $O_{AE}$ (the
  corresponding sending action is denoted $s_{O_{AE}}(d_{O_{E'}})$);
  \item The Adapter receives the response from the External Server through the channel $O_{AE}$ (the corresponding reading action is denoted $r_{O_AE}(d_{O_{E'}})$), then processes
  $d_{O_M}$ and $d_{O_{E'}}$ through a processing function $AF_2$ and generates the response $d_{O_A}$, and sends $d_{O_A}$ to the Client through the channel $O_{CA}$ (the corresponding sending action
  is denoted $s_{O_{CA}}(d_{O_A})$);
  \item The Client receives the response $d_{O_A}$ from the Adapter through the channel $O_{CA}$ (the corresponding reading action is denoted $r_{O_{CA}}(d_{O_A})$), then processes
  $d_{O_A}$ through a processing function $CF_2$ and generate the response $d_O$, and sends $d_O$ to the user through the channel $O$ (the corresponding sending action is denoted
  $s_O(d_O)$).
\end{enumerate}

\begin{figure}
    \centering
    \includegraphics{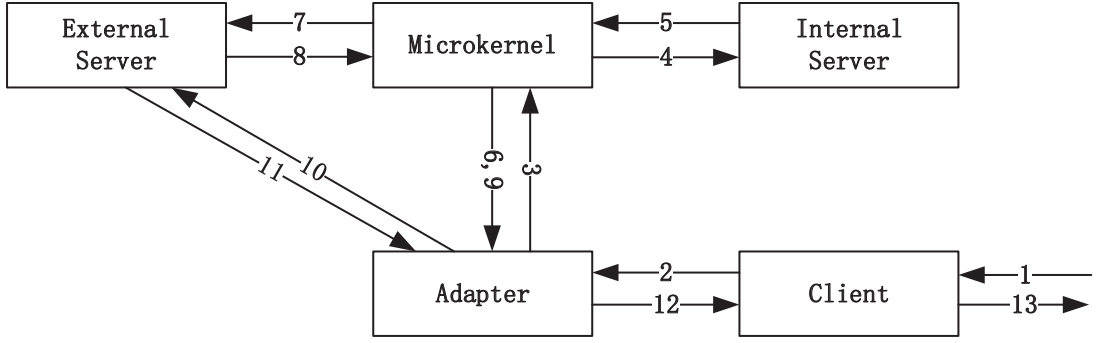}
    \caption{Typical process of Microkernel pattern}
    \label{MK3P}
\end{figure}

In the following, we verify the Microkernel pattern. We assume all data elements $d_{I}$, $d_{I_C}$, $d_{I_{A}}$, $d_{I_{A'}}$, $d_{I_M}$, $d_{O_I}$, $d_{O_E}$, $d_{O_{E'}}$, $d_{O_M}$, $d_O$, $d_{O_{A}}$ (for $1\leq i\leq n$) are from a finite set
$\Delta$.

The state transitions of the Client module
described by APTC are as follows.

$C=\sum_{d_{I}\in\Delta}(r_{I}(d_{I})\cdot C_{2})$

$C_{2}=CF_1\cdot C_{3}$

$C_{3}=\sum_{d_{I_C}\in\Delta}(s_{I_{CA}}(d_{I_C})\cdot C_{4})$

$C_{4}=\sum_{d_{O_A}\in\Delta}(r_{O_{CA}}(d_{O_A})\cdot C_{5})$

$C_{5}=CF_2\cdot C_{6}$

$C_{6}=\sum_{d_{O}\in\Delta}(s_{O}(d_{O})\cdot C)$

The state transitions of the Adapter module
described by APTC are as follows.

$A=\sum_{d_{I_C}\in\Delta}(r_{I_{CA}}(d_{I_C})\cdot A_{2})$

$A_{2}=AF_1\cdot A_{3}$

$A_{3}=\sum_{d_{I_A}\in\Delta}(s_{I_{AM}}(d_{I_A})\cdot A_{4})$

$A_{4}=\sum_{d_{O_M}\in\Delta}(r_{O_{AM}}(d_{O_M})\cdot A_{5})$

$A_{5}=\sum_{d_{I_{A'}}\in\Delta}(s_{I_{AE}}(d_{I_{A'}})\cdot A_{6})$

$A_{6}=\sum_{d_{O_{E'}}\in\Delta}(r_{O_{AE}}(d_{O_{E'}})\cdot A_{7})$

$A_{7}=AF_2\cdot A_{8}$

$A_{8}=\sum_{d_{O_A}\in\Delta}(s_{O_{CA}}(d_{O_A})\cdot A)$

The state transitions of the Microkernel module
described by APTC are as follows.

$M=\sum_{d_{I_A}\in\Delta}(r_{I_{AM}}(d_{I_A})\cdot M_{2})$

$M_{2}=MF_1\cdot M_{3}$

$M_{3}=\sum_{d_{I_M}\in\Delta}(s_{I_{MI}}(d_{I_M})\between s_{I_{EM}}(d_{I_M})\cdot M_{4})$

$M_{4}=\sum_{d_{O_I},d_{O_E}\in\Delta}(r_{O_{MI}}(d_{O_I})\between r_{O_{EM}}(d_{O_E})\cdot M_{5})$

$M_{5}=MF_2\cdot M_{6})$

$M_{6}=\sum_{d_{O_M}\in\Delta}(s_{O_{AM}}(d_{O_M})\cdot M)$

The state transitions of the Internal Server described by APTC are as follows.

$I=\sum_{d_{I_{M}}\in\Delta}(r_{I_{MI}}(d_{I_{M}})\cdot I_{2})$

$I_{2}=IF\cdot I_{3}$

$I_{3}=\sum_{d_{O_{I}}\in\Delta}(s_{O_{MI}}(d_{O_{I}})\cdot I)$

The state transitions of the External Server described by APTC are as follows.

$E=\sum_{d_{I_{M}}\in\Delta}(r_{I_{EM}}(d_{I_{M}})\cdot E_{2})$

$E_{2}=EF_1\cdot E_{3}$

$E_{3}=\sum_{d_{O_{E}}\in\Delta}(s_{O_{EM}}(d_{O_{E}})\cdot E_4)$

$E_4=\sum_{d_{I_{A'}}\in\Delta}(r_{I_{AE}}(d_{I_{A'}})\cdot E_{5})$

$E_{5}=EF_2\cdot E_{6}$

$E_{6}=\sum_{d_{O_{E'}}\in\Delta}(s_{O_{AE}}(d_{O_{E'}})\cdot E)$

The sending action and the reading action of the same data through the same channel can communicate with each other, otherwise, will cause a deadlock $\delta$. We define the following
communication functions between the Client and the Adapter.

$$\gamma(r_{I_{CA}}(d_{I_C}),s_{I_{CA}}(d_{I_C}))\triangleq c_{I_{CA}}(d_{I_C})$$

$$\gamma(r_{O_{CA}}(d_{O_A}),s_{O_{CA}}(d_{O_A}))\triangleq c_{O_{CA}}(d_{O_A})$$

There are two communication functions between the Adapter and the Microkernel as follows.

$$\gamma(r_{I_{AM}}(d_{I_A}),s_{I_{AM}}(d_{I_A}))\triangleq c_{I_{AM}}(d_{I_A})$$

$$\gamma(r_{O_{AM}}(d_{O_M}),s_{O_{AM}}(d_{O_M}))\triangleq c_{O_{AM}}(d_{O_M})$$

There are two communication functions between the Adapter and the External Server as follows.

$$\gamma(r_{I_{AE}}(d_{I_{A'}}),s_{I_{AE}}(d_{I_{A'}}))\triangleq c_{I_{AE}}(d_{I_{A'}})$$

$$\gamma(r_{O_{AE}}(d_{O_{E'}}),s_{O_{AE}}(d_{O_{E'}}))\triangleq c_{O_{AE}}(d_{O_{E'}})$$

There are two communication functions between the Internal Server and the Microkernel as follows.

$$\gamma(r_{I_{MI}}(d_{I_{M}}),s_{I{MI}}(d_{I_{M}}))\triangleq c_{I{MI}}(d_{I_{M}})$$

$$\gamma(r_{O_{MI}}(d_{O_{I}}),s_{O_{MI}}(d_{O_{I}}))\triangleq c_{O_{MI}}(d_{O_{I}})$$

There are two communication functions between the External Server and the Microkernel as follows.

$$\gamma(r_{I_{EM}}(d_{I_{M}}),s_{I_{EM}}(d_{I_{M}}))\triangleq c_{I_{EM}}(d_{I_{M}})$$

$$\gamma(r_{O_{EM}}(d_{O_{E}}),s_{O_{EM}}(d_{O_{E}}))\triangleq c_{O_{EM}}(d_{O_{E}})$$

Let all modules be in parallel, then the Microkernel pattern $C\quad A \quad M\quad I\quad E$ can be presented by the following process term.

$\tau_I(\partial_H(\Theta(C\between A\between M\between I\between E)))=\tau_I(\partial_H(C\between A\between M\between I\between E))$

where $H=\{r_{I_{CA}}(d_{I_C}),s_{I_{CA}}(d_{I_C}),r_{O_{CA}}(d_{O_A}),s_{O_{CA}}(d_{O_A}),r_{I_{AM}}(d_{I_A}),s_{I_{AM}}(d_{I_A}),\\
r_{O_{AM}}(d_{O_M}),s_{O_{AM}}(d_{O_M}),r_{I_{AE}}(d_{I_{A'}}),s_{I_{AE}}(d_{I_{A'}}),r_{O_{AE}}(d_{O_{E'}}),s_{O_{AE}}(d_{O_{E'}}),\\
r_{I_{MI}}(d_{I_{M}}),s_{I{MI}}(d_{I_{M}}),r_{O_{MI}}(d_{O_{I}}),s_{O_{MI}}(d_{O_{I}}),r_{I_{EM}}(d_{I_{M}}),s_{I_{EM}}(d_{I_{M}}),\\
r_{O_{EM}}(d_{O_{E}}),s_{O_{EM}}(d_{O_{E}})
|d_{I}, d_{I_C}, d_{I_{A}}, d_{I_{A'}}, d_{I_M}, d_{O_I}, d_{O_E}, d_{O_{E'}}, d_{O_M}, d_O, d_{O_{A}}\in\Delta\}$,

$I=\{c_{I_{CA}}(d_{I_C}),c_{O_{CA}}(d_{O_A}),c_{I_{AM}}(d_{I_A}),c_{O_{AM}}(d_{O_M}),c_{I_{AE}}(d_{I_{A'}}),c_{O_{AE}}(d_{O_{E'}}),\\
c_{I{MI}}(d_{I_{M}}),c_{O_{MI}}(d_{O_{I}}),c_{I_{EM}}(d_{I_{M}}),c_{O_{EM}}(d_{O_{E}}),CF_1,CF_2,AF_1,AF_2,MF_1,MF_2,IF,EF_1,EF_2\\
|d_{I}, d_{I_C}, d_{I_{A}}, d_{I_{A'}}, d_{I_M}, d_{O_I}, d_{O_E}, d_{O_{E'}}, d_{O_M}, d_O, d_{O_{A}}\in\Delta\}$.

Then we get the following conclusion on the Microkernel pattern.

\begin{theorem}[Correctness of the Microkernel pattern]
The Microkernel pattern $\tau_I(\partial_H(C\between A\between M\between I\between E))$ can exhibit desired external behaviors.
\end{theorem}

\begin{proof}
Based on the above state transitions of the above modules, by use of the algebraic laws of APTC, we can prove that

$\tau_I(\partial_H(C\between A\between M\between I\between E))=\sum_{d_{I},d_O\in\Delta}(r_{I}(d_{I})\cdot s_O(d_O))\cdot
\tau_I(\partial_H(C\between A\between M\between I\between E))$,

that is, the Microkernel pattern $\tau_I(\partial_H(C\between A\between M\between I\between E))$ can exhibit desired external behaviors.

For the details of proof, please refer to section \ref{app}, and we omit it.
\end{proof}

\subsubsection{Verification of the Reflection Pattern}

The Reflection pattern makes the system be able to change its structure and behaviors dynamically. There are two levels in the Reflection pattern: one is the meta level to encapsulate
the information of system properties and make the system self-aware; the other is the base level to implement the concrete application logic. The meta level modules include the Metaobject
Protocol and $n$ Metaobject. The Metaobject Protocol is used to configure the Metaobjects, and it interacts with Metaobject $i$ through the channels $I_{MP_i}$ and $O_{MP_i}$, and it exchanges
configuration information with outside through the input channel $I_M$ and $O_M$. The Metaobject encapsulate the system properties, and it interacts with the Metaobject Protocol, and with
the Component through the channels $I_{MC_i}$ and $O_{MC_i}$. The base level modules including concrete Components, which interact with the Metaobject, and with outside through the input
channel $I_C$ and $O_C$. As illustrates in Figure \ref{RE3}.

\begin{figure}
    \centering
    \includegraphics{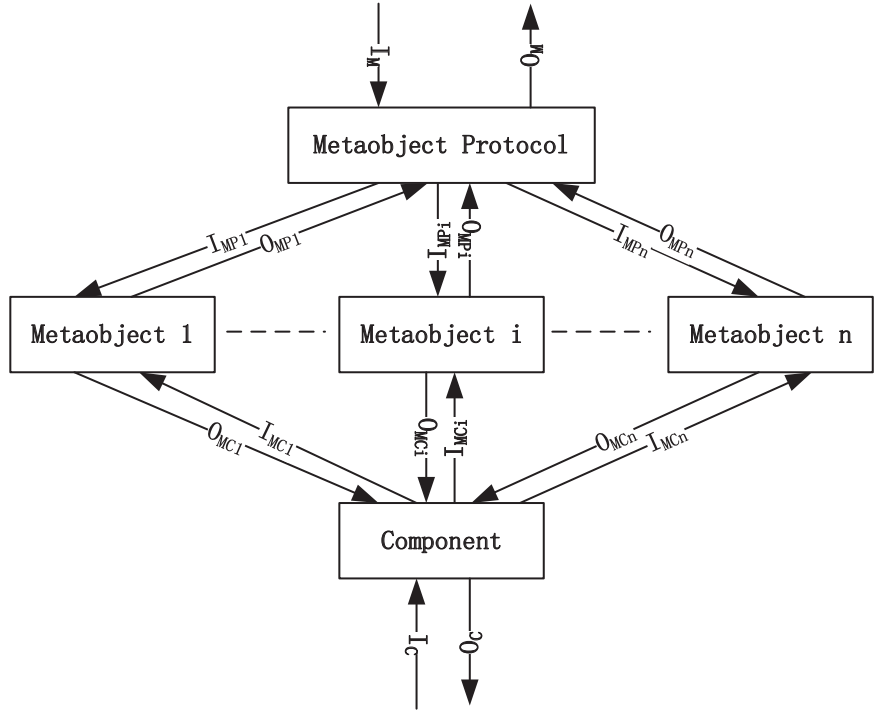}
    \caption{Reflection pattern}
    \label{RE3}
\end{figure}

The typical process of the Reflection pattern is shown in Figure \ref{RE3P} and as follows.

\begin{enumerate}
  \item The Metaobject Protocol receives the configuration information from the user through the channel $I_M$ (the corresponding reading action is denoted $r_{I_M}(d_{I_M})$), then
  processes the information through a processing function $PF_1$ and generates the configuration $d_{I_P}$, and sends $d_{I_P}$ to the Metaobject $i$ (for $1\leq i\leq n$) through the
  channel $I_{MP_i}$ (the corresponding sending action is denoted $s_{I_{MP_i}}(d_{I_P})$);
  \item The Metaobject $i$ receives the configuration $d_{I_P}$ from the Metaobject Protocol through the channel $I_{MP_i}$ (the corresponding reading action is denoted $r_{I_{MP_i}}(d_{I_P})$),
  then configures the properties through a configuration function $MF_{i1}$, and sends the configuration results $d_{O_{M_{i1}}}$ to the Metaobject Protocol through the channel $O_{MP_i}$
  (the corresponding sending action is denoted $s_{O_{MP_i}}(d_{O_{M_{i1}}})$);
  \item The Metaobject Protocol receives the configuration results from the Metaobject $i$ through the channel $O_{MP_i}$ (the corresponding reading action is denoted $r_{O_{MP_i}}(d_{O_{M_{i1}}})$), then
  processes the results through a processing function $PF_2$ and generates the result $d_{O_M}$, and sends $d_{O_M}$ to the outside through the
  channel $O_{M}$ (the corresponding sending action is denoted $s_{O_M}(d_{O_M})$);
  \item The Component receives the invacation from the user through the channel $I_C$ (the corresponding reading action is denoted $r_{I_C}(d_{I_C})$), then
  processes the invocation through a processing function $CF_1$ and generates the configuration $d_{I_C}$, and sends $d_{I_C}$ to the Metaobject $i$ (for $1\leq i\leq n$) through the
  channel $I_{MC_i}$ (the corresponding sending action is denoted $s_{I_{MC_i}}(d_{I_C})$);
  \item The Metaobject $i$ receives the invocation $d_{I_C}$ from the Component through the channel $I_{MC_i}$ (the corresponding reading action is denoted $r_{I_{MC_i}}(d_{I_C})$),
  then computes through a computational function $MF_{i2}$, and sends the computational results $d_{O_{M_{i2}}}$ to the Component through the channel $O_{MC_i}$
  (the corresponding sending action is denoted $s_{O_{MC_i}}(d_{O_{M_{i2}}})$);
  \item The Component receives the computational results from the Metaobject $i$ through the channel $O_{MC_i}$ (the corresponding reading action is denoted $r_{O_{MC_i}}(d_{O_{M_{i2}}})$), then
  processes the results through a processing function $CF_2$ and generates the result $d_{O_C}$, and sends $d_{O_C}$ to the outside through the
  channel $O_{C}$ (the corresponding sending action is denoted $s_{O_C}(d_{O_C})$).
\end{enumerate}

\begin{figure}
    \centering
    \includegraphics{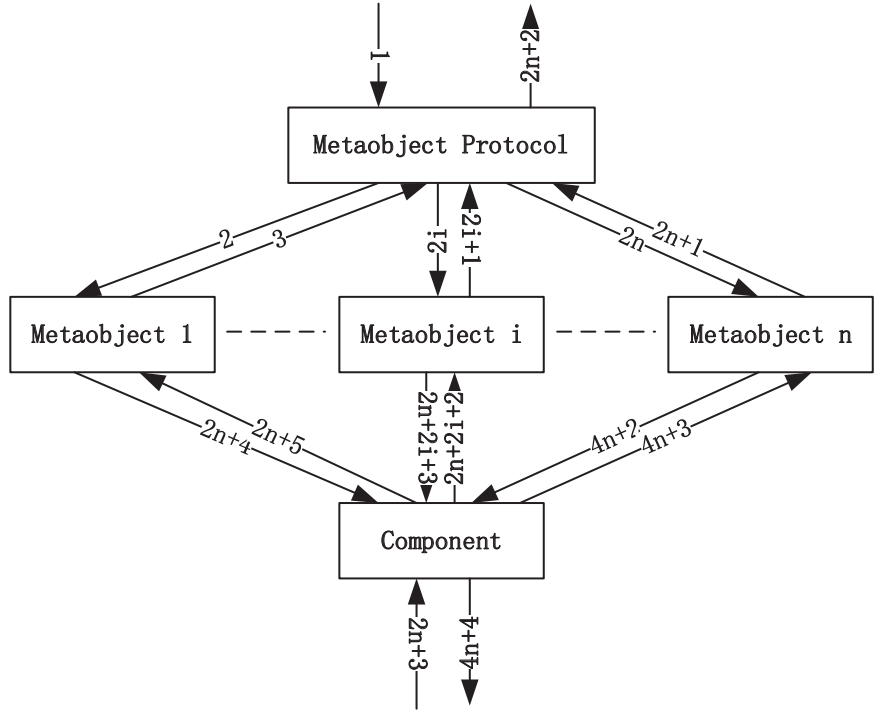}
    \caption{Typical process of Reflection pattern}
    \label{RE3P}
\end{figure}

In the following, we verify the Reflection pattern. We assume all data elements $d_{I_M}$, $d_{I_{C}}$, $d_{I_P}$, $d_{O_{M_{i1}}}$, $d_{O_{M_{i2}}}$, $d_{O_M}$, $d_{O_{C}}$ (for $1\leq i\leq n$) are from a finite set
$\Delta$.

The state transitions of the Metaobject Protocol module
described by APTC are as follows.

$P=\sum_{d_{I_M}\in\Delta}(r_{I_M}(d_{I_M})\cdot P_{2})$

$P_{2}=PF_1\cdot P_{3}$

$P_{3}=\sum_{d_{I_P}\in\Delta}(s_{I_{MP_1}}(d_{I_P})\between\cdots\between s_{I_{MP_n}}(d_{I_P})\cdot P_{4})$

$P_{4}=\sum_{d_{O_{M_{11}}},\cdots,d_{O_{M_{n1}}}\in\Delta}(r_{O_{MP_1}}(d_{O_{M_{11}}})\between\cdots\between r_{O_{MP_n}}(d_{O_{M_{n1}}})\cdot P_{5})$

$P_5=PF_2\cdot P_6$

$P_{6}=\sum_{d_{O_M}\in\Delta}(s_{O_M}(d_{O_M})\cdot P)$

The state transitions of the Component described by APTC are as follows.

$C=\sum_{d_{I_C}\in\Delta}(r_{I_C}(d_{I_C})\cdot C_{2})$

$C_{2}=CF_1\cdot C_{3}$

$C_{3}=\sum_{d_{I_P}\in\Delta}(s_{I_{MC_1}}(d_{I_C})\between\cdots\between s_{I_{MC_n}}(d_{I_C})\cdot C_{4})$

$C_{4}=\sum_{d_{O_{M_{12}}},\cdots,d_{O_{M_{n2}}}\in\Delta}(r_{O_{MC_1}}(d_{O_{M_{12}}})\between\cdots\between r_{O_{MC_n}}(d_{O_{M_{n2}}})\cdot C_{5})$

$C_5=CF_2\cdot C_6$

$C_{6}=\sum_{d_{O_C}\in\Delta}(s_{O_C}(d_{O_C})\cdot C)$

The state transitions of the Metaobject $i$ described by APTC are as follows.

$M_i=\sum_{d_{I_P}\in\Delta}(r_{I_{MP_i}}(d_{I_P})\cdot M_{i_2})$

$M_{i_2}=MF_{i1}\cdot M_{i_3}$

$M_{i_3}=\sum_{d_{O_{M_{i1}}}\in\Delta}(s_{O_{MP_{i}}}(d_{O_{M_{i1}}})\cdot M_{i_4})$

$M_{i_4}=\sum_{d_{I_C}\in\Delta}(r_{I_{MC_i}}(d_{I_C})\cdot M_{i_5})$

$M_{i_5}=MF_{i2}\cdot M_{i_6}$

$M_{i_6}=\sum_{d_{O_{M_{i2}}}\in\Delta}(s_{O_{MC_{i}}}(d_{O_{M_{i2}}})\cdot M_{i})$

The sending action and the reading action of the same data through the same channel can communicate with each other, otherwise, will cause a deadlock $\delta$. We define the following
communication functions of the Metaobject $i$ for $1\leq i\leq n$.

$$\gamma(r_{I_{MP_i}}(d_{I_P}),s_{I_{MP_i}}(d_{I_P}))\triangleq c_{I_{MP_i}}(d_{I_P})$$

$$\gamma(r_{O_{MP_{i}}}(d_{O_{M_{i1}}}),s_{O_{MP_{i}}}(d_{O_{M_{i1}}}))\triangleq c_{O_{MP_{i}}}(d_{O_{M_{i1}}})$$

$$\gamma(r_{I_{MC_i}}(d_{I_C}),s_{I_{MC_i}}(d_{I_C}))\triangleq c_{I_{MC_i}}(d_{I_C})$$

$$\gamma(r_{O_{MC_{i}}}(d_{O_{M_{i2}}}),s_{O_{MC_{i}}}(d_{O_{M_{i2}}}))\triangleq c_{O_{MC_{i}}}(d_{O_{M_{i2}}})$$

Let all modules be in parallel, then the Reflection pattern $C\quad P \quad M_1\cdots M_i\cdots M_n$ can be presented by the following process term.

$\tau_I(\partial_H(\Theta(C\between P\between M_1\between\cdots\between M_i\between\cdots\between M_n)))=\tau_I(\partial_H(C\between P\between M_1\between\cdots\between M_i\between\cdots\between M_n))$

where $H=\{r_{I_{MP_i}}(d_{I_P}),s_{I_{MP_i}}(d_{I_P})),r_{O_{MP_{i}}}(d_{O_{M_{i1}}}),s_{O_{MP_{i}}}(d_{O_{M_{i1}}}),\\
r_{I_{MC_i}}(d_{I_C}),s_{I_{MC_i}}(d_{I_C}),r_{O_{MC_{i}}}(d_{O_{M_{i2}}}),s_{O_{MC_{i}}}(d_{O_{M_{i2}}})
|d_{I_M}, d_{I_{C}}, d_{I_P}, d_{O_{M_{i1}}}, d_{O_{M_{i2}}}, d_{O_M}, d_{O_{C}}\in\Delta\}$ for $1\leq i\leq n$,

$I=\{c_{I_{MP_i}}(d_{I_P}),c_{O_{MP_{i}}}(d_{O_{M_{i1}}}),c_{I_{MC_i}}(d_{I_C}),c_{O_{MC_{i}}}(d_{O_{M_{i2}}}),PF_1,PF_2,CF_1,CF_2,MF_{i1},MF_{i2}\\
|d_{I_M}, d_{I_{C}}, d_{I_P}, d_{O_{M_{i1}}}, d_{O_{M_{i2}}}, d_{O_M}, d_{O_{C}}\in\Delta\}$ for $1\leq i\leq n$.

Then we get the following conclusion on the Reflection pattern.

\begin{theorem}[Correctness of the Reflection pattern]
The Reflection pattern $\tau_I(\partial_H(C\between P\between M_1\between\cdots\between M_i\between\cdots\between M_n))$ can exhibit desired external behaviors.
\end{theorem}

\begin{proof}
Based on the above state transitions of the above modules, by use of the algebraic laws of APTC, we can prove that

$\tau_I(\partial_H(C\between P\between M_1\between\cdots\between M_i\between\cdots\between M_n))=\sum_{d_{I_M},d_{I_C},d_{O_M},d_{O_C}\in\Delta}(r_{I_M}(d_{I_M})\parallel r_{I_C}(d_{I_C})\cdot s_{O_M}(d_{O_M})\parallel s_{O_C}(d_{O_C}))\cdot
\tau_I(\partial_H(C\between P\between M_1\between\cdots\between M_i\between\cdots\between M_n))$,

that is, the Reflection pattern $\tau_I(\partial_H(C\between P\between M_1\between\cdots\between M_i\between\cdots\between M_n))$ can exhibit desired external behaviors.

For the details of proof, please refer to section \ref{app}, and we omit it.
\end{proof}

\newpage\section{Verification of Design Patterns}

Design patterns are middle-level patterns, which are lower than the architecture patterns and higher than the programming language-specific idioms. Design patterns describe the 
architecture of the subsystems.

In this chapter, we verify the five categories of design patterns. In section \ref{SD4}, we verify the patterns related to structural decomposition. In section \ref{OW4}, we verify the 
patterns related to organization of work. We verify the patterns related to access control in section \ref{AC4} and verify management oriented patterns in section \ref{M4}. Finally, we 
verify the communication oriented patterns in section \ref{C4}.

\subsection{Structural Decomposition}\label{SD4}

In this subsection, we verify structural decomposition related patterns, including the Whole-Part pattern.

\subsubsection{Verification the Whole-Part Pattern}

The Whole-Part pattern is used to divide application logics into Parts and aggregate the Parts into a Whole. In this pattern, there are a Whole module and $n$ Part modules. The Whole
module interacts with outside through the channels $I$ and $O$, and with Part $i$ (for $1\leq i\leq n$) through the channels $I_{WP_i}$ and $O_{WP_i}$, as illustrated in Figure \ref{WP4}.

\begin{figure}
    \centering
    \includegraphics{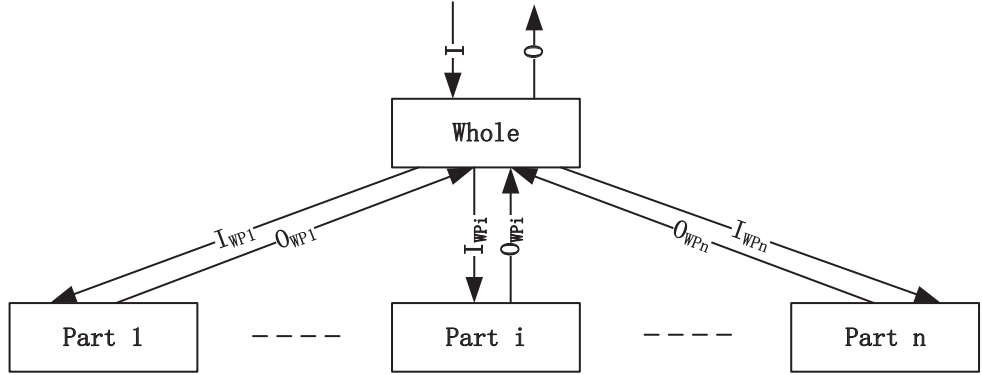}
    \caption{Whole-Part pattern}
    \label{WP4}
\end{figure}

The typical process of the Whole-Part pattern is shown in Figure \ref{WP4P} and as follows.

\begin{enumerate}
  \item The Whole receives the request $d_I$ from outside through the channel $I$ (the corresponding reading action is denoted $r_I(d_I)$), then processes the request through a processing
  function $WF_1$ and generates the request $d_{I_W}$, and sends the $d_{I_W}$ to the Part $i$ through the channel $I_{WP_i}$ (the corresponding sending action is denoted $s_{I_{WP_i}}(d_{I_W})$);
  \item The Part $i$ receives the request $d_{I_W}$ from the Whole through the channel $I_{WP_i}$ (the corresponding reading action is denoted $r_{I_{WP_i}}(d_{I_W})$), then processes the
  request through a processing function $PF_i$ and generates the response $d_{O_{P_i}}$, and sends the response to the Whole through the channel $O_{WP_i}$ (the corresponding sending
  action is denoted $s_{O_{WP_i}}(d_{O_{P_i}})$);
  \item The Whole receives the response $d_{O_{P_i}}$ from the Part $i$ through the channel $O_{WP_i}$ (the corresponding reading action is denoted $r_{O_{WP_i}}(d_{O_{P_i}})$), then processes the request through a processing
  function $WF_2$ and generates the request $d_{O}$, and sends the $d_{O}$ to the outside through the channel $O$ (the corresponding sending action is denoted $s_{O}(d_{O})$).
\end{enumerate}

\begin{figure}
    \centering
    \includegraphics{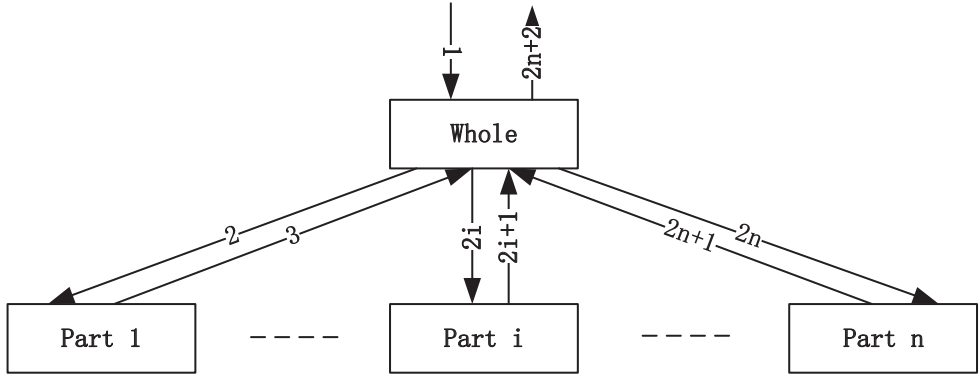}
    \caption{Typical process of Whole-Part pattern}
    \label{WP4P}
\end{figure}

In the following, we verify the Whole-Part pattern. We assume all data elements $d_{I}$, $d_{I_{W}}$, $d_{O_{P_{i}}}$, $d_{O}$ (for $1\leq i\leq n$) are from a finite set
$\Delta$.

The state transitions of the Whole module
described by APTC are as follows.

$W=\sum_{d_{I}\in\Delta}(r_{I}(d_{I})\cdot W_{2})$

$W_{2}=WF_1\cdot W_{3}$

$W_{3}=\sum_{d_{I_W}\in\Delta}(s_{I_{WP_1}}(d_{I_W})\between\cdots\between s_{I_{WP_n}}(d_{I_W})\cdot W_{4})$

$W_{4}=\sum_{d_{O_{P_{1}}},\cdots,d_{O_{P_{n}}}\in\Delta}(r_{O_{WP_1}}(d_{O_{P_{1}}})\between\cdots\between r_{O_{WP_n}}(d_{O_{P_{n}}})\cdot W_{5})$

$W_5=WF_2\cdot W_6$

$W_{6}=\sum_{d_{O}\in\Delta}(s_{O}(d_{O})\cdot W)$

The state transitions of the Part $i$ described by APTC are as follows.

$P_i=\sum_{d_{I_W}\in\Delta}(r_{I_{WP_i}}(d_{I_W})\cdot P_{i_2})$

$P_{i_2}=PF_{i}\cdot P_{i_3}$

$P_{i_3}=\sum_{d_{O_{P_{i}}}\in\Delta}(s_{O_{WP_{i}}}(d_{O_{P_{i}}})\cdot P)$

The sending action and the reading action of the same data through the same channel can communicate with each other, otherwise, will cause a deadlock $\delta$. We define the following
communication functions of the Part $i$ for $1\leq i\leq n$.

$$\gamma(r_{I_{WP_i}}(d_{I_W}),s_{I_{WP_i}}(d_{I_W}))\triangleq c_{I_{WP_i}}(d_{I_W})$$

$$\gamma(r_{O_{WP_{i}}}(d_{O_{P_{i}}}),s_{O_{WP_{i}}}(d_{O_{P_{i}}}))\triangleq c_{O_{WP_{i}}}(d_{O_{P_{i}}})$$

Let all modules be in parallel, then the Whole-Part pattern $Q\quad P_1\cdots P_i\cdots P_n$ can be presented by the following process term.

$\tau_I(\partial_H(\Theta(W\between P_1\between\cdots\between P_i\between\cdots\between P_n)))=\tau_I(\partial_H(W\between P_1\between\cdots\between P_i\between\cdots\between P_n))$

where $H=\{r_{I_{WP_i}}(d_{I_W}),s_{I_{WP_i}}(d_{I_W}),r_{O_{WP_{i}}}(d_{O_{P_{i}}}),s_{O_{WP_{i}}}(d_{O_{P_{i}}})\\
|d_{I}, d_{I_{W}}, d_{O_{P_{i}}}, d_{O}\in\Delta\}$ for $1\leq i\leq n$,

$I=\{c_{I_{WP_i}}(d_{I_W}),c_{O_{WP_{i}}}(d_{O_{P_{i}}}),WF_1,WF_2,PF_{i}\\
|d_{I}, d_{I_{W}}, d_{O_{P_{i}}}, d_{O}\in\Delta\}$ for $1\leq i\leq n$.

Then we get the following conclusion on the Whole-Part pattern.

\begin{theorem}[Correctness of the Whole-Part pattern]
The Whole-Part pattern $\tau_I(\partial_H(W\between P_1\between\cdots\between P_i\between\cdots\between P_n))$ can exhibit desired external behaviors.
\end{theorem}

\begin{proof}
Based on the above state transitions of the above modules, by use of the algebraic laws of APTC, we can prove that

$\tau_I(\partial_H(W\between P_1\between\cdots\between P_i\between\cdots\between P_n))=\sum_{d_{I_M},d_{I_C},d_{O_M},d_{O_C}\in\Delta}(r_{I_M}(d_{I_M})\parallel r_{I_C}(d_{I_C})\cdot s_{O_M}(d_{O_M})\parallel s_{O_C}(d_{O_C}))\cdot
\tau_I(\partial_H(W\between P_1\between\cdots\between P_i\between\cdots\between P_n))$,

that is, the Whole-Part pattern $\tau_I(\partial_H(W\between P_1\between\cdots\between P_i\between\cdots\between P_n))$ can exhibit desired external behaviors.

For the details of proof, please refer to section \ref{app}, and we omit it.
\end{proof}

\subsection{Organization of Work}\label{OW4}

\subsubsection{Verification of the Master-Slave Pattern}

The Master-Slave pattern is used to implement large scale computation. In this pattern, there are a Master module and $n$ Slave modules. The Slaves is used to implement concrete
computation and the Master is used to distribute computational tasks and collect the computational results. The Master
module interacts with outside through the channels $I$ and $O$, and with Slave $i$ (for $1\leq i\leq n$) through the channels $I_{MS_i}$ and $O_{MS_i}$, as illustrated in Figure \ref{MS4}.

\begin{figure}
    \centering
    \includegraphics{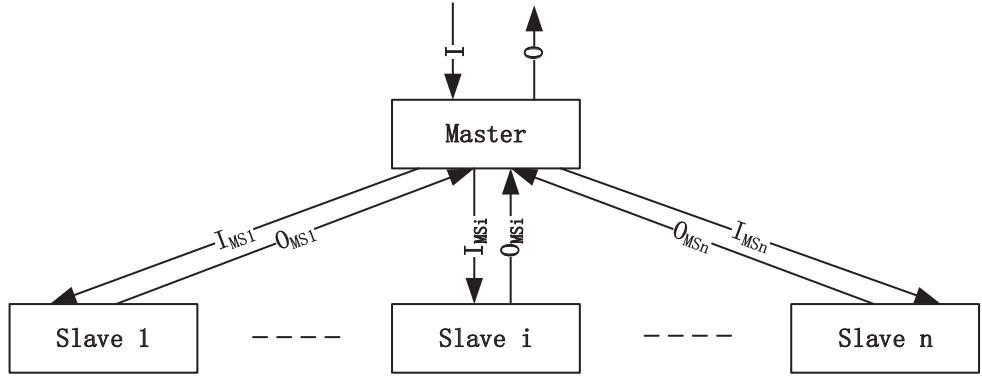}
    \caption{Master-Slave pattern}
    \label{MS4}
\end{figure}

The typical process of the Master-Slave pattern is shown in Figure \ref{MS4P} and as follows.

\begin{enumerate}
  \item The Master receives the request $d_I$ from outside through the channel $I$ (the corresponding reading action is denoted $r_I(d_I)$), then processes the request through a processing
  function $MF_1$ and generates the request $d_{I_M}$, and sends the $d_{I_M}$ to the Slave $i$ through the channel $I_{MS_i}$ (the corresponding sending action is denoted $s_{I_{MS_i}}(d_{I_M})$);
  \item The Slave $i$ receives the request $d_{I_M}$ from the Master through the channel $I_{MS_i}$ (the corresponding reading action is denoted $r_{I_{MS_i}}(d_{I_M})$), then processes the
  request through a processing function $SF_i$ and generates the response $d_{O_{S_i}}$, and sends the response to the Master through the channel $O_{MS_i}$ (the corresponding sending
  action is denoted $s_{O_{MS_i}}(d_{O_{S_i}})$);
  \item The Master receives the response $d_{O_{S_i}}$ from the Slave $i$ through the channel $O_{MS_i}$ (the corresponding reading action is denoted $r_{O_{MS_i}}(d_{O_{S_i}})$), then processes the request through a processing
  function $MF_2$ and generates the request $d_{O}$, and sends the $d_{O}$ to the outside through the channel $O$ (the corresponding sending action is denoted $s_{O}(d_{O})$).
\end{enumerate}

\begin{figure}
    \centering
    \includegraphics{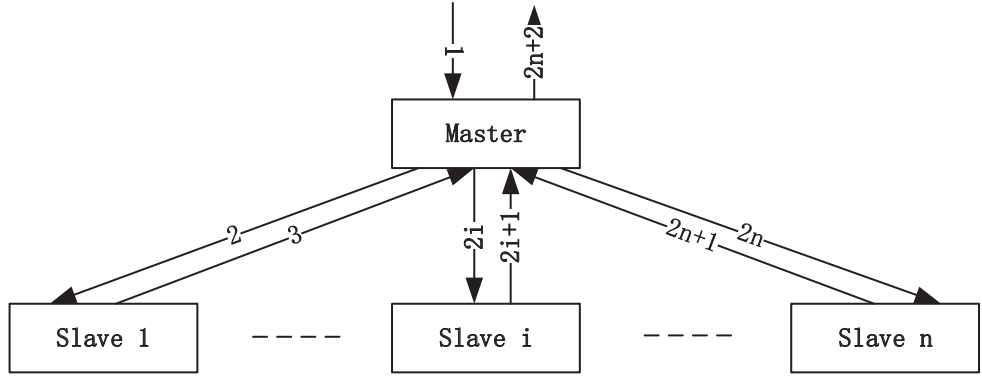}
    \caption{Typical process of Master-Slave pattern}
    \label{MS4P}
\end{figure}

In the following, we verify the Master-Slave pattern. We assume all data elements $d_{I}$, $d_{I_{M}}$, $d_{O_{S_{i}}}$, $d_{O}$ (for $1\leq i\leq n$) are from a finite set
$\Delta$.

The state transitions of the Master module
described by APTC are as follows.

$M=\sum_{d_{I}\in\Delta}(r_{I}(d_{I})\cdot M_{2})$

$M_{2}=MF_1\cdot M_{3}$

$M_{3}=\sum_{d_{I_M}\in\Delta}(s_{I_{MS_1}}(d_{I_M})\between\cdots\between s_{I_{MS_n}}(d_{I_M})\cdot M_{4})$

$M_{4}=\sum_{d_{O_{S_{1}}},\cdots,d_{O_{S_{n}}}\in\Delta}(r_{O_{MS_1}}(d_{O_{S_{1}}})\between\cdots\between r_{O_{MS_n}}(d_{O_{S_{n}}})\cdot M_{5})$

$M_5=MF_2\cdot M_6$

$M_{6}=\sum_{d_{O}\in\Delta}(s_{O}(d_{O})\cdot M)$

The state transitions of the Slave $i$ described by APTC are as follows.

$S_i=\sum_{d_{I_M}\in\Delta}(r_{I_{MS_i}}(d_{I_M})\cdot S_{i_2})$

$S_{i_2}=SF_{i}\cdot S_{i_3}$

$S_{i_3}=\sum_{d_{O_{S_{i}}}\in\Delta}(s_{O_{MS_{i}}}(d_{O_{S_{i}}})\cdot S)$

The sending action and the reading action of the same data through the same channel can communicate with each other, otherwise, will cause a deadlock $\delta$. We define the following
communication functions of the Slave $i$ for $1\leq i\leq n$.

$$\gamma(r_{I_{MS_i}}(d_{I_M}),s_{I_{MS_i}}(d_{I_M}))\triangleq c_{I_{MS_i}}(d_{I_M})$$

$$\gamma(r_{O_{MS_{i}}}(d_{O_{S_{i}}}),s_{O_{MS_{i}}}(d_{O_{S_{i}}}))\triangleq c_{O_{MS_{i}}}(d_{O_{S_{i}}})$$

Let all modules be in parallel, then the Master-Slave pattern $M\quad S_1\cdots S_i\cdots S_n$ can be presented by the following process term.

$\tau_I(\partial_H(\Theta(M\between S_1\between\cdots\between S_i\between\cdots\between S_n)))=\tau_I(\partial_H(M\between S_1\between\cdots\between S_i\between\cdots\between S_n))$

where $H=\{r_{I_{MS_i}}(d_{I_M}),s_{I_{MS_i}}(d_{I_M}),r_{O_{MS_{i}}}(d_{O_{S_{i}}}),s_{O_{MS_{i}}}(d_{O_{S_{i}}})\\
|d_{I}, d_{I_{M}}, d_{O_{S_{i}}}, d_{O}\in\Delta\}$ for $1\leq i\leq n$,

$I=\{c_{I_{MS_i}}(d_{I_M}),c_{O_{MS_{i}}}(d_{O_{S_{i}}}),MF_1,MF_2,SF_{i}\\
|d_{I}, d_{I_{M}}, d_{O_{S_{i}}}, d_{O}\in\Delta\}$ for $1\leq i\leq n$.

Then we get the following conclusion on the Master-Slave pattern.

\begin{theorem}[Correctness of the Master-Slave pattern]
The Master-Slave pattern $\tau_I(\partial_H(M\between S_1\between\cdots\between S_i\between\cdots\between S_n))$ can exhibit desired external behaviors.
\end{theorem}

\begin{proof}
Based on the above state transitions of the above modules, by use of the algebraic laws of APTC, we can prove that

$\tau_I(\partial_H(M\between S_1\between\cdots\between S_i\between\cdots\between S_n))=\sum_{d_{I},d_O\in\Delta}(r_{I}(d_{I})\cdot s_{O}(d_{O}))\cdot
\tau_I(\partial_H(M\between S_1\between\cdots\between S_i\between\cdots\between S_n))$,

that is, the Master-Slave pattern $\tau_I(\partial_H(M\between S_1\between\cdots\between S_i\between\cdots\between S_n))$ can exhibit desired external behaviors.

For the details of proof, please refer to section \ref{app}, and we omit it.
\end{proof}

\subsection{Access Control}\label{AC4}

\subsubsection{Verification of the Proxy Pattern}

The Proxy pattern is used to decouple the access of original components through a proxy. In this pattern, there are a Proxy module and $n$ Original modules. The Originals is used to implement concrete
computation and the Proxy is used to decouple the access to the Originals. The Proxy
module interacts with outside through the channels $I$ and $O$, and with Original $i$ (for $1\leq i\leq n$) through the channels $I_{PO_i}$ and $O_{PO_i}$, as illustrated in Figure \ref{Pro4}.

\begin{figure}
    \centering
    \includegraphics{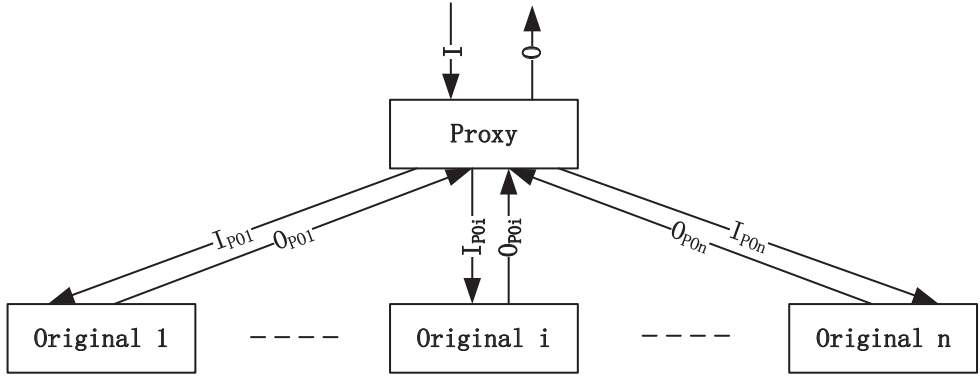}
    \caption{Proxy pattern}
    \label{Pro4}
\end{figure}

The typical process of the Proxy pattern is shown in Figure \ref{Pro4P} and as follows.

\begin{enumerate}
  \item The Proxy receives the request $d_I$ from outside through the channel $I$ (the corresponding reading action is denoted $r_I(d_I)$), then processes the request through a processing
  function $PF_1$ and generates the request $d_{I_P}$, and sends the $d_{I_P}$ to the Original $i$ through the channel $I_{PO_i}$ (the corresponding sending action is denoted $s_{I_{PO_i}}(d_{I_P})$);
  \item The Original $i$ receives the request $d_{I_P}$ from the Proxy through the channel $I_{PO_i}$ (the corresponding reading action is denoted $r_{I_{PO_i}}(d_{I_P})$), then processes the
  request through a processing function $OF_i$ and generates the response $d_{O_{O_i}}$, and sends the response to the Proxy through the channel $O_{PO_i}$ (the corresponding sending
  action is denoted $s_{O_{PO_i}}(d_{O_{O_i}})$);
  \item The Proxy receives the response $d_{O_{O_i}}$ from the Original $i$ through the channel $O_{PO_i}$ (the corresponding reading action is denoted $r_{O_{PO_i}}(d_{O_{O_i}})$), then processes the request through a processing
  function $PF_2$ and generates the request $d_{O}$, and sends the $d_{O}$ to the outside through the channel $O$ (the corresponding sending action is denoted $s_{O}(d_{O})$).
\end{enumerate}

\begin{figure}
    \centering
    \includegraphics{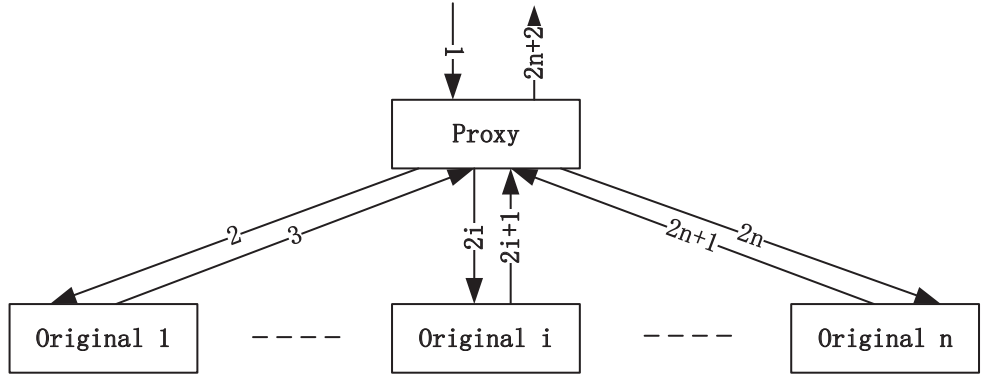}
    \caption{Typical process of Proxy pattern}
    \label{Pro4P}
\end{figure}

In the following, we verify the Proxy pattern. We assume all data elements $d_{I}$, $d_{I_{P}}$, $d_{O_{O_{i}}}$, $d_{O}$ (for $1\leq i\leq n$) are from a finite set
$\Delta$.

The state transitions of the Proxy module
described by APTC are as follows.

$P=\sum_{d_{I}\in\Delta}(r_{I}(d_{I})\cdot P_{2})$

$P_{2}=PF_1\cdot P_{3}$

$P_{3}=\sum_{d_{I_P}\in\Delta}(s_{I_{PO_1}}(d_{I_P})\between\cdots\between s_{I_{PO_n}}(d_{I_P})\cdot P_{4})$

$P_{4}=\sum_{d_{O_{O_{1}}},\cdots,d_{O_{O_{n}}}\in\Delta}(r_{O_{PO_1}}(d_{O_{O_{1}}})\between\cdots\between r_{O_{PO_n}}(d_{O_{O_{n}}})\cdot P_{5})$

$P_5=PF_2\cdot P_6$

$P_{6}=\sum_{d_{O}\in\Delta}(s_{O}(d_{O})\cdot P)$

The state transitions of the Original $i$ described by APTC are as follows.

$O_i=\sum_{d_{I_P}\in\Delta}(r_{I_{PO_i}}(d_{I_P})\cdot O_{i_2})$

$O_{i_2}=OF_{i}\cdot O_{i_3}$

$O_{i_3}=\sum_{d_{O_{O_{i}}}\in\Delta}(s_{O_{PO_{i}}}(d_{O_{O_{i}}})\cdot O_i)$

The sending action and the reading action of the same data through the same channel can communicate with each other, otherwise, will cause a deadlock $\delta$. We define the following
communication functions of the Original $i$ for $1\leq i\leq n$.

$$\gamma(r_{I_{PO_i}}(d_{I_P}),s_{I_{PO_i}}(d_{I_P}))\triangleq c_{I_{PO_i}}(d_{I_P})$$

$$\gamma(r_{O_{PO_{i}}}(d_{O_{O_{i}}}),s_{O_{PO_{i}}}(d_{O_{O_{i}}}))\triangleq c_{O_{PO_{i}}}(d_{O_{O_{i}}})$$

Let all modules be in parallel, then the Proxy pattern $P\quad O_1\cdots O_i\cdots O_n$ can be presented by the following process term.

$\tau_I(\partial_H(\Theta(P\between O_1\between\cdots\between O_i\between\cdots\between O_n)))=\tau_I(\partial_H(P\between O_1\between\cdots\between O_i\between\cdots\between O_n))$

where $H=\{r_{I_{PO_i}}(d_{I_P}),s_{I_{PO_i}}(d_{I_P}),r_{O_{PO_{i}}}(d_{O_{O_{i}}}),s_{O_{PO_{i}}}(d_{O_{O_{i}}})\\
|d_{I}, d_{I_{P}}, d_{O_{O_{i}}}, d_{O}\in\Delta\}$ for $1\leq i\leq n$,

$I=\{c_{I_{PO_i}}(d_{I_P}),c_{O_{PO_{i}}}(d_{O_{O_{i}}}),PF_1,PF_2,OF_{i}\\
|d_{I}, d_{I_{P}}, d_{O_{O_{i}}}, d_{O}\in\Delta\}$ for $1\leq i\leq n$.

Then we get the following conclusion on the Proxy pattern.

\begin{theorem}[Correctness of the Proxy pattern]
The Proxy pattern $\tau_I(\partial_H(P\between O_1\between\cdots\between O_i\between\cdots\between O_n))$ can exhibit desired external behaviors.
\end{theorem}

\begin{proof}
Based on the above state transitions of the above modules, by use of the algebraic laws of APTC, we can prove that

$\tau_I(\partial_H(P\between O_1\between\cdots\between O_i\between\cdots\between O_n))=\sum_{d_{I},d_O\in\Delta}(r_{I}(d_{I})\cdot s_{O}(d_{O}))\cdot
\tau_I(\partial_H(P\between O_1\between\cdots\between O_i\between\cdots\between O_n))$,

that is, the Proxy pattern $\tau_I(\partial_H(P\between O_1\between\cdots\between O_i\between\cdots\between O_n))$ can exhibit desired external behaviors.

For the details of proof, please refer to section \ref{app}, and we omit it.
\end{proof}

\subsection{Management}\label{M4}

\subsubsection{Verification of the Command Processor Pattern}

The Command Processor pattern is used to decouple the request and execution of a service. In this pattern, there are a Controller module, a Command Processor module,
and $n$ Command modules and $n$ Supplier modules. The Supplier is used to implement concrete
computation, the Command is used to encapsulate a Supplier into a command, and the Command Processor is used to manage Commands. The Controller
module interacts with outside through the channels $I$ and $O$, and with the Command Processor through the channels $I_{CP}$ and $O_{CP}$. The Command Processor interacts with
Command $i$ (for $1\leq i\leq n$) through the channels $I_{PC_i}$ and $O_{PC_i}$, and the Command $i$ interacts with the Supplier $i$ through the channels $I_{CS_i}$ and $O_{CS_i}$,
as illustrated in Figure \ref{CP4}.

\begin{figure}
    \centering
    \includegraphics{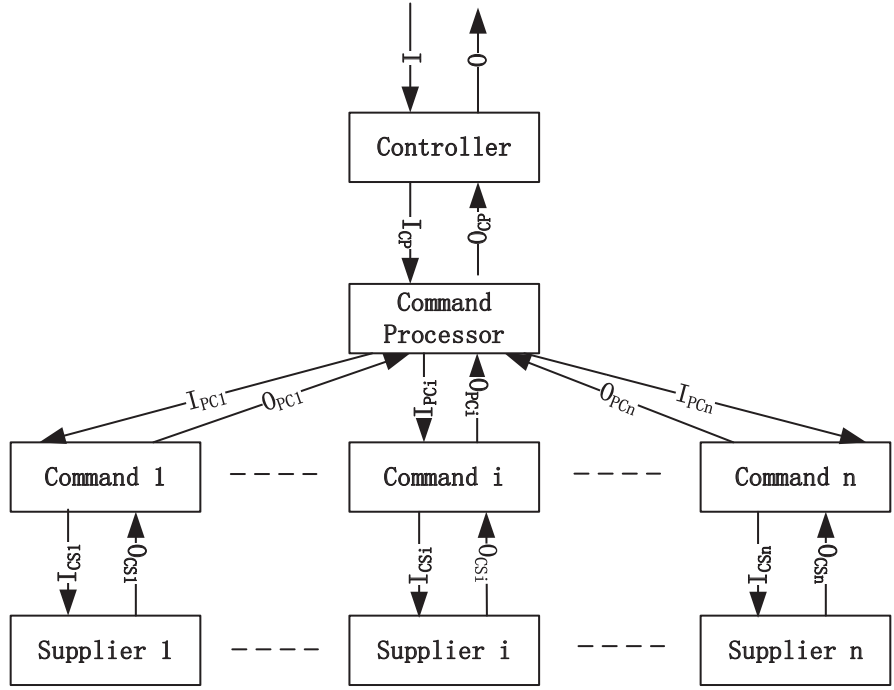}
    \caption{Command Processor pattern}
    \label{CP4}
\end{figure}

The typical process of the Command Processor pattern is shown in Figure \ref{CP4P} and as follows.

\begin{enumerate}
  \item The Controller receives the request $d_I$ from outside through the channel $I$ (the corresponding reading action is denoted $r_I(d_I)$), then processes the request through a processing
  function $CF_1$ and generates the request $d_{I_P}$, and sends the $d_{I_P}$ to the Command Processor through the channel $I_{CP}$ (the corresponding sending action is denoted $s_{I_{CP}}(d_{I_P})$);
  \item The Command Processor receives the request $d_{I_P}$ from the Controller through the channel $I_{CP}$ (the corresponding reading action is denoted $r_{I_{CP}}(d_{I_P})$), then processes the request through a processing
  function $PF_1$ and generates the request $d_{I_{Com_i}}$, and sends the $d_{I_{Com_i}}$ to the Command $i$ through the channel $I_{PC_i}$ (the corresponding sending action is denoted $s_{I_{PC_i}}(d_{I_{Com_i}})$);
  \item The Command $i$ receives the request $d_{I_{Com_i}}$ from the Command Processor through the channel $I_{PC_i}$ (the corresponding reading action is denoted $r_{I_{PC_i}}(d_{I_{Com_i}})$), then processes the
  request through a processing function $ComF_{i1}$ and generates the request $d_{I_{S_i}}$, and sends the request to the Supplier $i$ through the channel $I_{CS_i}$ (the corresponding sending
  action is denoted $s_{I_{CS_i}}(d_{I_{S_i}})$);
  \item The Supplier $i$ receives the request $d_{I_{S_i}}$ from the Command $i$ through the channel $I_{CS_i}$ (the corresponding reading action is denoted $r_{I_{CS_i}}(d_{I_{S_i}})$), then processes the
  request through a processing function $SF_i$ and generates the response $d_{O_{S_i}}$, and sends the response to the Command through the channel $O_{CS_i}$ (the corresponding sending
  action is denoted $s_{O_{CS_i}}(d_{O_{S_i}})$);
  \item The Command $i$ receives the response $d_{O_{S_i}}$ from the Supplier $i$ through the channel $O_{CS_i}$ (the corresponding reading action is denoted $r_{O_{CS_i}}(d_{O_{S_i}})$), then processes the
  request through a processing function $ComF_{i2}$ and generates the response $d_{O_{Com_i}}$, and sends the response to the Command Processor through the channel $O_{PC_i}$ (the corresponding sending
  action is denoted $s_{O_{PC_i}}(d_{O_{Com_i}})$);
  \item The Command Processor receives the response $d_{O_{Com_i}}$ from the Command $i$ through the channel $O_{PC_i}$ (the corresponding reading action is denoted $r_{O_{PC_i}}(d_{O_{Com_i}})$),
  then processes the response and generate the response $d_{O_{P}}$ through a processing function $PF_2$, and sends $d_{O_P}$ to the Controller through the channel $O_{CP}$ (the corresponding
  sending action is denoted $s_{O_{CP}}(d_{O_P})$);
  \item The Controller receives the response $d_{O_{P}}$ from the Command Processor through the channel $O_{CP}$ (the corresponding reading action is denoted $r_{O_{CP}}(d_{O_{P}})$), then processes the request through a processing
  function $CF_2$ and generates the request $d_{O}$, and sends the $d_{O}$ to the outside through the channel $O$ (the corresponding sending action is denoted $s_{O}(d_{O})$).
\end{enumerate}

\begin{figure}
    \centering
    \includegraphics{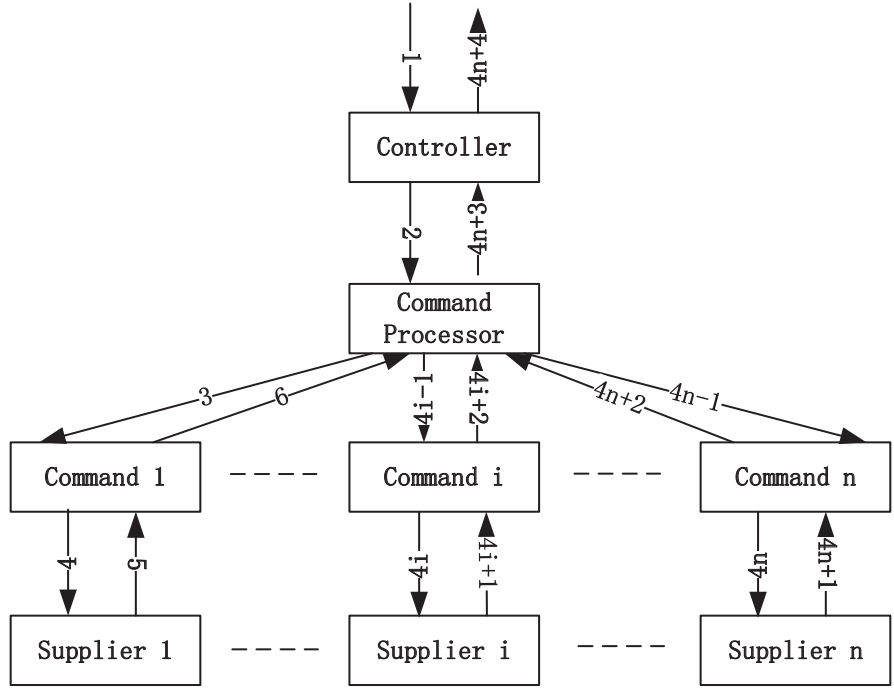}
    \caption{Typical process of Command Processor pattern}
    \label{CP4P}
\end{figure}

In the following, we verify the Command Processor pattern. We assume all data elements $d_{I}$, $d_{I_{P}}$, $d_{I_{Com_i}}$, $d_{I_{S_i}}$, $d_{O_{S_{i}}}$, $d_{O_{Com_i}}$, $d_{O_P}$, $d_{O}$ (for $1\leq i\leq n$) are from a finite set
$\Delta$.

The state transitions of the Controller module
described by APTC are as follows.

$C=\sum_{d_{I}\in\Delta}(r_{I}(d_{I})\cdot C_{2})$

$C_{2}=CF_1\cdot C_{3}$

$C_{3}=\sum_{d_{I_P}\in\Delta}(s_{I_{CP}}(d_{I_P})\cdot C_{4})$

$C_{4}=\sum_{d_{O_{P}}\in\Delta}(r_{O_{CP}}(d_{O_{P}})\cdot C_{5})$

$C_5=CF_2\cdot C_6$

$C_{6}=\sum_{d_{O}\in\Delta}(s_{O}(d_{O})\cdot C)$

The state transitions of the Command Processor module
described by APTC are as follows.

$P=\sum_{d_{I_P}\in\Delta}(r_{I_{CP}}(d_{I_P})\cdot P_{2})$

$P_{2}=PF_1\cdot P_{3}$

$P_{3}=\sum_{d_{I_{Com_1}},\cdots,d_{I_{Com_n}}\in\Delta}(s_{I_{PC_1}}(d_{I_{Com_1}})\between\cdots\between s_{I_{PC_n}}(d_{I_{Com_n}})\cdot P_{4})$

$P_{4}=\sum_{d_{O_{Com_{1}}},\cdots,d_{O_{Com_{n}}}\in\Delta}(r_{O_{PC_1}}(d_{O_{Com_{1}}})\between\cdots\between r_{O_{PC_n}}(d_{O_{Com_{n}}})\cdot P_{5})$

$P_5=PF_2\cdot P_6$

$P_{6}=\sum_{d_{O_P}\in\Delta}(s_{O_{CP}}(d_{O_P})\cdot P)$

The state transitions of the Command $i$ described by APTC are as follows.

$Com_i=\sum_{d_{I_{Com_i}}\in\Delta}(r_{I_{PC_i}}(d_{I_{Com_i}})\cdot Com_{i_2})$

$Com_{i_2}=ComF_{i1}\cdot Com_{i_3}$

$Com_{i_3}=\sum_{d_{I_{S_{i}}}\in\Delta}(s_{I_{CS_{i}}}(d_{I_{S_{i}}})\cdot Com_{i_4})$

$Com_{i_4}=\sum_{d_{O_{S_i}}\in\Delta}(r_{O_{CS_i}}(d_{O_{S_i}})\cdot Com_{i_5})$

$Com_{i_5}=ComF_{i2}\cdot Com_{i_6}$

$Com_{i_6}=\sum_{d_{O_{Com_{i}}}\in\Delta}(s_{O_{PC_{i}}}(d_{O_{Com_{i}}})\cdot Com_i)$

The state transitions of the Supplier $i$ described by APTC are as follows.

$S_i=\sum_{d_{I_{S_i}}\in\Delta}(r_{I_{CS_i}}(d_{I_{S_i}})\cdot S_{i_2})$

$S_{i_2}=SF_{i}\cdot S_{i_3}$

$S_{i_3}=\sum_{d_{O_{S_{i}}}\in\Delta}(s_{O_{CS_{i}}}(d_{O_{S_{i}}})\cdot S_i)$

The sending action and the reading action of the same data through the same channel can communicate with each other, otherwise, will cause a deadlock $\delta$. We define the following
communication functions of between the Controller the Command Processor.

$$\gamma(r_{I_{CP}}(d_{I_P}),s_{I_{CP}}(d_{I_P}))\triangleq c_{I_{CP}}(d_{I_P})$$

$$\gamma(r_{O_{CP}}(d_{O_{P}}),s_{O_{CP}}(d_{O_{P}}))\triangleq c_{O_{CP}}(d_{O_{P}})$$

There are two communication function between the Command Processor and the Command $i$ for $1\leq i\leq n$.

$$\gamma(r_{I_{PC_i}}(d_{I_{Com_i}}),s_{I_{PC_i}}(d_{I_{Com_i}}))\triangleq c_{I_{PC_i}}(d_{I_{Com_i}})$$

$$\gamma(r_{O_{PC_i}}(d_{O_{Com_{i}}}),s_{O_{PC_i}}(d_{O_{Com_{i}}}))\triangleq c_{O_{PC_i}}(d_{O_{Com_{i}}})$$

There are two communication function between the Supplier $i$ and the Command $i$ for $1\leq i\leq n$.

$$\gamma(r_{I_{CS_i}}(d_{I_{S_i}}),s_{I_{CS_i}}(d_{I_{S_i}}))\triangleq c_{I_{CS_i}}(d_{I_{S_i}})$$

$$\gamma(r_{O_{CS_{i}}}(d_{O_{S_{i}}}),s_{O_{CS_{i}}}(d_{O_{S_{i}}}))\triangleq c_{O_{CS_{i}}}(d_{O_{S_{i}}})$$

Let all modules be in parallel, then the Command Processor pattern

$C\quad P \quad Com_1\cdots\quad Com_i\quad\cdots Com_n\quad S_1\cdots S_i\cdots S_n$

can be presented by the following process term.

$\tau_I(\partial_H(\Theta(C\between P\between Com_1\between\cdots\between Com_i\between\cdots\between Com_n\between S_1\between\cdots\between S_i\between\cdots\between S_n)))=\tau_I(\partial_H(C\between P\between Com_1\between\cdots\between Com_i\between\cdots\between Com_n\between S_1\between\cdots\between S_i\between\cdots\between S_n))$

where $H=\{r_{I_{CP}}(d_{I_P}),s_{I_{CP}}(d_{I_P}),r_{O_{CP}}(d_{O_{P}}),s_{O_{CP}}(d_{O_{P}}),r_{I_{PC_i}}(d_{I_{Com_i}}),s_{I_{PC_i}}(d_{I_{Com_i}}),\\
r_{O_{PC_i}}(d_{O_{Com_{i}}}),s_{O_{PC_i}}(d_{O_{Com_{i}}}),r_{I_{CS_i}}(d_{I_{S_i}}),s_{I_{CS_i}}(d_{I_{S_i}}),r_{O_{CS_{i}}}(d_{O_{S_{i}}}),s_{O_{CS_{i}}}(d_{O_{S_{i}}})\\
|d_{I}, d_{I_{P}}, d_{I_{Com_i}}, d_{I_{S_i}}, d_{O_{S_{i}}}, d_{O_{Com_i}}, d_{O_P}, d_{O}\in\Delta\}$ for $1\leq i\leq n$,

$I=\{c_{I_{CP}}(d_{I_P}),c_{O_{CP}}(d_{O_{P}}),c_{I_{PC_i}}(d_{I_{Com_i}}),c_{O_{PC_i}}(d_{O_{Com_{i}}}),c_{I_{CS_i}}(d_{I_{S_i}}),\\
c_{O_{CS_{i}}}(d_{O_{S_{i}}}),CF_1,CF_2,PF_1,PF_2,ComF_{i1},ComF_{i2},SF_{i}\\
|d_{I}, d_{I_{P}}, d_{I_{Com_i}}, d_{I_{S_i}}, d_{O_{S_{i}}}, d_{O_{Com_i}}, d_{O_P}, d_{O}\in\Delta\}$ for $1\leq i\leq n$.

Then we get the following conclusion on the Command Processor pattern.

\begin{theorem}[Correctness of the Command Processor pattern]
The Command Processor pattern $\tau_I(\partial_H(C\between P\between Com_1\between\cdots\between Com_i\between\cdots\between Com_n\between S_1\between\cdots\between S_i\between\cdots\between S_n))$ can exhibit desired external behaviors.
\end{theorem}

\begin{proof}
Based on the above state transitions of the above modules, by use of the algebraic laws of APTC, we can prove that

$\tau_I(\partial_H(C\between P\between Com_1\between\cdots\between Com_i\between\cdots\between Com_n\between S_1\between\cdots\between S_i\between\cdots\between S_n))=\sum_{d_{I},d_O\in\Delta}(r_{I}(d_{I})\cdot s_{O}(d_{O}))\cdot
\tau_I(\partial_H(C\between P\between Com_1\between\cdots\between Com_i\between\cdots\between Com_n\between S_1\between\cdots\between S_i\between\cdots\between S_n))$,

that is, the Command Processor pattern $\tau_I(\partial_H(C\between P\between Com_1\between\cdots\between Com_i\between\cdots\between Com_n\between S_1\between\cdots\between S_i\between\cdots\between S_n))$ can exhibit desired external behaviors.

For the details of proof, please refer to section \ref{app}, and we omit it.
\end{proof}

\subsubsection{Verification of the View Handler Pattern}

The View Handler pattern is used to manage all views of the system, which has three components: the Supplier, the Views and the ViewHandler. The Supplier is used to contain the data and encapsulate the
core functionalities; the Views is to show the computational results to the user; and the ViewHandler interacts between the system and user, accepts the instructions and controls the Supplier
and the Views. The ViewHandler receives the instructions from the user through the channel $I$, then it sends the instructions to the Supplier through the channel $VS$ and to the View $i$
through the channel $VV_i$ for $1\leq i\leq n$; the model receives the instructions from the ViewHandler, updates the data and computes the results, and sends the results to the View $i$
through the channel $SV_i$ for $1\leq i\leq n$; When the View $i$ receives the results from the Supplier, it generates or updates the view to the user. As illustrates in Figure \ref{VH4}.

\begin{figure}
    \centering
    \includegraphics{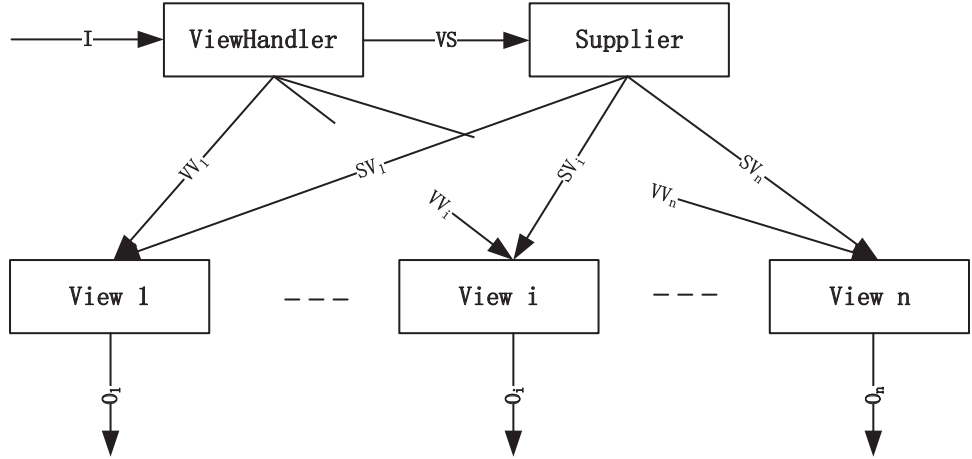}
    \caption{View Handler pattern}
    \label{VH4}
\end{figure}

The typical process of the View Handler pattern is shown in Figure \ref{VH4P} and following.

\begin{enumerate}
  \item The ViewHandler receives the instructions $d_I$ from the user through the channel $I$ (the corresponding reading action is denoted $r_I(D_I)$), processes the instructions through
  a processing function $VHF$, and generates the instructions to the Supplier $d_{I_S}$ and those to the View $i$ $d_{I_{V_i}}$ for $1\leq i\leq n$; it sends $d_{I_S}$ to the Supplier through the channel $CM$
  (the corresponding sending action is denoted $s_{VS}(d_{I_M})$) and sends $d_{I_{V_i}}$ to the View $i$ through the channel $VV_i$ (the corresponding sending action is denoted
  $s_{VV_i}(d_{I_{V_i}})$);
  \item The Supplier receives the instructions from the ViewHandler through the channel $VS$ (the corresponding reading action is denoted $r_{VS}(d_{I_S})$), processes the instructions through
  a processing function $SF$, generates the computational results to the View $i$ (for $1\leq i\leq n$) which is denoted $d_{O_{S_i}}$; then sends the results to the View $i$ through the
  channel $SV_i$ (the corresponding sending action is denoted $s_{MV_i}(d_{O_{M_i}})$);
  \item The View $i$ (for $1\leq i\leq n$) receives the instructions from the ViewHandler through the channel $VV_i$ (the corresponding reading action is denoted $r_{VV_i}(d_{I_{V_i}})$),
  processes the instructions through a processing function $VF_{i_1}$ to make ready to receive the computational results from the Supplier; then it receives the computational results from the Supplier
  through the channel $SV_i$ (the corresponding reading action is denoted $r_{SV_i}(d_{O_{S_i}})$), processes the results through a processing function $VF_{i_2}$, generates the output
  $d_{O_i}$, then sending the output through the channel $O_i$ (the corresponding sending action is denoted $s_{O_i}(d_{O_i})$).
\end{enumerate}

\begin{figure}
    \centering
    \includegraphics{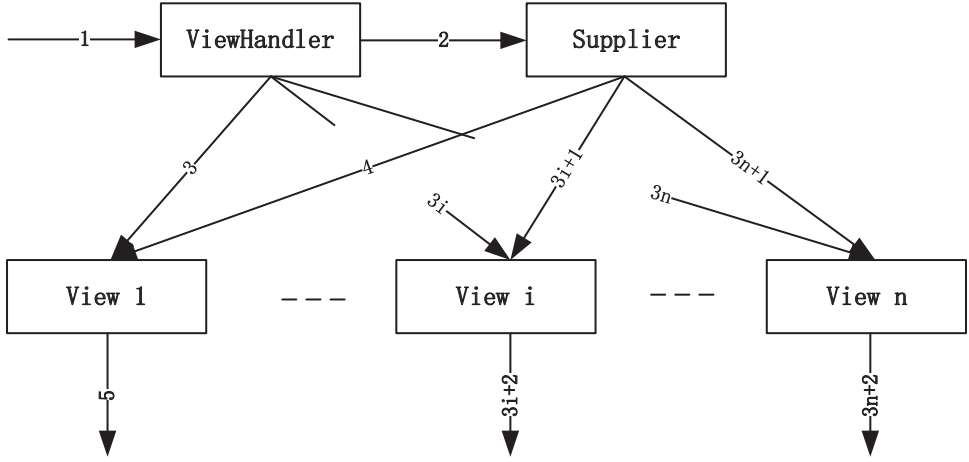}
    \caption{Typical process of View Handler pattern}
    \label{VH4P}
\end{figure}

In the following, we verify the View Handler pattern. We assume all data elements $d_{I}$, $d_{I_{S}}$, $d_{I_{V_i}}$, $d_{O_{S_i}}$, $d_{O_{i}}$ (for $1\leq i\leq n$) are from a finite set
$\Delta$.

The state transitions of the ViewHandler module
described by APTC are as follows.

$VH=\sum_{d_{I}\in\Delta}(r_{I}(d_{I})\cdot VH_{2})$

$VH_{2}=VHF\cdot VH_{3}$

$VH_{3}=\sum_{d_{I_S}\in\Delta}(s_{VS}(d_{I_S})\cdot VH_{4})$

$VH_{4}=\sum_{d_{I_{V_1}},\cdots,d_{I_{V_n}}\in\Delta}(s_{VV_1}(d_{I_{V_1}})\between\cdots\between s_{VV_n}(d_{I_{V_n}})\cdot VH)$

The state transitions of the Supplier described by APTC are as follows.

$S=\sum_{d_{I_{M}}\in\Delta}(r_{VS}(d_{I_{S}})\cdot S_{2})$

$S_{2}=SF\cdot S_{3}$

$S_{3}=\sum_{d_{O_{S_1}},\cdots,d_{O_{S_n}}\in\Delta}(s_{SV_1}(d_{O_{S_1}})\between\cdots\between s_{SV_n}(d_{O_{S_n}})\cdot S)$

The state transitions of the View $i$ described by APTC are as follows.

$V_i=\sum_{d_{I_{V_i}}\in\Delta}(r_{VV_i}(d_{I_{V_i}})\cdot V_{i_2})$

$V_{i_2}=VF_{i_1}\cdot V_{i_3}$

$V_{i_3}=\sum_{d_{O_{S_i}}\in\Delta}(r_{SV_i}(d_{O_{S_i}})\cdot V_{i_4})$

$V_{i_4}=VF_{i_2}\cdot V_{i_5}$

$V_{i_5}=\sum_{d_{O_{i}}\in\Delta}(s_{O_{i}}(d_{O_i})\cdot V_i)$

The sending action and the reading action of the same data through the same channel can communicate with each other, otherwise, will cause a deadlock $\delta$. We define the following
communication functions of the View $i$ for $1\leq i\leq n$.

$$\gamma(r_{VV_i}(d_{I_{V_i}}),s_{VV_i}(d_{I_{V_i}}))\triangleq c_{VV_i}(d_{I_{V_i}})$$
$$\gamma(r_{SV_i}(d_{O_{S_i}}),s_{SV_i}(d_{O_{S_i}}))\triangleq c_{SV_i}(d_{O_{S_i}})$$

There are one communication functions between the ViewHandler and the Supplier as follows.

$$\gamma(r_{VS}(d_{I_S}),s_{VS}(d_{I_S}))\triangleq c_{VS}(d_{I_S})$$

Let all modules be in parallel, then the View Handler pattern $VH\quad S\quad V_1\cdots V_i\cdots V_n$ can be presented by the following process term.

$\tau_I(\partial_H(\Theta(VH\between S\between V_1\between\cdots\between V_i\between\cdots\between V_n)))=\tau_I(\partial_H(VH\between S\between V_1\between\cdots\between V_i\between\cdots\between V_n))$

where $H=\{r_{VV_i}(d_{I_{V_i}}),s_{VV_i}(d_{I_{V_i}}),r_{SV_i}(d_{O_{S_i}}),s_{SV_i}(d_{O_{S_i}}),r_{VS}(d_{I_S}),s_{VS}(d_{I_S})\\
|d_{I}, d_{I_{S}}, d_{I_{V_i}}, d_{O_{S_i}}, d_{O_{i}}\in\Delta\}$ for $1\leq i\leq n$,

$I=\{c_{VV_i}(d_{I_{V_i}}),c_{VV_i}(d_{O_{M_i}}),c_{VS}(d_{I_S}),VHF,SF,VF_{1_1},VF_{1_2},\cdots,VF_{n_1},VF_{n_2}\\
|d_{I}, d_{I_{S}}, d_{I_{V_i}}, d_{O_{S_i}}, d_{O_{i}}\in\Delta\}$ for $1\leq i\leq n$.

Then we get the following conclusion on the View Handler pattern.

\begin{theorem}[Correctness of the View Handler pattern]
The View Handler pattern $\tau_I(\partial_H(VH\between S\between V_1\between\cdots\between V_i\between\cdots\between V_n))$ can exhibit desired external behaviors.
\end{theorem}

\begin{proof}
Based on the above state transitions of the above modules, by use of the algebraic laws of APTC, we can prove that

$\tau_I(\partial_H(VH\between S\between V_1\between\cdots\between V_i\between\cdots\between V_n))=\sum_{d_{I},d_{O_1},\cdots,d_{O_n}\in\Delta}(r_{I}(d_{I})\cdot s_{O_1}(d_{O_1})\parallel\cdots\parallel s_{O_i}(d_{O_i})\parallel\cdots\parallel s_{O_n}(d_{O_n}))\cdot
\tau_I(\partial_H(VH\between S\between V_1\between\cdots\between V_i\between\cdots\between V_n))$,

that is, the View Handler pattern $\tau_I(\partial_H(VH\between S\between V_1\between\cdots\between V_i\between\cdots\between V_n))$ can exhibit desired external behaviors.

For the details of proof, please refer to section \ref{app}, and we omit it.
\end{proof}

\subsection{Communication}\label{C4}

\subsubsection{Verification of the Forwarder-Receiver Pattern}

The Forwarder-Receiver pattern decouples the communication of two communicating peers. There are six modules in the Forwarder-Receiver pattern: the two Peers, the two Forwarders,
and the two Receivers. The Peers interact with the user through
the channels $I_1$, $I_2$ and $O_1$, $O_2$; with the Forwarder through the channels $PF_1$ and $PF_2$. The Receivers interact with Forwarders the through the channels $FR_1$ and $FR_2$,
and with the Peers through the channels $RP_1$ and $RP_2$. As illustrates in
Figure \ref{FR4}.

\begin{figure}
    \centering
    \includegraphics{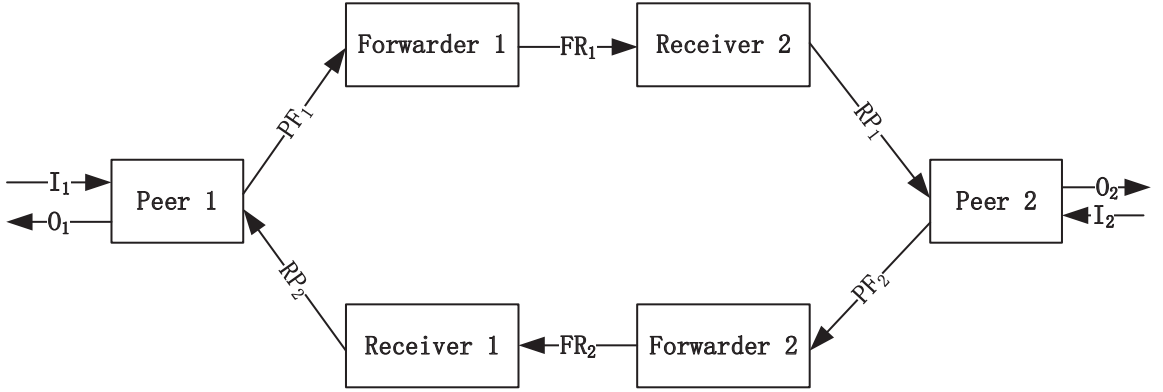}
    \caption{Forwarder-Receiver pattern}
    \label{FR4}
\end{figure}

The typical process of the Forwarder-Receiver pattern is shown in Figure \ref{FR4P} and as follows.

\begin{enumerate}
  \item The Peer 1 receives the request $d_{I_1}$ from the user through the channel $I_1$ (the corresponding reading action is denoted $r_{I_1}(d_{I_1})$), then processes the request $d_{I_1}$ through a processing
  function $P1F_1$, and sends the processed request $d_{I_{F_1}}$ to the Forwarder 1 through the channel $PF_1$ (the corresponding sending action is denoted $s_{PF_1}(d_{I_{F_1}})$);
  \item The Forwarder 1 receives $d_{I_{F_1}}$ from the Peer 1 through the channel $PF_1$ (the corresponding reading action is denoted $r_{PF_1}(d_{I_{F_1}})$), then processes the request
  through a processing function $F1F$, generates and sends the processed request $d_{I_{R_2}}$ to the Receiver 2 through the channel $FR_1$ (the corresponding sending action is denoted
  $s_{FR_1}(d_{I_{R_2}})$);
  \item The Receiver 2 receives the request $d_{I_{R_2}}$ from the Forwarder 1 through the channel $FR_1$ (the corresponding reading action is denoted $r_{FR_1}(d_{I_{R_2}})$), then processes
  the request through a processing function $R2F$, generates and sends the processed request $d_{I_{P_2}}$ to the Peer 2 through the channel $RP_1$ (the corresponding sending action is denoted $s_{RP_1}(d_{I_{P_2}})$);
  \item The Peer 2 receives the request $d_{I_{P_2}}$ from the Receiver 2 through the channel $RP_1$ (the corresponding reading action is denoted $r_{RP_1}(d_{I_{P_2}})$), then
  processes the request and generates the response $d_{O_2}$ through a processing function $P2F_2$, and sends the response to the outside through the channel $O_{2}$ (the corresponding sending action is denoted
  $s_{O_{2}}(d_{O_2})$).
\end{enumerate}

There is another symmetric process from Peer 2 to Peer 1, we omit it.

\begin{figure}
    \centering
    \includegraphics{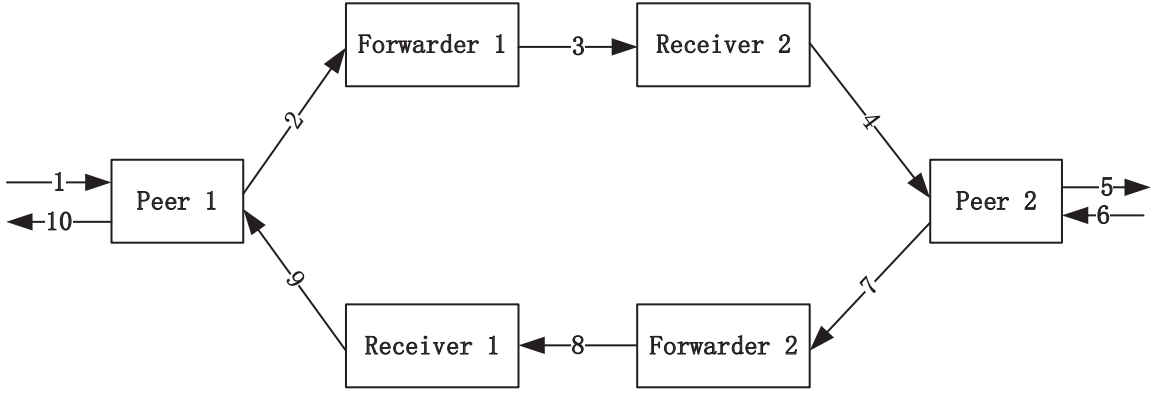}
    \caption{Typical process of Forwarder-Receiver pattern}
    \label{FR4P}
\end{figure}

In the following, we verify the Forwarder-Receiver pattern. We assume all data elements $d_{I_1}$, $d_{I_2}$, $d_{I_{F_1}}$, $d_{I_{F_2}}$, $d_{I_{R_1}}$, $d_{I_{R_2}}$, $d_{I_(P_1)}$, $d_{I_{P_2}}$, $d_{O_1}$, $d_{O_{2}}$ are from a finite set
$\Delta$. We only give the transitions of the first process.

The state transitions of the Peer 1 module
described by APTC are as follows.

$P1=\sum_{d_{I_1}\in\Delta}(r_{I_1}(d_{I_1})\cdot P1_{2})$

$P1_{2}=P1F_1\cdot P1_{3}$

$P1_{3}=\sum_{d_{I_{F_1}}\in\Delta}(s_{PF_1}(d_{I_{F_1}})\cdot P1)$

The state transitions of the Forwarder 1 module
described by APTC are as follows.

$F1=\sum_{d_{I_{F_1}}\in\Delta}(r_{PF_1}(d_{I_{F_1}})\cdot F1_{2})$

$F1_{2}=F1F\cdot F1_{3}$

$F1_{3}=\sum_{d_{I_{R_2}}\in\Delta}(s_{FR_1}(d_{I_{R_2}})\cdot F1)$

The state transitions of the Receiver 2 module
described by APTC are as follows.

$R2=\sum_{d_{I_{R_2}}\in\Delta}(r_{FR_1}(d_{I_{R_2}})\cdot R2_{2})$

$R2_{2}=R2F\cdot R2_{3}$

$R2_{3}=\sum_{d_{I_{P_2}}\in\Delta}(s_{RP_1}(d_{I_{P_2}})\cdot R2)$

The state transitions of the Peer 2 module
described by APTC are as follows.

$P2=\sum_{d_{I_{P2}}\in\Delta}(r_{RP_1}(d_{I_{P_2}})\cdot P2_{2})$

$P2_{2}=P2F_2\cdot P2_{3}$

$P2_{3}=\sum_{d_{O_2}\in\Delta}(s_{O_2}(d_{O_2})\cdot P2)$

The sending action and the reading action of the same data through the same channel can communicate with each other, otherwise, will cause a deadlock $\delta$. We define the following
communication functions between the Peer 1 and the Forwarder 1.

$$\gamma(r_{PF_1}(d_{I_{F_1}}),s_{PF_1}(d_{I_{F_1}}))\triangleq c_{PF_1}(d_{I_{F_1}})$$

There are one communication functions between the Forwarder 1 and the Receiver 2 as follows.

$$\gamma(r_{FR_1}(d_{I_{R_2}}),s_{FR_1}(d_{I_{R_2}}))\triangleq c_{FR_1}(d_{I_{R_2}})$$

There are one communication functions between the Receiver 2 and the Peer 2 as follows.

$$\gamma(r_{RP_1}(d_{I_{P_2}}),s_{RP_1}(d_{I_{P_2}}))\triangleq c_{RP_1}(d_{I_{P_2}})$$

We define the following communication functions between the Peer 2 and the Forwarder 2.

$$\gamma(r_{PF_2}(d_{I_{F_2}}),s_{PF_2}(d_{I_{F_2}}))\triangleq c_{PF_2}(d_{I_{F_2}})$$

There are one communication functions between the Forwarder 2 and the Receiver 1 as follows.

$$\gamma(r_{FR_2}(d_{I_{R_1}}),s_{FR_2}(d_{I_{R_1}}))\triangleq c_{FR_2}(d_{I_{R_1}})$$

There are one communication functions between the Receiver 1 and the Peer 1 as follows.

$$\gamma(r_{RP_2}(d_{I_{P_1}}),s_{RP_2}(d_{I_{P_1}}))\triangleq c_{RP_2}(d_{I_{P_1}})$$

Let all modules be in parallel, then the Forwarder-Receiver pattern $P1\quad F1 \quad R1\quad R2\quad F2\quad P2$ can be presented by the following process term.

$\tau_I(\partial_H(\Theta(P1\between F1\between R1\between R2\between F2\between P2)))=\tau_I(\partial_H(P1\between F1\between R1\between R2\between F2\between P2))$

where $H=\{r_{PF_1}(d_{I_{F_1}}),s_{PF_1}(d_{I_{F_1}}),r_{FR_1}(d_{I_{R_2}}),s_{FR_1}(d_{I_{R_2}}),r_{RP_1}(d_{I_{P_2}}),s_{RP_1}(d_{I_{P_2}}),\\
r_{PF_2}(d_{I_{F_2}}),s_{PF_2}(d_{I_{F_2}}),r_{FR_2}(d_{I_{R_1}}),s_{FR_2}(d_{I_{R_1}}),r_{RP_2}(d_{I_{P_1}}),s_{RP_2}(d_{I_{P_1}})\\
|d_{I_1}, d_{I_2}, d_{I_{F_1}}, d_{I_{F_2}}, d_{I_{R_1}}, d_{I_{R_2}}, d_{I_(P_1)}, d_{I_{P_2}}, d_{O_1}, d_{O_{2}}\in\Delta\}$,

$I=\{c_{PF_1}(d_{I_{F_1}}),c_{FR_1}(d_{I_{R_2}}),c_{RP_1}(d_{I_{P_2}}),c_{PF_2}(d_{I_{F_2}}),c_{FR_2}(d_{I_{R_1}}),c_{RP_2}(d_{I_{P_1}}),\\
P1F_1,P1F_2,P2F_1,P2F_2,F1F,F2F,R1F,R2F
|d_{I_1}, d_{I_2}, d_{I_{F_1}}, d_{I_{F_2}}, d_{I_{R_1}}, d_{I_{R_2}}, d_{I_(P_1)}, d_{I_{P_2}}, d_{O_1}, d_{O_{2}}\in\Delta\}$.

Then we get the following conclusion on the Forwarder-Receiver pattern.

\begin{theorem}[Correctness of the Forwarder-Receiver pattern]
The Forwarder-Receiver pattern $\tau_I(\partial_H(P1\between F1\between R1\between R2\between F2\between P2))$ can exhibit desired external behaviors.
\end{theorem}

\begin{proof}
Based on the above state transitions of the above modules, by use of the algebraic laws of APTC, we can prove that

$\tau_I(\partial_H(P1\between F1\between R1\between R2\between F2\between P2))=\sum_{d_{I_1},d_{I_2},d_{O_1},d_{O_2}\in\Delta}((r_{I_1}(d_{I_1})\cdot s_{O_2}(d_{O_2}))\parallel(r_{I_2}(d_{I_2})\cdot s_{O_1}(d_{O_1})))\cdot
\tau_I(\partial_H(P1\between F1\between R1\between R2\between F2\between P2))$,

that is, the Forwarder-Receiver pattern $\tau_I(\partial_H(P1\between F1\between R1\between R2\between F2\between P2))$ can exhibit desired external behaviors.

For the details of proof, please refer to section \ref{app}, and we omit it.
\end{proof}

\subsubsection{Verification of the Client-Dispatcher-Server Pattern}

The Client-Dispatcher-Server pattern decouples the invocation of the client and the server to introduce an intermediate dispatcher.
There are three modules in the Client-Dispatcher-Server pattern: the Client, the Dispatcher,
and the Server. The Client interacts with the user through
the channels $I$ and $O$; with the Dispatcher through the channels $I_{CD}$ and $O_{CD}$; with the Server through the channels $I_{CS}$ and $O_{CS}$. As illustrates in
Figure \ref{CDS4}.

\begin{figure}
    \centering
    \includegraphics{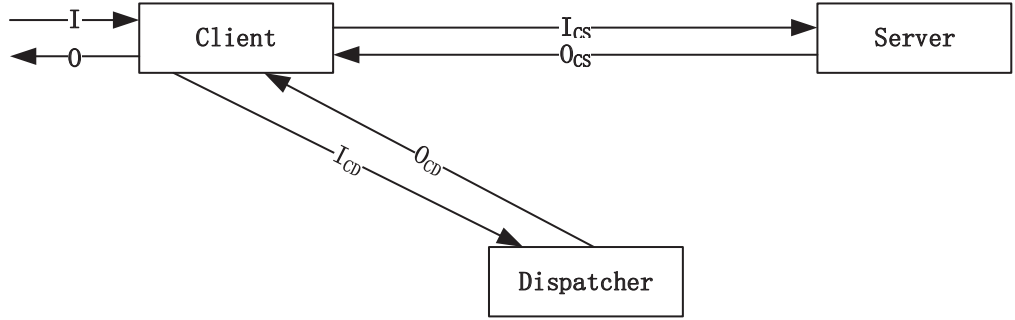}
    \caption{Client-Dispatcher-Server pattern}
    \label{CDS4}
\end{figure}

The typical process of the Client-Dispatcher-Server pattern is shown in Figure \ref{CDS4P} and as follows.

\begin{enumerate}
  \item The Client receives the request $d_{I}$ from the user through the channel $I$ (the corresponding reading action is denoted $r_{I}(d_{I})$), then processes the request $d_{I}$ through a processing
  function $CF_1$, and sends the processed request $d_{I_{D}}$ to the Dispatcher through the channel $I_{CD}$ (the corresponding sending action is denoted $s_{I_CD}(d_{I_{D}})$);
  \item The Dispatcher receives $d_{I_{D}}$ from the Client through the channel $I_{CD}$ (the corresponding reading action is denoted $r_{I_{CD}}(d_{I_{D}})$), then processes the request
  through a processing function $DF$, generates and sends the processed response $d_{O_{D}}$ to the Client through the channel $O_{CD}$ (the corresponding sending action is denoted
  $s_{O_{CD}}(d_{O_{D}})$);
  \item The Client receives the response $d_{O_D}$ from the Dispatcher through the channel $O_{CD}$ (the corresponding reading action is denoted $r_{O_{CD}}(d_{O_D})$), then processes
  the request through a processing function $CF_2$, generates and sends the processed request $d_{I_{S}}$ to the Server through the channel $I_{CS}$ (the corresponding sending action is denoted $s_{I_{CS}}(d_{I_{S}})$);
  \item The Server receives the request $d_{I_{S}}$ from the Client through the channel $I_{CS}$ (the corresponding reading action is denoted $r_{I_{CS}}(d_{I_{S}})$), then
  processes the request and generates the response $d_{O_S}$ through a processing function $SF$, and sends the response to the outside through the channel $O_{CS}$ (the corresponding sending action is denoted
  $s_{O_{CS}}(d_{O_S})$);
  \item The Client receives the response $d_{O_S}$ from the Server through the channel $O_{CS}$ (the corresponding reading action is denoted $r_{O_{CS}}(d_{O_S})$), then processes
  the request through a processing function $CF_3$, generates and sends the processed response $d_{O}$ to the user through the channel $O$ (the corresponding sending action is denoted $s_{O}(d_{O})$).
\end{enumerate}

\begin{figure}
    \centering
    \includegraphics{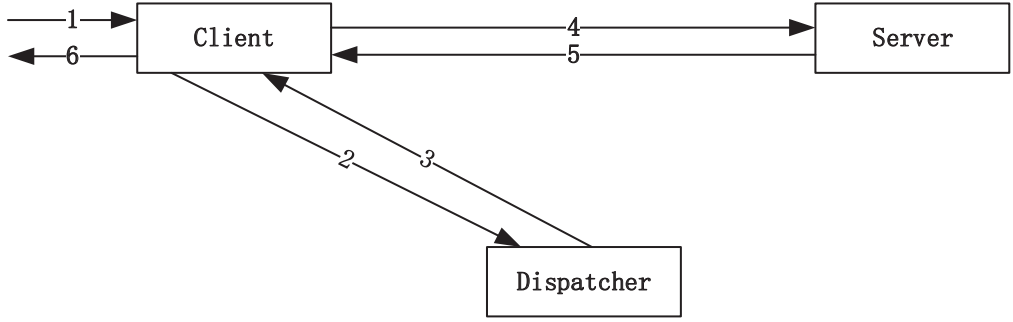}
    \caption{Typical process of Client-Dispatcher-Server pattern}
    \label{CDS4P}
\end{figure}

In the following, we verify the Client-Dispatcher-Server pattern. We assume all data elements $d_{I}$, $d_{I_D}$, $d_{I_{S}}$, $d_{O_{D}}$, $d_{O_{S}}$, $d_{O}$ are from a finite set
$\Delta$.

The state transitions of the Client module
described by APTC are as follows.

$C=\sum_{d_{I}\in\Delta}(r_{I}(d_{I})\cdot C_{2})$

$C_{2}=CF_1\cdot C_{3}$

$C_{3}=\sum_{d_{I_{D}}\in\Delta}(s_{I_{CD}}(d_{I_{D}})\cdot C_4)$

$C_4=\sum_{d_{O_D}\in\Delta}(r_{O_{CD}}(d_{O_D})\cdot C_{5})$

$C_{5}=CF_2\cdot C_{6}$

$C_{6}=\sum_{d_{I_{S}}\in\Delta}(s_{I_{CS}}(d_{I_{S}})\cdot C_7)$

$C_7=\sum_{d_{O_S}\in\Delta}(r_{O_{CS}}(d_{O_S})\cdot C_{8})$

$C_{8}=CF_3\cdot C_{9}$

$C_{9}=\sum_{d_{O}\in\Delta}(s_{O}(d_{O})\cdot C)$

The state transitions of the Dispatcher module
described by APTC are as follows.

$D=\sum_{d_{I_{D}}\in\Delta}(r_{I_{CD}}(d_{I_{D}})\cdot D_{2})$

$D_{2}=DF\cdot D_{3}$

$D_{3}=\sum_{d_{O_D}\in\Delta}(s_{O_{CD}}(d_{O_D})\cdot D)$

The state transitions of the Server module
described by APTC are as follows.

$S=\sum_{d_{I_{S}}\in\Delta}(r_{I_{CS}}(d_{I_S})\cdot S_{2})$

$S_{2}=SF\cdot S_{3}$

$S_{3}=\sum_{d_{O_S}\in\Delta}(s_{O_{CS}}(d_{O_{S}})\cdot S)$

The sending action and the reading action of the same data through the same channel can communicate with each other, otherwise, will cause a deadlock $\delta$. We define the following
communication functions between the Client and the Dispatcher.

$$\gamma(r_{I_{CD}}(d_{I_{D}}),s_{I_{CD}}(d_{I_{D}}))\triangleq c_{I_{CD}}(d_{I_{D}})$$

$$\gamma(r_{O_{CD}}(d_{O_D}),s_{O_{CD}}(d_{O_D}))\triangleq c_{O_{CD}}(d_{O_D})$$

There are two communication functions between the Client and the Server as follows.

$$\gamma(r_{I_{CS}}(d_{I_{S}}),s_{I_{CS}}(d_{I_{S}}))\triangleq c_{I_{CS}}(d_{I_{S}})$$

$$\gamma(r_{O_{CS}}(d_{O_S}),s_{O_{CS}}(d_{O_S}))\triangleq c_{O_{CS}}(d_{O_S})$$

Let all modules be in parallel, then the Client-Dispatcher-Server pattern $C\quad D \quad S$ can be presented by the following process term.

$\tau_I(\partial_H(\Theta(C\between D\between S)))=\tau_I(\partial_H(C\between D\between S))$

where $H=\{r_{I_{CD}}(d_{I_{D}}),s_{I_{CD}}(d_{I_{D}}),r_{O_{CD}}(d_{O_D}),s_{O_{CD}}(d_{O_D}),r_{I_{CS}}(d_{I_{S}}),s_{I_{CS}}(d_{I_{S}}),\\
r_{O_{CS}}(d_{O_S}),s_{O_{CS}}(d_{O_S})
|d_{I}, d_{I_D}, d_{I_{S}}, d_{O_{D}}, d_{O_{S}}, d_{O}\in\Delta\}$,

$I=\{c_{I_{CD}}(d_{I_{D}}),c_{O_{CD}}(d_{O_D}),c_{I_{CS}}(d_{I_{S}}),c_{O_{CS}}(d_{O_S}),CF_1,CF_2,CF_3,DF,SF\\
|d_{I}, d_{I_D}, d_{I_{S}}, d_{O_{D}}, d_{O_{S}}, d_{O}\in\Delta\}$.

Then we get the following conclusion on the Client-Dispatcher-Server pattern.

\begin{theorem}[Correctness of the Client-Dispatcher-Server pattern]
The Client-Dispatcher-Server pattern $\tau_I(\partial_H(C\between D\between S))$ can exhibit desired external behaviors.
\end{theorem}

\begin{proof}
Based on the above state transitions of the above modules, by use of the algebraic laws of APTC, we can prove that

$\tau_I(\partial_H(C\between D\between S))=\sum_{d_{I},d_{O}\in\Delta}(r_{I}(d_{I})\cdot s_{O}(d_{O}))\cdot
\tau_I(\partial_H(C\between D\between S))$,

that is, the Client-Dispatcher-Server pattern $\tau_I(\partial_H(C\between D\between S))$ can exhibit desired external behaviors.

For the details of proof, please refer to section \ref{app}, and we omit it.
\end{proof}

\subsubsection{Verification of the Publisher-Subscriber Pattern}

The Publisher-Subscriber pattern decouples the communication of the publisher and subscriber. There are four modules in the Publisher-Subscriber pattern: the Publisher, the Publisher Proxy,
the Subscriber Proxy and the Subscriber. The Publisher interacts with the outside through
the channel $I$; with the Publisher Proxy through the channel $PP$. The Publisher Proxy interacts with the Subscriber Proxy through the channel $PS$. The Subscriber interacts with the
Subscriber Proxy through the channel $SS$, and with the outside through the channels $O$. As illustrates in
Figure \ref{PS4}.

\begin{figure}
    \centering
    \includegraphics{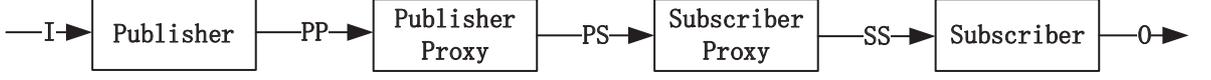}
    \caption{Publisher-Subscriber pattern}
    \label{PS4}
\end{figure}

The typical process of the Publisher-Subscriber pattern is shown in Figure \ref{PS4P} and as follows.

\begin{enumerate}
  \item The Publisher receives the input $d_{I}$ from the outside through the channel $I$ (the corresponding reading action is denoted $r_{I}(d_{I})$), then processes the input $d_{I}$ through a processing
  function $PF$, and sends the processed input $d_{I_{PP}}$ to the Publisher Proxy through the channel $PP$ (the corresponding sending action is denoted $s_{PP}(d_{I_{PP}})$);
  \item The Publisher Proxy receives $d_{I_{PP}}$ from the Publisher through the channel $PP$ (the corresponding reading action is denoted $r_{PP}(d_{I_{PP}})$), then processes the request
  through a processing function $PPF$, generates and sends the processed input $d_{I_{SP}}$ to the Subscriber Proxy through the channel $PS$ (the corresponding sending action is denoted
  $s_{PS}(d_{I_{SP}})$);
  \item The Subscriber Proxy receives the input $d_{I_{SP}}$ from the Publisher Proxy through the channel $PS$ (the corresponding reading action is denoted $r_{PS}(d_{I_{SP}})$), then processes
  the request through a processing function $SPF$, generates and sends the processed input $d_{I_{S}}$ to the Subscriber through the channel $SS$ (the corresponding sending action is denoted $s_{SS}(d_{I_{S}})$);
  \item The Subscriber receives the input $d_{I_{S}}$ from the Subscriber Proxy through the channel $SS$ (the corresponding reading action is denoted $r_{SS}(d_{I_{S}})$), then
  processes the request and generates the response $d_{O}$ through a processing function $SF$, and sends the response to the outside through the channel $O$ (the corresponding sending action is denoted
  $s_{O}(d_{O})$).
\end{enumerate}

\begin{figure}
    \centering
    \includegraphics{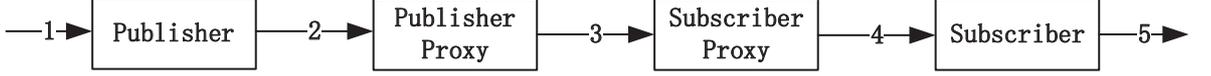}
    \caption{Typical process of Publisher-Subscriber pattern}
    \label{PS4P}
\end{figure}

In the following, we verify the Publisher-Subscriber pattern. We assume all data elements $d_{I}$, $d_{I_{PP}}$, $d_{I_{SP}}$, $d_{I_{S}}$, $d_{O}$ are from a finite set
$\Delta$.

The state transitions of the Publisher module
described by APTC is are follows.

$P=\sum_{d_{I}\in\Delta}(r_{I}(d_{I})\cdot P_{2})$

$P_{2}=PF\cdot P_{3}$

$P_{3}=\sum_{d_{I_{PP}}\in\Delta}(s_{PP}(d_{I_{PP}})\cdot P)$

The state transitions of the Publisher Proxy module
described by APTC are as follows.

$PP=\sum_{d_{I_{PP}}\in\Delta}(r_{PP}(d_{I_{PP}})\cdot PP_{2})$

$PP_{2}=PPF\cdot PP_{3}$

$PP_{3}=\sum_{d_{I_{SP}}\in\Delta}(s_{PS}(d_{I_{SP}})\cdot PP)$

The state transitions of the Subscriber Proxy module
described by APTC are as follows.

$SP=\sum_{d_{I_{SP}}\in\Delta}(r_{PS}(d_{I_{SP}})\cdot SP_{2})$

$SP_{2}=SPF\cdot SP_{3}$

$SP_{3}=\sum_{d_{I_{S}}\in\Delta}(s_{SS}(d_{I_{S}})\cdot SP)$

The state transitions of the Subscriber module
described by APTC are as follows.

$S=\sum_{d_{I_{S}}\in\Delta}(r_{SS}(d_{I_{S}})\cdot S_{2})$

$S_{2}=SF\cdot S_{3}$

$S_{3}=\sum_{d_{O}\in\Delta}(s_{O}(d_{O})\cdot S)$

The sending action and the reading action of the same data through the same channel can communicate with each other, otherwise, will cause a deadlock $\delta$. We define the following
communication functions between the Publisher and the Publisher Proxy.

$$\gamma(r_{PP}(d_{I_{PP}}),s_{PP}(d_{I_{PP}}))\triangleq c_{PP}(d_{I_{PP}})$$

There are one communication functions between the Publisher Proxy and the Subscriber Proxy as follows.

$$\gamma(r_{PS}(d_{I_{SP}}),s_{PS}(d_{I_{SP}}))\triangleq c_{PS}(d_{I_{SP}})$$

There are one communication functions between the Subscriber Proxy and the Subscriber as follows.

$$\gamma(r_{SS}(d_{I_{S}}),s_{SS}(d_{I_{S}}))\triangleq c_{SS}(d_{I_{S}})$$

Let all modules be in parallel, then the Publisher-Subscriber pattern $P\quad PP \quad SP\quad S$ can be presented by the following process term.

$\tau_I(\partial_H(\Theta(P\between PP\between SP\between S)))=\tau_I(\partial_H(P\between PP\between SP\between S))$

where $H=\{r_{PP}(d_{I_{PP}}),s_{PP}(d_{I_{PP}}),r_{PS}(d_{I_{SP}}),s_{PS}(d_{I_{SP}}),r_{SS}(d_{I_{S}}),s_{SS}(d_{I_{S}})\\
|d_{I}, d_{I_{PP}}, d_{I_{SP}}, d_{I_{S}}, d_{O}\in\Delta\}$,

$I=\{c_{PP}(d_{I_{PP}}),c_{PS}(d_{I_{SP}}),c_{SS}(d_{I_{S}}),PF,PPF,SPF,SF
|d_{I}, d_{I_{PP}}, d_{I_{SP}}, d_{I_{S}}, d_{O}\in\Delta\}$.

Then we get the following conclusion on the Publisher-Subscriber pattern.

\begin{theorem}[Correctness of the Publisher-Subscriber pattern]
The Publisher-Subscriber pattern $\tau_I(\partial_H(P\between PP\between SP\between S))$ can exhibit desired external behaviors.
\end{theorem}

\begin{proof}
Based on the above state transitions of the above modules, by use of the algebraic laws of APTC, we can prove that

$\tau_I(\partial_H(P\between PP\between SP\between S))=\sum_{d_{I},d_{O}\in\Delta}(r_{I}(d_{I})\cdot s_{O}(d_{O}))\cdot
\tau_I(\partial_H(P\between PP\between SP\between S))$,

that is, the Publisher-Subscriber pattern $\tau_I(\partial_H(P\between PP\between SP\between S))$ can exhibit desired external behaviors.

For the details of proof, please refer to section \ref{app}, and we omit it.
\end{proof}

\newpage\section{Verification of Idioms}

Idioms are the lowest-level patterns which are programming language-specific and implement some specific concrete problems.

There are almost numerous language-specific idioms, in this chapter, we only verify two idioms called the Singleton pattern and the Counted Pointer pattern.

\subsection{Verification of the Singleton Pattern}

The Singleton pattern ensures that there only one instance in runtime for an object. In Singleton pattern, there is only one module: The Singleton. The Singleton interacts with the
outside through the input channels $I_i$ and the output channels $O_i$ for $1\leq i\leq n$, as illustrated in Figure \ref{Sin5}.

\begin{figure}
    \centering
    \includegraphics{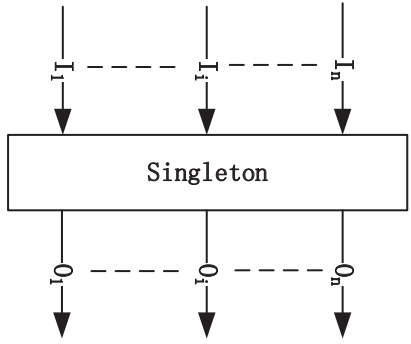}
    \caption{Singleton pattern}
    \label{Sin5}
\end{figure}

The typical process is shown in Figure \ref{Sin5P} and as follows.

\begin{enumerate}
  \item The Singleton receives the input $d_{I_i}$ from the outside through the channel $I_i$ (the corresponding reading action is denoted $r_{I_i}(d_{I_i})$);
  \item Then it processes the input and generates the output $d_{O_i}$ through a processing function $SF_i$;
  \item Then it sends the output to the outside through the channel $O_i$ (the corresponding sending action is denoted $s_{O_i}(d_{O_i})$).
\end{enumerate}

\begin{figure}
    \centering
    \includegraphics{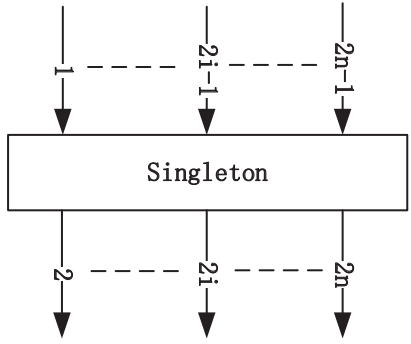}
    \caption{Typical process of Singleton pattern}
    \label{Sin5P}
\end{figure}

In the following, we verify the Singleton pattern. We assume all data elements $d_{I_i}$, $d_{O_i}$ for $1\leq i\leq n$ are from a finite set
$\Delta$.

The state transitions of the Singleton module
described by APTC are as follows.

$S=\sum_{d_{I_1},\cdots,d_{I_n}\in\Delta}(r_{I_1}(d_{I_1})\between\cdots\between r_{I_n}(d_{I_n})\cdot S_{2})$

$S_{2}=SF_1\between\cdots\between SF_n\cdot S_{3}$

$S_{3}=\sum_{d_{O_1},\cdots,d_{O_n}\in\Delta}(s_{O_1}(d_{O_1})\between\cdots\between s_{O_n}(d_{O_n})\cdot S)$

There is no communications in the Singleton pattern.

Let all modules be in parallel, then the Singleton pattern $S$ can be presented by the following process term.

$\tau_I(\partial_H(\Theta(S)))=\tau_I(\partial_H(S))$

where $H=\emptyset$, $I=\{SF_i\}$ for $1\leq i\leq n$.

Then we get the following conclusion on the Singleton pattern.

\begin{theorem}[Correctness of the Singleton pattern]
The Singleton pattern $\tau_I(\partial_H(S))$ can exhibit desired external behaviors.
\end{theorem}

\begin{proof}
Based on the above state transitions of the above modules, by use of the algebraic laws of APTC, we can prove that

$\tau_I(\partial_H(S))=\sum_{d_{I_1},d_{O_1},\cdots,d_{I_n},d_{O_n}\in\Delta}(r_{I_1}(d_{I_1})\parallel\cdots\parallel r_{I_n}(d_{I_n})\cdot s_{O_1}(d_{O_1})\parallel\cdots\parallel s_{O_n}(d_{O_n}))\cdot
\tau_I(\partial_H(S))$,

that is, the Singleton pattern $\tau_I(\partial_H(S))$ can exhibit desired external behaviors.

For the details of proof, please refer to section \ref{app}, and we omit it.
\end{proof}

\subsection{Verification of the Counted Pointer Pattern}

The Counted Pointer pattern makes memory management (implemented as Handle) of shared objects (implemented as Bodys) easier in C++. There are three modules in the Counted Pointer pattern: the Client, the Handle,
and the Body. The Client interacts with the outside through
the channels $I$ and $O$; with the Handle through the channel $I_{CH}$ and $O_{CH}$. The Handle interacts with the Body through the channels $I_{HB}$ and $O_{HB}$. As illustrates in
Figure \ref{CP5}.

\begin{figure}
    \centering
    \includegraphics{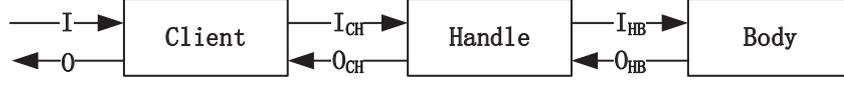}
    \caption{Counted Pointer pattern}
    \label{CP5}
\end{figure}

The typical process of the Counted Pointer pattern is shown in Figure \ref{CP5P} and as follows.

\begin{enumerate}
  \item The Client receives the input $d_{I}$ from the outside through the channel $I$ (the corresponding reading action is denoted $r_{I}(d_{I})$), then processes the input $d_{I}$ through a processing
  function $CF_1$, and sends the processed input $d_{I_{H}}$ to the Handle through the channel $I_{CH}$ (the corresponding sending action is denoted $s_{I_{CH}}(d_{I_{H}})$);
  \item The Handle receives $d_{I_{H}}$ from the Client through the channel $I_{CH}$ (the corresponding reading action is denoted $r_{I_{CH}}(d_{I_{H}})$), then processes the request
  through a processing function $HF_1$, generates and sends the processed input $d_{I_{B}}$ to the Body through the channel $I_{HB}$ (the corresponding sending action is denoted
  $s_{I_{HB}}(d_{I_{B}})$);
  \item The Body receives the input $d_{I_{B}}$ from the Handle through the channel $I_{HB}$ (the corresponding reading action is denoted $r_{I_{HB}}(d_{I_{B}})$), then processes
  the input through a processing function $BF$, generates and sends the response $d_{O_{B}}$ to the Handle through the channel $O_{HB}$ (the corresponding sending action is denoted $s_{O_{HB}}(d_{O_{B}})$);
  \item The Handle receives the response $d_{O_B}$ from the Body through the channel $O_{HB}$ (the corresponding reading action is denoted $r_{O_{HB}}(d_{O_B})$), then processes
  the response through a processing function $HF_2$, generates and sends the response $d_{O_H}$ (the corresponding sending action is denoted $s_{O_{CH}}(d_{O_H})$);
  \item The Client receives the response $d_{O_{H}}$ from the Handle through the channel $O_{CH}$ (the corresponding reading action is denoted $r_{O_{CH}}(d_{O_{H}})$), then
  processes the request and generates the response $d_{O}$ through a processing function $CF_2$, and sends the response to the outside through the channel $O$ (the corresponding sending action is denoted
  $s_{O}(d_{O})$).
\end{enumerate}

\begin{figure}
    \centering
    \includegraphics{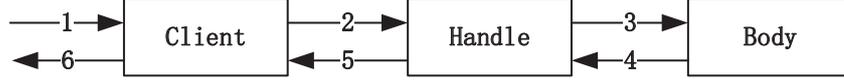}
    \caption{Typical process of Counted Pointer pattern}
    \label{CP5P}
\end{figure}

In the following, we verify the Counted Pointer pattern. We assume all data elements $d_{I}$, $d_{I_{H}}$, $d_{I_{B}}$, $d_{O_{B}}$, $d_{O_H}$, $d_{O}$ are from a finite set
$\Delta$.

The state transitions of the Client module
described by APTC are as follows.

$C=\sum_{d_{I}\in\Delta}(r_{I}(d_{I})\cdot C_{2})$

$C_{2}=CF_1\cdot C_{3}$

$C_{3}=\sum_{d_{I_{H}}\in\Delta}(s_{I_{CH}}(d_{I_{H}})\cdot C_4)$

$C_4=\sum_{d_{O_H}\in\Delta}(r_{O_{CH}}(d_{O_H})\cdot C_{5})$

$C_{5}=CF_2\cdot C_{6}$

$C_{6}=\sum_{d_{O}\in\Delta}(s_{O}(d_{O})\cdot C)$

The state transitions of the Handle module
described by APTC are as follows.

$H=\sum_{d_{I_{H}}\in\Delta}(r_{I_{CH}}(d_{I_{H}})\cdot H_{2})$

$H_{2}=HF_1\cdot H_{3}$

$H_{3}=\sum_{d_{I_{B}}\in\Delta}(s_{I_{HB}}(d_{I_{B}})\cdot H_4)$

$H_4=\sum_{d_{O_{B}}\in\Delta}(r_{O_{HB}}(d_{O_{B}})\cdot H_{5})$

$H_{5}=HF_2\cdot H_{6}$

$H_{6}=\sum_{d_{O_{H}}\in\Delta}(s_{O_{CH}}(d_{O_{H}})\cdot H)$

The state transitions of the Body module
described by APTC are as follows.

$B=\sum_{d_{I_{B}}\in\Delta}(r_{I_{HB}}(d_{I_B})\cdot B_{2})$

$B_{2}=BF\cdot B_{3}$

$B_{3}=\sum_{d_{O_{B}}\in\Delta}(s_{O_{HB}}(d_{O_{B}})\cdot B)$

The sending action and the reading action of the same data through the same channel can communicate with each other, otherwise, will cause a deadlock $\delta$. We define the following
communication functions between the Client and the Handle Proxy.

$$\gamma(r_{I_{CH}}(d_{I_{H}}),s_{I_{CH}}(d_{I_{H}}))\triangleq c_{I_{CH}}(d_{I_{H}})$$

$$\gamma(r_{O_{CH}}(d_{O_H}),s_{O_{CH}}(d_{O_H}))\triangleq c_{O_{CH}}(d_{O_H})$$

There are two communication functions between the Handle and the Body as follows.

$$\gamma(r_{I_{HB}}(d_{I_{B}}),s_{I_{HB}}(d_{I_{B}}))\triangleq c_{I_{HB}}(d_{I_{B}})$$

$$\gamma(r_{O_{HB}}(d_{O_{B}}),s_{O_{HB}}(d_{O_{B}}))\triangleq c_{O_{HB}}(d_{O_{B}})$$

Let all modules be in parallel, then the Counted Pointer pattern $C\quad H \quad B$ can be presented by the following process term.

$\tau_I(\partial_H(\Theta(C\between H\between B)))=\tau_I(\partial_H(C\between H\between B))$

where $H=\{r_{I_{CH}}(d_{I_{H}}),s_{I_{CH}}(d_{I_{H}}),r_{O_{CH}}(d_{O_H}),s_{O_{CH}}(d_{O_H}),r_{I_{HB}}(d_{I_{B}}),s_{I_{HB}}(d_{I_{B}}),\\
r_{O_{HB}}(d_{O_{B}}),s_{O_{HB}}(d_{O_{B}})
|d_{I}, d_{I_{H}}, d_{I_{B}}, d_{O_{B}}, d_{O_H}, d_{O}\in\Delta\}$,

$I=\{c_{I_{CH}}(d_{I_{H}}),c_{O_{CH}}(d_{O_H}),c_{I_{HB}}(d_{I_{B}}),c_{O_{HB}}(d_{O_{B}}),CF_1,CF_2,HF_1,HF_2,BF\\
|d_{I}, d_{I_{H}}, d_{I_{B}}, d_{O_{B}}, d_{O_H}, d_{O}\in\Delta\}$.

Then we get the following conclusion on the Counted Pointer pattern.

\begin{theorem}[Correctness of the Counted Pointer pattern]
The Counted Pointer pattern $\tau_I(\partial_H(C\between H\between B))$ can exhibit desired external behaviors.
\end{theorem}

\begin{proof}
Based on the above state transitions of the above modules, by use of the algebraic laws of APTC, we can prove that

$\tau_I(\partial_H(C\between H\between B))=\sum_{d_{I},d_{O}\in\Delta}(r_{I}(d_{I})\cdot s_{O}(d_{O}))\cdot
\tau_I(\partial_H(C\between H\between B))$,

that is, the Counted Pointer pattern $\tau_I(\partial_H(C\between H\between B))$ can exhibit desired external behaviors.

For the details of proof, please refer to section \ref{app}, and we omit it.
\end{proof}

\newpage\section{Verification of Patterns for Concurrent and Networked Objects}

Patterns for concurrent and networked objects can be used both in higher-level and lower-level systems and applications. 

In this chapter, we verify patterns for concurrent and networked objects. In section \ref{SAC6}, we verify service access and configuration patterns. In section \ref{EH6}, we verify
patterns related to event handling. We verify synchronization patterns in section \ref{S6} and concurrency patterns in section \ref{C6}.

\subsection{Service Access and Configuration Patterns}\label{SAC6}

In this subsection, we verify patterns for service access and configuration, including the Wrapper Facade pattern, the Component Configurator pattern, the Interceptor pattern, and the
Extension Interface pattern.

\subsubsection{Verification of the Wrapper Facade Pattern}

The Wrapper Facade pattern encapsulates the non-object-oriented APIs into the object-oriented ones. There are two classes of modules in the Wrapper Facade pattern: the Wrapper Facade
 and $n$ API Functions. The Wrapper Facade interacts with API Function $i$ through the channels $I_{WA_i}$ and $O_{WA_i}$, and it exchanges
information with outside through the input channel $I$ and $O$. As illustrates in Figure \ref{WF6}.

\begin{figure}
    \centering
    \includegraphics{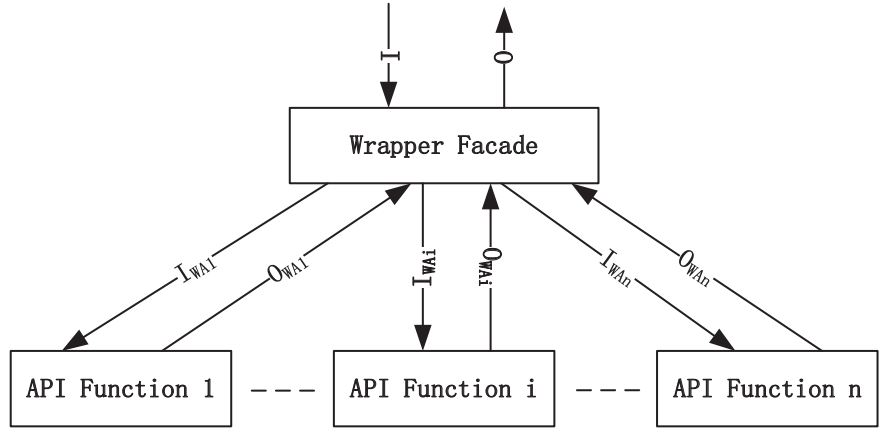}
    \caption{Wrapper Facade pattern}
    \label{WF6}
\end{figure}

The typical process of the Wrapper Facade pattern is shown in Figure \ref{WF6P} and as follows.

\begin{enumerate}
  \item The Wrapper Facade receives the input from the user through the channel $I$ (the corresponding reading action is denoted $r_{I}(d_{I})$), then
  processes the input through a processing function $WF_1$ and generates the input $d_{I_{A_i}}$, and sends $d_{I_{A_i}}$ to the API Function $i$ (for $1\leq i\leq n$) through the
  channel $I_{WA_i}$ (the corresponding sending action is denoted $s_{I_{WA_i}}(d_{I_{A_i}})$);
  \item The API Function $i$ receives the input $d_{I_{A_i}}$ from the Wrapper Facade through the channel $I_{WA_i}$ (the corresponding reading action is denoted $r_{I_{WA_i}}(d_{I_{A_i}})$),
  then processes the input through a processing function $AF_{i}$, and sends the results $d_{O_{A_{i}}}$ to the Wrapper Facade through the channel $O_{WA_i}$
  (the corresponding sending action is denoted $s_{O_{WA_i}}(d_{O_{A_{i}}})$);
  \item The Wrapper Facade receives the computational results from the API Function $i$ through the channel $O_{WA_i}$ (the corresponding reading action is denoted $r_{O_{WA_i}}(d_{O_{A_{i}}})$), then
  processes the results through a processing function $WF_2$ and generates the result $d_{O}$, and sends $d_{O}$ to the outside through the
  channel $O$ (the corresponding sending action is denoted $s_{O}(d_{O})$).
\end{enumerate}

\begin{figure}
    \centering
    \includegraphics{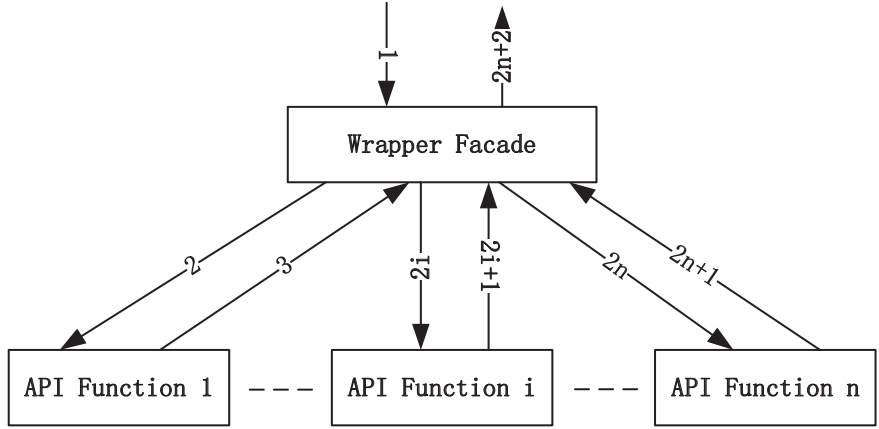}
    \caption{Typical process of Wrapper Facade pattern}
    \label{WF6P}
\end{figure}

In the following, we verify the Wrapper Facade pattern. We assume all data elements $d_{I}$, $d_{I_{A_i}}$, $d_{O_{A_{i}}}$, $d_{O}$ (for $1\leq i\leq n$) are from a finite set
$\Delta$.

The state transitions of the Wrapper Facade module
described by APTC are as follows.

$W=\sum_{d_{I}\in\Delta}(r_{I}(d_{I})\cdot W_{2})$

$W_{2}=WF_1\cdot W_{3}$

$W_{3}=\sum_{d_{I_{A_1}},\cdot,d_{I_{A_n}}\in\Delta}(s_{I_{WA_1}}(d_{I_{A_1}})\between\cdots\between s_{I_{WA_n}}(d_{I_{A_n}})\cdot W_{4})$

$W_{4}=\sum_{d_{O_{A_{1}}},\cdots,d_{O_{A_{n}}}\in\Delta}(r_{O_{WA_1}}(d_{O_{A_{1}}})\between\cdots\between r_{O_{WA_n}}(d_{O_{A_{n}}})\cdot W_{5})$

$W_5=WF_2\cdot W_6$

$W_{6}=\sum_{d_{O}\in\Delta}(s_{O}(d_{O})\cdot W)$

The state transitions of the API Function $i$ described by APTC are as follows.

$A_i=\sum_{d_{I_{A_i}}\in\Delta}(r_{I_{WA_i}}(d_{I_{A_i}})\cdot A_{i_2})$

$A_{i_2}=AF_{i}\cdot A_{i_3}$

$A_{i_3}=\sum_{d_{O_{A_{i}}}\in\Delta}(s_{O_{WA_{i}}}(d_{O_{A_{i}}})\cdot A_{i})$

The sending action and the reading action of the same data through the same channel can communicate with each other, otherwise, will cause a deadlock $\delta$. We define the following
communication functions of the API Function $i$ for $1\leq i\leq n$.

$$\gamma(r_{I_{WA_i}}(d_{I_{A_i}}),s_{I_{WA_i}}(d_{I_{A_i}}))\triangleq c_{I_{WA_i}}(d_{I_{A_i}})$$

$$\gamma(r_{O_{WA_{i}}}(d_{O_{A_{i}}}),s_{O_{WA_{i}}}(d_{O_{A_{i}}}))\triangleq c_{O_{WA_{i}}}(d_{O_{A_{i}}})$$

Let all modules be in parallel, then the Wrapper Facade pattern $W\quad A_1\cdots A_i\cdots A_n$ can be presented by the following process term.

$\tau_I(\partial_H(\Theta(W\between A_1\between\cdots\between A_i\between\cdots\between A_n)))=\tau_I(\partial_H(W\between A_1\between\cdots\between A_i\between\cdots\between A_n))$

where $H=\{r_{I_{WA_i}}(d_{I_{A_i}}),s_{I_{WA_i}}(d_{I_{A_i}}),r_{O_{WA_{i}}}(d_{O_{A_{i}}}),s_{O_{WA_{i}}}(d_{O_{A_{i}}})\\
|d_{I}, d_{I_{A_i}}, d_{O_{A_{i}}}, d_{O}\in\Delta\}$ for $1\leq i\leq n$,

$I=\{c_{I_{WA_i}}(d_{I_{A_i}}),c_{O_{WA_{i}}}(d_{O_{A_{i}}}),WF_1,WF_2,AF_{i}\\
|d_{I}, d_{I_{A_i}}, d_{O_{A_{i}}}, d_{O}\in\Delta\}$ for $1\leq i\leq n$.

Then we get the following conclusion on the Wrapper Facade pattern.

\begin{theorem}[Correctness of the Wrapper Facade pattern]
The Wrapper Facade pattern $\tau_I(\partial_H(W\between A_1\between\cdots\between A_i\between\cdots\between A_n))$ can exhibit desired external behaviors.
\end{theorem}

\begin{proof}
Based on the above state transitions of the above modules, by use of the algebraic laws of APTC, we can prove that

$\tau_I(\partial_H(W\between A_1\between\cdots\between A_i\between\cdots\between A_n))=\sum_{d_{I},d_O\in\Delta}(r_{I}(d_{I})\cdot s_{O}(d_{O}))\cdot
\tau_I(\partial_H(W\between A_1\between\cdots\between A_i\between\cdots\between A_n))$,

that is, the Wrapper Facade pattern $\tau_I(\partial_H(W\between A_1\between\cdots\between A_i\between\cdots\between A_n))$ can exhibit desired external behaviors.

For the details of proof, please refer to section \ref{app}, and we omit it.
\end{proof}

\subsubsection{Verification of the Component Configurator Pattern}

The Component Configurator pattern allows to configure the components dynamically. There are three classes of modules in the Component Configurator pattern: the Component Configurator
 and $n$ Components and the Component Repository. The Component Configurator interacts with Component $i$ through the channels $I_{CC_i}$ and $O_{CC_i}$, and it exchanges
information with outside through the input channel $I$ and $O$, and with the Component Repository through the channels $I_{CR}$ and $O_{CR}$. As illustrates in Figure \ref{CC6}.

\begin{figure}
    \centering
    \includegraphics{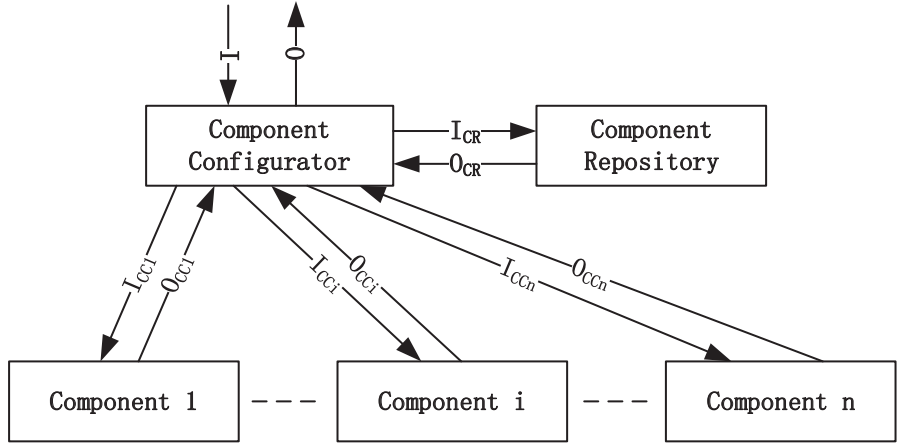}
    \caption{Component Configurator pattern}
    \label{CC6}
\end{figure}

The typical process of the Component Configurator pattern is shown in Figure \ref{CC6P} and as follows.

\begin{enumerate}
  \item The Component Configurator receives the input from the user through the channel $I$ (the corresponding reading action is denoted $r_{I}(d_{I})$), then
  processes the input through a processing function $CCF_1$ and generates the input $d_{I_{C_i}}$, and sends $d_{I_{C_i}}$ to the Component $i$ (for $1\leq i\leq n$) through the
  channel $I_{CC_i}$ (the corresponding sending action is denoted $s_{I_{CC_i}}(d_{I_{C_i}})$);
  \item The Component $i$ receives the input $d_{I_{C_i}}$ from the Component Configurator through the channel $I_{CC_i}$ (the corresponding reading action is denoted $r_{I_{CC_i}}(d_{I_{C_i}})$),
  then processes the input through a processing function $CF_{i}$, and sends the results $d_{O_{C_{i}}}$ to the Component Configurator through the channel $O_{CC_i}$
  (the corresponding sending action is denoted $s_{O_{CC_i}}(d_{O_{C_{i}}})$);
  \item The Component Configurator receives the configurational results from the Component $i$ through the channel $O_{CC_i}$ (the corresponding reading action is denoted $r_{O_{CC_i}}(d_{O_{C_{i}}})$), then
  processes the results through a processing function $CCF_2$ and generates the configurational information $d_{I_{R}}$, and sends $d_{I_R}$ to the Component Repository through the channel
  $I_{CR}$ (the corresponding sending action is denoted $s_{I_{CR}}(d_{I_R})$);
  \item The Component Repository receives the configurational information $d_{I_R}$ through the channel $I_{CR}$ (the corresponding reading action is denoted $r_{I_{CR}}(d_{I_R})$), then
  processes the information and generates the results $d_{O_{R}}$ through a processing function $RF$, and sends the results $d_{O_R}$ to the Component Configurator through the channels
  $O_{CR}$ (the corresponding sending action is denoted $s_{O_{CR}}(d_{O_R})$);
  \item The Component Configurator receives the results $d_{O_R}$ from the Component Repository through the channel $O_{CR}$ (the corresponding reading action is denoted $r_{O_{CR}}(d_{O_R})$),
  the processes the results and generates the results $d_{O}$ through a processing function $CCF_3$, and sends $d_{O}$ to the outside through the
  channel $O$ (the corresponding sending action is denoted $s_{O}(d_{O})$).
\end{enumerate}

\begin{figure}
    \centering
    \includegraphics{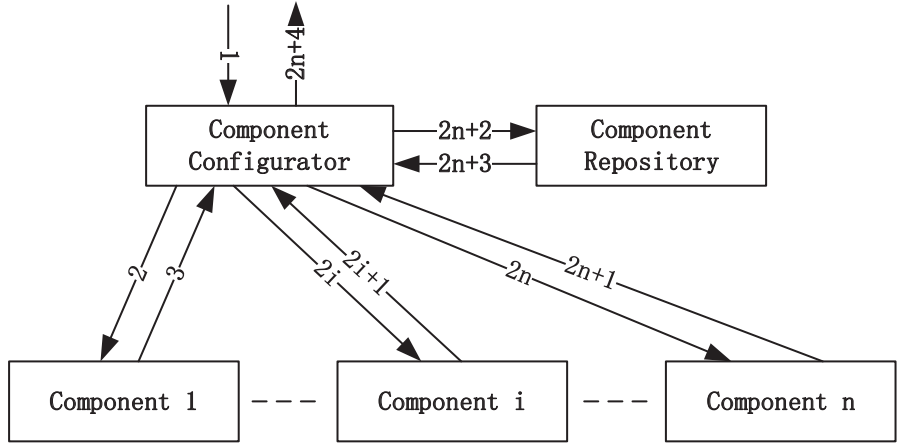}
    \caption{Typical process of Component Configurator pattern}
    \label{CC6P}
\end{figure}

In the following, we verify the Component Configurator pattern. We assume all data elements $d_{I}$, $d_{I_R}$, $d_{I_{C_i}}$, $d_{O_{C_{i}}}$, $d_{O_R}$, $d_{O}$ (for $1\leq i\leq n$) are from a finite set
$\Delta$.

The state transitions of the Component Configurator module
described by APTC are as follows.

$CC=\sum_{d_{I}\in\Delta}(r_{I}(d_{I})\cdot CC_{2})$

$CC_{2}=CCF_1\cdot CC_{3}$

$CC_{3}=\sum_{d_{I_{C_1}},\cdot,d_{I_{C_n}}\in\Delta}(s_{I_{CC_1}}(d_{I_{C_1}})\between\cdots\between s_{I_{CC_n}}(d_{I_{C_n}})\cdot CC_{4})$

$CC_{4}=\sum_{d_{O_{C_{1}}},\cdots,d_{O_{C_{n}}}\in\Delta}(r_{O_{CC_1}}(d_{O_{C_{1}}})\between\cdots\between r_{O_{CC_n}}(d_{O_{C_{n}}})\cdot CC_{5})$

$CC_5=CCF_2\cdot CC_6$

$CC_{6}=\sum_{d_{I_R}\in\Delta}(s_{I_{CR}}(d_{I_R})\cdot CC_7)$

$CC_{7}=\sum_{d_{O_{R}}\in\Delta}(r_{O_{CR}}(d_{O_{R}})\cdot CC_{8})$

$CC_8=CCF_3\cdot CC_9$

$CC_{9}=\sum_{d_{O}\in\Delta}(s_{O}(d_{O})\cdot CC)$

The state transitions of the Component $i$ described by APTC are as follows.

$C_i=\sum_{d_{I_{C_i}}\in\Delta}(r_{I_{CC_i}}(d_{I_{C_i}})\cdot C_{i_2})$

$C_{i_2}=CF_{i}\cdot C_{i_3}$

$C_{i_3}=\sum_{d_{O_{C_{i}}}\in\Delta}(s_{O_{CC_{i}}}(d_{O_{C_{i}}})\cdot C_{i})$

The state transitions of the Component Repository described by APTC are as follows.

$R=\sum_{d_{I_{R}}\in\Delta}(r_{I_{CR}}(d_{I_{R}})\cdot R_{2})$

$R_{2}=RF\cdot R_{3}$

$R_{3}=\sum_{d_{O_{R}}\in\Delta}(s_{O_{CR}}(d_{O_{R}})\cdot R)$

The sending action and the reading action of the same data through the same channel can communicate with each other, otherwise, will cause a deadlock $\delta$. We define the following
communication functions of the Component Configurator for $1\leq i\leq n$.

$$\gamma(r_{I_{CC_i}}(d_{I_{C_i}}),s_{I_{CC_i}}(d_{I_{C_i}}))\triangleq c_{I_{CC_i}}(d_{I_{C_i}})$$

$$\gamma(r_{O_{CC_{i}}}(d_{O_{C_{i}}}),s_{O_{CC_{i}}}(d_{O_{C_{i}}}))\triangleq c_{O_{CC_{i}}}(d_{O_{C_{i}}})$$

$$\gamma(r_{I_{CR}}(d_{I_{R}}),s_{I_{CR}}(d_{I_{R}}))\triangleq c_{I_{CR}}(d_{I_{R}})$$

$$\gamma(r_{O_{CR}}(d_{O_{R}}),s_{O_{CR}}(d_{O_{R}}))\triangleq c_{O_{CR}}(d_{O_{R}})$$

Let all modules be in parallel, then the Component Configurator pattern $CC\quad R\quad C_1\cdots C_i\cdots C_n$ can be presented by the following process term.

$\tau_I(\partial_H(\Theta(CC\between R\between C_1\between\cdots\between C_i\between\cdots\between C_n)))=\tau_I(\partial_H(CC\between R\between C_1\between\cdots\between C_i\between\cdots\between C_n))$

where $H=\{r_{I_{CC_i}}(d_{I_{C_i}}),s_{I_{CC_i}}(d_{I_{C_i}}),r_{O_{CC_{i}}}(d_{O_{C_{i}}}),s_{O_{CC_{i}}}(d_{O_{C_{i}}}),r_{I_{CR}}(d_{I_{R}}),s_{I_{CR}}(d_{I_{R}}),\\
r_{O_{CR}}(d_{O_{R}}),s_{O_{CR}}(d_{O_{R}})
|d_{I}, d_{I_R}, d_{I_{C_i}}, d_{O_{C_{i}}}, d_{O_R}, d_{O}\in\Delta\}$ for $1\leq i\leq n$,

$I=\{c_{I_{CC_i}}(d_{I_{C_i}}),c_{O_{CC_{i}}}(d_{O_{C_{i}}}),c_{I_{CR}}(d_{I_{R}}),c_{O_{CR}}(d_{O_{R}}),CCF_1,CCF_2,CCF_3,CF_{i},RF\\
|d_{I}, d_{I_R}, d_{I_{C_i}}, d_{O_{C_{i}}}, d_{O_R}, d_{O}\in\Delta\}$ for $1\leq i\leq n$.

Then we get the following conclusion on the Component Configurator pattern.

\begin{theorem}[Correctness of the Component Configurator pattern]
The Component Configurator pattern $\tau_I(\partial_H(CC\between R\between C_1\between\cdots\between C_i\between\cdots\between C_n))$ can exhibit desired external behaviors.
\end{theorem}

\begin{proof}
Based on the above state transitions of the above modules, by use of the algebraic laws of APTC, we can prove that

$\tau_I(\partial_H(CC\between R\between C_1\between\cdots\between C_i\between\cdots\between C_n))=\sum_{d_{I},d_O\in\Delta}(r_{I}(d_{I})\cdot s_{O}(d_{O}))\cdot
\tau_I(\partial_H(CC\between R\between C_1\between\cdots\between C_i\between\cdots\between C_n))$,

that is, the Component Configurator pattern $\tau_I(\partial_H(CC\between R\between C_1\between\cdots\between C_i\between\cdots\between C_n))$ can exhibit desired external behaviors.

For the details of proof, please refer to section \ref{app}, and we omit it.
\end{proof}

\subsubsection{Verification of the Interceptor Pattern}

The Interceptor pattern adds functionalities to the concrete framework to introduce an intermediate Dispatcher and an Interceptor.
There are three modules in the Interceptor pattern: the Concrete Framework, the Dispatcher,
and the Interceptor. The Concrete Framework interacts with the user through
the channels $I$ and $O$; with the Dispatcher through the channel $CD$; with the Interceptor through the channels $I_{IC}$ and $O_{IC}$. The Dispatcher interacts with the Interceptor
through the channel $DI$. As illustrates in
Figure \ref{In6}.

\begin{figure}
    \centering
    \includegraphics{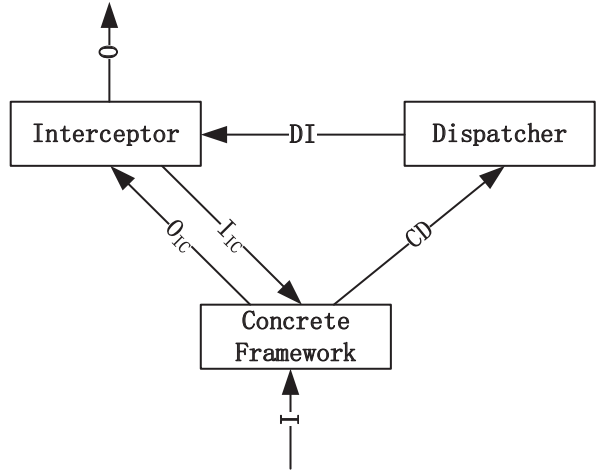}
    \caption{Interceptor pattern}
    \label{In6}
\end{figure}

The typical process of the Interceptor pattern is shown in Figure \ref{In6P} and as follows.

\begin{enumerate}
  \item The Concrete Framework receives the request $d_{I}$ from the user through the channel $I$ (the corresponding reading action is denoted $r_{I}(d_{I})$), then processes the request $d_{I}$ through a processing
  function $CF_1$, and sends the processed request $d_{I_{D}}$ to the Dispatcher through the channel $CD$ (the corresponding sending action is denoted $s_{CD}(d_{I_{D}})$);
  \item The Dispatcher receives $d_{I_{D}}$ from the Concrete Framework through the channel $CD$ (the corresponding reading action is denoted $r_{CD}(d_{I_{D}})$), then processes the request
  through a processing function $DF$, generates and sends the processed request $d_{O_{D}}$ to the Interceptor through the channel $DI$ (the corresponding sending action is denoted
  $s_{DI}(d_{O_{D}})$);
  \item The Interceptor receives the request $d_{O_D}$ from the Dispatcher through the channel $DI$ (the corresponding reading action is denoted $r_{DI}(d_{O_D})$), then
  processes the request and generates the request $d_{I_C}$ through a processing function $IF_1$, and sends the request to the Concrete Framework through the channel $I_{IC}$ (the corresponding sending action is denoted
  $s_{I_{IC}}(d_{I_C})$);
  \item The Concrete Framework receives the request $d_{I_C}$ from the Interceptor through the channel $I_{IC}$ (the corresponding reading action is denoted $r_{I_{IC}}(d_{I_C})$), then processes
  the request through a processing function $CF_2$, generates and sends the response $d_{O_C}$ to the Interceptor through the channel $O_{IC}$ (the corresponding sending action is denoted $s_{O_{IC}}(d_{O_C})$);
  \item The Interceptor receives the response $d_{O_C}$ from the Concrete Framework through the channel $O_{IC}$ (the corresponding reading action is denoted $r_{O_{IC}}(d_{O_C})$), then processes
  the request through a processing function $IF_2$, generates and sends the processed response $d_{O}$ to the user through the channel $O$ (the corresponding sending action is denoted $s_{O}(d_{O})$).
\end{enumerate}

\begin{figure}
    \centering
    \includegraphics{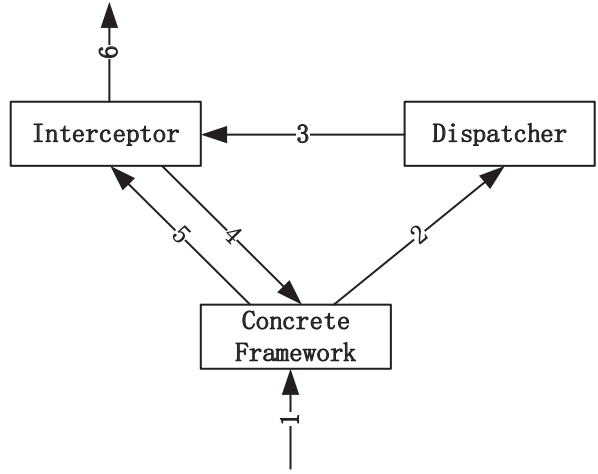}
    \caption{Typical process of Interceptor pattern}
    \label{In6P}
\end{figure}

In the following, we verify the Interceptor pattern. We assume all data elements $d_{I}$, $d_{I_D}$, $d_{I_{C}}$, $d_{O_{D}}$, $d_{O_{C}}$, $d_{O}$ are from a finite set
$\Delta$.

The state transitions of the Concrete Framework module
described by APTC are as follows.

$C=\sum_{d_{I}\in\Delta}(r_{I}(d_{I})\cdot C_{2})$

$C_{2}=CF_1\cdot C_{3}$

$C_{3}=\sum_{d_{I_{D}}\in\Delta}(s_{CD}(d_{I_{D}})\cdot C_4)$

$C_4=\sum_{d_{I_C}\in\Delta}(r_{I_{IC}}(d_{I_C})\cdot C_{5})$

$C_{5}=CF_2\cdot C_{6}$

$C_{6}=\sum_{d_{O_C}\in\Delta}(s_{O_{IC}}(d_{O_C})\cdot C)$

The state transitions of the Dispatcher module
described by APTC are as follows.

$D=\sum_{d_{I_{D}}\in\Delta}(r_{CD}(d_{I_{D}})\cdot D_{2})$

$D_{2}=DF\cdot D_{3}$

$D_{3}=\sum_{d_{O_D}\in\Delta}(s_{DI}(d_{O_D})\cdot D)$

The state transitions of the Interceptor module
described by APTC are as follows.

$I=\sum_{d_{O_D}\in\Delta}(r_{DI}(d_{O_D})\cdot I_{2})$

$I_{2}=IF_1\cdot I_{3}$

$I_{3}=\sum_{d_{I_C}\in\Delta}(s_{I_{IC}}(d_{I_C})\cdot I_4)$

$I_4=\sum_{d_{O_C}\in\Delta}(r_{O_{IC}}(d_{O_C})\cdot I_{5})$

$I_{5}=IF_2\cdot I_{6}$

$I_{6}=\sum_{d_{O}\in\Delta}(s_{O}(d_{O})\cdot I)$

The sending action and the reading action of the same data through the same channel can communicate with each other, otherwise, will cause a deadlock $\delta$. We define the following
communication functions between the Concrete Framework and the Dispatcher.

$$\gamma(r_{CD}(d_{I_{D}}),s_{CD}(d_{I_{D}}))\triangleq c_{CD}(d_{I_{D}})$$

There are two communication functions between the Concrete Framework and the Interceptor as follows.

$$\gamma(r_{I_{IC}}(d_{I_C}),s_{I_{IC}}(d_{I_C}))\triangleq c_{I_{IC}}(d_{I_C})$$

$$\gamma(r_{O_{IC}}(d_{O_C}),s_{O_{IC}}(d_{O_C}))\triangleq c_{O_{IC}}(d_{O_C})$$

There are one communication function between the Dispatcher and the Interceptor as follows.

$$\gamma(r_{DI}(d_{O_D}),s_{DI}(d_{O_D}))\triangleq c_{DI}(d_{O_D})$$

Let all modules be in parallel, then the Interceptor pattern $C\quad D \quad I$ can be presented by the following process term.

$\tau_I(\partial_H(\Theta(C\between D\between I)))=\tau_I(\partial_H(C\between D\between I))$

where $H=\{r_{CD}(d_{I_{D}}),s_{CD}(d_{I_{D}}),r_{I_{IC}}(d_{I_C}),s_{I_{IC}}(d_{I_C}),r_{O_{IC}}(d_{O_C}),s_{O_{IC}}(d_{O_C}),\\
r_{DI}(d_{O_D}),s_{DI}(d_{O_D})
|d_{I}, d_{I_D}, d_{I_{C}}, d_{O_{D}}, d_{O_{C}}, d_{O}\in\Delta\}$,

$I=\{c_{CD}(d_{I_{D}}),c_{I_{IC}}(d_{I_C}),c_{O_{IC}}(d_{O_C}),c_{DI}(d_{O_D}),CF_1,CF_2,DF,IF_1,IF_2\\
|d_{I}, d_{I_D}, d_{I_{C}}, d_{O_{D}}, d_{O_{C}}, d_{O}\in\Delta\}$.

Then we get the following conclusion on the Interceptor pattern.

\begin{theorem}[Correctness of the Interceptor pattern]
The Interceptor pattern $\tau_I(\partial_H(C\between D\between I))$ can exhibit desired external behaviors.
\end{theorem}

\begin{proof}
Based on the above state transitions of the above modules, by use of the algebraic laws of APTC, we can prove that

$\tau_I(\partial_H(C\between D\between I))=\sum_{d_{I},d_{O}\in\Delta}(r_{I}(d_{I})\cdot s_{O}(d_{O}))\cdot
\tau_I(\partial_H(C\between D\between I))$,

that is, the Interceptor pattern $\tau_I(\partial_H(C\between D\between I))$ can exhibit desired external behaviors.

For the details of proof, please refer to section \ref{app}, and we omit it.
\end{proof}

\subsubsection{Verification of the Extension Interface Pattern}

The Extension Interface pattern allows to export multiple interface of a component to extend or modify the functionalities of the component.
There are three classes of modules in the Extension Interface pattern: the Component Factory
 and $n$ Extension Interfaces and the Component. The Component Factory interacts with Extension Interface $i$ through the channels $I_{FE_i}$ and $O_{FE_i}$, and it exchanges
information with outside through the input channel $I$ and the output channel $O$, and with the Component through the channels $I_{FC}$ and $O_{FC}$. As illustrates in Figure \ref{EI6}.

\begin{figure}
    \centering
    \includegraphics{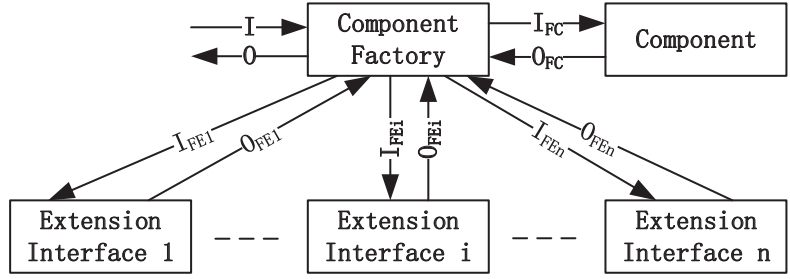}
    \caption{Extension Interface pattern}
    \label{EI6}
\end{figure}

The typical process of the Extension Interface pattern is shown in Figure \ref{EI6P} and as follows.

\begin{enumerate}
  \item The Component Factory receives the input from the user through the channel $I$ (the corresponding reading action is denoted $r_{I}(d_{I})$), then
  processes the input through a processing function $FF_1$ and generates the input $d_{I_{C}}$, and sends $d_{I_C}$ to the Component through the channel
  $I_{FC}$ (the corresponding sending action is denoted $s_{I_{FC}}(d_{I_C})$);
  \item The Component receives the input $d_{I_C}$ through the channel $I_{FC}$ (the corresponding reading action is denoted $r_{I_{FC}}(d_{I_C})$), then
  processes the information and generates the results $d_{O_{C}}$ through a processing function $CF$, and sends the results $d_{O_C}$ to the Component Factory through the channels
  $O_{FC}$ (the corresponding sending action is denoted $s_{O_{FC}}(d_{O_C})$);
  \item The Component Factory receives the results $d_{O_C}$ from the Component through the channel $O_{FC}$ (the corresponding reading action is denoted $r_{O_{FC}}(d_{O_C})$), then
  processes the results through a processing function $FF_2$ and generates the input $d_{I_{E_i}}$, and sends $d_{I_{E_i}}$ to the Extension Interface $i$ (for $1\leq i\leq n$) through the
  channel $I_{FE_i}$ (the corresponding sending action is denoted $s_{I_{FE_i}}(d_{I_{E_i}})$);
  \item The Extension Interface $i$ receives the input $d_{I_{E_i}}$ from the Component Factory through the channel $I_{FE_i}$ (the corresponding reading action is denoted $r_{I_{FE_i}}(d_{I_{E_i}})$),
  then processes the input through a processing function $EF_{i}$, and sends the results $d_{O_{E_{i}}}$ to the Component Factory through the channel $O_{FE_i}$
  (the corresponding sending action is denoted $s_{O_{FE_i}}(d_{O_{E_{i}}})$);
  \item The Component Factory receives the results from the Extension Interface $i$ through the channel $O_{FE_i}$ (the corresponding reading action is denoted $r_{O_{FE_i}}(d_{O_{E_{i}}})$),
  the processes the results and generates the results $d_{O}$ through a processing function $FF_3$, and sends $d_{O}$ to the outside through the
  channel $O$ (the corresponding sending action is denoted $s_{O}(d_{O})$).
\end{enumerate}

\begin{figure}
    \centering
    \includegraphics{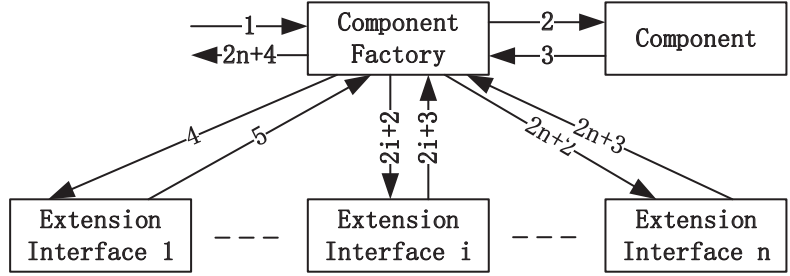}
    \caption{Typical process of Extension Interface pattern}
    \label{EI6P}
\end{figure}

In the following, we verify the Extension Interface pattern. We assume all data elements $d_{I}$, $d_{I_C}$, $d_{I_{E_i}}$, $d_{O_{E_{i}}}$, $d_{O_C}$, $d_{O}$ (for $1\leq i\leq n$) are from a finite set
$\Delta$.

The state transitions of the Component Factory module
described by APTC are as follows.

$F=\sum_{d_{I}\in\Delta}(r_{I}(d_{I})\cdot F_{2})$

$F_{2}=FF_1\cdot F_{3}$

$F_{3}=\sum_{d_{I_C}\in\Delta}(s_{I_{FC}}(d_{I_C})\cdot F_4)$

$F_{4}=\sum_{d_{O_{C}}\in\Delta}(r_{O_{FC}}(d_{O_{C}})\cdot F_{5})$

$F_5=FF_2\cdot F_6$

$F_{6}=\sum_{d_{I_{E_1}},\cdot,d_{I_{E_n}}\in\Delta}(s_{I_{FE_1}}(d_{I_{E_1}})\between\cdots\between s_{I_{FE_n}}(d_{I_{E_n}})\cdot F_{7})$

$F_{7}=\sum_{d_{O_{E_{1}}},\cdots,d_{O_{E_{n}}}\in\Delta}(r_{O_{FE_1}}(d_{O_{E_{1}}})\between\cdots\between r_{O_{FE_n}}(d_{O_{E_{n}}})\cdot F_{8})$

$F_8=FF_3\cdot F_9$

$F_{9}=\sum_{d_{O}\in\Delta}(s_{O}(d_{O})\cdot F)$

The state transitions of the Extension Interface $i$ described by APTC are as follows.

$E_i=\sum_{d_{I_{E_i}}\in\Delta}(r_{I_{FE_i}}(d_{I_{E_i}})\cdot E_{i_2})$

$E_{i_2}=EF_{i}\cdot E_{i_3}$

$E_{i_3}=\sum_{d_{O_{E_{i}}}\in\Delta}(s_{O_{FE_{i}}}(d_{O_{E_{i}}})\cdot E_{i})$

The state transitions of the Component described by APTC are as follows.

$C=\sum_{d_{I_{C}}\in\Delta}(r_{I_{FC}}(d_{I_{C}})\cdot C_{2})$

$C_{2}=CF\cdot C_{3}$

$C_{3}=\sum_{d_{O_{C}}\in\Delta}(s_{O_{FC}}(d_{O_{C}})\cdot C)$

The sending action and the reading action of the same data through the same channel can communicate with each other, otherwise, will cause a deadlock $\delta$. We define the following
communication functions of the Component Factory for $1\leq i\leq n$.

$$\gamma(r_{I_{FE_i}}(d_{I_{E_i}}),s_{I_{FE_i}}(d_{I_{E_i}}))\triangleq c_{I_{FE_i}}(d_{I_{E_i}})$$

$$\gamma(r_{O_{FE_{i}}}(d_{O_{E_{i}}}),s_{O_{FE_{i}}}(d_{O_{E_{i}}}))\triangleq c_{O_{FE_{i}}}(d_{O_{E_{i}}})$$

$$\gamma(r_{I_{FC}}(d_{I_{C}}),s_{I_{FC}}(d_{I_{C}}))\triangleq c_{I_{FC}}(d_{I_{C}})$$

$$\gamma(r_{O_{FC}}(d_{O_{C}}),s_{O_{FC}}(d_{O_{C}}))\triangleq c_{O_{FC}}(d_{O_{C}})$$

Let all modules be in parallel, then the Extension Interface pattern $F\quad C\quad E_1\cdots E_i\cdots E_n$ can be presented by the following process term.

$\tau_I(\partial_H(\Theta(F\between C\between E_1\between\cdots\between E_i\between\cdots\between E_n)))=\tau_I(\partial_H(F\between C\between E_1\between\cdots\between E_i\between\cdots\between E_n))$

where $H=\{r_{I_{FE_i}}(d_{I_{E_i}}),s_{I_{FE_i}}(d_{I_{E_i}}),r_{O_{FE_{i}}}(d_{O_{E_{i}}}),s_{O_{FE_{i}}}(d_{O_{E_{i}}}),r_{I_{FC}}(d_{I_{C}}),s_{I_{FC}}(d_{I_{C}}),\\
r_{O_{FC}}(d_{O_{C}}),s_{O_{FC}}(d_{O_{C}})
|d_{I}, d_{I_C}, d_{I_{E_i}}, d_{O_{E_{i}}}, d_{O_C}, d_{O}\in\Delta\}$ for $1\leq i\leq n$,

$I=\{c_{I_{FE_i}}(d_{I_{E_i}}),c_{O_{FE_{i}}}(d_{O_{E_{i}}}),c_{I_{FC}}(d_{I_{C}}),c_{O_{FC}}(d_{O_{C}}),FF_1,FF_2,FF_3,EF_{i},CF\\
|d_{I}, d_{I_C}, d_{I_{E_i}}, d_{O_{E_{i}}}, d_{O_C}, d_{O}\in\Delta\}$ for $1\leq i\leq n$.

Then we get the following conclusion on the Extension Interface pattern.

\begin{theorem}[Correctness of the Extension Interface pattern]
The Extension Interface pattern $\tau_I(\partial_H(F\between C\between E_1\between\cdots\between E_i\between\cdots\between E_n))$ can exhibit desired external behaviors.
\end{theorem}

\begin{proof}
Based on the above state transitions of the above modules, by use of the algebraic laws of APTC, we can prove that

$\tau_I(\partial_H(F\between C\between E_1\between\cdots\between E_i\between\cdots\between E_n))=\sum_{d_{I},d_O\in\Delta}(r_{I}(d_{I})\cdot s_{O}(d_{O}))\cdot
\tau_I(\partial_H(F\between C\between E_1\between\cdots\between E_i\between\cdots\between E_n))$,

that is, the Extension Interface pattern $\tau_I(\partial_H(F\between C\between E_1\between\cdots\between E_i\between\cdots\between E_n))$ can exhibit desired external behaviors.

For the details of proof, please refer to section \ref{app}, and we omit it.
\end{proof}

\subsection{Event Handling Patterns}\label{EH6}

In this subsection, we verify patterns related to event handling, including the Reactor pattern, the Proactor pattern, the Asynchronous Completion Token pattern, and the Acceptor-Connector
pattern.

\subsubsection{Verification of the Reactor Pattern}

The Reactor pattern allows to demultiplex and dispatch the request event to the event-driven applications.
There are three classes of modules in the Reactor pattern: the Handle Set
 and $n$ Event Handlers and the Reactor. The Handle Set interacts with Event Handler $i$ through the channel $EH_i$, and it exchanges
information with outside through the input channel $I$ and the output channel $O$, and with the Reactor through the channel $HR$. The Reactor interacts with the Event Handler $i$
through the channel $RE_i$. As illustrates in Figure \ref{Re6}.

\begin{figure}
    \centering
    \includegraphics{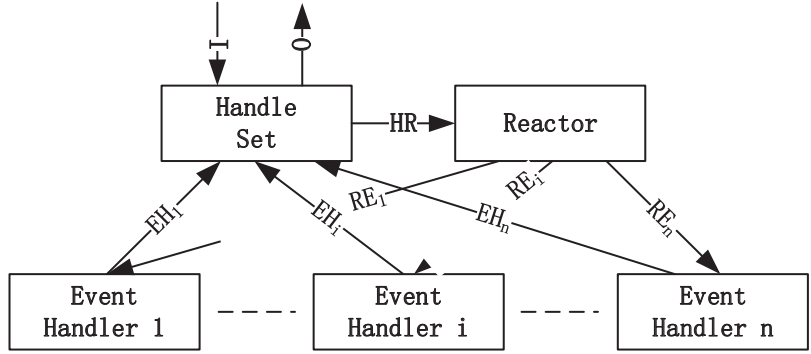}
    \caption{Reactor pattern}
    \label{Re6}
\end{figure}

The typical process of the Reactor pattern is shown in Figure \ref{Re6P} and as follows.

\begin{enumerate}
  \item The Handle Set receives the input from the user through the channel $I$ (the corresponding reading action is denoted $r_{I}(d_{I})$), then
  processes the input through a processing function $HF_1$ and generates the input $d_{I_{R}}$, and sends $d_{I_R}$ to the Reactor through the channel
  $HR$ (the corresponding sending action is denoted $s_{HR}(d_{I_R})$);
  \item The Reactor receives the input $d_{I_R}$ through the channel $HR$ (the corresponding reading action is denoted $r_{HR}(d_{I_R})$), then
  processes the information and generates the results $d_{I_{E_i}}$ through a processing function $RF$, and sends the results $d_{I_{E_i}}$ to the Event Handler $i$ through the channels
  $RE_i$ (the corresponding sending action is denoted $s_{RE_i}(d_{I_{E_i}})$);
  \item The Event Handler $i$ receives the input $d_{I_{E_i}}$ from the Reactor through the channel $RE_i$ (the corresponding reading action is denoted $r_{RE_i}(d_{I_{E_i}})$),
  then processes the input through a processing function $EF_{i}$, and sends the results $d_{O_{E_{i}}}$ to the Handle Set through the channel $EH_i$
  (the corresponding sending action is denoted $s_{EH_i}(d_{O_{E_{i}}})$);
  \item The Handle Set receives the results from the Event Handler $i$ through the channel $EH_i$ (the corresponding reading action is denoted $r_{EH_i}(d_{O_{E_{i}}})$),
  the processes the results and generates the results $d_{O}$ through a processing function $HF_2$, and sends $d_{O}$ to the outside through the
  channel $O$ (the corresponding sending action is denoted $s_{O}(d_{O})$).
\end{enumerate}

\begin{figure}
    \centering
    \includegraphics{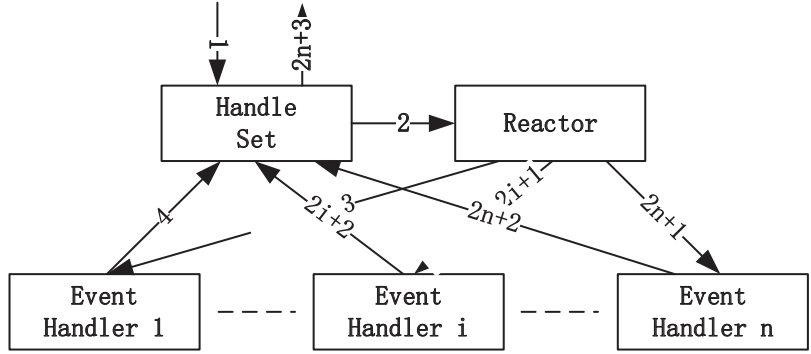}
    \caption{Typical process of Reactor pattern}
    \label{Re6P}
\end{figure}

In the following, we verify the Reactor pattern. We assume all data elements $d_{I}$, $d_{I_R}$, $d_{I_{E_i}}$, $d_{O_{E_{i}}}$, $d_{O}$ (for $1\leq i\leq n$) are from a finite set
$\Delta$.

The state transitions of the Handle Set module
described by APTC are as follows.

$H=\sum_{d_{I}\in\Delta}(r_{I}(d_{I})\cdot F_{2})$

$H_{2}=HF_1\cdot H_{3}$

$H_{3}=\sum_{d_{I_R}\in\Delta}(s_{HR}(d_{I_R})\cdot H_4)$

$H_{4}=\sum_{d_{O_{E_i}}\in\Delta}(r_{EH_i}(d_{O_{E_i}})\cdot H_{5})$

$H_5=HF_2\cdot H_6$

$H_{6}=\sum_{d_{O}\in\Delta}(s_{O}(d_{O})\cdot H)$

The state transitions of the Event Handler $i$ described by APTC are as follows.

$E_i=\sum_{d_{I_{E_i}}\in\Delta}(r_{RE_i}(d_{I_{E_i}})\cdot E_{i_2})$

$E_{i_2}=EF_{i}\cdot E_{i_3}$

$E_{i_3}=\sum_{d_{O_{E_{i}}}\in\Delta}(s_{EH_i}(d_{O_{E_{i}}})\cdot E_{i})$

The state transitions of the Reactor described by APTC are as follows.

$R=\sum_{d_{I_{R}}\in\Delta}(r_{HR}(d_{I_{R}})\cdot R_{2})$

$R_{2}=RF\cdot R_{3}$

$R_{3}=\sum_{d_{I_{E_1}},\cdots,d_{I_{E_n}}\in\Delta}(s_{RE_1}(d_{I_{E_1}})\between\cdots\between s_{RE_n}(d_{I_{E_n}})\cdot R)$

The sending action and the reading action of the same data through the same channel can communicate with each other, otherwise, will cause a deadlock $\delta$. We define the following
communication functions of the Handle Set for $1\leq i\leq n$.

$$\gamma(r_{HR}(d_{I_R}),s_{HR}(d_{I_R}))\triangleq c_{HR}(d_{I_R})$$

$$\gamma(r_{EH_i}(d_{O_{E_i}}),s_{EH_i}(d_{O_{E_i}}))\triangleq c_{EH_i}(d_{O_{E_i}})$$

There is one communication function between the Reactor and the Event Handler $i$.

$$\gamma(r_{RE_i}(d_{I_{E_i}}),s_{RE_i}(d_{I_{E_i}}))\triangleq c_{RE_i}(d_{I_{E_i}})$$

Let all modules be in parallel, then the Reactor pattern $H\quad R\quad E_1\cdots E_i\cdots E_n$ can be presented by the following process term.

$\tau_I(\partial_H(\Theta(H\between R\between E_1\between\cdots\between E_i\between\cdots\between E_n)))=\tau_I(\partial_H(H\between R\between E_1\between\cdots\between E_i\between\cdots\between E_n))$

where $H=\{r_{HR}(d_{I_R}),s_{HR}(d_{I_R}),r_{EH_i}(d_{O_{E_i}}),s_{EH_i}(d_{O_{E_i}}),r_{RE_i}(d_{I_{E_i}}),s_{RE_i}(d_{I_{E_i}})\\
|d_{I}, d_{I_R}, d_{I_{E_i}}, d_{O_{E_{i}}}, d_{O}\in\Delta\}$ for $1\leq i\leq n$,

$I=\{c_{HR}(d_{I_R}),c_{EH_i}(d_{O_{E_i}}),c_{RE_i}(d_{I_{E_i}}),HF_1,HF_2,EF_{i},RF\\
|d_{I}, d_{I_R}, d_{I_{E_i}}, d_{O_{E_{i}}}, d_{O}\in\Delta\}$ for $1\leq i\leq n$.

Then we get the following conclusion on the Reactor pattern.

\begin{theorem}[Correctness of the Reactor pattern]
The Reactor pattern $\tau_I(\partial_H(H\between R\between E_1\between\cdots\between E_i\between\cdots\between E_n))$ can exhibit desired external behaviors.
\end{theorem}

\begin{proof}
Based on the above state transitions of the above modules, by use of the algebraic laws of APTC, we can prove that

$\tau_I(\partial_H(H\between R\between E_1\between\cdots\between E_i\between\cdots\between E_n))=\sum_{d_{I},d_O\in\Delta}(r_{I}(d_{I})\cdot s_{O}(d_{O}))\cdot
\tau_I(\partial_H(H\between R\between E_1\between\cdots\between E_i\between\cdots\between E_n))$,

that is, the Reactor pattern $\tau_I(\partial_H(H\between R\between E_1\between\cdots\between E_i\between\cdots\between E_n))$ can exhibit desired external behaviors.

For the details of proof, please refer to section \ref{app}, and we omit it.
\end{proof}

\subsubsection{Verification of the Proactor Pattern}

The Proactor pattern also decouples the delivery the events between the event-driven applications and clients, but the events are triggered by the completion of asynchronous operations,
which has four classes of components: the Asynchronous Operation Processor, the Asynchronous Operation, the Proactor and $n$ Completion Handlers.
The Asynchronous Operation Processor interacts with the outside through the channel $I$; with the Asynchronous Operation through the channels $I_{PO}$ and $O_{PO}$; with the Proactor with
the channel $PP$. The Proactor interacts with the Completion Handler $i$ with the channel $PC_i$. The Completion Handler $i$ interacts with the outside through the channel $O_i$.
As illustrates in Figure \ref{Pro6}.

\begin{figure}
    \centering
    \includegraphics{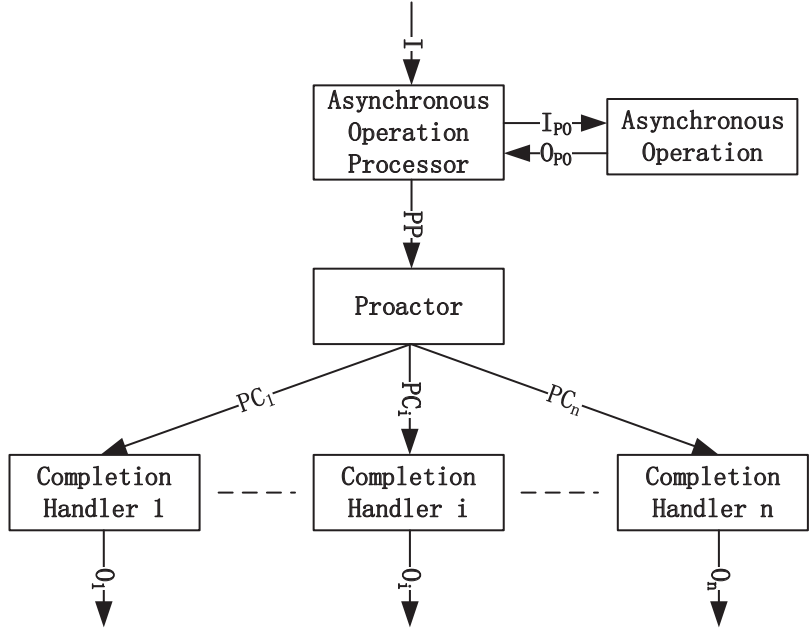}
    \caption{Proactor pattern}
    \label{Pro6}
\end{figure}

The typical process of the Proactor pattern is shown in Figure \ref{Pro6P} and following.

\begin{enumerate}
  \item The Asynchronous Operation Processor receives the input $d_I$ from the user through the channel $I$ (the corresponding reading action is denoted $r_I(D_I)$), processes the input through
  a processing function $AOPF_1$, and generates the input to the Asynchronous Operation $d_{I_{AO}}$ and it sends $d_{I_{AO}}$ to the Asynchronous Operation through the channel $I_{PO}$
  (the corresponding sending action is denoted $s_{I_{PO}}(d_{I_{AO}})$);
  \item The Asynchronous Operation receives the input from the Asynchronous Operation Processor through the channel $I_{PO}$ (the corresponding reading action is denoted $r_{I_{PO}}(d_{I_{AO}})$), processes the input through
  a processing function $AOF$, generates the computational results to the Asynchronous Operation Processor which is denoted $d_{O_{AO}}$; then sends the results to the Asynchronous Operation Processor through the
  channel $O_{PO}$ (the corresponding sending action is denoted $s_{O_{PO}}(d_{O_{AO}})$);
  \item The Asynchronous Operation Processor receives the results from the Asynchronous Operation through the channel $O_{PO}$ (the corresponding reading action is denoted $r_{O_{PO}}(d_{O_{AO}})$),
  then processes the results and generates the events $d_{I_{P}}$ through a processing function $AOPF_2$, and sends it to the Proactor through the channel $PP$ (the corresponding
  sending action is denoted $s_{PP}(d_{I_{P}})$);
  \item The Proactor receives the events $d_{I_P}$ from the Asynchronous Operation Processor through the channel $PP$ (the corresponding reading action is denoted $r_{PP}(d_{I_P})$),
  then processes the events through a processing function $PF$, and sends the processed events to the Completion Handler $i$ $d_{I_{C_i}}$ for $1\leq i\leq n$ through the channel $PC_i$
  (the corresponding sending action is denoted $s_{PC_i}(d_{I_{C_i}})$);
  \item The Completion Handler $i$ (for $1\leq i\leq n$) receives the events from the Proactor through the channel $PC_i$ (the corresponding reading action is denoted $r_{PC_i}(d_{I_{C_i}})$),
  processes the events through a processing function $CF_{i}$, generates the output
  $d_{O_i}$, then sending the output through the channel $O_i$ (the corresponding sending action is denoted $s_{O_i}(d_{O_i})$).
\end{enumerate}

\begin{figure}
    \centering
    \includegraphics{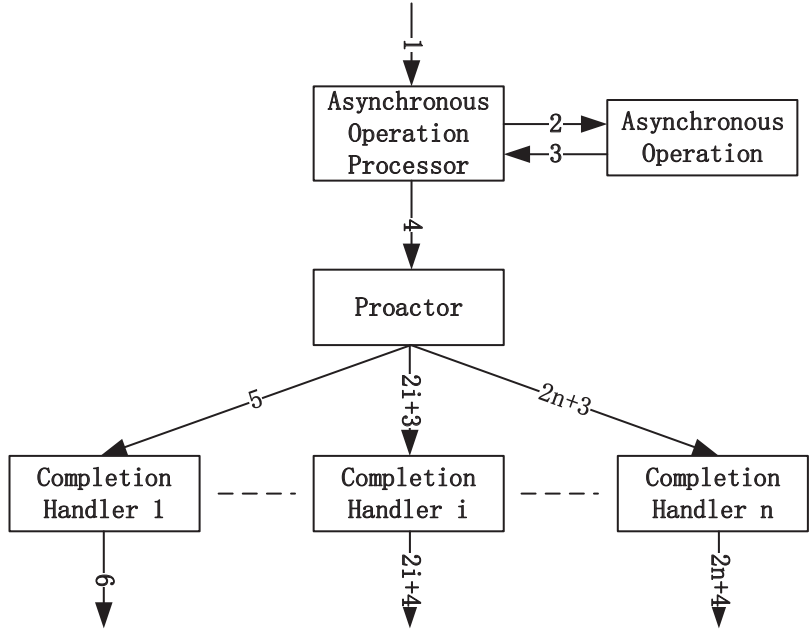}
    \caption{Typical process of Proactor pattern}
    \label{Pro6P}
\end{figure}

In the following, we verify the Proactor pattern. We assume all data elements $d_{I}$, $d_{I_{AO}}$, $d_{I_P}$, $d_{I_{C_i}}$, $d_{O_{AO}}$, $d_{O_{i}}$ (for $1\leq i\leq n$) are from a finite set
$\Delta$.

The state transitions of the Asynchronous Operation Processor module
described by APTC are as follows.

$AOP=\sum_{d_{I}\in\Delta}(r_{I}(d_{I})\cdot AOP_{2})$

$AOP_{2}=AOPF_1\cdot AOP_{3}$

$AOP_{3}=\sum_{d_{I_{AO}}\in\Delta}(s_{I_{PO}}(d_{I_{AO}})\cdot AOP_{4})$

$AOP_4=\sum_{d_{O_{AO}}\in\Delta}(r_{O_{PO}}(d_{O_{AO}})\cdot AOP_{5})$

$AOP_{5}=AOPF_2\cdot AOP_{6}$

$AOP_{6}=\sum_{d_{I_{P}}\in\Delta}(s_{PP}(d_{I_{P}})\cdot AOP)$

The state transitions of the Asynchronous Operation described by APTC are as follows.

$AO=\sum_{d_{I_{AO}}\in\Delta}(r_{I_{PO}}(d_{I_{AO}})\cdot AO_{2})$

$AO_{2}=AOF\cdot AO_{3}$

$AO_{3}=\sum_{d_{O_{AO}}\in\Delta}(s_{O_{PO}}(d_{O_{AO}})\cdot AO)$

The state transitions of the Proactor described by APTC are as follows.

$P=\sum_{d_{I_{P}}\in\Delta}(r_{PP}(d_{I_{P}})\cdot P_{2})$

$P_{2}=PF\cdot P_{3}$

$P_{3}=\sum_{d_{I_{C_1}},\cdots,d_{I_{c_n}}\in\Delta}(s_{PC_1}(d_{I_{C_1}})\between\cdots\between s_{PC_n}(d_{I_{C_n}})\cdot P)$

The state transitions of the Completion Handler $i$ described by APTC are as follows.

$C_i=\sum_{d_{I_{C_i}}\in\Delta}(r_{PC_i}(d_{I_{C_i}})\cdot C_{i_2})$

$C_{i_2}=CF_{i}\cdot C_{i_3}$

$C_{i_3}=\sum_{d_{O_{i}}\in\Delta}(s_{O_{i}}(d_{O_i})\cdot C_i)$

The sending action must occur before the reading action of the same data through the same channel, then they can asynchronously communicate with each other, otherwise, will cause a deadlock $\delta$. We define the following
communication constraint of the Completion Handler $i$ for $1\leq i\leq n$.

$$s_{PC_i}(d_{I_{C_i}})\leq r_{PC_i}(d_{I_{C_i}})$$

Here, $\leq$ is a causality relation.

There are two communication constraints between the Asynchronous Operation Processor and the Asynchronous Operation as follows.

$$s_{I_{PO}}(d_{I_{AO}})\leq r_{PO}(d_{I_{AO}})$$

$$s_{O_{PO}}(d_{O_{AO}})\leq r_{O_{PO}}(d_{O_{AO}})$$

There is one communication constraint between the Asynchronous Operation Processor and the Proactor as follows.

$$s_{PP}(d_{I_{P}})\leq r_{PP}(d_{I_{P}})$$

Let all modules be in parallel, then the Proactor pattern $AOP\quad AO\quad P\quad C_1\cdots C_i\cdots C_n$ can be presented by the following process term.

$\tau_I(\partial_H(\Theta(AOP\between AO\between P\between C_1\between\cdots\between C_i\between\cdots\between C_n)))=\tau_I(\partial_H(AOP\between AO\between P\between C_1\between\cdots\between C_i\between\cdots\between C_n))$

where $H=\{s_{PC_i}(d_{I_{C_i}}), r_{PC_i}(d_{I_{C_i}}),s_{I_{PO}}(d_{I_{AO}}), r_{PO}(d_{I_{AO}}),s_{O_{PO}}(d_{O_{AO}}), r_{O_{PO}}(d_{O_{AO}}),
s_{PP}(d_{I_{P}}), r_{PP}(d_{I_{P}})\\
|s_{PC_i}(d_{I_{C_i}})\nleq r_{PC_i}(d_{I_{C_i}}),
s_{I_{PO}}(d_{I_{AO}})\nleq r_{PO}(d_{I_{AO}}),
s_{O_{PO}}(d_{O_{AO}})\nleq r_{O_{PO}}(d_{O_{AO}}),
s_{PP}(d_{I_{P}})\nleq r_{PP}(d_{I_{P}}),\\
d_{I}, d_{I_{AO}}, d_{I_P}, d_{I_{C_i}}, d_{O_{AO}}, d_{O_{i}}\in\Delta\}$ for $1\leq i\leq n$,

$I=\{s_{PC_i}(d_{I_{C_i}}), r_{PC_i}(d_{I_{C_i}}),s_{I_{PO}}(d_{I_{AO}}), r_{PO}(d_{I_{AO}}),s_{O_{PO}}(d_{O_{AO}}), r_{O_{PO}}(d_{O_{AO}}),\\
s_{PP}(d_{I_{P}}), r_{PP}(d_{I_{P}}), AOPF_1,AOPF_2,AOF,CF_i\\
|s_{PC_i}(d_{I_{C_i}})\leq r_{PC_i}(d_{I_{C_i}}),
s_{I_{PO}}(d_{I_{AO}})\leq r_{PO}(d_{I_{AO}}),
s_{O_{PO}}(d_{O_{AO}})\leq r_{O_{PO}}(d_{O_{AO}}),
s_{PP}(d_{I_{P}})\leq r_{PP}(d_{I_{P}}),\\
d_{I}, d_{I_{AO}}, d_{I_P}, d_{I_{C_i}}, d_{O_{AO}}, d_{O_{i}}\in\Delta\}$ for $1\leq i\leq n$.

Then we get the following conclusion on the Proactor pattern.

\begin{theorem}[Correctness of the Proactor pattern]
The Proactor pattern $\tau_I(\partial_H(AOP\between AO\between P\between C_1\between\cdots\between C_i\between\cdots\between C_n))$ can exhibit desired external behaviors.
\end{theorem}

\begin{proof}
Based on the above state transitions of the above modules, by use of the algebraic laws of APTC, we can prove that

$\tau_I(\partial_H(AOP\between AO\between P\between C_1\between\cdots\between C_i\between\cdots\between C_n))=\sum_{d_{I},d_{O_1},\cdots,d_{O_n}\in\Delta}(r_{I}(d_{I})\cdot s_{O_1}(d_{O_1})\parallel\cdots\parallel s_{O_i}(d_{O_i})\parallel\cdots\parallel s_{O_n}(d_{O_n}))\cdot
\tau_I(\partial_H(AOP\between AO\between P\between C_1\between\cdots\between C_i\between\cdots\between C_n))$,

that is, the Proactor pattern $\tau_I(\partial_H(AOP\between AO\between P\between C_1\between\cdots\between C_i\between\cdots\between C_n))$ can exhibit desired external behaviors.

For the details of proof, please refer to section \ref{app}, and we omit it.
\end{proof}

\subsubsection{Verification of the Asynchronous Completion Token Pattern}

The Asynchronous Completion Token pattern also decouples the delivery the events between the event-driven applications and clients, but the events are triggered by the completion of asynchronous operations,
which has four classes of components: the Initiator, the Asynchronous Operation, the Service and $n$ Completion Handlers.
The Initiator interacts with the outside through the channel $I$; with the Service through the channels $I_{IS}$ and $O_{IS}$; with the Completion Handler $i$ with the channel $IC_i$.
The Service interacts with the Asynchronous Operation through the channel $I_{SO}$ and $O_{SO}$. The Completion Handler $i$ interacts with the outside through the channel $O_i$.
As illustrates in Figure \ref{ACT6}.

\begin{figure}
    \centering
    \includegraphics{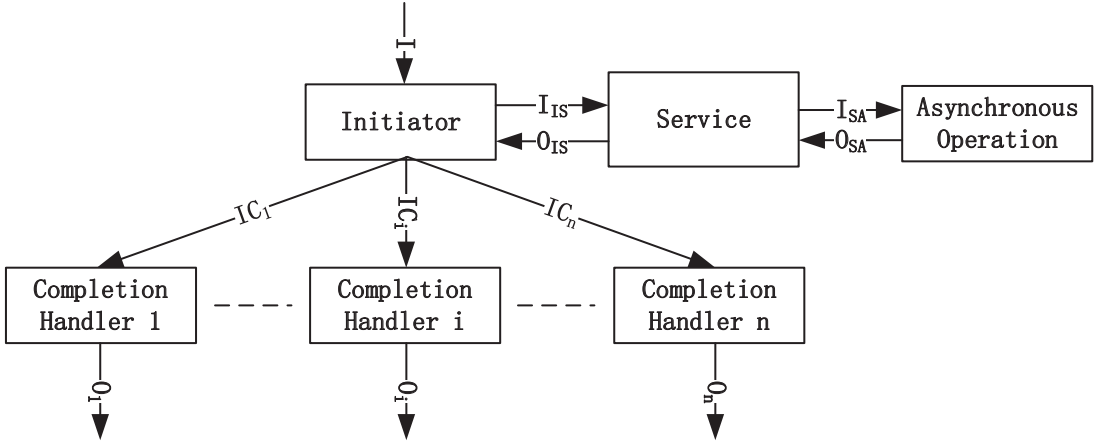}
    \caption{Asynchronous Completion Token pattern}
    \label{ACT6}
\end{figure}

The typical process of the Asynchronous Completion Token pattern is shown in Figure \ref{ACT6P} and following.

\begin{enumerate}
  \item The Initiator receives the input $d_I$ from the user through the channel $I$ (the corresponding reading action is denoted $r_I(D_I)$), processes the input through
  a processing function $IF_1$, and generates the input to the Service $d_{I_{S}}$ and it sends $d_{I_{S}}$ to the Asynchronous Operation through the channel $I_{IS}$
  (the corresponding sending action is denoted $s_{I_{IS}}(d_{I_{S}})$);
  \item The Service receives the input from the Initiator through the channel $I_{IS}$ (the corresponding reading action is denoted $r_{I_{IS}}(d_{I_{S}})$), processes the input through
  a processing function $SF_1$, generates the input to the Asynchronous Operation which is denoted $d_{I_{A}}$; then sends the input to the Asynchronous Operation through the
  channel $I_{SA}$ (the corresponding sending action is denoted $s_{O_{SA}}(d_{O_{A}})$);
  \item The Asynchronous Operation receives the input from the Service through the channel $I_{SA}$ (the corresponding reading action is denoted $r_{I_{SA}}(d_{I_A})$), then processes
  the input and generate the results $d_{O_A}$ through a processing function $AF$, and sends the results to the Service through the channel $O_{SA}$ (the corresponding sending action
  is denoted $s_{O_{SA}}(d_{O_A})$);
  \item the Service receives the results $d_{O_A}$ from the Asynchronous Operation through the channel $O_{SA}$ (the corresponding reading action is denoted $r_{O_{SA}}(d_{O_A})$), then
  processes the results and generates the results $d_{O_S}$ through a processing function $SF_2$, and sends the results to the Initiator through the channel $O_{IS}$ (the corresponding
  sending action is denoted $s_{O_{IS}}(d_{O_S})$);
  \item The Initiator receives the results $d_{O_S}$ from the Service through the channel $O_{IS}$ (the corresponding reading action is denoted $r_{O_{IS}}(d_{O_S})$), then processes
  the results and generates the events $d_{I_{C_i}}$ through a processing function $IF_2$, and sends the processed events to the Completion Handler $i$ $d_{I_{C_i}}$ for $1\leq i\leq n$ through the channel $IC_i$
  (the corresponding sending action is denoted $s_{IC_i}(d_{I_{C_i}})$);
  \item The Completion Handler $i$ (for $1\leq i\leq n$) receives the events from the Initiator through the channel $IC_i$ (the corresponding reading action is denoted $r_{IC_i}(d_{I_{C_i}})$),
  processes the events through a processing function $CF_{i}$, generates the output
  $d_{O_i}$, then sending the output through the channel $O_i$ (the corresponding sending action is denoted $s_{O_i}(d_{O_i})$).
\end{enumerate}

\begin{figure}
    \centering
    \includegraphics{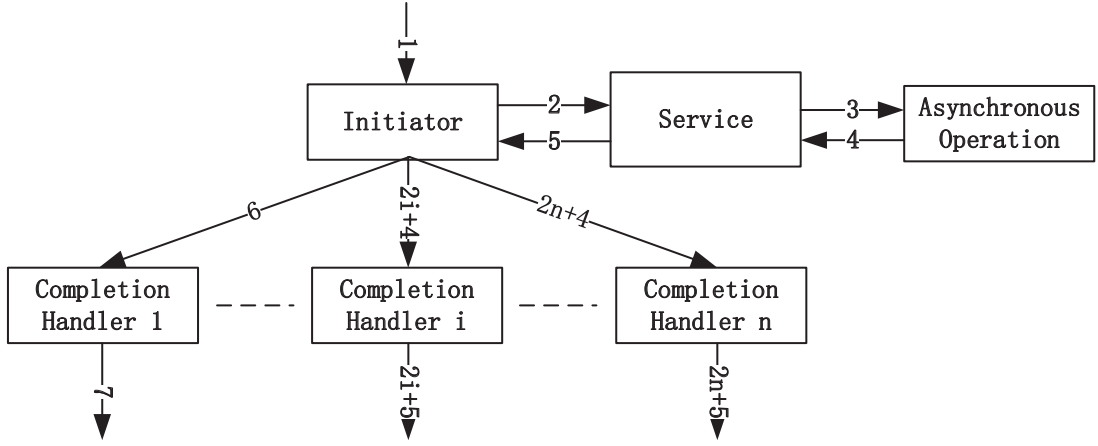}
    \caption{Typical process of Asynchronous Completion Token pattern}
    \label{ACT6P}
\end{figure}

In the following, we verify the Asynchronous Completion Token pattern. We assume all data elements $d_{I}$, $d_{I_S}$, $d_{I_{A}}$, $d_{I_{C_i}}$, $d_{O_{A}}$, $d_{O_S}$, $d_{O_{i}}$ (for $1\leq i\leq n$) are from a finite set
$\Delta$.

The state transitions of the Initiator module
described by APTC are as follows.

$I=\sum_{d_{I}\in\Delta}(r_{I}(d_{I})\cdot I_{2})$

$I_{2}=IF_1\cdot I_{3}$

$I_{3}=\sum_{d_{I_{S}}\in\Delta}(s_{I_{IS}}(d_{I_{IS}})\cdot I_{4})$

$I_4=\sum_{d_{O_{S}}\in\Delta}(r_{O_{IS}}(d_{O_{S}})\cdot I_{5})$

$I_{5}=IF_2\cdot I_{6}$

$I_{6}=\sum_{d_{I_{C_1}},d_{I_{C_n}}\in\Delta}(s_{IC_i}(d_{I_{C_i}})\cdot I)$

The state transitions of the Service described by APTC are as follows.

$S=\sum_{d_{I_{S}}\in\Delta}(r_{I_{IS}}(d_{I_{S}})\cdot S_{2})$

$S_{2}=SF_1\cdot S_{3}$

$S_{3}=\sum_{d_{I_{A}}\in\Delta}(s_{I_{SA}}(d_{I_{A}})\cdot S_4)$

$S_4=\sum_{d_{O_{A}}\in\Delta}(r_{O_{SA}}(d_{O_{A}})\cdot S_{5})$

$S_{5}=SF_2\cdot S_{6}$

$S_{6}=\sum_{d_{O_{S}}\in\Delta}(s_{O_{IS}}(d_{O_{S}})\cdot S)$

The state transitions of the Asynchronous Operation described by APTC are as follows.

$A=\sum_{d_{I_{A}}\in\Delta}(r_{I_{SA}}(d_{I_{A}})\cdot A_{2})$

$A_{2}=AF\cdot A_{3}$

$A_{3}=\sum_{d_{O_{A}}\in\Delta}(s_{O_{SA}}(d_{O_{A}})\cdot A)$

The state transitions of the Completion Handler $i$ described by APTC are as follows.

$C_i=\sum_{d_{I_{C_i}}\in\Delta}(r_{IC_i}(d_{I_{C_i}})\cdot C_{i_2})$

$C_{i_2}=CF_{i}\cdot C_{i_3}$

$C_{i_3}=\sum_{d_{O_{i}}\in\Delta}(s_{O_{i}}(d_{O_i})\cdot C_i)$

The sending action must occur before the reading action of the same data through the same channel, then they can asynchronously communicate with each other, otherwise, will cause a deadlock $\delta$. We define the following
communication constraint of the Completion Handler $i$ for $1\leq i\leq n$.

$$s_{IC_i}(d_{I_{C_i}})\leq r_{IC_i}(d_{I_{C_i}})$$

Here, $\leq$ is a causality relation.

There are two communication constraints between the Initiator and the Service as follows.

$$s_{I_{IS}}(d_{I_{S}})\leq r_{I_{IS}}(d_{I_{S}})$$

$$s_{O_{IS}}(d_{O_{S}})\leq r_{O_{IS}}(d_{O_{S}})$$

There is two communication constraints between the Service and the Asynchronous Operation as follows.

$$s_{I_{SA}}(d_{I_{A}})\leq r_{I_{SA}}(d_{I_{A}})$$

$$s_{O_{SA}}(d_{O_{A}})\leq r_{O_{SA}}(d_{O_{A}})$$

Let all modules be in parallel, then the Asynchronous Completion Token pattern $I\quad S\quad A\quad C_1\cdots C_i\cdots C_n$ can be presented by the following process term.

$\tau_I(\partial_H(\Theta(I\between S\between A\between C_1\between\cdots\between C_i\between\cdots\between C_n)))=\tau_I(\partial_H(I\between S\between A\between C_1\between\cdots\between C_i\between\cdots\between C_n))$

where $H=\{s_{IC_i}(d_{I_{C_i}}), r_{IC_i}(d_{I_{C_i}}),s_{I_{IS}}(d_{I_{S}}), r_{I_{IS}}(d_{I_{S}}),s_{O_{IS}}(d_{O_{S}}), r_{O_{IS}}(d_{O_{S}}),\\
s_{I_{SA}}(d_{I_{A}}), r_{I_{SA}}(d_{I_{A}}),s_{O_{SA}}(d_{O_{A}}), r_{O_{SA}}(d_{O_{A}})\\
|s_{IC_i}(d_{I_{C_i}})\nleq r_{IC_i}(d_{I_{C_i}}),
s_{I_{IS}}(d_{I_{S}})\nleq r_{I_{IS}}(d_{I_{S}}),
s_{O_{IS}}(d_{O_{S}})\nleq r_{O_{IS}}(d_{O_{S}}),
s_{I_{SA}}(d_{I_{A}})\nleq r_{I_{SA}}(d_{I_{A}}),\\
s_{O_{SA}}(d_{O_{A}})\nleq r_{O_{SA}}(d_{O_{A}}),
d_{I}, d_{I_S}, d_{I_{A}}, d_{I_{C_i}}, d_{O_{A}}, d_{O_S}, d_{O_{i}}\in\Delta\}$ for $1\leq i\leq n$,

$I=\{s_{IC_i}(d_{I_{C_i}}), r_{IC_i}(d_{I_{C_i}}),s_{I_{IS}}(d_{I_{S}}), r_{I_{IS}}(d_{I_{S}}),s_{O_{IS}}(d_{O_{S}}), r_{O_{IS}}(d_{O_{S}}),\\
s_{I_{SA}}(d_{I_{A}}), r_{I_{SA}}(d_{I_{A}}),s_{O_{SA}}(d_{O_{A}}), r_{O_{SA}}(d_{O_{A}}), IF_1,IF_2,SF_1,SF_2,AF,CF_i\\
|s_{IC_i}(d_{I_{C_i}})\leq r_{IC_i}(d_{I_{C_i}}),
s_{I_{IS}}(d_{I_{S}})\leq r_{I_{IS}}(d_{I_{S}}),
s_{O_{IS}}(d_{O_{S}})\leq r_{O_{IS}}(d_{O_{S}}),
s_{I_{SA}}(d_{I_{A}})\leq r_{I_{SA}}(d_{I_{A}}),\\
s_{O_{SA}}(d_{O_{A}})\leq r_{O_{SA}}(d_{O_{A}}),
d_{I}, d_{I_S}, d_{I_{A}}, d_{I_{C_i}}, d_{O_{A}}, d_{O_S}, d_{O_{i}}\in\Delta\}$ for $1\leq i\leq n$.

Then we get the following conclusion on the Asynchronous Completion Token pattern.

\begin{theorem}[Correctness of the Asynchronous Completion Token pattern]
The Asynchronous Completion Token pattern $\tau_I(\partial_H(I\between S\between A\between C_1\between\cdots\between C_i\between\cdots\between C_n))$ can exhibit desired external behaviors.
\end{theorem}

\begin{proof}
Based on the above state transitions of the above modules, by use of the algebraic laws of APTC, we can prove that

$\tau_I(\partial_H(I\between S\between A\between C_1\between\cdots\between C_i\between\cdots\between C_n))=\sum_{d_{I},d_{O_1},\cdots,d_{O_n}\in\Delta}(r_{I}(d_{I})\cdot s_{O_1}(d_{O_1})\parallel\cdots\parallel s_{O_i}(d_{O_i})\parallel\cdots\parallel s_{O_n}(d_{O_n}))\cdot
\tau_I(\partial_H(I\between S\between A\between C_1\between\cdots\between C_i\between\cdots\between C_n))$,

that is, the Asynchronous Completion Token pattern $\tau_I(\partial_H(I\between S\between A\between C_1\between\cdots\between C_i\between\cdots\between C_n))$ can exhibit desired external behaviors.

For the details of proof, please refer to section \ref{app}, and we omit it.
\end{proof}

\subsubsection{Verification of the Acceptor-Connector Pattern}

The Acceptor-Connector pattern decouples the connection and initialization of two cooperating peers. There are six modules in the Acceptor-Connector pattern: the two Service Handlers, the two
Dispatchers, and the two initiator: the Connector and the Acceptor. The Service Handlers interact with the user through
the channels $I_1$, $I_2$ and $O_1$, $O_2$; with the Dispatcher through the channels $DS_1$ and $DS_2$; with each other through the channels $I_{SS_1}$ and $I_{SS_2}$.
The Connector interacts with Dispatcher 1 the through the channels $CD$, and with the outside through the channels $I_C$. The Acceptor interacts with the Dispatcher 2 through the
channel $AD$; with the outside through the channel $I_A$. The Dispatchers interact with the Service Handlers through the channels $DS_1$ and $DS_2$. As illustrates in
Figure \ref{AC6}.

\begin{figure}
    \centering
    \includegraphics{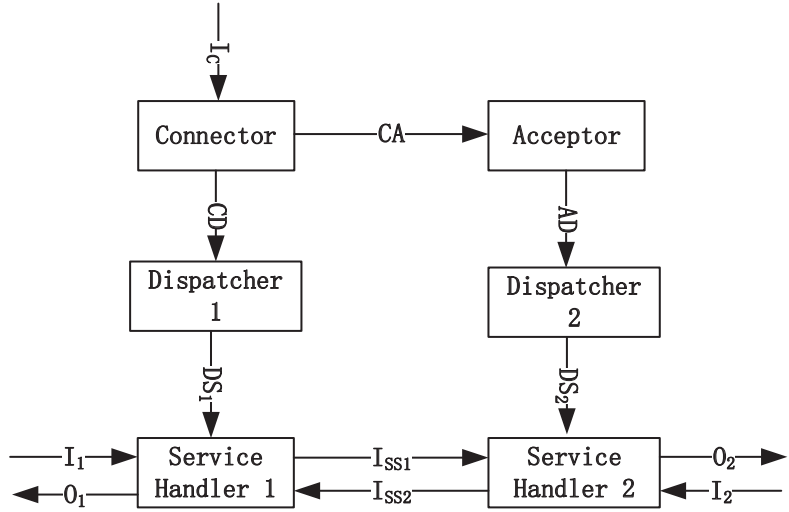}
    \caption{Acceptor-Connector pattern}
    \label{AC6}
\end{figure}

The typical process of the Acceptor-Connector pattern is shown in Figure \ref{AC6P} and as follows.

\begin{enumerate}
  \item The Connector receives the request $d_{I_C}$ from the outside through the channel $I_C$ (the corresponding reading action is denoted $r_{I_C}(d_{I_C})$), then processes the request
  and generates the request $d_{I_{D_1}}$ and $d_{I_A}$ through a processing function $CF$, and sends the request to the Dispatcher 1 through the channel $CD$ (the corresponding sending action is
  denoted $s_{CD}(d_{I_{D_1}})$) and sends the request to the Acceptor through the channel $CA$ (the corresponding sending action is
  denoted $s_{CA}(d_{I_{A}})$);
  \item The Dispatcher 1 receives the request $d_{I_{D_1}}$ from the Connector through the channel $CD$ (the corresponding reading action is denoted $r_{CD}(d_{I_{D_1}})$), the  processes
  the request and generates the request $d_{I_{S_1}}$ through a processing function $D1F$, and sends the request to the Service Handler 1 through the channel $DS_1$ (the corresponding
  sending action is denoted $s_{DS_1}(d_{I_{S_1}})$);
  \item The Service Handler 1 receives the request $d_{I_{S_1}}$ from the Dispatcher 1 through the channel $DS_1$ (the corresponding reading action is denoted $r_{DS_1}(d_{I_{S_1}})$),
  then processes the request through a processing function $S1F_1$ and make ready to accept the request from the outside;
  \item The Acceptor receives the request $d_{I_A}$ from the Connector through the channel $CA$ (the corresponding reading action is denoted $r_{CA}(d_{I_A})$), then processes the request
  and generates the request $d_{I_{D_2}}$ through a processing function $AF$, and sends the request to the Dispatcher 1 through the channel $AD$ (the corresponding sending action is
  denoted $s_{AD}(d_{I_{D_2}})$);
  \item The Dispatcher 2 receives the request $d_{I_{D_2}}$ from the Acceptor through the channel $AD$ (the corresponding reading action is denoted $r_{AD}(d_{I_{D_2}})$), the  processes
  the request and generates the request $d_{I_{S_2}}$ through a processing function $D2F$, and sends the request to the Service Handler 2 through the channel $DS_2$ (the corresponding
  sending action is denoted $s_{DS_2}(d_{I_{S_2}})$);
  \item The Service Handler 2 receives the request $d_{I_{S_2}}$ from the Dispatcher 2 through the channel $DS_2$ (the corresponding reading action is denoted $r_{DS_2}(d_{I_{S_2}})$),
  then processes the request through a processing function $S2F_1$ and make ready to accept the request from the outside;
  \item The Service Handler 1 receives the request $d_{I_1}$ from the user through the channel $I_1$ (the corresponding reading action is denoted $r_{I_1}(d_{I_1})$), then processes the request $d_{I_1}$ through a processing
  function $S1F_2$, and sends the processed request $d_{I_{SS_2}}$ to the Service Handler 2 through the channel $I_{SS_1}$ (the corresponding sending action is denoted $s_{I_{SS_1}}(d_{I_{SS_2}})$);
  \item The Service Handler 2 receives the request $d_{I_{SS_2}}$ from the Service Handler 1 through the channel $I_{SS_1}$ (the corresponding reading action is denoted $r_{I_{SS_1}}(d_{I_{SS_2}})$), then
  processes the request and generates the response $d_{O_2}$ through a processing function $S2F_3$, and sends the response to the outside through the channel $O_{2}$ (the corresponding sending action is denoted
  $s_{O_{2}}(d_{O_2})$);
  \item The Service Handler 2 receives the request $d_{I_2}$ from the user through the channel $I_2$ (the corresponding reading action is denoted $r_{I_2}(d_{I_2})$), then processes the request $d_{I_2}$ through a processing
  function $S2F_2$, and sends the processed request $d_{I_{SS_1}}$ to the Service Handler 1 through the channel $I_{SS_2}$ (the corresponding sending action is denoted $s_{I_{SS_2}}(d_{I_{SS_1}})$);
  \item The Service Handler 1 receives the request $d_{I_{SS_1}}$ from the Service Handler 2 through the channel $I_{SS_2}$ (the corresponding reading action is denoted $r_{I_{SS_2}}(d_{I_{SS_1}})$), then
  processes the request and generates the response $d_{O_1}$ through a processing function $S1F_3$, and sends the response to the outside through the channel $O_{1}$ (the corresponding sending action is denoted
  $s_{O_{1}}(d_{O_1})$).
\end{enumerate}

\begin{figure}
    \centering
    \includegraphics{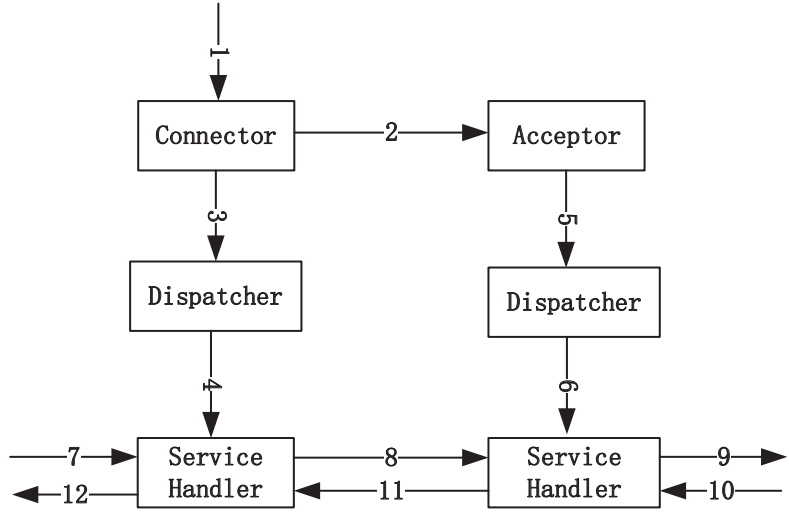}
    \caption{Typical process of Acceptor-Connector pattern}
    \label{AC6P}
\end{figure}

In the following, we verify the Acceptor-Connector pattern. We assume all data elements $d_{I_1}$, $d_{I_2}$, $d_{I_C}$, $d_{I_A}$, $d_{I_{D_1}}$, $d_{I_{D_2}}$, $d_{I_{S_1}}$, $d_{I_{S_2}}$, $d_{I_{SS_1}}$, $d_{I_{SS_2}}$, $d_{O_1}$, $d_{O_{2}}$ are from a finite set
$\Delta$. We only give the transitions of the first process.

The state transitions of the Connector module
described by APTC are as follows.

$C=\sum_{d_{I_{C}}\in\Delta}(r_{I_C}(d_{I_{C}})\cdot C_{2})$

$C_{2}=CF\cdot C_{3}$

$C_{3}=\sum_{d_{I_{A}},d_{I_{D_1}}\in\Delta}(s_{CA}(d_{I_{A}})\between s_{CD}(d_{I_{D_1}})\cdot C)$

The state transitions of the Dispatcher 1 module
described by APTC are as follows.

$D1=\sum_{d_{I_{D_1}}\in\Delta}(r_{CD}(d_{I_{D_1}})\cdot D1_{2})$

$D1_{2}=D1F\cdot D1_{3}$

$D1_{3}=\sum_{d_{I_{S_1}}\in\Delta}(s_{DS_1}(d_{I_{S_1}})\cdot D1)$

The state transitions of the Service Handler 1 module
described by APTC are as follows.

$S1=\sum_{d_{I_1},d_{I_{S_1}},d_{I_{SS_1}}\in\Delta}(r_{I_1}(d_{I_1})\between r_{DS_1}(d_{I_{S_1}})\between r_{I_{SS_2}}(d_{I_{SS_1}})\cdot S1_{2})$

$S1_{2}=S1F_1\between S1F_2\between S1F_3\cdot S1_{3}$

$S1_{3}=\sum_{d_{I_{SS_2}},d_{O_1}\in\Delta}(s_{I_{SS_1}}(d_{I_{SS_2}})\between s_{O_1}(d_{O_1})\cdot S1)$

The state transitions of the Acceptor module
described by APTC are as follows.

$A=\sum_{d_{I_{A}}\in\Delta}(r_{CA}(d_{I_{A}})\cdot A_{2})$

$A_{2}=AF\cdot A_{3}$

$A_{3}=\sum_{d_{I_{D_2}}\in\Delta}(s_{AD}(d_{I_{D_2}})\cdot A)$

The state transitions of the Dispatcher 2 module
described by APTC are as follows.

$D2=\sum_{d_{I_{D_2}}\in\Delta}(r_{AD}(d_{I_{D_2}})\cdot D2_{2})$

$D2_{2}=D2F\cdot D2_{3}$

$D2_{3}=\sum_{d_{I_{S_2}}\in\Delta}(s_{DS_2}(d_{I_{S_2}})\cdot D2)$

The state transitions of the Service Handler 2 module
described by APTC are as follows.

$S2=\sum_{d_{I_2},d_{I_{S_2}},d_{I_{SS_2}}\in\Delta}(r_{I_2}(d_{I_2})\between r_{DS_2}(d_{I_{S_2}})\between r_{I_{SS_1}}(d_{I_{SS_2}})\cdot S2_{2})$

$S2_{2}=S2F_1\between S2F_2\between S2F_3\cdot S2_{3}$

$S2_{3}=\sum_{d_{I_{SS_1}},d_{O_2}\in\Delta}(s_{I_{SS_2}}(d_{I_{SS_1}})\between s_{O_2}(d_{O_2})\cdot S2)$

The sending action and the reading action of the same data through the same channel can communicate with each other, otherwise, will cause a deadlock $\delta$. We define the following
communication functions between the Connector and the Acceptor.

$$\gamma(r_{CA}(d_{I_{A}}),s_{CA}(d_{I_{A}}))\triangleq c_{CA}(d_{I_{A}})$$

There are one communication functions between the Connector and the Dispatcher 1 as follows.

$$\gamma(r_{CD}(d_{I_{D_1}}),s_{CD}(d_{I_{D_1}}))\triangleq c_{CD}(d_{I_{D_1}})$$

There are one communication functions between the Dispatcher 1 and the Service Handler 1 as follows.

$$\gamma(r_{DS_1}(d_{I_{S_1}}),s_{DS_1}(d_{I_{S_1}}))\triangleq c_{DS_1}(d_{I_{S_1}})$$

We define the following communication functions between the Acceptor and the Dispatcher 2.

$$\gamma(r_{AD}(d_{I_{D_2}}),s_{AD}(d_{I_{D_2}}))\triangleq c_{AD}(d_{I_{D_2}})$$

There are one communication functions between the Dispatcher 2 and the Service Handler 2 as follows.

$$\gamma(r_{DS_2}(d_{I_{S_2}}),s_{DS_2}(d_{I_{S_2}}))\triangleq c_{DS_2}(d_{I_{S_2}})$$

There are one communication functions between the Service Handler 1 and the Service Handler 2 as follows.

$$\gamma(r_{I_{SS_1}}(d_{I_{SS_2}}),s_{I_{SS_1}}(d_{I_{SS_2}}))\triangleq c_{I_{SS_1}}(d_{I_{SS_2}})$$

$$\gamma(r_{I_{SS_2}}(d_{I_{SS_1}}),s_{I_{SS_2}}(d_{I_{SS_1}}))\triangleq c_{I_{SS_2}}(d_{I_{SS_1}})$$

Let all modules be in parallel, then the Acceptor-Connector pattern $C\quad D1 \quad S1\quad A\quad D2\quad S2$ can be presented by the following process term.

$\tau_I(\partial_H(\Theta(C\between D1\between S1\between A\between D2\between S2)))=\tau_I(\partial_H(C\between D1\between S1\between A\between D2\between S2))$

where $H=\{r_{CA}(d_{I_{A}}),s_{CA}(d_{I_{A}}),r_{CD}(d_{I_{D_1}}),s_{CD}(d_{I_{D_1}}),r_{DS_1}(d_{I_{S_1}}),s_{DS_1}(d_{I_{S_1}}),r_{AD}(d_{I_{D_2}}),s_{AD}(d_{I_{D_2}}),\\
r_{DS_2}(d_{I_{S_2}}),s_{DS_2}(d_{I_{S_2}}),r_{I_{SS_1}}(d_{I_{SS_2}}),s_{I_{SS_1}}(d_{I_{SS_2}}),r_{I_{SS_2}}(d_{I_{SS_1}}),s_{I_{SS_2}}(d_{I_{SS_1}})\\
|d_{I_1}, d_{I_2}, d_{I_C}, d_{I_A}, d_{I_{D_1}}, d_{I_{D_2}}, d_{I_{S_1}}, d_{I_{S_2}}, d_{I_{SS_1}}, d_{I_{SS_2}}, d_{O_1}, d_{O_{2}}\in\Delta\}$,

$I=\{c_{CA}(d_{I_{A}}),c_{CD}(d_{I_{D_1}}),c_{DS_1}(d_{I_{S_1}}),c_{AD}(d_{I_{D_2}}),c_{DS_2}(d_{I_{S_2}}),c_{I_{SS_1}}(d_{I_{SS_2}}),c_{I_{SS_2}}(d_{I_{SS_1}}),\\
CF,AF,D1F,D2F,S1F_1,S1F_2,S1F_3,S2F_1,S2F_2,S2F_3\\
|d_{I_1}, d_{I_2}, d_{I_C}, d_{I_A}, d_{I_{D_1}}, d_{I_{D_2}}, d_{I_{S_1}}, d_{I_{S_2}}, d_{I_{SS_1}}, d_{I_{SS_2}}, d_{O_1}, d_{O_{2}}\in\Delta\}$.

Then we get the following conclusion on the Acceptor-Connector pattern.

\begin{theorem}[Correctness of the Acceptor-Connector pattern]
The Acceptor-Connector pattern $\tau_I(\partial_H(C\between D1\between S1\between A\between D2\between S2))$ can exhibit desired external behaviors.
\end{theorem}

\begin{proof}
Based on the above state transitions of the above modules, by use of the algebraic laws of APTC, we can prove that

$\tau_I(\partial_H(C\between D1\between S1\between A\between D2\between S2))=\sum_{d_{I_C},d_{I_1},d_{I_2},d_{O_1},d_{O_2}\in\Delta}(r_{I_C}(d_{I_C})\parallel(r_{I_1}(d_{I_1})\cdot s_{O_2}(d_{O_2}))\parallel(r_{I_2}(d_{I_2})\cdot s_{O_1}(d_{O_1})))\cdot
\tau_I(\partial_H(C\between D1\between S1\between A\between D2\between S2))$,

that is, the Acceptor-Connector pattern $\tau_I(\partial_H(C\between D1\between S1\between A\between D2\between S2))$ can exhibit desired external behaviors.

For the details of proof, please refer to section \ref{app}, and we omit it.
\end{proof}

\subsection{Synchronization Patterns}\label{S6}

In this subsection, we verify the synchronization patterns, including the Scoped Locking pattern, the Strategized Locking pattern, the Thread-Safe Interface pattern, and the
Double-Checked Locking Optimization pattern.

\subsubsection{Verification of the Scoped Locking Pattern}

The Scoped Locking pattern ensures that a lock is acquired automatically when control enters a scope and released when control leaves the scope.
In Scoped Locking pattern, there are two classes of modules: The $n$ Controls and the Guard. The Control $i$ interacts with the
outside through the input channel $I_i$ and the output channel $O_i$; with the Guard through the channel $CG_i$ for $1\leq i\leq n$, as illustrated in Figure \ref{SL6}.

\begin{figure}
    \centering
    \includegraphics{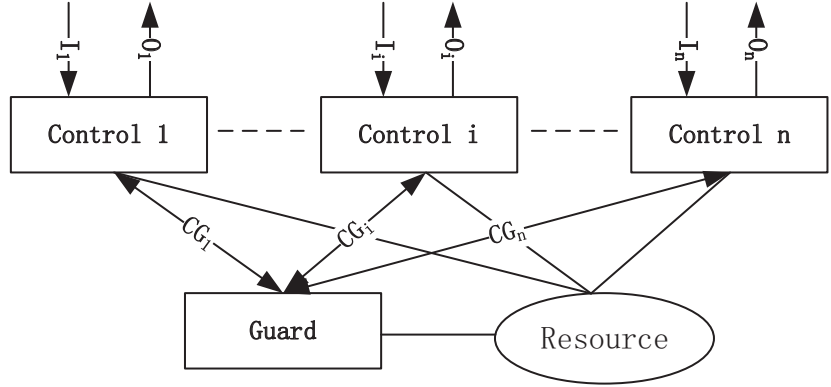}
    \caption{Scoped Locking pattern}
    \label{SL6}
\end{figure}

The typical process is shown in Figure \ref{SL6P} and as follows.

\begin{enumerate}
  \item The Control $i$ receives the input $d_{I_i}$ from the outside through the channel $I_i$ (the corresponding reading action is denoted $r_{I_i}(d_{I_i})$), then it processes the
  input and generates the input $d_{I_{G_i}}$ through a processing function $CF_{i1}$, and it sends the input to the Guard through the channel $CG_i$
  (the corresponding sending action is denoted $s_{CG_i}(d_{I_{G_i}})$);
  \item The Guard receives the input $d_{I_{G_i}}$ from the Control $i$ through the channel $CG_i$ (the corresponding reading action is denoted $r_{CG_i}(d_{I_{G_i}})$) for $1\leq i\leq n$,
  then processes the request and generates the output $d_{O_{G_i}}$ through a processing function $GF_i$, (note that, after the processing, a lock is acquired), and sends the output
  to the Control $i$ through the channel $CG_i$ (the corresponding sending action is denoted $s_{CG_i}(d_{O_{G_i}})$);
  \item The Control $i$ receives the output from the Guard through the channel $CG_i$ (the corresponding reading action is denoted $r_{CG_i}(d_{O_{G_i}})$), then processes the output
  and generate the output $d_{O_i}$ through a processing function $CF_{i2}$ (accessing the resource), and sends the output to the outside through the channel $O_i$ (the corresponding sending action is denoted
  $s_{O_i}(d_{O_i})$).
\end{enumerate}

\begin{figure}
    \centering
    \includegraphics{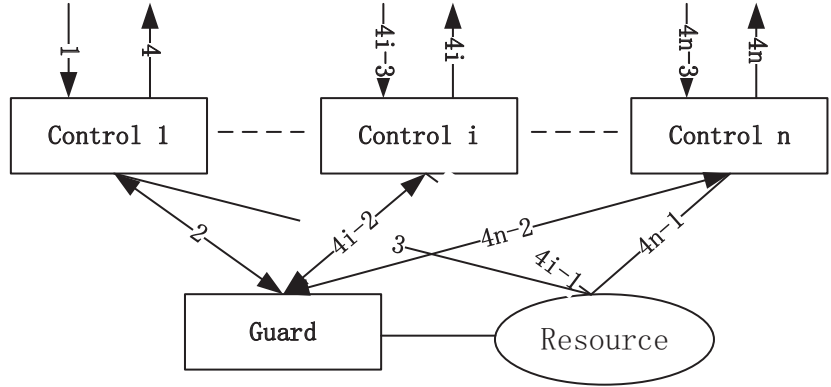}
    \caption{Typical process of Scoped Locking pattern}
    \label{SL6P}
\end{figure}

In the following, we verify the Scoped Locking pattern. We assume all data elements $d_{I_i}$, $d_{O_i}$, $d_{I_{G_i}}$, $d_{O_{G_i}}$ for $1\leq i\leq n$ are from a finite set
$\Delta$.

The state transitions of the Control $i$ module
described by APTC are as follows.

$C_i=\sum_{d_{I_i}\in\Delta}(r_{I_i}(d_{I_i})\cdot C_{i_2})$

$C_{i_2}=CF_{i1}\cdot C_{i_3}$

$C_{i_3}=\sum_{d_{I_{G_i}}\in\Delta}(s_{CG_i}(d_{I_{G_i}})\cdot C_{i_4})$

$C_{i_4}=\sum_{d_{O_{G_i}}\in\Delta}(r_{CG_i}(d_{O_{G_i}})\cdot C_{i_5})$

$C_{i_5}=CF_{i2}\cdot C_{i_6}\quad CF_{12}\%\cdots\%CF_{n2}$

$C_{i_6}=\sum_{d_{O_i}\in\Delta}(s_{O_i}(d_{O_i})\cdot C_{i})$

The state transitions of the Guard module
described by APTC are as follows.

$G=\sum_{d_{I_{G_1}},\cdots,d_{I_{G_n}}\in\Delta}(r_{CG_1}(d_{I_{G_1}})\between\cdots\between r_{CG_n}(d_{I_{G_n}})\cdot G_{2})$

$G_{2}=GF_1\between\cdots\between GF_n\cdot G_{3}\quad (GF_1\%\cdots\% GF_n)$

$G_{3}=\sum_{d_{O_{G_1}},\cdots,d_{O_{G_n}}\in\Delta}(s_{CG_1}(d_{O_{G_1}})\between\cdots\between s_{CG_n}(d_{O_{G_n}})\cdot G)$

The sending action and the reading action of the same data through the same channel can communicate with each other, otherwise, will cause a deadlock $\delta$. We define the following
communication functions between the Control $i$ and the Guard.

$$\gamma(r_{CG_i}(d_{I_{G_i}}),s_{CG_i}(d_{I_{G_i}}))\triangleq c_{CG_i}(d_{I_{G_i}})$$

$$\gamma(r_{CG_i}(d_{O_{G_i}}),s_{CG_i}(d_{O_{G_i}}))\triangleq c_{CG_i}(d_{O_{G_i}})$$

Let all modules be in parallel, then the Scoped Locking pattern $C_1\cdots C_n\quad G$ can be presented by the following process term.

$\tau_I(\partial_H(\Theta(C_1\between \cdots\between C_n\between G)))=\tau_I(\partial_H(C_1\between \cdots\between C_n\between G))$

where $H=\{r_{CG_i}(d_{I_{G_i}}),s_{CG_i}(d_{I_{G_i}}),r_{CG_i}(d_{O_{G_i}}),s_{CG_i}(d_{O_{G_i}})|d_{I_i}, d_{O_i}, d_{I_{G_i}}, d_{O_{G_i}}\in\Delta\}$,

$I=\{c_{CG_i}(d_{I_{G_i}}),c_{CG_i}(d_{O_{G_i}}),CF_{i1},CF_{i2},GF_i|d_{I_i}, d_{O_i}, d_{I_{G_i}}, d_{O_{G_i}}\in\Delta\}$ for $1\leq i\leq n$.

Then we get the following conclusion on the Scoped Locking pattern.

\begin{theorem}[Correctness of the Scoped Locking pattern]
The Scoped Locking pattern $\tau_I(\partial_H(C_1\between \cdots\between C_n\between G))$ can exhibit desired external behaviors.
\end{theorem}

\begin{proof}
Based on the above state transitions of the above modules, by use of the algebraic laws of APTC, we can prove that

$\tau_I(\partial_H(C_1\between \cdots\between C_n\between G))=\sum_{d_{I_1},d_{O_1},\cdots,d_{I_n},d_{O_n}\in\Delta}(r_{I_1}(d_{I_1})\parallel\cdots\parallel r_{I_n}(d_{I_n})\cdot s_{O_1}(d_{O_1})\parallel\cdots\parallel s_{O_n}(d_{O_n}))\cdot
\tau_I(\partial_H(C_1\between \cdots\between C_n\between G))$,

that is, the Scoped Locking pattern $\tau_I(\partial_H(C_1\between \cdots\between C_n\between G))$ can exhibit desired external behaviors.

For the details of proof, please refer to section \ref{app}, and we omit it.
\end{proof}

\subsubsection{Verification of the Strategized Locking Pattern}

The Strategized Locking pattern uses a component (the LockStrategy) to parameterize the synchronization for protecting the concurrent access to the critical section.
In Strategized Locking pattern, there are two classes of modules: The $n$ Components and the $n$ LockStrategies. The Component $i$ interacts with the
outside through the input channel $I_i$ and the output channel $O_i$; with the LockStrategy $i$ through the channel $CL_i$ for $1\leq i\leq n$, as illustrated in Figure \ref{StrL6}.

\begin{figure}
    \centering
    \includegraphics{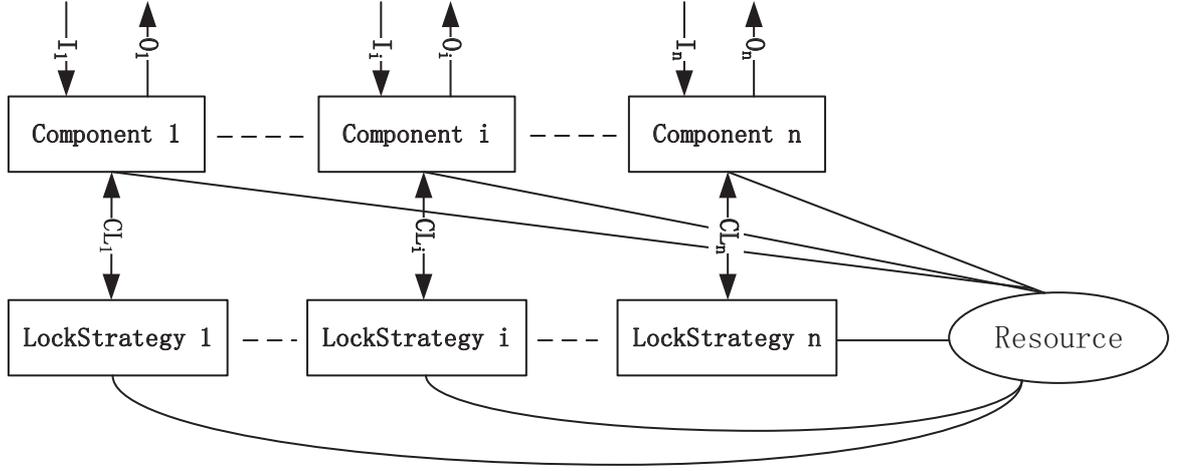}
    \caption{Strategized Locking pattern}
    \label{StrL6}
\end{figure}

The typical process is shown in Figure \ref{StrL6P} and as follows.

\begin{enumerate}
  \item The Component $i$ receives the input $d_{I_i}$ from the outside through the channel $I_i$ (the corresponding reading action is denoted $r_{I_i}(d_{I_i})$), then it processes the
  input and generates the input $d_{I_{L_i}}$ through a processing function $CF_{i1}$, and it sends the input to the LockStrategy through the channel $CL_i$
  (the corresponding sending action is denoted $s_{CL_i}(d_{I_{L_i}})$);
  \item The LockStrategy receives the input $d_{I_{L_i}}$ from the Component $i$ through the channel $CL_i$ (the corresponding reading action is denoted $r_{CL_i}(d_{I_{L_i}})$) for $1\leq i\leq n$,
  then processes the request and generates the output $d_{O_{L_i}}$ through a processing function $LF_i$, (note that, after the processing, a lock is acquired), and sends the output
  to the Component $i$ through the channel $CL_i$ (the corresponding sending action is denoted $s_{CL_i}(d_{O_{L_i}})$);
  \item The Component $i$ receives the output from the LockStrategy through the channel $CL_i$ (the corresponding reading action is denoted $r_{CL_i}(d_{O_{L_i}})$), then processes the output
  and generate the output $d_{O_i}$ through a processing function $CF_{i2}$, and sends the output to the outside through the channel $O_i$ (the corresponding sending action is denoted
  $s_{O_i}(d_{O_i})$).
\end{enumerate}

\begin{figure}
    \centering
    \includegraphics{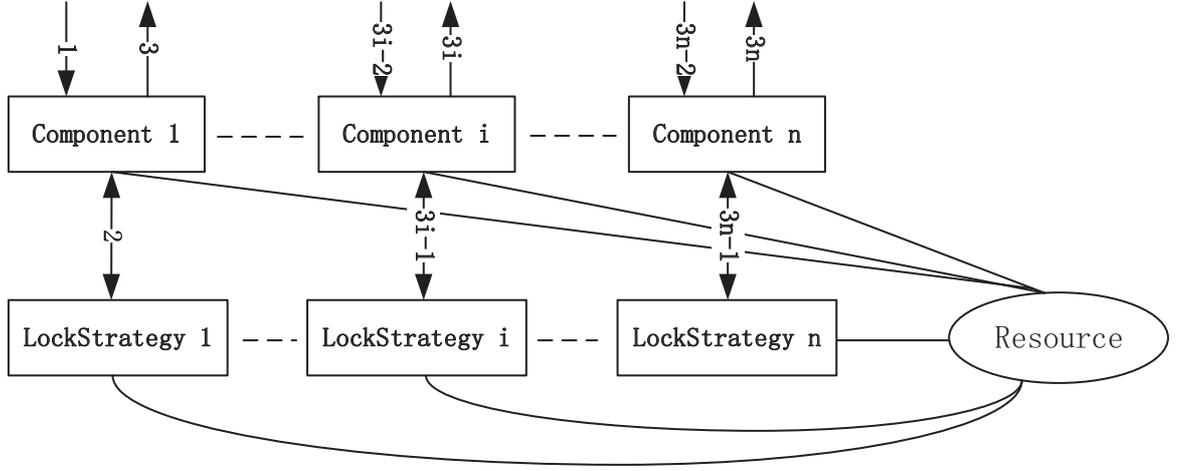}
    \caption{Typical process of Strategized Locking pattern}
    \label{StrL6P}
\end{figure}

In the following, we verify the Strategized Locking pattern. We assume all data elements $d_{I_i}$, $d_{O_i}$, $d_{I_{L_i}}$, $d_{O_{L_i}}$ for $1\leq i\leq n$ are from a finite set
$\Delta$.

The state transitions of the Component $i$ module
described by APTC are as follows.

$C_i=\sum_{d_{I_i}\in\Delta}(r_{I_i}(d_{I_i})\cdot C_{i_2})$

$C_{i_2}=CF_{i1}\cdot C_{i_3}$

$C_{i_3}=\sum_{d_{I_{L_i}}\in\Delta}(s_{CL_i}(d_{I_{L_i}})\cdot C_{i_4})$

$C_{i_4}=\sum_{d_{O_{L_i}}\in\Delta}(r_{CL_i}(d_{O_{L_i}})\cdot C_{i_5})$

$C_{i_5}=CF_{i2}\cdot C_{i_6}\quad (CF_{12}\%\cdots\% CF_{n2})$

$C_{i_6}=\sum_{d_{O_i}\in\Delta}(s_{O_i}(d_{O_i})\cdot C_{i})$

The state transitions of the LockStrategy $i$ module
described by APTC are as follows.

$L_i=\sum_{d_{I_{L_i}}\in\Delta}(r_{CL_i}(d_{I_{L_i}})\cdot L_{i_2})$

$L_{i_2}=LF_i\cdot L_{i_3}\quad (LF_1\%\cdots\% LF_n)$

$L_{i_3}=\sum_{d_{O_{L_i}}\in\Delta}(s_{CL_i}(d_{O_{L_i}})\cdot L_i)$

The sending action and the reading action of the same data through the same channel can communicate with each other, otherwise, will cause a deadlock $\delta$. We define the following
communication functions between the Component $i$ and the LockStrategy $i$.

$$\gamma(r_{CL_i}(d_{I_{L_i}}),s_{CL_i}(d_{I_{L_i}}))\triangleq c_{CL_i}(d_{I_{L_i}})$$

$$\gamma(r_{CL_i}(d_{O_{L_i}}),s_{CL_i}(d_{O_{L_i}}))\triangleq c_{CL_i}(d_{O_{L_i}})$$

Let all modules be in parallel, then the Strategized Locking pattern $C_1\cdots C_n\quad L_1\cdots L_n$ can be presented by the following process term.

$\tau_I(\partial_H(\Theta(C_1\between \cdots\between C_n\between L_1\between\cdots\between L_n)))=\tau_I(\partial_H(C_1\between \cdots\between C_n\between L_1\between\cdots\between L_n))$

where $H=\{r_{CL_i}(d_{I_{L_i}}),s_{CL_i}(d_{I_{L_i}}),r_{CL_i}(d_{O_{L_i}}),s_{CL_i}(d_{O_{L_i}})|d_{I_i}, d_{O_i}, d_{I_{L_i}}, d_{O_{L_i}}\in\Delta\}$,

$I=\{c_{CL_i}(d_{I_{L_i}}),c_{CL_i}(d_{O_{L_i}}),CF_{i1},CF_{i2},LF_i|d_{I_i}, d_{O_i}, d_{I_{L_i}}, d_{O_{L_i}}\in\Delta\}$ for $1\leq i\leq n$.

Then we get the following conclusion on the Strategized Locking pattern.

\begin{theorem}[Correctness of the Strategized Locking pattern]
The Strategized Locking pattern $\tau_I(\partial_H(C_1\between \cdots\between C_n\between L_1\between\cdots\between L_n))$ can exhibit desired external behaviors.
\end{theorem}

\begin{proof}
Based on the above state transitions of the above modules, by use of the algebraic laws of APTC, we can prove that

$\tau_I(\partial_H(C_1\between \cdots\between C_n\between L_1\between\cdots\between L_n))=\sum_{d_{I_1},d_{O_1},\cdots,d_{I_n},d_{O_n}\in\Delta}(r_{I_1}(d_{I_1})\parallel\cdots\parallel r_{I_n}(d_{I_n})\cdot s_{O_1}(d_{O_1})\parallel\cdots\parallel s_{O_n}(d_{O_n}))\cdot
\tau_I(\partial_H(C_1\between \cdots\between C_n\between L_1\between\cdots\between L_n))$,

that is, the Strategized Locking pattern $\tau_I(\partial_H(C_1\between \cdots\between C_n\between L_1\between\cdots\between L_n))$ can exhibit desired external behaviors.

For the details of proof, please refer to section \ref{app}, and we omit it.
\end{proof}

\subsubsection{Verification of the Double-Checked Locking Optimization Pattern}

The Double-Checked Locking Optimization pattern ensures that a lock is acquired in a thread-safe manner.
In Double-Checked Locking Optimization pattern, there are two classes of modules: The $n$ Threads and the Singleton Lock. The Thread $i$ interacts with the
outside through the input channel $I_i$ and the output channel $O_i$; with the Singleton Lock through the channel $TS_i$ for $1\leq i\leq n$, as illustrated in Figure \ref{DCLO6}.

\begin{figure}
    \centering
    \includegraphics{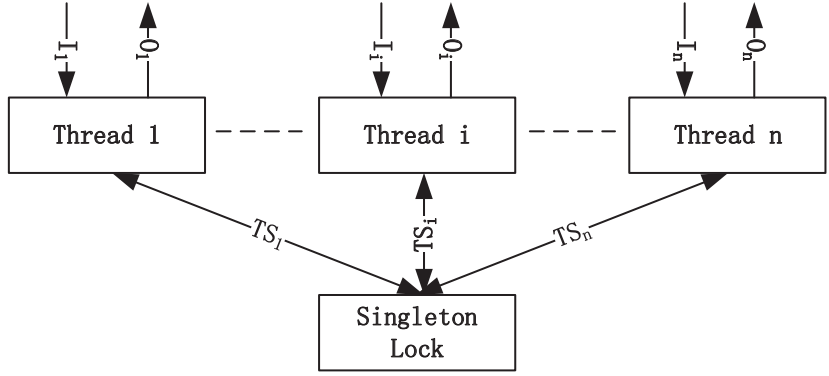}
    \caption{Double-Checked Locking Optimization pattern}
    \label{DCLO6}
\end{figure}

The typical process is shown in Figure \ref{DCLO6P} and as follows.

\begin{enumerate}
  \item The Thread $i$ receives the input $d_{I_i}$ from the outside through the channel $I_i$ (the corresponding reading action is denoted $r_{I_i}(d_{I_i})$), then it processes the
  input and generates the input $d_{I_{S_i}}$ through a processing function $TF_{i1}$, and it sends the input to the Singleton Lock through the channel $TS_i$
  (the corresponding sending action is denoted $s_{TS_i}(d_{I_{S_i}})$);
  \item The Singleton Lock receives the input $d_{I_{S_i}}$ from the Thread $i$ through the channel $TS_i$ (the corresponding reading action is denoted $r_{TS_i}(d_{I_{S_i}})$) for $1\leq i\leq n$,
  then processes the request and generates the output $d_{O_{S_i}}$ through a processing function $SF_i$, (note that, after the processing, a lock is acquired), and sends the output
  to the Thread $i$ through the channel $TS_i$ (the corresponding sending action is denoted $s_{TS_i}(d_{O_{S_i}})$);
  \item The Thread $i$ receives the output from the Singleton Lock through the channel $TS_i$ (the corresponding reading action is denoted $r_{TS_i}(d_{O_{S_i}})$), then processes the output
  and generate the output $d_{O_i}$ through a processing function $TF_{i2}$ (accessing the resource), and sends the output to the outside through the channel $O_i$ (the corresponding sending action is denoted
  $s_{O_i}(d_{O_i})$).
\end{enumerate}

\begin{figure}
    \centering
    \includegraphics{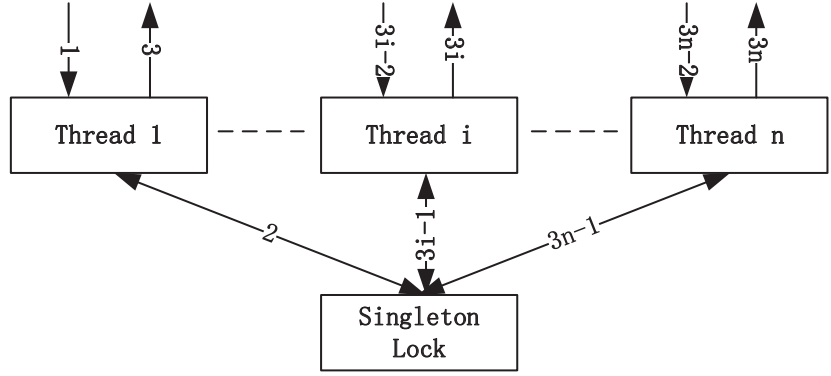}
    \caption{Typical process of Double-Checked Locking Optimization pattern}
    \label{DCLO6P}
\end{figure}

In the following, we verify the Double-Checked Locking Optimization pattern. We assume all data elements $d_{I_i}$, $d_{O_i}$, $d_{I_{S_i}}$, $d_{O_{S_i}}$ for $1\leq i\leq n$ are from a finite set
$\Delta$.

The state transitions of the Thread $i$ module
described by APTC are as follows.

$T_i=\sum_{d_{I_i}\in\Delta}(r_{I_i}(d_{I_i})\cdot T_{i_2})$

$T_{i_2}=TF_{i1}\cdot T_{i_3}$

$T_{i_3}=\sum_{d_{I_{S_i}}\in\Delta}(s_{TS_i}(d_{I_{S_i}})\cdot T_{i_4})$

$T_{i_4}=\sum_{d_{O_{S_i}}\in\Delta}(r_{TS_i}(d_{O_{S_i}})\cdot T_{i_5})$

$T_{i_5}=TF_{i2}\cdot T_{i_6}\quad (TF_{12}\%\cdots\%TF_{n2})$

$T_{i_6}=\sum_{d_{O_i}\in\Delta}(s_{O_i}(d_{O_i})\cdot T_{i})$

The state transitions of the Singleton Lock module
described by APTC are as follows.

$S=\sum_{d_{I_{S_1}},\cdots,d_{I_{S_n}}\in\Delta}(r_{TS_1}(d_{I_{S_1}})\between\cdots\between r_{TS_n}(d_{I_{S_n}})\cdot S_{2})$

$S_{2}=SF_1\between\cdots\between SF_n\cdot S_{3}\quad (SF_1\%\cdots\% SF_n)$

$S_{3}=\sum_{d_{O_{S_1}},\cdots,d_{O_{S_n}}\in\Delta}(s_{TS_1}(d_{O_{S_1}})\between\cdots\between s_{TS_n}(d_{O_{S_n}})\cdot S)$

The sending action and the reading action of the same data through the same channel can communicate with each other, otherwise, will cause a deadlock $\delta$. We define the following
communication functions between the Thread $i$ and the Singleton Lock.

$$\gamma(r_{TS_i}(d_{I_{S_i}}),s_{TS_i}(d_{I_{S_i}}))\triangleq c_{TS_i}(d_{I_{S_i}})$$

$$\gamma(r_{TS_i}(d_{O_{S_i}}),s_{TS_i}(d_{O_{S_i}}))\triangleq c_{TS_i}(d_{O_{S_i}})$$

Let all modules be in parallel, then the Double-Checked Locking Optimization pattern $T_1\cdots T_n\quad S$ can be presented by the following process term.

$\tau_I(\partial_H(\Theta(T_1\between \cdots\between T_n\between S)))=\tau_I(\partial_H(T_1\between \cdots\between T_n\between S))$

where $H=\{r_{TS_i}(d_{I_{S_i}}),s_{TS_i}(d_{I_{S_i}}),r_{TS_i}(d_{O_{S_i}}),s_{TS_i}(d_{O_{S_i}})|d_{I_i}, d_{O_i}, d_{I_{S_i}}, d_{O_{S_i}}\in\Delta\}$,

$I=\{c_{TS_i}(d_{I_{S_i}}),c_{TS_i}(d_{O_{S_i}}),TF_{i1},TF_{i2},SF_i|d_{I_i}, d_{O_i}, d_{I_{S_i}}, d_{O_{S_i}}\in\Delta\}$ for $1\leq i\leq n$.

Then we get the following conclusion on the Double-Checked Locking Optimization pattern.

\begin{theorem}[Correctness of the Double-Checked Locking Optimization pattern]
The Double-Checked Locking Optimization pattern $\tau_I(\partial_H(T_1\between \cdots\between T_n\between S))$ can exhibit desired external behaviors.
\end{theorem}

\begin{proof}
Based on the above state transitions of the above modules, by use of the algebraic laws of APTC, we can prove that

$\tau_I(\partial_H(T_1\between \cdots\between T_n\between S))=\sum_{d_{I_1},d_{O_1},\cdots,d_{I_n},d_{O_n}\in\Delta}(r_{I_1}(d_{I_1})\parallel\cdots\parallel r_{I_n}(d_{I_n})\cdot s_{O_1}(d_{O_1})\parallel\cdots\parallel s_{O_n}(d_{O_n}))\cdot
\tau_I(\partial_H(T_1\between \cdots\between T_n\between S))$,

that is, the Double-Checked Locking Optimization pattern $\tau_I(\partial_H(T_1\between \cdots\between T_n\between S))$ can exhibit desired external behaviors.

For the details of proof, please refer to section \ref{app}, and we omit it.
\end{proof}

\subsection{Concurrency Patterns}\label{C6}

In this subsection, we verify concurrency related patterns, including the Active Object pattern, the Monitor Object pattern, the Half-Sync/Harf-Async pattern, the Leader/Followers pattern,
and the Thread-Specific Storage pattern.

\subsubsection{Verification of the Active Object Pattern}

The Active Object pattern is used to decouple the method request and method execution of an object. In this pattern, there are a Proxy module, a Scheduler module,
and $n$ Method Request modules and $n$ Servant modules. The Servant is used to implement concrete
computation, the Method Request is used to encapsulate a Servant, and the Scheduler is used to manage Method Requests. The Proxy
module interacts with outside through the channels $I$ and $O$, and with the Scheduler through the channels $I_{PS}$ and $O_{PS}$. The Scheduler interacts with
Method Request $i$ (for $1\leq i\leq n$) through the channels $I_{SM_i}$ and $O_{SM_i}$, and the Method Request $i$ interacts with the Servant $i$ through the channels $I_{MS_i}$ and $O_{CS_i}$,
as illustrated in Figure \ref{AO6}.

\begin{figure}
    \centering
    \includegraphics{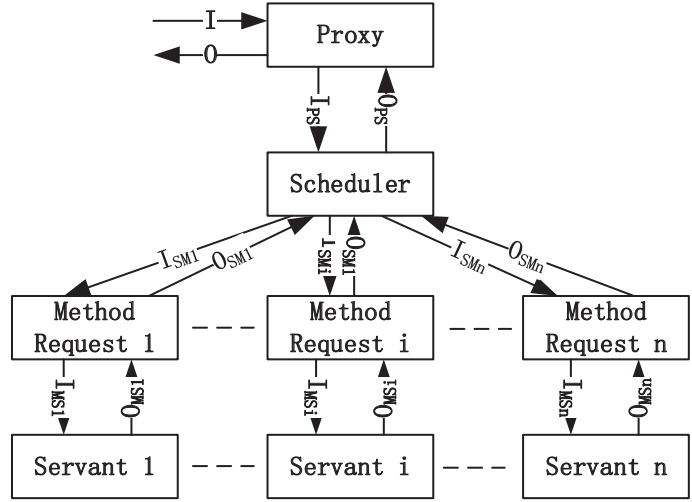}
    \caption{Active Object pattern}
    \label{AO6}
\end{figure}

The typical process of the Active Object pattern is shown in Figure \ref{AO6P} and as follows.

\begin{enumerate}
  \item The Proxy receives the request $d_I$ from outside through the channel $I$ (the corresponding reading action is denoted $r_I(d_I)$), then processes the request through a processing
  function $PF_1$ and generates the request $d_{I_{Sh}}$, and sends the $d_{I_{Sh}}$ to the Scheduler through the channel $I_{PS}$ (the corresponding sending action is denoted $s_{I_{PS}}(d_{I_{Sh}})$);
  \item The Scheduler receives the request $d_{I_{Sh}}$ from the Proxy through the channel $I_{PS}$ (the corresponding reading action is denoted $r_{I_{PS}}(d_{I_{Sh}})$), then processes the request through a processing
  function $ShF_1$ and generates the request $d_{I_{M_i}}$, and sends the $d_{I_{M_i}}$ to the Method Request $i$ through the channel $I_{SM_i}$ (the corresponding sending action is denoted $s_{I_{SM_i}}(d_{I_{M_i}})$);
  \item The Method Request $i$ receives the request $d_{I_{M_i}}$ from the Scheduler through the channel $I_{SM_i}$ (the corresponding reading action is denoted $r_{I_{SM_i}}(d_{I_{M_i}})$), then processes the
  request through a processing function $MF_{i1}$ and generates the request $d_{I_{S_i}}$, and sends the request to the Servant $i$ through the channel $I_{MS_i}$ (the corresponding sending
  action is denoted $s_{I_{MS_i}}(d_{I_{S_i}})$);
  \item The Servant $i$ receives the request $d_{I_{S_i}}$ from the Method Request $i$ through the channel $I_{MS_i}$ (the corresponding reading action is denoted $r_{I_{MS_i}}(d_{I_{S_i}})$), then processes the
  request through a processing function $SF_i$ and generates the response $d_{O_{S_i}}$, and sends the response to the Method Request through the channel $O_{MS_i}$ (the corresponding sending
  action is denoted $s_{O_{MS_i}}(d_{O_{S_i}})$);
  \item The Method Request $i$ receives the response $d_{O_{S_i}}$ from the Servant $i$ through the channel $O_{MS_i}$ (the corresponding reading action is denoted $r_{O_{MS_i}}(d_{O_{S_i}})$), then processes the
  request through a processing function $MF_{i2}$ and generates the response $d_{O_{M_i}}$, and sends the response to the Scheduler through the channel $O_{SM_i}$ (the corresponding sending
  action is denoted $s_{O_{SM_i}}(d_{O_{M_i}})$);
  \item The Scheduler receives the response $d_{O_{M_i}}$ from the Method Request $i$ through the channel $O_{SM_i}$ (the corresponding reading action is denoted $r_{O_{SM_i}}(d_{O_{M_i}})$),
  then processes the response and generate the response $d_{O_{Sh}}$ through a processing function $ShF_2$, and sends $d_{O_{Sh}}$ to the Proxy through the channel $O_{PS}$ (the corresponding
  sending action is denoted $s_{O_{PS}}(d_{O_{Sh}})$);
  \item The Proxy receives the response $d_{O_{Sh}}$ from the Scheduler through the channel $O_{PS}$ (the corresponding reading action is denoted $r_{O_{PS}}(d_{O_{Sh}})$), then processes the request through a processing
  function $PF_2$ and generates the request $d_{O}$, and sends the $d_{O}$ to the outside through the channel $O$ (the corresponding sending action is denoted $s_{O}(d_{O})$).
\end{enumerate}

\begin{figure}
    \centering
    \includegraphics{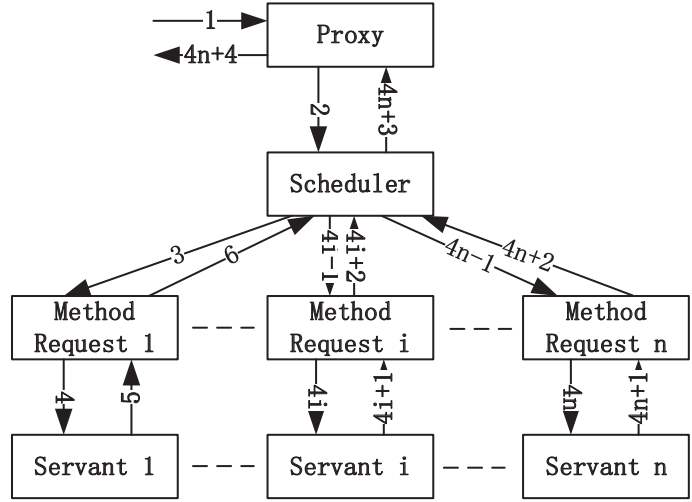}
    \caption{Typical process of Active Object pattern}
    \label{AO6P}
\end{figure}

In the following, we verify the Active Object pattern. We assume all data elements $d_{I}$, $d_{I_{Sh}}$, $d_{I_{M_i}}$, $d_{I_{S_i}}$, $d_{O_{S_{i}}}$, $d_{O_{M_i}}$, $d_{O_{Sh}}$, $d_{O}$ (for $1\leq i\leq n$) are from a finite set
$\Delta$.

The state transitions of the Proxy module
described by APTC are as follows.

$P=\sum_{d_{I}\in\Delta}(r_{I}(d_{I})\cdot P_{2})$

$P_{2}=PF_1\cdot P_{3}$

$P_{3}=\sum_{d_{I_{Sh}}\in\Delta}(s_{I_{PS}}(d_{I_{Sh}})\cdot P_{4})$

$P_{4}=\sum_{d_{O_{Sh}}\in\Delta}(r_{O_{PS}}(d_{O_{Sh}})\cdot P_{5})$

$P_5=PF_2\cdot P_6$

$P_{6}=\sum_{d_{O}\in\Delta}(s_{O}(d_{O})\cdot P)$

The state transitions of the Scheduler module
described by APTC are as follows.

$Sh=\sum_{d_{I_{Sh}}\in\Delta}(r_{I_{PS}}(d_{I_{Sh}})\cdot Sh_{2})$

$Sh_{2}=ShF_1\cdot Sh_{3}$

$Sh_{3}=\sum_{d_{I_{M_1}},\cdots,d_{I_{M_n}}\in\Delta}(s_{I_{SM_1}}(d_{I_{M_1}})\between\cdots\between s_{I_{SM_n}}(d_{I_{M_n}})\cdot Sh_{4})$

$Sh_{4}=\sum_{d_{O_{M_{1}}},\cdots,d_{O_{M_{n}}}\in\Delta}(r_{O_{SM_1}}(d_{O_{M_{1}}})\between\cdots\between r_{O_{SM_n}}(d_{O_{M_{n}}})\cdot Sh_{5})$

$Sh_5=ShF_2\cdot Sh_6$

$Sh_{6}=\sum_{d_{O_{Sh}}\in\Delta}(s_{O_{PS}}(d_{O_{Sh}})\cdot Sh)$

The state transitions of the Method Request $i$ described by APTC are as follows.

$M_i=\sum_{d_{I_{M_i}}\in\Delta}(r_{I_{SM_i}}(d_{I_{M_i}})\cdot M_{i_2})$

$M_{i_2}=MF_{i1}\cdot M_{i_3}$

$M_{i_3}=\sum_{d_{I_{S_{i}}}\in\Delta}(s_{I_{MS_{i}}}(d_{I_{S_{i}}})\cdot M_{i_4})$

$M_{i_4}=\sum_{d_{O_{S_i}}\in\Delta}(r_{O_{MS_i}}(d_{O_{S_i}})\cdot M_{i_5})$

$M_{i_5}=MF_{i2}\cdot M_{i_6}$

$M_{i_6}=\sum_{d_{O_{M_{i}}}\in\Delta}(s_{O_{SM_{i}}}(d_{O_{M_{i}}})\cdot M_i)$

The state transitions of the Servant $i$ described by APTC are as follows.

$S_i=\sum_{d_{I_{S_i}}\in\Delta}(r_{I_{CS_i}}(d_{I_{S_i}})\cdot S_{i_2})$

$S_{i_2}=SF_{i}\cdot S_{i_3}$

$S_{i_3}=\sum_{d_{O_{S_{i}}}\in\Delta}(s_{O_{CS_{i}}}(d_{O_{S_{i}}})\cdot S_i)$

The sending action and the reading action of the same data through the same channel can communicate with each other, otherwise, will cause a deadlock $\delta$. We define the following
communication functions of between the Proxy the Scheduler.

$$\gamma(r_{I_{PS}}(d_{I_{Sh}}),s_{I_{PS}}(d_{I_{Sh}}))\triangleq c_{I_{PS}}(d_{I_{Sh}})$$

$$\gamma(r_{O_{PS}}(d_{O_{Sh}}),s_{O_{PS}}(d_{O_{Sh}}))\triangleq c_{O_{PS}}(d_{O_{Sh}})$$

There are two communication function between the Scheduler  and the Method Request $i$ for $1\leq i\leq n$.

$$\gamma(r_{I_{SM_i}}(d_{I_{M_i}}),s_{I_{SM_i}}(d_{I_{M_i}}))\triangleq c_{I_{SM_i}}(d_{I_{M_i}})$$

$$\gamma(r_{O_{SM_i}}(d_{O_{M_{i}}}),s_{O_{SM_i}}(d_{O_{M_{i}}}))\triangleq c_{O_{SM_i}}(d_{O_{M_{i}}})$$

There are two communication function between the Servant $i$ and the Method Request $i$ for $1\leq i\leq n$.

$$\gamma(r_{I_{MS_i}}(d_{I_{S_i}}),s_{I_{MS_i}}(d_{I_{S_i}}))\triangleq c_{I_{MS_i}}(d_{I_{S_i}})$$

$$\gamma(r_{O_{MS_{i}}}(d_{O_{S_{i}}}),s_{O_{MS_{i}}}(d_{O_{S_{i}}}))\triangleq c_{O_{MS_{i}}}(d_{O_{S_{i}}})$$

Let all modules be in parallel, then the Active Object pattern

$P\quad Sh \quad M_1\cdots\quad M_i\quad\cdots M_n\quad S_1\cdots S_i\cdots S_n$

can be presented by the following process term.

$\tau_I(\partial_H(\Theta(P\between Sh\between M_1\between\cdots\between M_i\between\cdots\between M_n\between S_1\between\cdots\between S_i\between\cdots\between S_n)))=\tau_I(\partial_H(P\between Sh\between M_1\between\cdots\between M_i\between\cdots\between M_n\between S_1\between\cdots\between S_i\between\cdots\between S_n))$

where $H=\{r_{I_{PS}}(d_{I_{Sh}}),s_{I_{PS}}(d_{I_{Sh}}),r_{O_{PS}}(d_{O_{Sh}}),s_{O_{PS}}(d_{O_{Sh}}),r_{I_{SM_i}}(d_{I_{M_i}}),s_{I_{SM_i}}(d_{I_{M_i}}),\\
r_{O_{SM_i}}(d_{O_{M_{i}}}),s_{O_{SM_i}}(d_{O_{M_{i}}}),r_{I_{MS_i}}(d_{I_{S_i}}),s_{I_{MS_i}}(d_{I_{S_i}}),r_{O_{MS_{i}}}(d_{O_{S_{i}}}),s_{O_{MS_{i}}}(d_{O_{S_{i}}})\\
|d_{I}, d_{I_{Sh}}, d_{I_{M_i}}, d_{I_{S_i}}, d_{O_{S_{i}}}, d_{O_{M_i}}, d_{O_{Sh}}, d_{O}\in\Delta\}$ for $1\leq i\leq n$,

$I=\{c_{I_{PS}}(d_{I_{Sh}}),c_{O_{PS}}(d_{O_{Sh}}),c_{I_{SM_i}}(d_{I_{M_i}}),c_{O_{SM_i}}(d_{O_{M_{i}}}),c_{I_{MS_i}}(d_{I_{S_i}}),\\
c_{O_{MS_{i}}}(d_{O_{S_{i}}}),PF_1,PF_2,ShF_1,ShF_2,MF_{i1},MF_{i2},SF_{i}\\
|d_{I}, d_{I_{Sh}}, d_{I_{M_i}}, d_{I_{S_i}}, d_{O_{S_{i}}}, d_{O_{M_i}}, d_{O_{Sh}}, d_{O}\in\Delta\}$ for $1\leq i\leq n$.

Then we get the following conclusion on the Active Object pattern.

\begin{theorem}[Correctness of the Active Object pattern]
The Active Object pattern $\tau_I(\partial_H(P\between Sh\between M_1\between\cdots\between M_i\between\cdots\between M_n\between S_1\between\cdots\between S_i\between\cdots\between S_n))$ can exhibit desired external behaviors.
\end{theorem}

\begin{proof}
Based on the above state transitions of the above modules, by use of the algebraic laws of APTC, we can prove that

$\tau_I(\partial_H(P\between Sh\between M_1\between\cdots\between M_i\between\cdots\between M_n\between S_1\between\cdots\between S_i\between\cdots\between S_n))=\sum_{d_{I},d_O\in\Delta}(r_{I}(d_{I})\cdot s_{O}(d_{O}))\cdot
\tau_I(\partial_H(P\between Sh\between M_1\between\cdots\between M_i\between\cdots\between M_n\between S_1\between\cdots\between S_i\between\cdots\between S_n))$,

that is, the Active Object pattern $\tau_I(\partial_H(P\between Sh\between M_1\between\cdots\between M_i\between\cdots\between M_n\between S_1\between\cdots\between S_i\between\cdots\between S_n))$ can exhibit desired external behaviors.

For the details of proof, please refer to section \ref{app}, and we omit it.
\end{proof}

\subsubsection{Verification of the Monitor Object Pattern}

The Monitor Object pattern synchronizes concurrent method execution to ensure that only methods runs at a time.
In Monitor Object pattern, there are two classes of modules: The $n$ Client Threads and the Monitor Object. The Client Thread $i$ interacts with the
outside through the input channel $I_i$ and the output channel $O_i$; with the Monitor Object through the channel $CM_i$ for $1\leq i\leq n$, as illustrated in Figure \ref{MO6}.

\begin{figure}
    \centering
    \includegraphics{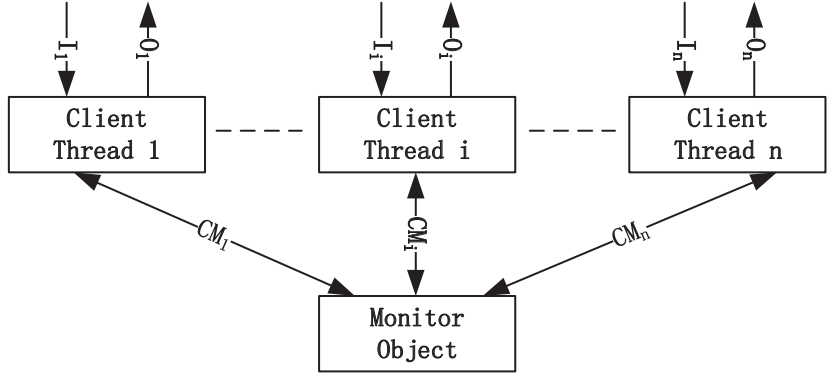}
    \caption{Monitor Object pattern}
    \label{MO6}
\end{figure}

The typical process is shown in Figure \ref{MO6P} and as follows.

\begin{enumerate}
  \item The Client Thread $i$ receives the input $d_{I_i}$ from the outside through the channel $I_i$ (the corresponding reading action is denoted $r_{I_i}(d_{I_i})$), then it processes the
  input and generates the input $d_{I_{M_i}}$ through a processing function $CF_{i1}$, and it sends the input to the Monitor Object through the channel $CM_i$
  (the corresponding sending action is denoted $s_{CM_i}(d_{I_{M_i}})$);
  \item The Monitor Object receives the input $d_{I_{M_i}}$ from the Client Thread $i$ through the channel $CM_i$ (the corresponding reading action is denoted $r_{CM_i}(d_{I_{M_i}})$) for $1\leq i\leq n$,
  then processes the request and generates the output $d_{O_{M_i}}$ through a processing function $MF_i$, and sends the output
  to the Client Thread $i$ through the channel $CM_i$ (the corresponding sending action is denoted $s_{CM_i}(d_{O_{M_i}})$);
  \item The Client Thread $i$ receives the output from the Monitor Object through the channel $CM_i$ (the corresponding reading action is denoted $r_{CM_i}(d_{O_{M_i}})$), then processes the output
  and generate the output $d_{O_i}$ through a processing function $CF_{i2}$ (accessing the resource), and sends the output to the outside through the channel $O_i$ (the corresponding sending action is denoted
  $s_{O_i}(d_{O_i})$).
\end{enumerate}

\begin{figure}
    \centering
    \includegraphics{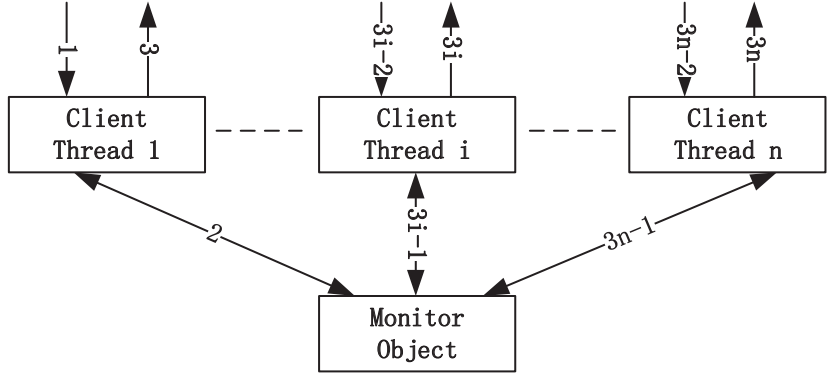}
    \caption{Typical process of Monitor Object pattern}
    \label{MO6P}
\end{figure}

In the following, we verify the Monitor Object pattern. We assume all data elements $d_{I_i}$, $d_{O_i}$, $d_{I_{M_i}}$, $d_{O_{M_i}}$ for $1\leq i\leq n$ are from a finite set
$\Delta$.

The state transitions of the Client Thread $i$ module
described by APTC are as follows.

$C_i=\sum_{d_{I_i}\in\Delta}(r_{I_i}(d_{I_i})\cdot C_{i_2})$

$C_{i_2}=CF_{i1}\cdot C_{i_3}$

$C_{i_3}=\sum_{d_{I_{M_i}}\in\Delta}(s_{CM_i}(d_{I_{M_i}})\cdot C_{i_4})$

$C_{i_4}=\sum_{d_{O_{M_i}}\in\Delta}(r_{CM_i}(d_{O_{M_i}})\cdot C_{i_5})$

$C_{i_5}=CF_{i2}\cdot C_{i_6}\quad (CF_{12}\%\cdots\%CF_{n2})$

$C_{i_6}=\sum_{d_{O_i}\in\Delta}(s_{O_i}(d_{O_i})\cdot C_{i})$

The state transitions of the Monitor Object module
described by APTC are as follows.

$M=\sum_{d_{I_{M_1}},\cdots,d_{I_{M_n}}\in\Delta}(r_{CM_1}(d_{I_{M_1}})\between\cdots\between r_{CM_n}(d_{I_{M_n}})\cdot M_{2})$

$M_{2}=MF_1\between\cdots\between MF_n\cdot M_{3}\quad (MF_1\%\cdots\% MF_n)$

$M_{3}=\sum_{d_{O_{M_1}},\cdots,d_{O_{M_n}}\in\Delta}(s_{CM_1}(d_{O_{M_1}})\between\cdots\between s_{CM_n}(d_{O_{M_n}})\cdot M)$

The sending action and the reading action of the same data through the same channel can communicate with each other, otherwise, will cause a deadlock $\delta$. We define the following
communication functions between the Client Thread $i$ and the Monitor Object.

$$\gamma(r_{CM_i}(d_{I_{M_i}}),s_{CM_i}(d_{I_{M_i}}))\triangleq c_{CM_i}(d_{I_{M_i}})$$

$$\gamma(r_{CM_i}(d_{O_{M_i}}),s_{CM_i}(d_{O_{M_i}}))\triangleq c_{CM_i}(d_{O_{M_i}})$$

Let all modules be in parallel, then the Monitor Object pattern $C_1\cdots C_n\quad M$ can be presented by the following process term.

$\tau_I(\partial_H(\Theta(C_1\between \cdots\between C_n\between M)))=\tau_I(\partial_H(C_1\between \cdots\between C_n\between M))$

where $H=\{r_{CM_i}(d_{I_{M_i}}),s_{CM_i}(d_{I_{M_i}}),r_{CM_i}(d_{O_{M_i}}),s_{CM_i}(d_{O_{M_i}})|d_{I_i}, d_{O_i}, d_{I_{M_i}}, d_{O_{M_i}}\in\Delta\}$,

$I=\{c_{CM_i}(d_{I_{M_i}}),c_{CM_i}(d_{O_{M_i}}),CF_{i1},CF_{i2},MF_i|d_{I_i}, d_{O_i}, d_{I_{M_i}}, d_{O_{M_i}}\in\Delta\}$ for $1\leq i\leq n$.

Then we get the following conclusion on the Monitor Object pattern.

\begin{theorem}[Correctness of the Monitor Object pattern]
The Monitor Object pattern $\tau_I(\partial_H(C_1\between \cdots\between C_n\between M))$ can exhibit desired external behaviors.
\end{theorem}

\begin{proof}
Based on the above state transitions of the above modules, by use of the algebraic laws of APTC, we can prove that

$\tau_I(\partial_H(C_1\between \cdots\between C_n\between M))=\sum_{d_{I_1},d_{O_1},\cdots,d_{I_n},d_{O_n}\in\Delta}(r_{I_1}(d_{I_1})\parallel\cdots\parallel r_{I_n}(d_{I_n})\cdot s_{O_1}(d_{O_1})\parallel\cdots\parallel s_{O_n}(d_{O_n}))\cdot
\tau_I(\partial_H(C_1\between \cdots\between C_n\between M))$,

that is, the Monitor Object pattern $\tau_I(\partial_H(C_1\between \cdots\between C_n\between M))$ can exhibit desired external behaviors.

For the details of proof, please refer to section \ref{app}, and we omit it.
\end{proof}

\subsubsection{Verification of the Half-Sync/Half-Async Pattern}

The Half-Sync/Half-Async pattern decouples the asynchronous and synchronous processings, which has two classes of components: $n$ Synchronous Services and the Asynchronous Service.
The Asynchronous Service receives the inputs asynchronously from the user through the channel $I$, then the Asynchronous Service sends the results to the Synchronous Service $i$
through the channel $AS_i$ synchronously for $1\leq i\leq n$;
When the Synchronous Service $i$ receives the input from the Asynchronous Service, it generates and sends the results out to the user through the channel $O_i$.
As illustrates in Figure \ref{HSHA6}.

\begin{figure}
    \centering
    \includegraphics{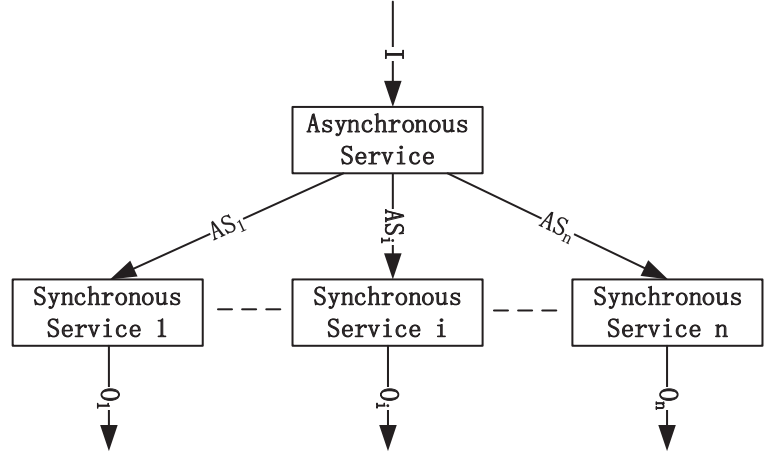}
    \caption{Half-Sync/Harf-Async pattern}
    \label{HSHA6}
\end{figure}

The typical process of the Half-Sync/Half-Async pattern is shown in Figure \ref{HSHA6P} and following.

\begin{enumerate}
  \item The Asynchronous Service receives the input $d_I$ from the user through the channel $I$ (the corresponding reading action is denoted $r_I(D_I)$), processes the input through
  a processing function $AF$, and generate the input to the Synchronous Service $i$ (for $1\leq i\leq n$) which is denoted $d_{I_{S_i}}$; then sends the input to the Synchronous Service $i$ through the
  channel $AS_i$ (the corresponding sending action is denoted $s_{AS_i}(d_{I_{S_i}})$);
  \item The Synchronous Service $i$ (for $1\leq i\leq n$) receives the input from the Asynchronous Service
  through the channel $AS_i$ (the corresponding reading action is denoted $r_{AS_i}(d_{I_{S_i}})$), processes the results through a processing function $SF_{i}$, generates the output
  $d_{O_i}$, then sending the output through the channel $O_i$ (the corresponding sending action is denoted $s_{O_i}(d_{O_i})$).
\end{enumerate}

\begin{figure}
    \centering
    \includegraphics{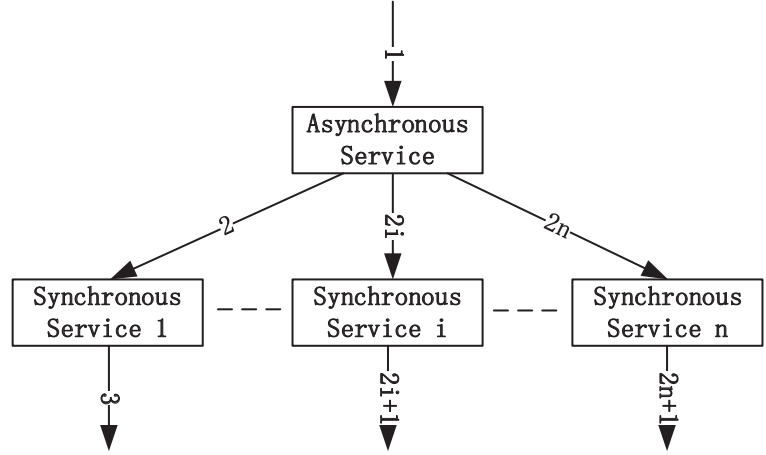}
    \caption{Typical process of Half-Sync/Harf-Async pattern}
    \label{HSHA6P}
\end{figure}

In the following, we verify the Half-Sync/Half-Async pattern. We assume all data elements $d_{I}$, $d_{I_{S_i}}$, $d_{O_{i}}$ (for $1\leq i\leq n$) are from a finite set
$\Delta$.

The state transitions of the Asynchronous Service module
described by APTC are as follows.

$A=\sum_{d_{I}\in\Delta}(r_{I}(d_{I})\cdot A_{2})$

$A_{2}=AF\cdot A_{3}$

$A_{3}=\sum_{d_{I_{S_1}},\cdots,d_{I_{S_n}}\in\Delta}(s_{AS_1}(d_{I_{S_1}})\between\cdots\between s_{AS_n}(d_{I_{S_n}})\cdot A)$

The state transitions of the Synchronous Service $i$ described by APTC are as follows.

$S_i=\sum_{d_{I_{S_i}}\in\Delta}(r_{AS_i}(d_{I_{S_i}})\cdot S_{i_2})$

$S_{i_2}=SF_i\cdot S_{i_3}$

$S_{i_3}=\sum_{d_{O_{i}}\in\Delta}(s_{O_{i}}(d_{O_i})\cdot S_i)$

The sending action and the reading action of the same data through the same channel can communicate with each other, otherwise, will cause a deadlock $\delta$. We define the following
communication functions of the Synchronous Service $i$ for $1\leq i\leq n$.

$$\gamma(r_{AS_i}(d_{I_{S_i}}),s_{AS_i}(d_{I_{S_i}}))\triangleq c_{AS_i}(d_{I_{S_i}})$$

Let all modules be in parallel, then the Half-Sync/Half-Async pattern $A \quad S_1\cdots S_i\cdots S_n$ can be presented by the following process term.

$\tau_I(\partial_H(\Theta(A\between S_1\between\cdots\between S_i\between\cdots\between S_n)))=\tau_I(\partial_H(A\between S_1\between\cdots\between S_i\between\cdots\between S_n))$

where $H=\{r_{AS_i}(d_{O_{S_i}}),s_{AS_i}(d_{O_{S_i}})|d_{I}, d_{I_{S_i}}, d_{O_{i}}\in\Delta\}$ for $1\leq i\leq n$,

$I=\{c_{AS_i}(d_{I_{S_i}}),AF,SF_{i}|d_{I}, d_{I_{S_i}}, d_{O_{i}}\in\Delta\}$ for $1\leq i\leq n$.

Then we get the following conclusion on the Half-Sync/Half-Async pattern.

\begin{theorem}[Correctness of the Half-Sync/Half-Async pattern]
The Half-Sync/Half-Async pattern $\tau_I(\partial_H(A\between S_1\between\cdots\between S_i\between\cdots\between S_n))$ can exhibit desired external behaviors.
\end{theorem}

\begin{proof}
Based on the above state transitions of the above modules, by use of the algebraic laws of APTC, we can prove that

$\tau_I(\partial_H(A\between S_1\between\cdots\between S_i\between\cdots\between S_n))=\sum_{d_{I},d_{O_1},\cdots,d_{O_n}\in\Delta}(r_{I}(d_{I})\cdot s_{O_1}(d_{O_1})\parallel\cdots\parallel s_{O_i}(d_{O_i})\parallel\cdots\parallel s_{O_n}(d_{O_n}))\cdot
\tau_I(\partial_H(A\between S_1\between\cdots\between S_i\between\cdots\between S_n))$,

that is, the Half-Sync/Half-Async pattern $\tau_I(\partial_H(A\between S_1\between\cdots\between S_i\between\cdots\between S_n))$ can exhibit desired external behaviors.

For the details of proof, please refer to section \ref{app}, and we omit it.
\end{proof}

\subsubsection{Verification of the Leader/Followers Pattern}

The Leader/Followers pattern decouples the event delivery between the event source and event handler. There are four modules in the Leader/Followers pattern: the Handle Set, the Leader,
the Follower and the Event Handler. The Handle Set interacts with the outside through
the channel $I$; with the Leader through the channel $HL$. The Leader interacts with the Follower through the channel $LF$. The Event Handler interacts with the
Follower through the channel $FE$, and with the outside through the channels $O$. As illustrates in
Figure \ref{LF6}.

\begin{figure}
    \centering
    \includegraphics{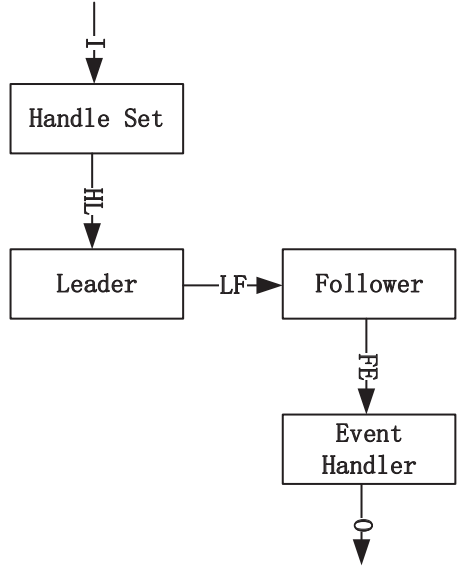}
    \caption{Leader/Followers pattern}
    \label{LF6}
\end{figure}

The typical process of the Leader/Followers pattern is shown in Figure \ref{LF6P} and as follows.

\begin{enumerate}
  \item The Handle Set receives the input $d_{I}$ from the outside through the channel $I$ (the corresponding reading action is denoted $r_{I}(d_{I})$), then processes the input $d_{I}$ through a processing
  function $HF$, and sends the processed input $d_{I_{HL}}$ to the Leader through the channel $HL$ (the corresponding sending action is denoted $s_{PP}(d_{I_{HL}})$);
  \item The Leader receives $d_{I_{HL}}$ from the Handle Set through the channel $HL$ (the corresponding reading action is denoted $r_{HL}(d_{I_{HL}})$), then processes the request
  through a processing function $LF$, generates and sends the processed input $d_{I_{LF}}$ to the Follower through the channel $LF$ (the corresponding sending action is denoted
  $s_{LF}(d_{I_{LF}})$);
  \item The Follower receives the input $d_{I_{LF}}$ from the Leader through the channel $LF$ (the corresponding reading action is denoted $r_{LF}(d_{I_{LF}})$), then processes
  the request through a processing function $FF$, generates and sends the processed input $d_{I_{FE}}$ to the Event Handler through the channel $FE$ (the corresponding sending action is denoted $s_{FE}(d_{I_{FE}})$);
  \item The Event Handler receives the input $d_{I_{FE}}$ from the Follower through the channel $FE$ (the corresponding reading action is denoted $r_{FE}(d_{I_{FE}})$), then
  processes the request and generates the response $d_{O}$ through a processing function $EF$, and sends the response to the outside through the channel $O$ (the corresponding sending action is denoted
  $s_{O}(d_{O})$).
\end{enumerate}

\begin{figure}
    \centering
    \includegraphics{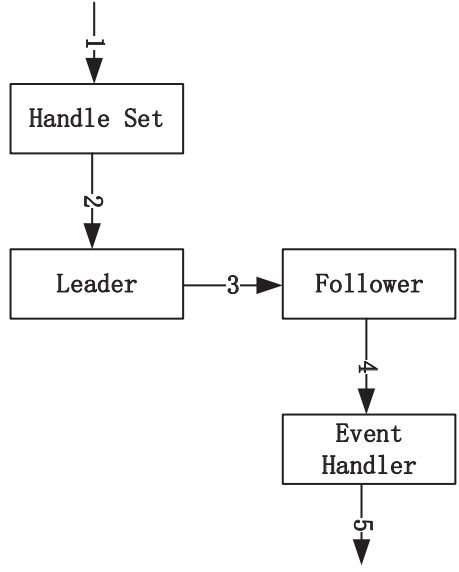}
    \caption{Typical process of Leader/Followers pattern}
    \label{LF6P}
\end{figure}

In the following, we verify the Leader/Followers pattern. We assume all data elements $d_{I}$, $d_{I_{HL}}$, $d_{I_{LF}}$, $d_{I_{FE}}$, $d_{O}$ are from a finite set
$\Delta$.

The state transitions of the Handle Set module
described by APTC are as follows.

$H=\sum_{d_{I}\in\Delta}(r_{I}(d_{I})\cdot H_{2})$

$H_{2}=HF\cdot H_{3}$

$H_{3}=\sum_{d_{I_{HL}}\in\Delta}(s_{HL}(d_{I_{HL}})\cdot H)$

The state transitions of the Leader module
described by APTC are as follows.

$L=\sum_{d_{I_{HL}}\in\Delta}(r_{HL}(d_{I_{HL}})\cdot L_{2})$

$L_{2}=LF\cdot L_{3}$

$L_{3}=\sum_{d_{I_{LF}}\in\Delta}(s_{LF}(d_{I_{LF}})\cdot L)$

The state transitions of the Follower module
described by APTC are as follows.

$F=\sum_{d_{I_{LF}}\in\Delta}(r_{LF}(d_{I_{LF}})\cdot F_{2})$

$F_{2}=FF\cdot F_{3}$

$F_{3}=\sum_{d_{I_{FE}}\in\Delta}(s_{FE}(d_{I_{FE}})\cdot F)$

The state transitions of the Event Handler module
described by APTC are as follows.

$E=\sum_{d_{I_{FE}}\in\Delta}(r_{FE}(d_{I_{FE}})\cdot E_{2})$

$E_{2}=EF\cdot E_{3}$

$E_{3}=\sum_{d_{O}\in\Delta}(s_{O}(d_{O})\cdot E)$

The sending action and the reading action of the same data through the same channel can communicate with each other, otherwise, will cause a deadlock $\delta$. We define the following
communication functions between the Handle Set and the Leader.

$$\gamma(r_{HL}(d_{I_{HL}}),s_{HL}(d_{I_{HL}}))\triangleq c_{HL}(d_{I_{HL}})$$

There are one communication functions between the Leader and the Follower as follows.

$$\gamma(r_{LF}(d_{I_{LF}}),s_{LF}(d_{I_{LF}}))\triangleq c_{LF}(d_{I_{LF}})$$

There are one communication functions between the Follower and the Event Handler as follows.

$$\gamma(r_{FE}(d_{I_{FE}}),s_{FE}(d_{I_{FE}}))\triangleq c_{FE}(d_{I_{FE}})$$

Let all modules be in parallel, then the Leader/Followers pattern $H\quad L \quad F\quad E$ can be presented by the following process term.

$\tau_I(\partial_H(\Theta(H\between L\between F\between E)))=\tau_I(\partial_H(H\between L\between F\between E))$

where $H=\{r_{HL}(d_{I_{HL}}),s_{HL}(d_{I_{HL}}),r_{LF}(d_{I_{LF}}),s_{LF}(d_{I_{LF}}),r_{FE}(d_{I_{FE}}),s_{FE}(d_{I_{FE}})\\
|d_{I}, d_{I_{HL}}, d_{I_{LF}}, d_{I_{FE}}, d_{O}\in\Delta\}$,

$I=\{c_{HL}(d_{I_{HL}}),c_{LF}(d_{I_{LF}}),c_{FE}(d_{I_{FE}}),HF,LF,FF,EF
|d_{I}, d_{I_{HL}}, d_{I_{LF}}, d_{I_{FE}}, d_{O}\in\Delta\}$.

Then we get the following conclusion on the Leader/Followers pattern.

\begin{theorem}[Correctness of the Leader/Followers pattern]
The Leader/Followers pattern $\tau_I(\partial_H(H\between L\between F\between E))$ can exhibit desired external behaviors.
\end{theorem}

\begin{proof}
Based on the above state transitions of the above modules, by use of the algebraic laws of APTC, we can prove that

$\tau_I(\partial_H(H\between L\between F\between E))=\sum_{d_{I},d_{O}\in\Delta}(r_{I}(d_{I})\cdot s_{O}(d_{O}))\cdot
\tau_I(\partial_H(H\between L\between F\between E))$,

that is, the Leader/Followers pattern $\tau_I(\partial_H(H\between L\between F\between E))$ can exhibit desired external behaviors.

For the details of proof, please refer to section \ref{app}, and we omit it.
\end{proof}

\subsubsection{Verification of the Thread-Specific Storage Pattern}

The Thread-Specific Storage pattern allows the application threads to get a global access point to a local object.
There are four modules in the Thread-Specific Storage pattern: the Application Thread, the Thread Specific Object Proxy, the Key Factory,
and the Thread Specific Object. The Application Thread interacts with the outside through
the channels $I$ and $O$; with the Thread Specific Object Proxy through the channel $I_{AP}$ and $O_{AP}$.
The Thread Specific Object Proxy interacts with the Thread Specific Object through the channels $I_{PO}$ and $O_{PO}$; with the Key Factory through the channels $I_{PF}$ and $O_{PF}$.
As illustrates in Figure \ref{TSS6}.

\begin{figure}
    \centering
    \includegraphics{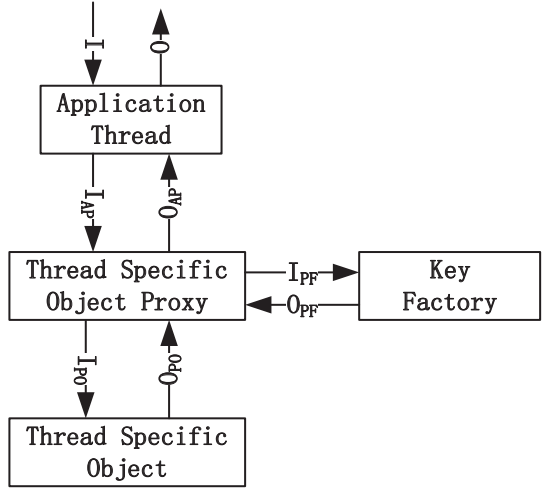}
    \caption{Thread-Specific Storage pattern}
    \label{TSS6}
\end{figure}

The typical process of the Thread-Specific Storage pattern is shown in Figure \ref{TSS6P} and as follows.

\begin{enumerate}
  \item The Application Thread receives the input $d_{I}$ from the outside through the channel $I$ (the corresponding reading action is denoted $r_{I}(d_{I})$), then processes the input $d_{I}$ through a processing
  function $AF_1$, and sends the processed input $d_{I_{P}}$ to the Thread Specific Object Proxy through the channel $I_{AP}$ (the corresponding sending action is denoted $s_{I_{AP}}(d_{I_{P}})$);
  \item The Thread Specific Object Proxy receives $d_{I_{P}}$ from the Application Thread through the channel $I_{AP}$ (the corresponding reading action is denoted $r_{I_{AP}}(d_{I_{P}})$), then processes the request
  through a processing function $PF_1$, generates and sends the processed input $d_{I_{F}}$ to the Key Factory through the channel $I_{PF}$ (the corresponding sending action is denoted
  $s_{I_{PF}}(d_{I_F})$);
  \item The Key Factory receives the input $d_{I_F}$ from the Thread Specific Object Proxy through the channel $I_{PF}$ (the corresponding reading action is denoted $r_{I_{PF}}(d_{I_F})$),
  then processes the request and generates the result $d_{O_F}$ through a processing function $FF$, and sends the result to the Thread Specific Object Proxy through the channel $O_{PF}$
  (the corresponding sending action is denoted $s_{O_{PF}}(d_{O_F})$);
  \item The Thread Specific Object Proxy receives the result from the Key Factory through the channel $O_{PF}$ (the corresponding reading action is denoted $r_{O_{PF}}(d_{O_F})$), then
  processes the results and generate the request $d_{I_O}$ to the Thread Specific Object through a processing function $PF_2$, and sends the request to the Thread Specific Object
   through the channel $I_{PO}$ (the corresponding sending action is denoted $s_{I_{PO}}(d_{I_{O}})$);
  \item The Thread Specific Object receives the input $d_{I_{O}}$ from the Thread Specific Object Proxy through the channel $I_{PO}$ (the corresponding reading action is denoted $r_{I_{PO}}(d_{I_{O}})$), then processes
  the input through a processing function $OF$, generates and sends the response $d_{O_{O}}$ to the Thread Specific Object Proxy through the channel $O_{PO}$ (the corresponding sending action is denoted $s_{O_{PO}}(d_{O_{O}})$);
  \item The Thread Specific Object Proxy receives the response $d_{O_O}$ from the Thread Specific Object through the channel $O_{PO}$ (the corresponding reading action is denoted $r_{O_{PO}}(d_{O_O})$), then processes
  the response through a processing function $PF_3$, generates and sends the response $d_{O_P}$ (the corresponding sending action is denoted $s_{O_{AP}}(d_{O_P})$);
  \item The Application Thread receives the response $d_{O_{P}}$ from the Thread Specific Object Proxy through the channel $O_{AP}$ (the corresponding reading action is denoted $r_{O_{AP}}(d_{O_{P}})$), then
  processes the request and generates the response $d_{O}$ through a processing function $AF_2$, and sends the response to the outside through the channel $O$ (the corresponding sending action is denoted
  $s_{O}(d_{O})$).
\end{enumerate}

\begin{figure}
    \centering
    \includegraphics{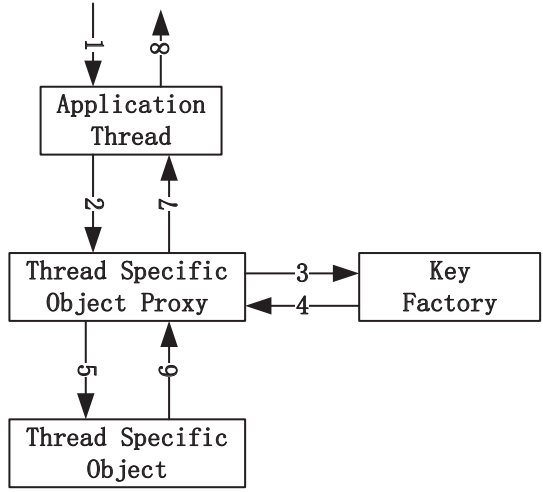}
    \caption{Typical process of Thread-Specific Storage pattern}
    \label{TSS6P}
\end{figure}

In the following, we verify the Thread-Specific Storage pattern. We assume all data elements $d_{I}$, $d_{I_{P}}$, $d_{I_{F}}$, $d_{I_O}$, $d_{O_{P}}$, $d_{O_F}$, $d_{O_O}$, $d_{O}$ are from a finite set
$\Delta$.

The state transitions of the Application Thread module
described by APTC are as follows.

$A=\sum_{d_{I}\in\Delta}(r_{I}(d_{I})\cdot A_{2})$

$A_{2}=AF_1\cdot A_{3}$

$A_{3}=\sum_{d_{I_{P}}\in\Delta}(s_{I_{AP}}(d_{I_{P}})\cdot A_4)$

$A_4=\sum_{d_{O_P}\in\Delta}(r_{O_{AP}}(d_{O_P})\cdot A_{5})$

$A_{5}=AF_2\cdot A_{6}$

$A_{6}=\sum_{d_{O}\in\Delta}(s_{O}(d_{O})\cdot A)$

The state transitions of the Thread Specific Object Proxy module
described by APTC are as follows.

$P=\sum_{d_{I_{P}}\in\Delta}(r_{I_{AP}}(d_{I_{P}})\cdot P_{2})$

$P_{2}=PF_1\cdot P_{3}$

$P_{3}=\sum_{d_{I_{F}}\in\Delta}(s_{I_{PF}}(d_{I_{F}})\cdot P_4)$

$P_4=\sum_{d_{O_{F}}\in\Delta}(r_{O_{PF}}(d_{O_{F}})\cdot P_{5})$

$P_{5}=PF_2\cdot P_{6}$

$P_{6}=\sum_{d_{I_{O}}\in\Delta}(s_{I_{PO}}(d_{I_{O}})\cdot P_7)$

$P_7=\sum_{d_{O_{O}}\in\Delta}(r_{O_{PO}}(d_{O_{O}})\cdot P_{8})$

$P_{8}=PF_3\cdot P_{9}$

$P_{9}=\sum_{d_{O_{P}}\in\Delta}(s_{O_{AP}}(d_{O_{P}})\cdot P)$

The state transitions of the Key Factory module
described by APTC are as follows.

$F=\sum_{d_{I_{F}}\in\Delta}(r_{I_{PF}}(d_{I_F})\cdot F_{2})$

$F_{2}=FF\cdot F_{3}$

$F_{3}=\sum_{d_{O_{F}}\in\Delta}(s_{O_{PF}}(d_{O_{F}})\cdot F)$

The state transitions of the Thread Specific Object module
described by APTC are as follows.

$O=\sum_{d_{I_{O}}\in\Delta}(r_{I_{PO}}(d_{I_O})\cdot O_{2})$

$O_{2}=OF\cdot O_{3}$

$O_{3}=\sum_{d_{O_{O}}\in\Delta}(s_{O_{PO}}(d_{O_{O}})\cdot O)$

The sending action and the reading action of the same data through the same channel can communicate with each other, otherwise, will cause a deadlock $\delta$. We define the following
communication functions between the Application Thread and the Thread Specific Object Proxy Proxy.

$$\gamma(r_{I_{AP}}(d_{I_{P}}),s_{I_{AP}}(d_{I_{P}}))\triangleq c_{I_{AP}}(d_{I_{P}})$$

$$\gamma(r_{O_{AP}}(d_{O_P}),s_{O_{AP}}(d_{O_P}))\triangleq c_{O_{AP}}(d_{O_P})$$

There are two communication functions between the Thread Specific Object Proxy and the Key Factory as follows.

$$\gamma(r_{I_{PF}}(d_{I_{F}}),s_{I_{PF}}(d_{I_{F}}))\triangleq c_{I_{PF}}(d_{I_{F}})$$

$$\gamma(r_{O_{PF}}(d_{O_{F}}),s_{O_{PF}}(d_{O_{F}}))\triangleq c_{O_{PF}}(d_{O_{F}})$$

There are two communication functions between the Thread Specific Object Proxy and the Thread Specific Object as follows.

$$\gamma(r_{I_{PO}}(d_{I_{O}}),s_{I_{PO}}(d_{I_{O}}))\triangleq c_{I_{PO}}(d_{I_{O}})$$

$$\gamma(r_{O_{PO}}(d_{O_{O}}),s_{O_{PO}}(d_{O_{O}}))\triangleq c_{O_{PO}}(d_{O_{O}})$$

Let all modules be in parallel, then the Thread-Specific Storage pattern $A\quad P \quad F\quad O$ can be presented by the following process term.

$\tau_I(\partial_H(\Theta(A\between P\between F\between O)))=\tau_I(\partial_H(A\between P\between F\between O))$

where $H=\{r_{I_{AP}}(d_{I_{P}}),s_{I_{AP}}(d_{I_{P}}),r_{O_{AP}}(d_{O_P}),s_{O_{AP}}(d_{O_P}),r_{I_{PF}}(d_{I_{F}}),s_{I_{PF}}(d_{I_{F}}),\\
r_{O_{PF}}(d_{O_{F}}),s_{O_{PF}}(d_{O_{F}}),r_{I_{PO}}(d_{I_{O}}),s_{I_{PO}}(d_{I_{O}}),r_{O_{PO}}(d_{O_O}),s_{O_{PO}}(d_{O_O})\\
|d_{I}, d_{I_{P}}, d_{I_{F}}, d_{I_O}, d_{O_{P}}, d_{O_F}, d_{O_O}, d_{O}\in\Delta\}$,

$I=\{c_{I_{AP}}(d_{I_{P}}),c_{O_{AP}}(d_{O_P}),c_{I_{PF}}(d_{I_{F}}),c_{O_{PF}}(d_{O_{F}}),c_{I_{PO}}(d_{I_{O}}),c_{O_{PO}}(d_{O_{O}}),\\
AF_1,AF_2,PF_1,PF_2,PF_3,FF,OF
|d_{I}, d_{I_{P}}, d_{I_{F}}, d_{I_O}, d_{O_{P}}, d_{O_F}, d_{O_O}, d_{O}\in\Delta\}$.

Then we get the following conclusion on the Thread-Specific Storage pattern.

\begin{theorem}[Correctness of the Thread-Specific Storage pattern]
The Thread-Specific Storage pattern $\tau_I(\partial_H(A\between P\between F\between O))$ can exhibit desired external behaviors.
\end{theorem}

\begin{proof}
Based on the above state transitions of the above modules, by use of the algebraic laws of APTC, we can prove that

$\tau_I(\partial_H(A\between P\between F\between O))=\sum_{d_{I},d_{O}\in\Delta}(r_{I}(d_{I})\cdot s_{O}(d_{O}))\cdot
\tau_I(\partial_H(A\between P\between F\between O))$,

that is, the Thread-Specific Storage pattern $\tau_I(\partial_H(A\between P\between F\between O))$ can exhibit desired external behaviors.

For the details of proof, please refer to section \ref{app}, and we omit it.
\end{proof}

\newpage\section{Verification of Patterns for Resource Management}\label{RMP}

Patterns for resource management are patterns related to resource management, and can be used in higher-level and lower-level systems and applications.

In this chapter, we verify patterns for resource management. In section \ref{RA7}, we verify patterns related to resource acquisition. We verify patterns for resource lifecycle in section
\ref{RL7} and patterns for resource release in section \ref{RR7}.

\subsection{Resource Acquisition}\label{RA7}

In this subsection, we verify patterns for resource acquisition, including the Lookup pattern, the Lazy Acquisition pattern, the Eager Acquisition pattern, and the Partial Acquisition
pattern.

\subsubsection{Verification of the Lookup Pattern}

The Lookup pattern uses a mediating lookup service to find and access resources.
There are four modules in the Lookup pattern: the Resource User, the Resource Provider, the Lookup Service,
and the Resource. The Resource User interacts with the outside through
the channels $I$ and $O$; with the Resource Provider through the channel $I_{UP}$ and $O_{UP}$; with the Resource through the channels $I_{UR}$ and $O_{UR}$;
with the Lookup Service through the channels $I_{US}$ and $O_{US}$.
As illustrates in Figure \ref{Lo7}.

\begin{figure}
    \centering
    \includegraphics{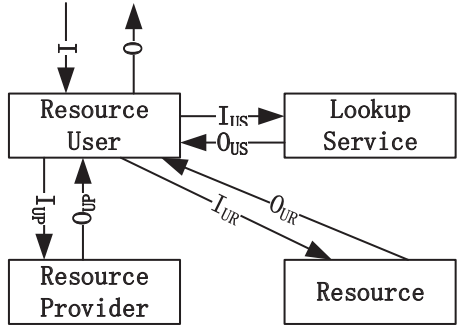}
    \caption{Lookup pattern}
    \label{Lo7}
\end{figure}

The typical process of the Lookup pattern is shown in Figure \ref{Lo7P} and as follows.

\begin{enumerate}
  \item The Resource User receives the input $d_{I}$ from the outside through the channel $I$ (the corresponding reading action is denoted $r_{I}(d_{I})$), then processes the input $d_{I}$ through a processing
  function $UF_1$, and sends the processed input $d_{I_{S}}$ to the Lookup Service through the channel $I_{US}$ (the corresponding sending action is denoted $s_{I_{US}}(d_{I_{S}})$);
  \item The Lookup Service receives $d_{I_{S}}$ from the Resource User through the channel $I_{US}$ (the corresponding reading action is denoted $r_{I_{US}}(d_{I_{S}})$), then processes the request
  through a processing function $SF$, generates and sends the processed output $d_{O_{S}}$ to the Resource User through the channel $O_{US}$ (the corresponding sending action is denoted
  $s_{O_{US}}(d_{O_S})$);
  \item The Resource User receives the output $d_{O_S}$ from the Lookup Service through the channel $O_{US}$ (the corresponding reading action is denoted $r_{O_{US}}(d_{O_S})$),
  then processes the output and generates the input $d_{I_P}$ through a processing function $UF_2$, and sends the input to the Resource Provider through the channel $I_{UP}$
  (the corresponding sending action is denoted $s_{I_{UP}}(d_{I_P})$);
  \item The Resource Provider receives the input from the Resource User through the channel $I_{UP}$ (the corresponding reading action is denoted $r_{I_{UP}}(d_{I_P})$), then
  processes the input and generate the output $d_{O_P}$ to the Resource User through a processing function $PF$, and sends the output to the Resource User
   through the channel $O_{UP}$ (the corresponding sending action is denoted $s_{O_{UP}}(d_{O_{P}})$);
  \item The Resource User receives the output $d_{O_{P}}$ from the Resource Provider through the channel $O_{UP}$ (the corresponding reading action is denoted $r_{O_{UP}}(d_{O_{P}})$), then processes
  the input through a processing function $UF_3$, generates and sends the input $d_{I_{R}}$ to the Resource through the channel $I_{UR}$ (the corresponding sending action is denoted $s_{I_{UR}}(d_{I_{R}})$);
  \item The Resource receives the input $d_{I_R}$ from the Resource User through the channel $I_{UR}$ (the corresponding reading action is denoted $r_{I_{UR}}(d_{I_R})$), then processes
  the input through a processing function $RF$, generates and sends the response $d_{O_R}$ (the corresponding sending action is denoted $s_{O_{UR}}(d_{O_R})$);
  \item The Resource User receives the response $d_{O_{R}}$ from the Resource through the channel $O_{UR}$ (the corresponding reading action is denoted $r_{O_{UR}}(d_{O_{R}})$), then
  processes the response and generates the response $d_{O}$ through a processing function $UF_4$, and sends the response to the outside through the channel $O$ (the corresponding sending action is denoted
  $s_{O}(d_{O})$).
\end{enumerate}

\begin{figure}
    \centering
    \includegraphics{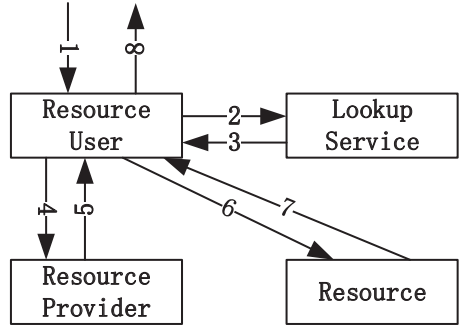}
    \caption{Typical process of Lookup pattern}
    \label{Lo7P}
\end{figure}

In the following, we verify the Lookup pattern. We assume all data elements $d_{I}$, $d_{I_{S}}$, $d_{I_{P}}$, $d_{I_R}$, $d_{O_{S}}$, $d_{O_P}$, $d_{O_R}$, $d_{O}$ are from a finite set
$\Delta$.

The state transitions of the Resource User module
described by APTC are as follows.

$U=\sum_{d_{I}\in\Delta}(r_{I}(d_{I})\cdot U_{2})$

$U_{2}=UF_1\cdot U_{3}$

$U_{3}=\sum_{d_{I_{S}}\in\Delta}(s_{I_{US}}(d_{I_{S}})\cdot U_4)$

$U_4=\sum_{d_{O_S}\in\Delta}(r_{O_{US}}(d_{O_S})\cdot U_{5})$

$U_{5}=UF_2\cdot U_{6}$

$U_{6}=\sum_{d_{I_{P}}\in\Delta}(s_{I_{UP}}(d_{I_{P}})\cdot U_7)$

$U_7=\sum_{d_{O_P}\in\Delta}(r_{O_{UP}}(d_{O_P})\cdot U_{8})$

$U_{8}=UF_3\cdot U_{9}$

$U_{9}=\sum_{d_{I_{R}}\in\Delta}(s_{I_{UR}}(d_{I_{R}})\cdot U_{10})$

$U_{10}=\sum_{d_{O_R}\in\Delta}(r_{O_{UR}}(d_{O_R})\cdot U_{11})$

$U_{11}=UF_4\cdot U_{12}$

$U_{12}=\sum_{d_{O}\in\Delta}(s_{O}(d_{O})\cdot U)$

The state transitions of the Resource Provider module
described by APTC are as follows.

$P=\sum_{d_{I_{P}}\in\Delta}(r_{I_{UP}}(d_{I_P})\cdot P_{2})$

$P_{2}=PF\cdot P_{3}$

$P_{3}=\sum_{d_{O_{P}}\in\Delta}(s_{O_{UP}}(d_{O_{P}})\cdot P)$

The state transitions of the Lookup Service module
described by APTC are as follows.

$S=\sum_{d_{I_{S}}\in\Delta}(r_{I_{US}}(d_{I_S})\cdot S_{2})$

$S_{2}=SF\cdot S_{3}$

$S_{3}=\sum_{d_{O_{S}}\in\Delta}(s_{O_{US}}(d_{O_{S}})\cdot S)$

The state transitions of the Resource module
described by APTC are as follows.

$R=\sum_{d_{I_{R}}\in\Delta}(r_{I_{UR}}(d_{I_R})\cdot R_{2})$

$R_{2}=RF\cdot R_{3}$

$R_{3}=\sum_{d_{O_{R}}\in\Delta}(s_{O_{UR}}(d_{O_{R}})\cdot R)$

The sending action and the reading action of the same data through the same channel can communicate with each other, otherwise, will cause a deadlock $\delta$. We define the following
communication functions between the Resource User and the Resource Provider Proxy.

$$\gamma(r_{I_{UP}}(d_{I_{P}}),s_{I_{UP}}(d_{I_{P}}))\triangleq c_{I_{UP}}(d_{I_{P}})$$

$$\gamma(r_{O_{UP}}(d_{O_P}),s_{O_{UP}}(d_{O_P}))\triangleq c_{O_{UP}}(d_{O_P})$$

There are two communication functions between the Resource User and the Lookup Service as follows.

$$\gamma(r_{I_{US}}(d_{I_{S}}),s_{I_{US}}(d_{I_{S}}))\triangleq c_{I_{US}}(d_{I_{S}})$$

$$\gamma(r_{O_{US}}(d_{O_{S}}),s_{O_{US}}(d_{O_{S}}))\triangleq c_{O_{US}}(d_{O_{S}})$$

There are two communication functions between the Resource User and the Resource as follows.

$$\gamma(r_{I_{UR}}(d_{I_{R}}),s_{I_{UR}}(d_{I_{R}}))\triangleq c_{I_{UR}}(d_{I_{R}})$$

$$\gamma(r_{O_{UR}}(d_{O_{R}}),s_{O_{UR}}(d_{O_{R}}))\triangleq c_{O_{UR}}(d_{O_{R}})$$

Let all modules be in parallel, then the Lookup pattern $U\quad S \quad P\quad R$ can be presented by the following process term.

$\tau_I(\partial_H(\Theta(U\between S\between P\between R)))=\tau_I(\partial_H(U\between S\between P\between R))$

where $H=\{r_{I_{UP}}(d_{I_{P}}),s_{I_{UP}}(d_{I_{P}}),r_{O_{UP}}(d_{O_P}),s_{O_{UP}}(d_{O_P}),r_{I_{US}}(d_{I_{S}}),s_{I_{US}}(d_{I_{S}}),\\
r_{O_{US}}(d_{O_{S}}),s_{O_{US}}(d_{O_{S}}),r_{I_{UR}}(d_{I_{R}}),s_{I_{UR}}(d_{I_{R}}),r_{O_{UR}}(d_{O_R}),s_{O_{UR}}(d_{O_R})\\
|d_{I}, d_{I_{P}}, d_{I_{S}}, d_{I_R}, d_{O_{P}}, d_{O_S}, d_{O_R}, d_{O}\in\Delta\}$,

$I=\{c_{I_{UP}}(d_{I_{P}}),c_{O_{UP}}(d_{O_P}),c_{I_{US}}(d_{I_{S}}),c_{O_{US}}(d_{O_{S}}),c_{I_{UR}}(d_{I_{R}}),c_{O_{UR}}(d_{O_{R}}),\\
UF_1,UF_2,UF_3,UF_4,PF,SF,RF
|d_{I}, d_{I_{P}}, d_{I_{S}}, d_{I_R}, d_{O_{P}}, d_{O_S}, d_{O_R}, d_{O}\in\Delta\}$.

Then we get the following conclusion on the Lookup pattern.

\begin{theorem}[Correctness of the Lookup pattern]
The Lookup pattern $\tau_I(\partial_H(U\between S\between P\between R))$ can exhibit desired external behaviors.
\end{theorem}

\begin{proof}
Based on the above state transitions of the above modules, by use of the algebraic laws of APTC, we can prove that

$\tau_I(\partial_H(U\between S\between P\between R))=\sum_{d_{I},d_{O}\in\Delta}(r_{I}(d_{I})\cdot s_{O}(d_{O}))\cdot
\tau_I(\partial_H(U\between S\between P\between R))$,

that is, the Lookup pattern $\tau_I(\partial_H(U\between S\between P\between R))$ can exhibit desired external behaviors.

For the details of proof, please refer to section \ref{app}, and we omit it.
\end{proof}

\subsubsection{Verification of the Lazy Acquisition Pattern}

The Lazy Acquisition pattern defers the acquisitions of resources to the latest time.
There are four modules in the Lazy Acquisition pattern: the Resource User, the Resource Provider, the Resource Proxy,
and the Resource. The Resource User interacts with the outside through
the channels $I$ and $O$; with the Resource Proxy through the channel $I_{UP}$ and $O_{UP}$. The Resource Proxy interacts with the Resource through the channels $I_{PR}$ and $O_{PR}$;
with the Resource Provider through the channels $I_{PP}$ and $O_{PP}$.
As illustrates in Figure \ref{LA7}.

\begin{figure}
    \centering
    \includegraphics{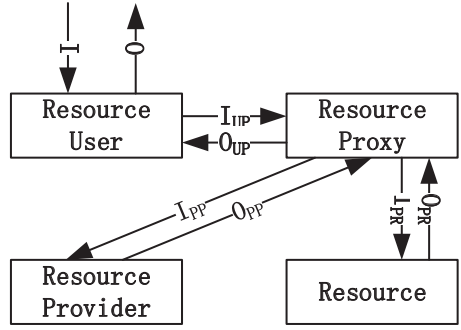}
    \caption{Lazy Acquisition pattern}
    \label{LA7}
\end{figure}

The typical process of the Lazy Acquisition pattern is shown in Figure \ref{LA7P} and as follows.

\begin{enumerate}
  \item The Resource User receives the input $d_{I}$ from the outside through the channel $I$ (the corresponding reading action is denoted $r_{I}(d_{I})$), then processes the input $d_{I}$ through a processing
  function $UF_1$, and sends the processed input $d_{I_{P}}$ to the Resource Proxy through the channel $I_{UP}$ (the corresponding sending action is denoted $s_{I_{UP}}(d_{I_{P}})$);
  \item The Resource Proxy receives $d_{I_{P}}$ from the Resource User through the channel $I_{UP}$ (the corresponding reading action is denoted $r_{I_{UP}}(d_{I_{P}})$), then processes the request
  through a processing function $PF_1$, generates and sends the processed input $d_{I_{RP}}$ to the Resource Provider through the channel $I_{PP}$ (the corresponding sending action is denoted
  $s_{I_{PP}}(d_{I_{RP}})$);
  \item The Resource Provider receives the input $d_{I_{RP}}$ from the Resource Proxy through the channel $I_{PP}$ (the corresponding reading action is denoted $r_{I_{PP}}(d_{I_{RP}})$),
  then processes the input and generates the output $d_{O_{RP}}$ through a processing function $RPF$, and sends the output to the Resource Proxy through the channel $O_{PP}$
  (the corresponding sending action is denoted $s_{O_{PP}}(d_{O_{RP}})$);
  \item The Resource Proxy receives the output from the Resource Provider through the channel $O_{PP}$ (the corresponding reading action is denoted $r_{O_{PP}}(d_{O_{RP}})$), then
  processes the results and generate the input $d_{I_R}$ to the Resource through a processing function $PF_2$, and sends the input to the Resource
  through the channel $I_{PR}$ (the corresponding sending action is denoted $s_{I_{PR}}(d_{I_{R}})$);
  \item The Resource receives the input $d_{I_{R}}$ from the Resource Proxy through the channel $I_{PR}$ (the corresponding reading action is denoted $r_{I_{PR}}(d_{I_{R}})$), then processes
  the input through a processing function $RF$, generates and sends the output $d_{O_{R}}$ to the Resource Proxy through the channel $O_{PR}$ (the corresponding sending action is denoted $s_{O_{PR}}(d_{O_{R}})$);
  \item The Resource Proxy receives the output $d_{O_R}$ from the Resource through the channel $O_{PR}$ (the corresponding reading action is denoted $r_{O_{PR}}(d_{O_R})$), then processes
  the response through a processing function $PF_3$, generates and sends the response $d_{O_P}$ (the corresponding sending action is denoted $s_{O_{UP}}(d_{O_P})$);
  \item The Resource User receives the response $d_{O_{P}}$ from the Resource Proxy through the channel $O_{UP}$ (the corresponding reading action is denoted $r_{O_{UP}}(d_{O_{P}})$), then
  processes the response and generates the response $d_{O}$ through a processing function $UF_2$, and sends the response to the outside through the channel $O$ (the corresponding sending action is denoted
  $s_{O}(d_{O})$).
\end{enumerate}

\begin{figure}
    \centering
    \includegraphics{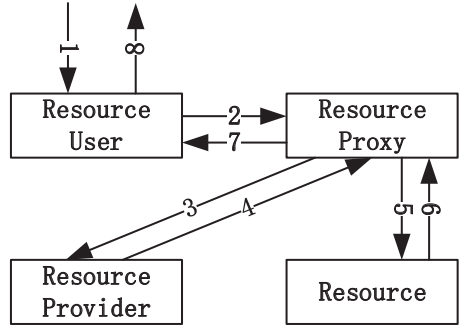}
    \caption{Typical process of Lazy Acquisition pattern}
    \label{LA7P}
\end{figure}

In the following, we verify the Lazy Acquisition pattern. We assume all data elements $d_{I}$, $d_{I_{S}}$, $d_{I_{P}}$, $d_{I_R}$, $d_{O_{S}}$, $d_{O_P}$, $d_{O_R}$, $d_{O}$ are from a finite set
$\Delta$.

The state transitions of the Resource User module
described by APTC are as follows.

$U=\sum_{d_{I}\in\Delta}(r_{I}(d_{I})\cdot U_{2})$

$U_{2}=UF_1\cdot U_{3}$

$U_{3}=\sum_{d_{I_{P}}\in\Delta}(s_{I_{UP}}(d_{I_{P}})\cdot U_4)$

$U_4=\sum_{d_{O_P}\in\Delta}(r_{O_{UP}}(d_{O_P})\cdot U_{5})$

$U_{5}=UF_2\cdot U_{6}$

$U_{6}=\sum_{d_{O}\in\Delta}(s_{O}(d_{O})\cdot U)$

The state transitions of the Resource Proxy module
described by APTC are as follows.

$P=\sum_{d_{I_{P}}\in\Delta}(r_{I_{UP}}(d_{I_P})\cdot P_{2})$

$P_{2}=PF_1\cdot P_{3}$

$P_{3}=\sum_{d_{I_{RP}}\in\Delta}(s_{I_{PP}}(d_{I_{RP}})\cdot P_4)$

$P_4=\sum_{d_{O_{RP}}\in\Delta}(r_{O_{PP}}(d_{O_{RP}})\cdot P_{5})$

$P_{5}=PF_2\cdot P_{6}$

$P_{6}=\sum_{d_{I_{R}}\in\Delta}(s_{I_{PR}}(d_{I_{R}})\cdot P_7)$

$P_7=\sum_{d_{O_{R}}\in\Delta}(r_{O_{PR}}(d_{O_{R}})\cdot P_{8})$

$P_{8}=PF_3\cdot P_{9}$

$P_{9}=\sum_{d_{O_{P}}\in\Delta}(s_{O_{UP}}(d_{O_{P}})\cdot P)$

The state transitions of the Resource Provider module
described by APTC are as follows.

$RP=\sum_{d_{I_{RP}}\in\Delta}(r_{I_{PP}}(d_{I_{RP}})\cdot RP_{2})$

$RP_{2}=RPF\cdot RP_{3}$

$RP_{3}=\sum_{d_{O_{RP}}\in\Delta}(s_{O_{PP}}(d_{O_{RP}})\cdot RP)$

The state transitions of the Resource module
described by APTC are as follows.

$R=\sum_{d_{I_{R}}\in\Delta}(r_{I_{PR}}(d_{I_R})\cdot R_{2})$

$R_{2}=RF\cdot R_{3}$

$R_{3}=\sum_{d_{O_{R}}\in\Delta}(s_{O_{PR}}(d_{O_{R}})\cdot R)$

The sending action and the reading action of the same data through the same channel can communicate with each other, otherwise, will cause a deadlock $\delta$. We define the following
communication functions between the Resource User and the Resource Proxy.

$$\gamma(r_{I_{UP}}(d_{I_{P}}),s_{I_{UP}}(d_{I_{P}}))\triangleq c_{I_{UP}}(d_{I_{P}})$$

$$\gamma(r_{O_{UP}}(d_{O_P}),s_{O_{UP}}(d_{O_P}))\triangleq c_{O_{UP}}(d_{O_P})$$

There are two communication functions between the Resource Provider and the Resource Proxy as follows.

$$\gamma(r_{I_{PP}}(d_{I_{RP}}),s_{I_{PP}}(d_{I_{RP}}))\triangleq c_{I_{PP}}(d_{I_{RP}})$$

$$\gamma(r_{O_{PP}}(d_{O_{RP}}),s_{O_{PP}}(d_{O_{RP}}))\triangleq c_{O_{PP}}(d_{O_{RP}})$$

There are two communication functions between the Resource Proxy and the Resource as follows.

$$\gamma(r_{I_{PR}}(d_{I_{R}}),s_{I_{PR}}(d_{I_{R}}))\triangleq c_{I_{PR}}(d_{I_{R}})$$

$$\gamma(r_{O_{PR}}(d_{O_{R}}),s_{O_{PR}}(d_{O_{R}}))\triangleq c_{O_{PR}}(d_{O_{R}})$$

Let all modules be in parallel, then the Lazy Acquisition pattern $U\quad P \quad RP\quad R$ can be presented by the following process term.

$\tau_I(\partial_H(\Theta(U\between P\between RP\between R)))=\tau_I(\partial_H(U\between P\between RP\between R))$

where $H=\{r_{I_{UP}}(d_{I_{P}}),s_{I_{UP}}(d_{I_{P}}),r_{O_{UP}}(d_{O_P}),s_{O_{UP}}(d_{O_P}),r_{I_{PP}}(d_{I_{RP}}),s_{I_{PP}}(d_{I_{RP}}),\\
r_{O_{PP}}(d_{O_{RP}}),s_{O_{PP}}(d_{O_{RP}}),r_{I_{PR}}(d_{I_{R}}),s_{I_{PR}}(d_{I_{R}}),r_{O_{PR}}(d_{O_R}),s_{O_{PR}}(d_{O_R})\\
|d_{I}, d_{I_{P}}, d_{I_{RP}}, d_{I_R}, d_{O_{P}}, d_{O_{RP}}, d_{O_R}, d_{O}\in\Delta\}$,

$I=\{c_{I_{UP}}(d_{I_{P}}),c_{O_{UP}}(d_{O_P}),c_{I_{PP}}(d_{I_{RP}}),c_{O_{PP}}(d_{O_{RP}}),c_{I_{PR}}(d_{I_{R}}),c_{O_{PR}}(d_{O_{R}}),\\
UF_1,UF_2,PF_1,PF_2,PF_3,RPF,RF
|d_{I}, d_{I_{P}}, d_{I_{RP}}, d_{I_R}, d_{O_{P}}, d_{O_{RP}}, d_{O_R}, d_{O}\in\Delta\}$.

Then we get the following conclusion on the Lazy Acquisition pattern.

\begin{theorem}[Correctness of the Lazy Acquisition pattern]
The Lazy Acquisition pattern $\tau_I(\partial_H(U\between P\between RP\between R))$ can exhibit desired external behaviors.
\end{theorem}

\begin{proof}
Based on the above state transitions of the above modules, by use of the algebraic laws of APTC, we can prove that

$\tau_I(\partial_H(U\between P\between RP\between R))=\sum_{d_{I},d_{O}\in\Delta}(r_{I}(d_{I})\cdot s_{O}(d_{O}))\cdot
\tau_I(\partial_H(U\between P\between RP\between R))$,

that is, the Lazy Acquisition pattern $\tau_I(\partial_H(U\between P\between RP\between R))$ can exhibit desired external behaviors.

For the details of proof, please refer to section \ref{app}, and we omit it.
\end{proof}

\subsubsection{Verification of the Eager Acquisition Pattern}

The Eager Acquisition pattern acquires the resources eagerly.
There are four modules in the Eager Acquisition pattern: the Resource User, the Resource Provider, the Resource Proxy,
and the Resource. The Resource User interacts with the outside through
the channels $I$ and $O$; with the Resource Proxy through the channel $I_{UP}$ and $O_{UP}$. The Resource Proxy interacts with the Resource through the channels $I_{PR}$ and $O_{PR}$;
with the Resource Provider through the channels $I_{PP}$ and $O_{PP}$.
As illustrates in Figure \ref{EA7}.

\begin{figure}
    \centering
    \includegraphics{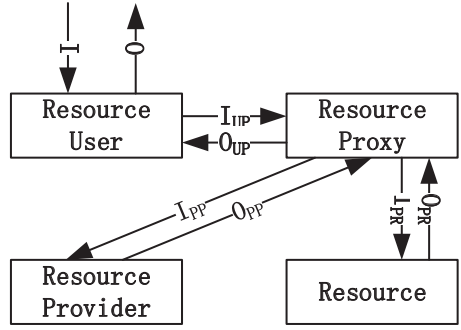}
    \caption{Eager Acquisition pattern}
    \label{EA7}
\end{figure}

The typical process of the Eager Acquisition pattern is shown in Figure \ref{EA7P} and as follows.

\begin{enumerate}
  \item The Resource User receives the input $d_{I}$ from the outside through the channel $I$ (the corresponding reading action is denoted $r_{I}(d_{I})$), then processes the input $d_{I}$ through a processing
  function $UF_1$, and sends the processed input $d_{I_{P}}$ to the Resource Proxy through the channel $I_{UP}$ (the corresponding sending action is denoted $s_{I_{UP}}(d_{I_{P}})$);
  \item The Resource Proxy receives $d_{I_{P}}$ from the Resource User through the channel $I_{UP}$ (the corresponding reading action is denoted $r_{I_{UP}}(d_{I_{P}})$), then processes the request
  through a processing function $PF_1$, generates and sends the processed input $d_{I_{RP}}$ to the Resource Provider through the channel $I_{PP}$ (the corresponding sending action is denoted
  $s_{I_{PP}}(d_{I_{RP}})$);
  \item The Resource Provider receives the input $d_{I_{RP}}$ from the Resource Proxy through the channel $I_{PP}$ (the corresponding reading action is denoted $r_{I_{PP}}(d_{I_{RP}})$),
  then processes the input and generates the output $d_{O_{RP}}$ through a processing function $RPF$, and sends the output to the Resource Proxy through the channel $O_{PP}$
  (the corresponding sending action is denoted $s_{O_{PP}}(d_{O_{RP}})$);
  \item The Resource Proxy receives the output from the Resource Provider through the channel $O_{PP}$ (the corresponding reading action is denoted $r_{O_{PP}}(d_{O_{RP}})$), then
  processes the results and generate the input $d_{I_R}$ to the Resource through a processing function $PF_2$, and sends the input to the Resource
  through the channel $I_{PR}$ (the corresponding sending action is denoted $s_{I_{PR}}(d_{I_{R}})$);
  \item The Resource receives the input $d_{I_{R}}$ from the Resource Proxy through the channel $I_{PR}$ (the corresponding reading action is denoted $r_{I_{PR}}(d_{I_{R}})$), then processes
  the input through a processing function $RF$, generates and sends the output $d_{O_{R}}$ to the Resource Proxy through the channel $O_{PR}$ (the corresponding sending action is denoted $s_{O_{PR}}(d_{O_{R}})$);
  \item The Resource Proxy receives the output $d_{O_R}$ from the Resource through the channel $O_{PR}$ (the corresponding reading action is denoted $r_{O_{PR}}(d_{O_R})$), then processes
  the response through a processing function $PF_3$, generates and sends the response $d_{O_P}$ (the corresponding sending action is denoted $s_{O_{UP}}(d_{O_P})$);
  \item The Resource User receives the response $d_{O_{P}}$ from the Resource Proxy through the channel $O_{UP}$ (the corresponding reading action is denoted $r_{O_{UP}}(d_{O_{P}})$), then
  processes the response and generates the response $d_{O}$ through a processing function $UF_2$, and sends the response to the outside through the channel $O$ (the corresponding sending action is denoted
  $s_{O}(d_{O})$).
\end{enumerate}

\begin{figure}
    \centering
    \includegraphics{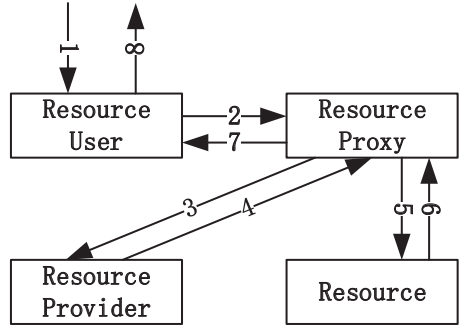}
    \caption{Typical process of Eager Acquisition pattern}
    \label{EA7P}
\end{figure}

In the following, we verify the Eager Acquisition pattern. We assume all data elements $d_{I}$, $d_{I_{S}}$, $d_{I_{P}}$, $d_{I_R}$, $d_{O_{S}}$, $d_{O_P}$, $d_{O_R}$, $d_{O}$ are from a finite set
$\Delta$.

The state transitions of the Resource User module
described by APTC are as follows.

$U=\sum_{d_{I}\in\Delta}(r_{I}(d_{I})\cdot U_{2})$

$U_{2}=UF_1\cdot U_{3}$

$U_{3}=\sum_{d_{I_{P}}\in\Delta}(s_{I_{UP}}(d_{I_{P}})\cdot U_4)$

$U_4=\sum_{d_{O_P}\in\Delta}(r_{O_{UP}}(d_{O_P})\cdot U_{5})$

$U_{5}=UF_2\cdot U_{6}$

$U_{6}=\sum_{d_{O}\in\Delta}(s_{O}(d_{O})\cdot U)$

The state transitions of the Resource Proxy module
described by APTC are as follows.

$P=\sum_{d_{I_{P}}\in\Delta}(r_{I_{UP}}(d_{I_P})\cdot P_{2})$

$P_{2}=PF_1\cdot P_{3}$

$P_{3}=\sum_{d_{I_{RP}}\in\Delta}(s_{I_{PP}}(d_{I_{RP}})\cdot P_4)$

$P_4=\sum_{d_{O_{RP}}\in\Delta}(r_{O_{PP}}(d_{O_{RP}})\cdot P_{5})$

$P_{5}=PF_2\cdot P_{6}$

$P_{6}=\sum_{d_{I_{R}}\in\Delta}(s_{I_{PR}}(d_{I_{R}})\cdot P_7)$

$P_7=\sum_{d_{O_{R}}\in\Delta}(r_{O_{PR}}(d_{O_{R}})\cdot P_{8})$

$P_{8}=PF_3\cdot P_{9}$

$P_{9}=\sum_{d_{O_{P}}\in\Delta}(s_{O_{UP}}(d_{O_{P}})\cdot P)$

The state transitions of the Resource Provider module
described by APTC are as follows.

$RP=\sum_{d_{I_{RP}}\in\Delta}(r_{I_{PP}}(d_{I_{RP}})\cdot RP_{2})$

$RP_{2}=RPF\cdot RP_{3}$

$RP_{3}=\sum_{d_{O_{RP}}\in\Delta}(s_{O_{PP}}(d_{O_{RP}})\cdot RP)$

The state transitions of the Resource module
described by APTC are as follows.

$R=\sum_{d_{I_{R}}\in\Delta}(r_{I_{PR}}(d_{I_R})\cdot R_{2})$

$R_{2}=RF\cdot R_{3}$

$R_{3}=\sum_{d_{O_{R}}\in\Delta}(s_{O_{PR}}(d_{O_{R}})\cdot R)$

The sending action and the reading action of the same data through the same channel can communicate with each other, otherwise, will cause a deadlock $\delta$. We define the following
communication functions between the Resource User and the Resource Proxy.

$$\gamma(r_{I_{UP}}(d_{I_{P}}),s_{I_{UP}}(d_{I_{P}}))\triangleq c_{I_{UP}}(d_{I_{P}})$$

$$\gamma(r_{O_{UP}}(d_{O_P}),s_{O_{UP}}(d_{O_P}))\triangleq c_{O_{UP}}(d_{O_P})$$

There are two communication functions between the Resource Provider and the Resource Proxy as follows.

$$\gamma(r_{I_{PP}}(d_{I_{RP}}),s_{I_{PP}}(d_{I_{RP}}))\triangleq c_{I_{PP}}(d_{I_{RP}})$$

$$\gamma(r_{O_{PP}}(d_{O_{RP}}),s_{O_{PP}}(d_{O_{RP}}))\triangleq c_{O_{PP}}(d_{O_{RP}})$$

There are two communication functions between the Resource Proxy and the Resource as follows.

$$\gamma(r_{I_{PR}}(d_{I_{R}}),s_{I_{PR}}(d_{I_{R}}))\triangleq c_{I_{PR}}(d_{I_{R}})$$

$$\gamma(r_{O_{PR}}(d_{O_{R}}),s_{O_{PR}}(d_{O_{R}}))\triangleq c_{O_{PR}}(d_{O_{R}})$$

Let all modules be in parallel, then the Eager Acquisition pattern $U\quad P \quad RP\quad R$ can be presented by the following process term.

$\tau_I(\partial_H(\Theta(U\between P\between RP\between R)))=\tau_I(\partial_H(U\between P\between RP\between R))$

where $H=\{r_{I_{UP}}(d_{I_{P}}),s_{I_{UP}}(d_{I_{P}}),r_{O_{UP}}(d_{O_P}),s_{O_{UP}}(d_{O_P}),r_{I_{PP}}(d_{I_{RP}}),s_{I_{PP}}(d_{I_{RP}}),\\
r_{O_{PP}}(d_{O_{RP}}),s_{O_{PP}}(d_{O_{RP}}),r_{I_{PR}}(d_{I_{R}}),s_{I_{PR}}(d_{I_{R}}),r_{O_{PR}}(d_{O_R}),s_{O_{PR}}(d_{O_R})\\
|d_{I}, d_{I_{P}}, d_{I_{RP}}, d_{I_R}, d_{O_{P}}, d_{O_{RP}}, d_{O_R}, d_{O}\in\Delta\}$,

$I=\{c_{I_{UP}}(d_{I_{P}}),c_{O_{UP}}(d_{O_P}),c_{I_{PP}}(d_{I_{RP}}),c_{O_{PP}}(d_{O_{RP}}),c_{I_{PR}}(d_{I_{R}}),c_{O_{PR}}(d_{O_{R}}),\\
UF_1,UF_2,PF_1,PF_2,PF_3,RPF,RF
|d_{I}, d_{I_{P}}, d_{I_{RP}}, d_{I_R}, d_{O_{P}}, d_{O_{RP}}, d_{O_R}, d_{O}\in\Delta\}$.

Then we get the following conclusion on the Eager Acquisition pattern.

\begin{theorem}[Correctness of the Eager Acquisition pattern]
The Eager Acquisition pattern $\tau_I(\partial_H(U\between P\between RP\between R))$ can exhibit desired external behaviors.
\end{theorem}

\begin{proof}
Based on the above state transitions of the above modules, by use of the algebraic laws of APTC, we can prove that

$\tau_I(\partial_H(U\between P\between RP\between R))=\sum_{d_{I},d_{O}\in\Delta}(r_{I}(d_{I})\cdot s_{O}(d_{O}))\cdot
\tau_I(\partial_H(U\between P\between RP\between R))$,

that is, the Eager Acquisition pattern $\tau_I(\partial_H(U\between P\between RP\between R))$ can exhibit desired external behaviors.

For the details of proof, please refer to section \ref{app}, and we omit it.
\end{proof}

\subsubsection{Verification of the Partial Acquisition Pattern}

The Partial Acquisition pattern partially acquires the resources into multi stages to optimize resource management.
There are three modules in the Partial Acquisition pattern: the Resource User, the Resource Provider,
and the Resource. The Resource User interacts with the outside through
the channels $I$ and $O$; with the Resource Provider through the channels $I_{UP}$ and $O_{UP}$; with the Resource through the channels $I_{UR}$ and $O_{UR}$.
As illustrates in Figure \ref{PA7}.

\begin{figure}
    \centering
    \includegraphics{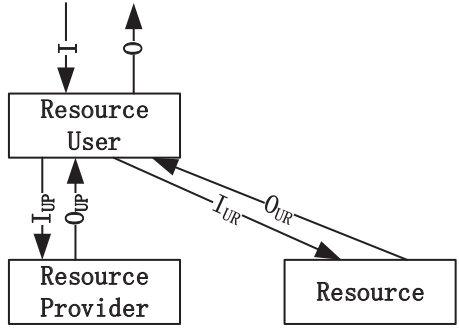}
    \caption{Partial Acquisition pattern}
    \label{PA7}
\end{figure}

The typical process of the Partial Acquisition pattern is shown in Figure \ref{PA7P} and as follows.

\begin{enumerate}
  \item The Resource User receives the input $d_{I}$ from the outside through the channel $I$ (the corresponding reading action is denoted $r_{I}(d_{I})$), then processes the input $d_{I}$ through a processing
  function $UF_1$, and generates the input $d_{I_P}$, and sends the input to the Resource Provider through the channel $I_{UP}$
  (the corresponding sending action is denoted $s_{I_{UP}}(d_{I_P})$);
  \item The Resource Provider receives the input from the Resource User through the channel $I_{UP}$ (the corresponding reading action is denoted $r_{I_{UP}}(d_{I_P})$), then
  processes the input and generate the output $d_{O_P}$ to the Resource User through a processing function $PF$, and sends the output to the Resource User
   through the channel $O_{UP}$ (the corresponding sending action is denoted $s_{O_{UP}}(d_{O_{P}})$);
  \item The Resource User receives the output $d_{O_{P}}$ from the Resource Provider through the channel $O_{UP}$ (the corresponding reading action is denoted $r_{O_{UP}}(d_{O_{P}})$), then processes
  the output through a processing function $UF_2$, generates and sends the input $d_{I_{R}}$ to the Resource through the channel $I_{UR}$ (the corresponding sending action is denoted $s_{I_{UR}}(d_{I_{R}})$);
  \item The Resource receives the input $d_{I_R}$ from the Resource User through the channel $I_{UR}$ (the corresponding reading action is denoted $r_{I_{UR}}(d_{I_R})$), then processes
  the input through a processing function $RF$, generates and sends the response $d_{O_R}$ (the corresponding sending action is denoted $s_{O_{UR}}(d_{O_R})$);
  \item The Resource User receives the response $d_{O_{R}}$ from the Resource through the channel $O_{UR}$ (the corresponding reading action is denoted $r_{O_{UR}}(d_{O_{R}})$), then
  processes the response and generates the response $d_{O}$ through a processing function $UF_3$, and sends the response to the outside through the channel $O$ (the corresponding sending action is denoted
  $s_{O}(d_{O})$).
\end{enumerate}

\begin{figure}
    \centering
    \includegraphics{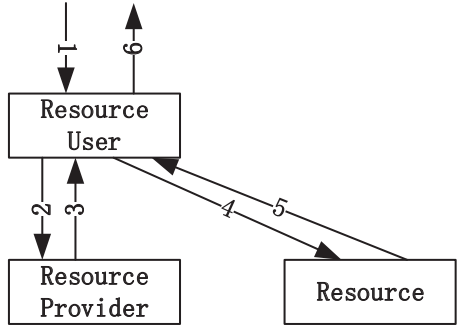}
    \caption{Typical process of Partial Acquisition pattern}
    \label{PA7P}
\end{figure}

In the following, we verify the Partial Acquisition pattern. We assume all data elements $d_{I}$, $d_{I_{P}}$, $d_{I_R}$, $d_{O_P}$, $d_{O_R}$, $d_{O}$ are from a finite set
$\Delta$.

The state transitions of the Resource User module
described by APTC are as follows.

$U=\sum_{d_{I}\in\Delta}(r_{I}(d_{I})\cdot U_{2})$

$U_{2}=UF_1\cdot U_{3}$

$U_{3}=\sum_{d_{I_{P}}\in\Delta}(s_{I_{UP}}(d_{I_{P}})\cdot U_4)$

$U_4=\sum_{d_{O_P}\in\Delta}(r_{O_{UP}}(d_{O_P})\cdot U_{5})$

$U_{5}=UF_2\cdot U_{6}$

$U_{6}=\sum_{d_{I_{R}}\in\Delta}(s_{I_{UR}}(d_{I_{R}})\cdot U_{7})$

$U_{7}=\sum_{d_{O_R}\in\Delta}(r_{O_{UR}}(d_{O_R})\cdot U_{8})$

$U_{8}=UF_3\cdot U_{9}$

$U_{9}=\sum_{d_{O}\in\Delta}(s_{O}(d_{O})\cdot U)$

The state transitions of the Resource Provider module
described by APTC are as follows.

$P=\sum_{d_{I_{P}}\in\Delta}(r_{I_{UP}}(d_{I_P})\cdot P_{2})$

$P_{2}=PF\cdot P_{3}$

$P_{3}=\sum_{d_{O_{P}}\in\Delta}(s_{O_{UP}}(d_{O_{P}})\cdot P)$

The state transitions of the Resource module
described by APTC are as follows.

$R=\sum_{d_{I_{R}}\in\Delta}(r_{I_{UR}}(d_{I_R})\cdot R_{2})$

$R_{2}=RF\cdot R_{3}$

$R_{3}=\sum_{d_{O_{R}}\in\Delta}(s_{O_{UR}}(d_{O_{R}})\cdot R)$

The sending action and the reading action of the same data through the same channel can communicate with each other, otherwise, will cause a deadlock $\delta$. We define the following
communication functions between the Resource User and the Resource Provider.

$$\gamma(r_{I_{UP}}(d_{I_{P}}),s_{I_{UP}}(d_{I_{P}}))\triangleq c_{I_{UP}}(d_{I_{P}})$$

$$\gamma(r_{O_{UP}}(d_{O_P}),s_{O_{UP}}(d_{O_P}))\triangleq c_{O_{UP}}(d_{O_P})$$

There are two communication functions between the Resource User and the Resource as follows.

$$\gamma(r_{I_{UR}}(d_{I_{R}}),s_{I_{UR}}(d_{I_{R}}))\triangleq c_{I_{UR}}(d_{I_{R}})$$

$$\gamma(r_{O_{UR}}(d_{O_{R}}),s_{O_{UR}}(d_{O_{R}}))\triangleq c_{O_{UR}}(d_{O_{R}})$$

Let all modules be in parallel, then the Partial Acquisition pattern $U\quad P\quad R$ can be presented by the following process term.

$\tau_I(\partial_H(\Theta(U\between P\between R)))=\tau_I(\partial_H(U\between P\between R))$

where $H=\{r_{I_{UP}}(d_{I_{P}}),s_{I_{UP}}(d_{I_{P}}),r_{O_{UP}}(d_{O_P}),s_{O_{UP}}(d_{O_P}),r_{I_{UR}}(d_{I_{R}}),s_{I_{UR}}(d_{I_{R}}),\\
r_{O_{UR}}(d_{O_R}),s_{O_{UR}}(d_{O_R})
|d_{I}, d_{I_{P}}, d_{I_R}, d_{O_{P}}, d_{O_R}, d_{O}\in\Delta\}$,

$I=\{c_{I_{UP}}(d_{I_{P}}),c_{O_{UP}}(d_{O_P}),c_{I_{UR}}(d_{I_{R}}),c_{O_{UR}}(d_{O_{R}}),\\
UF_1,UF_2,UF_3,PF,RF
|d_{I}, d_{I_{P}}, d_{I_R}, d_{O_{P}}, d_{O_R}, d_{O}\in\Delta\}$.

Then we get the following conclusion on the Partial Acquisition pattern.

\begin{theorem}[Correctness of the Partial Acquisition pattern]
The Partial Acquisition pattern $\tau_I(\partial_H(U\between P\between R))$ can exhibit desired external behaviors.
\end{theorem}

\begin{proof}
Based on the above state transitions of the above modules, by use of the algebraic laws of APTC, we can prove that

$\tau_I(\partial_H(U\between P\between R))=\sum_{d_{I},d_{O}\in\Delta}(r_{I}(d_{I})\cdot s_{O}(d_{O}))\cdot
\tau_I(\partial_H(U\between P\between R))$,

that is, the Partial Acquisition pattern $\tau_I(\partial_H(U\between P\between R))$ can exhibit desired external behaviors.

For the details of proof, please refer to section \ref{app}, and we omit it.
\end{proof}

\subsection{Resource Lifecycle}\label{RL7}

In this subsection, we verify patterns related to resource liftcycle, including the Caching pattern, the Pooling pattern, the Coordinator pattern, and the Resource Lifecycle Manager
pattern.

\subsubsection{Verification of the Caching Pattern}

The Caching pattern allows to cache the resources to avoid re-acquisitions of the resources.
There are four modules in the Caching pattern: the Resource User, the Resource Provider, the Resource Cache,
and the Resource. The Resource User interacts with the outside through
the channels $I$ and $O$; with the Resource Provider through the channel $I_{UP}$ and $O_{UP}$; with the Resource through the channels $I_{UR}$ and $O_{UR}$;
with the Resource Cache through the channels $I_{UC}$ and $O_{UC}$.
As illustrates in Figure \ref{Ca7}.

\begin{figure}
    \centering
    \includegraphics{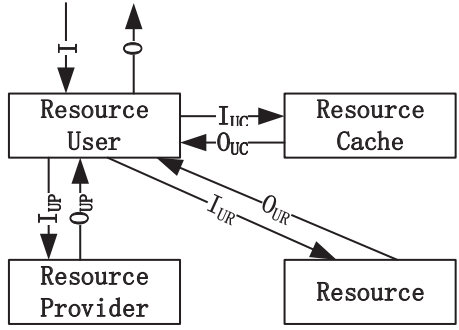}
    \caption{Caching pattern}
    \label{Ca7}
\end{figure}

The typical process of the Caching pattern is shown in Figure \ref{Ca7P} and as follows.

\begin{enumerate}
  \item The Resource User receives the input $d_{I}$ from the outside through the channel $I$ (the corresponding reading action is denoted $r_{I}(d_{I})$), then processes the input
   and generates the input $d_{I_P}$ through a processing function $UF_1$, and sends the input to the Resource Provider through the channel $I_{UP}$
  (the corresponding sending action is denoted $s_{I_{UP}}(d_{I_P})$);
  \item The Resource Provider receives the input from the Resource User through the channel $I_{UP}$ (the corresponding reading action is denoted $r_{I_{UP}}(d_{I_P})$), then
  processes the input and generate the output $d_{O_P}$ to the Resource User through a processing function $PF$, and sends the output to the Resource User
   through the channel $O_{UP}$ (the corresponding sending action is denoted $s_{O_{UP}}(d_{O_{P}})$);
  \item The Resource User receives the output $d_{O_{P}}$ from the Resource Provider through the channel $O_{UP}$ (the corresponding reading action is denoted $r_{O_{UP}}(d_{O_{P}})$), then processes
  the input through a processing function $UF_2$, generates and sends the input $d_{I_{R}}$ to the Resource through the channel $I_{UR}$ (the corresponding sending action is denoted $s_{I_{UR}}(d_{I_{R}})$);
  \item The Resource receives the input $d_{I_R}$ from the Resource User through the channel $I_{UR}$ (the corresponding reading action is denoted $r_{I_{UR}}(d_{I_R})$), then processes
  the input through a processing function $RF$, generates and sends the response $d_{O_R}$ (the corresponding sending action is denoted $s_{O_{UR}}(d_{O_R})$);
  \item The Resource User receives the response $d_{O_{R}}$ from the Resource through the channel $O_{UR}$ (the corresponding reading action is denoted $r_{O_{UR}}(d_{O_{R}})$),
  then processes the input $d_{O_R}$ through a processing function $UF_3$, and sends the processed input $d_{I_{C}}$ to the Resource Cache through the channel $I_{UC}$ (the corresponding sending action is denoted $s_{I_{UC}}(d_{I_{C}})$);
  \item The Resource Cache receives $d_{I_{C}}$ from the Resource User through the channel $I_{UC}$ (the corresponding reading action is denoted $r_{I_{UC}}(d_{I_{C}})$), then processes the request
  through a processing function $CF$, generates and sends the processed output $d_{O_{C}}$ to the Resource User through the channel $O_{UC}$ (the corresponding sending action is denoted
  $s_{O_{UC}}(d_{O_C})$);
  \item The Resource User receives the output $d_{O_C}$ from the Resource Cache through the channel $O_{UC}$ (the corresponding reading action is denoted $r_{O_{UC}}(d_{O_C})$),
   then processes the response and generates the response $d_{O}$ through a processing function $UF_4$, and sends the response to the outside through the channel $O$ (the corresponding sending action is denoted
  $s_{O}(d_{O})$).
\end{enumerate}

\begin{figure}
    \centering
    \includegraphics{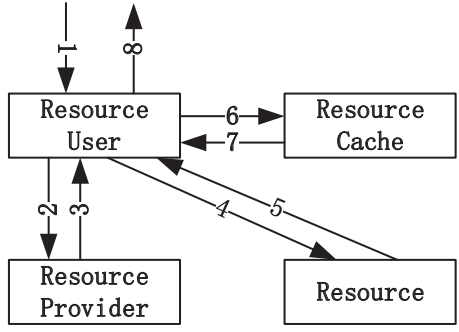}
    \caption{Typical process of Caching pattern}
    \label{Ca7P}
\end{figure}

In the following, we verify the Caching pattern. We assume all data elements $d_{I}$, $d_{I_{C}}$, $d_{I_{P}}$, $d_{I_R}$, $d_{O_{C}}$, $d_{O_P}$, $d_{O_R}$, $d_{O}$ are from a finite set
$\Delta$.

The state transitions of the Resource User module
described by APTC are as follows.

$U=\sum_{d_{I}\in\Delta}(r_{I}(d_{I})\cdot U_{2})$

$U_{2}=UF_1\cdot U_{3}$

$U_{3}=\sum_{d_{I_{P}}\in\Delta}(s_{I_{UP}}(d_{I_{P}})\cdot U_4)$

$U_4=\sum_{d_{O_P}\in\Delta}(r_{O_{UP}}(d_{O_P})\cdot U_{5})$

$U_{5}=UF_2\cdot U_{6}$

$U_{6}=\sum_{d_{I_{R}}\in\Delta}(s_{I_{UR}}(d_{I_{R}})\cdot U_{7})$

$U_{7}=\sum_{d_{O_R}\in\Delta}(r_{O_{UR}}(d_{O_R})\cdot U_{8})$

$U_{8}=UF_3\cdot U_{9}$

$U_{9}=\sum_{d_{I_{C}}\in\Delta}(s_{I_{UC}}(d_{I_{C}})\cdot U_{10})$

$U_{10}=\sum_{d_{O_C}\in\Delta}(r_{O_{UC}}(d_{O_C})\cdot U_{11})$

$U_{11}=UF_4\cdot U_{12}$

$U_{12}=\sum_{d_{O}\in\Delta}(s_{O}(d_{O})\cdot U)$

The state transitions of the Resource Provider module
described by APTC are as follows.

$P=\sum_{d_{I_{P}}\in\Delta}(r_{I_{UP}}(d_{I_P})\cdot P_{2})$

$P_{2}=PF\cdot P_{3}$

$P_{3}=\sum_{d_{O_{P}}\in\Delta}(s_{O_{UP}}(d_{O_{P}})\cdot P)$

The state transitions of the Resource Cache module
described by APTC are as follows.

$C=\sum_{d_{I_{C}}\in\Delta}(r_{I_{UC}}(d_{I_C})\cdot C_{2})$

$C_{2}=CF\cdot C_{3}$

$C_{3}=\sum_{d_{O_{C}}\in\Delta}(s_{O_{UC}}(d_{O_{C}})\cdot C)$

The state transitions of the Resource module
described by APTC are as follows.

$R=\sum_{d_{I_{R}}\in\Delta}(r_{I_{UR}}(d_{I_R})\cdot R_{2})$

$R_{2}=RF\cdot R_{3}$

$R_{3}=\sum_{d_{O_{R}}\in\Delta}(s_{O_{UR}}(d_{O_{R}})\cdot R)$

The sending action and the reading action of the same data through the same channel can communicate with each other, otherwise, will cause a deadlock $\delta$. We define the following
communication functions between the Resource User and the Resource Provider.

$$\gamma(r_{I_{UP}}(d_{I_{P}}),s_{I_{UP}}(d_{I_{P}}))\triangleq c_{I_{UP}}(d_{I_{P}})$$

$$\gamma(r_{O_{UP}}(d_{O_P}),s_{O_{UP}}(d_{O_P}))\triangleq c_{O_{UP}}(d_{O_P})$$

There are two communication functions between the Resource User and the Resource Cache as follows.

$$\gamma(r_{I_{UC}}(d_{I_{C}}),s_{I_{UC}}(d_{I_{C}}))\triangleq c_{I_{UC}}(d_{I_{C}})$$

$$\gamma(r_{O_{UC}}(d_{O_{C}}),s_{O_{UC}}(d_{O_{C}}))\triangleq c_{O_{UC}}(d_{O_{C}})$$

There are two communication functions between the Resource User and the Resource as follows.

$$\gamma(r_{I_{UR}}(d_{I_{R}}),s_{I_{UR}}(d_{I_{R}}))\triangleq c_{I_{UR}}(d_{I_{R}})$$

$$\gamma(r_{O_{UR}}(d_{O_{R}}),s_{O_{UR}}(d_{O_{R}}))\triangleq c_{O_{UR}}(d_{O_{R}})$$

Let all modules be in parallel, then the Caching pattern $U\quad C \quad P\quad R$ can be presented by the following process term.

$\tau_I(\partial_H(\Theta(U\between C\between P\between R)))=\tau_I(\partial_H(U\between C\between P\between R))$

where $H=\{r_{I_{UP}}(d_{I_{P}}),s_{I_{UP}}(d_{I_{P}}),r_{O_{UP}}(d_{O_P}),s_{O_{UP}}(d_{O_P}),r_{I_{UC}}(d_{I_{C}}),s_{I_{UC}}(d_{I_{C}}),\\
r_{O_{UC}}(d_{O_{C}}),s_{O_{UC}}(d_{O_{C}}),r_{I_{UR}}(d_{I_{R}}),s_{I_{UR}}(d_{I_{R}}),r_{O_{UR}}(d_{O_R}),s_{O_{UR}}(d_{O_R})\\
|d_{I}, d_{I_{P}}, d_{I_{C}}, d_{I_R}, d_{O_{P}}, d_{O_C}, d_{O_R}, d_{O}\in\Delta\}$,

$I=\{c_{I_{UP}}(d_{I_{P}}),c_{O_{UP}}(d_{O_P}),c_{I_{UC}}(d_{I_{C}}),c_{O_{UC}}(d_{O_{C}}),c_{I_{UR}}(d_{I_{R}}),c_{O_{UR}}(d_{O_{R}}),\\
UF_1,UF_2,UF_3,UF_4,PF,CF,RF
|d_{I}, d_{I_{P}}, d_{I_{C}}, d_{I_R}, d_{O_{P}}, d_{O_C}, d_{O_R}, d_{O}\in\Delta\}$.

Then we get the following conclusion on the Caching pattern.

\begin{theorem}[Correctness of the Caching pattern]
The Caching pattern $\tau_I(\partial_H(U\between C\between P\between R))$ can exhibit desired external behaviors.
\end{theorem}

\begin{proof}
Based on the above state transitions of the above modules, by use of the algebraic laws of APTC, we can prove that

$\tau_I(\partial_H(U\between C\between P\between R))=\sum_{d_{I},d_{O}\in\Delta}(r_{I}(d_{I})\cdot s_{O}(d_{O}))\cdot
\tau_I(\partial_H(U\between C\between P\between R))$,

that is, the Caching pattern $\tau_I(\partial_H(U\between C\between P\between R))$ can exhibit desired external behaviors.

For the details of proof, please refer to section \ref{app}, and we omit it.
\end{proof}

\subsubsection{Verification of the Pooling Pattern}

The Pooling pattern allows to recycle the resources to avoid re-acquisitions of the resources.
There are four modules in the Pooling pattern: the Resource User, the Resource Provider, the Resource Pool,
and the Resource. The Resource User interacts with the outside through
the channels $I$ and $O$; with the Resource Provider through the channel $I_{UP}$ and $O_{UP}$; with the Resource through the channels $I_{UR}$ and $O_{UR}$.
The Resource Pool interacts with the Resource Provider through the channels $I_{PP}$ and $O_{PP}$.
As illustrates in Figure \ref{Po7}.

\begin{figure}
    \centering
    \includegraphics{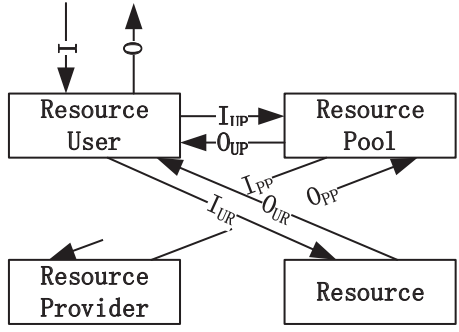}
    \caption{Pooling pattern}
    \label{Po7}
\end{figure}

The typical process of the Pooling pattern is shown in Figure \ref{Po7P} and as follows.

\begin{enumerate}
  \item The Resource User receives the input $d_{I}$ from the outside through the channel $I$ (the corresponding reading action is denoted $r_{I}(d_{I})$), then processes the input
   and generates the input $d_{I_P}$ through a processing function $UF_1$, and sends the input to the Resource Pool through the channel $I_{UP}$
  (the corresponding sending action is denoted $s_{I_{UP}}(d_{I_P})$);
  \item The Resource Pool receives the input from the Resource User through the channel $I_{UP}$ (the corresponding reading action is denoted $r_{I_{UP}}(d_{I_P})$), then
  processes the input and generate the input $d_{I_{RP}}$ to the Resource Provider through a processing function $PF_1$, and sends the input to the Resource Provider
   through the channel $I_{PP}$ (the corresponding sending action is denoted $s_{I_{PP}}(d_{I_{RP}})$);
  \item The Resource Provider receives the output $d_{I_{RP}}$ from the Resource Pool through the channel $I_{PP}$ (the corresponding reading action is denoted $r_{I_{PP}}(d_{I_{RP}})$), then processes
  the input through a processing function $RPF$, generates and sends the output $d_{O_{RP}}$ to the Resource Pool through the channel $O_{PP}$ (the corresponding sending action is denoted $s_{O_{PP}}(d_{O_{RP}})$);
  \item The Resource Pool receives the input $d_{O_{RP}}$ from the Resource Provider through the channel $O_{PP}$ (the corresponding reading action is denoted $r_{O_{PP}}(d_{O_{RP}})$), then processes
  the output through a processing function $PF_2$, generates and sends the response $d_{O_P}$ (the corresponding sending action is denoted $s_{O_{UP}}(d_{O_P})$);
  \item The Resource User receives the response $d_{O_{P}}$ from the Resource Pool through the channel $O_{UP}$ (the corresponding reading action is denoted $r_{O_{UP}}(d_{O_{P}})$),
  then processes the output $d_{O_P}$ through a processing function $UF_2$, and sends the processed input $d_{I_{R}}$ to the Resource through the channel $I_{UR}$ (the corresponding sending action is denoted $s_{I_{UR}}(d_{I_{R}})$);
  \item The Resource receives $d_{I_{R}}$ from the Resource User through the channel $I_{UR}$ (the corresponding reading action is denoted $r_{I_{UR}}(d_{I_{R}})$), then processes the request
  through a processing function $RF$, generates and sends the processed output $d_{O_{R}}$ to the Resource User through the channel $O_{UR}$ (the corresponding sending action is denoted
  $s_{O_{UR}}(d_{O_R})$);
  \item The Resource User receives the output $d_{O_R}$ from the Resource through the channel $O_{UR}$ (the corresponding reading action is denoted $r_{O_{UR}}(d_{O_R})$),
   then processes the response and generates the response $d_{O}$ through a processing function $UF_3$, and sends the response to the outside through the channel $O$ (the corresponding sending action is denoted
  $s_{O}(d_{O})$).
\end{enumerate}

\begin{figure}
    \centering
    \includegraphics{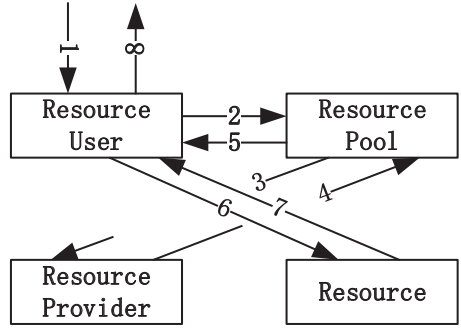}
    \caption{Typical process of Pooling pattern}
    \label{Po7P}
\end{figure}

In the following, we verify the Pooling pattern. We assume all data elements $d_{I}$, $d_{I_{P}}$, $d_{I_{RP}}$, $d_{I_R}$, $d_{O_{RP}}$, $d_{O_P}$, $d_{O_R}$, $d_{O}$ are from a finite set
$\Delta$.

The state transitions of the Resource User module
described by APTC are as follows.

$U=\sum_{d_{I}\in\Delta}(r_{I}(d_{I})\cdot U_{2})$

$U_{2}=UF_1\cdot U_{3}$

$U_{3}=\sum_{d_{I_{P}}\in\Delta}(s_{I_{UP}}(d_{I_{P}})\cdot U_4)$

$U_4=\sum_{d_{O_P}\in\Delta}(r_{O_{UP}}(d_{O_P})\cdot U_{5})$

$U_{5}=UF_2\cdot U_{6}$

$U_{6}=\sum_{d_{I_{R}}\in\Delta}(s_{I_{UR}}(d_{I_{R}})\cdot U_{7})$

$U_{7}=\sum_{d_{O_R}\in\Delta}(r_{O_{UR}}(d_{O_R})\cdot U_{8})$

$U_{8}=UF_3\cdot U_{9}$

$U_{9}=\sum_{d_{O}\in\Delta}(s_{O}(d_{O})\cdot U)$

The state transitions of the Resource Provider module
described by APTC are as follows.

$RP=\sum_{d_{I_{RP}}\in\Delta}(r_{I_{PP}}(d_{I_{RP}})\cdot RP_{2})$

$RP_{2}=RPF\cdot RP_{3}$

$RP_{3}=\sum_{d_{O_{RP}}\in\Delta}(s_{O_{PP}}(d_{O_{RP}})\cdot RP)$

The state transitions of the Resource Pool module
described by APTC are as follows.

$P=\sum_{d_{I_{P}}\in\Delta}(r_{I_{UP}}(d_{I_P})\cdot P_{2})$

$P_{2}=PF_1\cdot P_{3}$

$P_{3}=\sum_{d_{I_{RP}}\in\Delta}(s_{I_{PP}}(d_{I_{RP}})\cdot P_4)$

$P_4=\sum_{d_{O_{RP}}\in\Delta}(r_{O_{PP}}(d_{O_{RP}})\cdot P_{5})$

$P_{5}=PF_2\cdot P_{6}$

$P_{6}=\sum_{d_{O_{P}}\in\Delta}(s_{O_{UP}}(d_{O_{P}})\cdot P)$

The state transitions of the Resource module
described by APTC are as follows.

$R=\sum_{d_{I_{R}}\in\Delta}(r_{I_{UR}}(d_{I_R})\cdot R_{2})$

$R_{2}=RF\cdot R_{3}$

$R_{3}=\sum_{d_{O_{R}}\in\Delta}(s_{O_{UR}}(d_{O_{R}})\cdot R)$

The sending action and the reading action of the same data through the same channel can communicate with each other, otherwise, will cause a deadlock $\delta$. We define the following
communication functions between the Resource User and the Resource Pool.

$$\gamma(r_{I_{UP}}(d_{I_{P}}),s_{I_{UP}}(d_{I_{P}}))\triangleq c_{I_{UP}}(d_{I_{P}})$$

$$\gamma(r_{O_{UP}}(d_{O_P}),s_{O_{UP}}(d_{O_P}))\triangleq c_{O_{UP}}(d_{O_P})$$

There are two communication functions between the Resource Provider and the Resource Pool as follows.

$$\gamma(r_{I_{PP}}(d_{I_{RP}}),s_{I_{PP}}(d_{I_{RP}}))\triangleq c_{I_{PP}}(d_{I_{RP}})$$

$$\gamma(r_{O_{PP}}(d_{O_{RP}}),s_{O_{PP}}(d_{O_{RP}}))\triangleq c_{O_{PP}}(d_{O_{RP}})$$

There are two communication functions between the Resource User and the Resource as follows.

$$\gamma(r_{I_{UR}}(d_{I_{R}}),s_{I_{UR}}(d_{I_{R}}))\triangleq c_{I_{UR}}(d_{I_{R}})$$

$$\gamma(r_{O_{UR}}(d_{O_{R}}),s_{O_{UR}}(d_{O_{R}}))\triangleq c_{O_{UR}}(d_{O_{R}})$$

Let all modules be in parallel, then the Pooling pattern $U\quad RP \quad P\quad R$ can be presented by the following process term.

$\tau_I(\partial_H(\Theta(U\between RP\between P\between R)))=\tau_I(\partial_H(U\between RP\between P\between R))$

where $H=\{r_{I_{UP}}(d_{I_{P}}),s_{I_{UP}}(d_{I_{P}}),r_{O_{UP}}(d_{O_P}),s_{O_{UP}}(d_{O_P}),r_{I_{PP}}(d_{I_{RP}}),s_{I_{PP}}(d_{I_{RP}}),\\
r_{O_{PP}}(d_{O_{RP}}),s_{O_{PP}}(d_{O_{RP}}),r_{I_{UR}}(d_{I_{R}}),s_{I_{UR}}(d_{I_{R}}),r_{O_{UR}}(d_{O_R}),s_{O_{UR}}(d_{O_R})\\
|d_{I}, d_{I_{P}}, d_{I_{RP}}, d_{I_R}, d_{O_{P}}, d_{O_{RP}}, d_{O_R}, d_{O}\in\Delta\}$,

$I=\{c_{I_{UP}}(d_{I_{P}}),c_{O_{UP}}(d_{O_P}),c_{I_{PP}}(d_{I_{RP}}),c_{O_{PP}}(d_{O_{RP}}),c_{I_{UR}}(d_{I_{R}}),c_{O_{UR}}(d_{O_{R}}),\\
UF_1,UF_2,UF_3,PF_1,PF_2,RPF,RF
|d_{I}, d_{I_{P}}, d_{I_{RP}}, d_{I_R}, d_{O_{P}}, d_{O_{RP}}, d_{O_R}, d_{O}\in\Delta\}$.

Then we get the following conclusion on the Pooling pattern.

\begin{theorem}[Correctness of the Pooling pattern]
The Pooling pattern $\tau_I(\partial_H(U\between RP\between P\between R))$ can exhibit desired external behaviors.
\end{theorem}

\begin{proof}
Based on the above state transitions of the above modules, by use of the algebraic laws of APTC, we can prove that

$\tau_I(\partial_H(U\between RP\between P\between R))=\sum_{d_{I},d_{O}\in\Delta}(r_{I}(d_{I})\cdot s_{O}(d_{O}))\cdot
\tau_I(\partial_H(U\between RP\between P\between R))$,

that is, the Pooling pattern $\tau_I(\partial_H(U\between RP\between P\between R))$ can exhibit desired external behaviors.

For the details of proof, please refer to section \ref{app}, and we omit it.
\end{proof}

\subsubsection{Verification of the Coordinator Pattern}

The Coordinator pattern gives a solution to maintain the consistency by coordinating the completion of tasks involving multi participants,
which has two classes of components: $n$ Synchronous Services and the Coordinator.
The Coordinator receives the inputs from the user through the channel $I$, then the Coordinator sends the results to the Participant $i$
through the channel $CP_i$ for $1\leq i\leq n$;
When the Participant $i$ receives the input from the Coordinator, it generates and sends the results out to the user through the channel $O_i$.
As illustrates in Figure \ref{Co7}.

\begin{figure}
    \centering
    \includegraphics{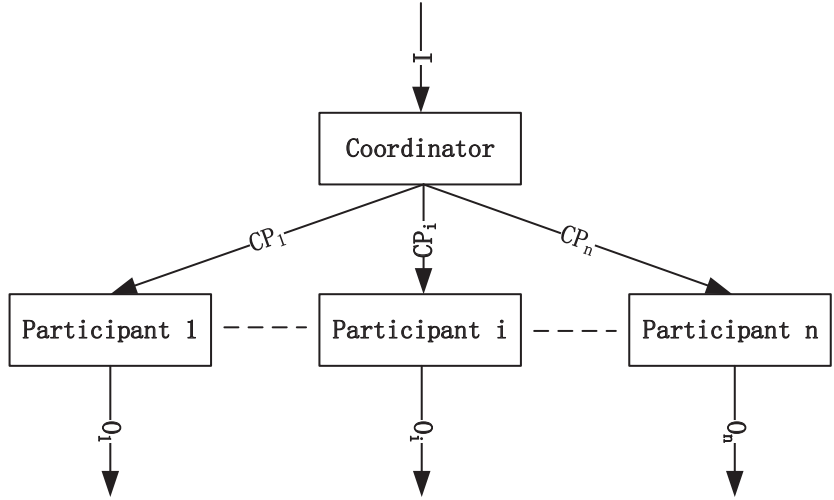}
    \caption{Coordinator pattern}
    \label{Co7}
\end{figure}

The typical process of the Coordinator pattern is shown in Figure \ref{Co7P} and following.

\begin{enumerate}
  \item The Coordinator receives the input $d_I$ from the user through the channel $I$ (the corresponding reading action is denoted $r_I(D_I)$), processes the input through
  a processing function $CF$, and generate the input to the Participant $i$ (for $1\leq i\leq n$) which is denoted $d_{I_{P_i}}$; then sends the input to the Participant $i$ through the
  channel $CP_i$ (the corresponding sending action is denoted $s_{CP_i}(d_{I_{P_i}})$);
  \item The Participant $i$ (for $1\leq i\leq n$) receives the input from the Coordinator
  through the channel $CP_i$ (the corresponding reading action is denoted $r_{CP_i}(d_{I_{P_i}})$), processes the results through a processing function $PF_{i}$, generates the output
  $d_{O_i}$, then sending the output through the channel $O_i$ (the corresponding sending action is denoted $s_{O_i}(d_{O_i})$).
\end{enumerate}

\begin{figure}
    \centering
    \includegraphics{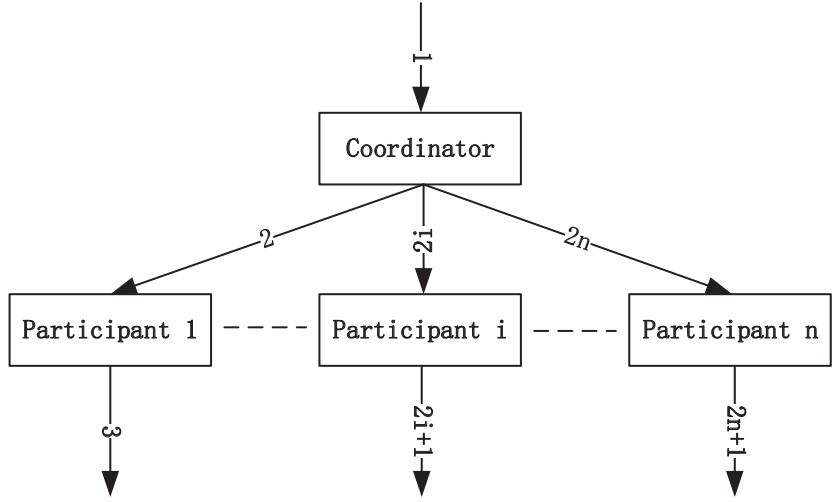}
    \caption{Typical process of Coordinator pattern}
    \label{Co7P}
\end{figure}

In the following, we verify the Coordinator pattern. We assume all data elements $d_{I}$, $d_{I_{P_i}}$, $d_{O_{i}}$ (for $1\leq i\leq n$) are from a finite set
$\Delta$.

The state transitions of the Coordinator module
described by APTC are as follows.

$C=\sum_{d_{I}\in\Delta}(r_{I}(d_{I})\cdot C_{2})$

$C_{2}=CF\cdot C_{3}$

$C_{3}=\sum_{d_{I_{P_1}},\cdots,d_{I_{P_n}}\in\Delta}(s_{CP_1}(d_{I_{P_1}})\between\cdots\between s_{CP_n}(d_{I_{P_n}})\cdot C)$

The state transitions of the Participant $i$ described by APTC are as follows.

$P_i=\sum_{d_{I_{P_i}}\in\Delta}(r_{CP_i}(d_{I_{P_i}})\cdot P_{i_2})$

$P_{i_2}=PF_i\cdot P_{i_3}$

$P_{i_3}=\sum_{d_{O_{i}}\in\Delta}(s_{O_{i}}(d_{O_i})\cdot P_i)$

The sending action and the reading action of the same data through the same channel can communicate with each other, otherwise, will cause a deadlock $\delta$. We define the following
communication functions of the Participant $i$ for $1\leq i\leq n$.

$$\gamma(r_{CP_i}(d_{I_{P_i}}),s_{CP_i}(d_{I_{P_i}}))\triangleq c_{CP_i}(d_{I_{P_i}})$$

Let all modules be in parallel, then the Coordinator pattern $C \quad P_1\cdots P_i\cdots P_n$ can be presented by the following process term.

$\tau_I(\partial_H(\Theta(C\between P_1\between\cdots\between P_i\between\cdots\between P_n)))=\tau_I(\partial_H(C\between P_1\between\cdots\between P_i\between\cdots\between P_n))$

where $H=\{r_{CP_i}(d_{O_{P_i}}),s_{CP_i}(d_{O_{P_i}})|d_{I}, d_{I_{P_i}}, d_{O_{i}}\in\Delta\}$ for $1\leq i\leq n$,

$I=\{c_{CP_i}(d_{I_{P_i}}),CF,PF_{i}|d_{I}, d_{I_{P_i}}, d_{O_{i}}\in\Delta\}$ for $1\leq i\leq n$.

Then we get the following conclusion on the Coordinator pattern.

\begin{theorem}[Correctness of the Coordinator pattern]
The Coordinator pattern $\tau_I(\partial_H(C\between P_1\between\cdots\between P_i\between\cdots\between P_n))$ can exhibit desired external behaviors.
\end{theorem}

\begin{proof}
Based on the above state transitions of the above modules, by use of the algebraic laws of APTC, we can prove that

$\tau_I(\partial_H(C\between P_1\between\cdots\between P_i\between\cdots\between P_n))=\sum_{d_{I},d_{O_1},\cdots,d_{O_n}\in\Delta}(r_{I}(d_{I})\cdot s_{O_1}(d_{O_1})\parallel\cdots\parallel s_{O_i}(d_{O_i})\parallel\cdots\parallel s_{O_n}(d_{O_n}))\cdot
\tau_I(\partial_H(C\between P_1\between\cdots\between P_i\between\cdots\between P_n))$,

that is, the Coordinator pattern $\tau_I(\partial_H(A\between S_1\between\cdots\between S_i\between\cdots\between S_n))$ can exhibit desired external behaviors.

For the details of proof, please refer to section \ref{app}, and we omit it.
\end{proof}

\subsubsection{Verification of the Resource Lifecycle Manager Pattern}

The Resource Lifecycle Manager pattern decouples the lifecyle management by introduce a Resource Lifecycle Manager.
There are four modules in the Resource Lifecycle Manager pattern: the Resource User, the Resource Provider, the Resource Lifecycle Manager,
and the Resource. The Resource User interacts with the outside through
the channels $I$ and $O$; with the Resource Provider through the channel $I_{UM}$ and $O_{UM}$; with the Resource through the channels $I_{UR}$ and $O_{UR}$.
The Resource Lifecycle Manager interacts with the Resource Provider through the channels $I_{MP}$ and $O_{MP}$.
As illustrates in Figure \ref{RLM7}.

\begin{figure}
    \centering
    \includegraphics{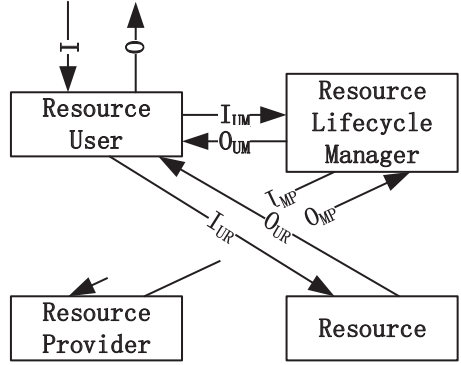}
    \caption{Resource Lifecycle Manager pattern}
    \label{RLM7}
\end{figure}

The typical process of the Resource Lifecycle Manager pattern is shown in Figure \ref{RLM7P} and as follows.

\begin{enumerate}
  \item The Resource User receives the input $d_{I}$ from the outside through the channel $I$ (the corresponding reading action is denoted $r_{I}(d_{I})$), then processes the input
   and generates the input $d_{I_M}$ through a processing function $UF_1$, and sends the input to the Resource Lifecycle Manager through the channel $I_{UM}$
  (the corresponding sending action is denoted $s_{I_{UM}}(d_{I_M})$);
  \item The Resource Lifecycle Manager receives the input from the Resource User through the channel $I_{UM}$ (the corresponding reading action is denoted $r_{I_{UM}}(d_{I_M})$), then
  processes the input and generate the input $d_{I_{P}}$ to the Resource Provider through a processing function $MF_1$, and sends the input to the Resource Provider
   through the channel $I_{MP}$ (the corresponding sending action is denoted $s_{I_{MP}}(d_{I_{P}})$);
  \item The Resource Provider receives the output $d_{I_{P}}$ from the Resource Lifecycle Manager through the channel $I_{MP}$ (the corresponding reading action is denoted $r_{I_{MP}}(d_{I_{P}})$), then processes
  the input through a processing function $PF$, generates and sends the output $d_{O_{P}}$ to the Resource Lifecycle Manager through the channel $O_{MP}$ (the corresponding sending action is denoted $s_{O_{MP}}(d_{O_{P}})$);
  \item The Resource Lifecycle Manager receives the input $d_{O_{P}}$ from the Resource Provider through the channel $O_{MP}$ (the corresponding reading action is denoted $r_{O_{MP}}(d_{O_{P}})$), then processes
  the output through a processing function $MF_2$, generates and sends the response $d_{O_M}$ (the corresponding sending action is denoted $s_{O_{UM}}(d_{O_M})$);
  \item The Resource User receives the response $d_{O_{M}}$ from the Resource Lifecycle Manager through the channel $O_{UM}$ (the corresponding reading action is denoted $r_{O_{UM}}(d_{O_{M}})$),
  then processes the output $d_{O_M}$ through a processing function $UF_2$, and sends the processed input $d_{I_{R}}$ to the Resource through the channel $I_{UR}$ (the corresponding sending action is denoted $s_{I_{UR}}(d_{I_{R}})$);
  \item The Resource receives $d_{I_{R}}$ from the Resource User through the channel $I_{UR}$ (the corresponding reading action is denoted $r_{I_{UR}}(d_{I_{R}})$), then processes the request
  through a processing function $RF$, generates and sends the processed output $d_{O_{R}}$ to the Resource User through the channel $O_{UR}$ (the corresponding sending action is denoted
  $s_{O_{UR}}(d_{O_R})$);
  \item The Resource User receives the output $d_{O_R}$ from the Resource through the channel $O_{UR}$ (the corresponding reading action is denoted $r_{O_{UR}}(d_{O_R})$),
   then processes the response and generates the response $d_{O}$ through a processing function $UF_3$, and sends the response to the outside through the channel $O$ (the corresponding sending action is denoted
  $s_{O}(d_{O})$).
\end{enumerate}

\begin{figure}
    \centering
    \includegraphics{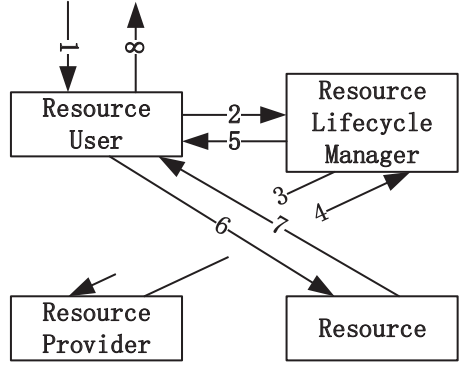}
    \caption{Typical process of Resource Lifecycle Manager pattern}
    \label{RLM7P}
\end{figure}

In the following, we verify the Resource Lifecycle Manager pattern. We assume all data elements $d_{I}$, $d_{I_{P}}$, $d_{I_{M}}$, $d_{I_R}$, $d_{O_{M}}$, $d_{O_P}$, $d_{O_R}$, $d_{O}$ are from a finite set
$\Delta$.

The state transitions of the Resource User module
described by APTC are as follows.

$U=\sum_{d_{I}\in\Delta}(r_{I}(d_{I})\cdot U_{2})$

$U_{2}=UF_1\cdot U_{3}$

$U_{3}=\sum_{d_{I_{M}}\in\Delta}(s_{I_{UM}}(d_{I_{M}})\cdot U_4)$

$U_4=\sum_{d_{O_M}\in\Delta}(r_{O_{UM}}(d_{O_M})\cdot U_{5})$

$U_{5}=UF_2\cdot U_{6}$

$U_{6}=\sum_{d_{I_{R}}\in\Delta}(s_{I_{UR}}(d_{I_{R}})\cdot U_{7})$

$U_{7}=\sum_{d_{O_R}\in\Delta}(r_{O_{UR}}(d_{O_R})\cdot U_{8})$

$U_{8}=UF_3\cdot U_{9}$

$U_{9}=\sum_{d_{O}\in\Delta}(s_{O}(d_{O})\cdot U)$

The state transitions of the Resource Provider module
described by APTC are as follows.

$P=\sum_{d_{I_{P}}\in\Delta}(r_{I_{MP}}(d_{I_{P}})\cdot P_{2})$

$P_{2}=PF\cdot P_{3}$

$P_{3}=\sum_{d_{O_{P}}\in\Delta}(s_{O_{MP}}(d_{O_{P}})\cdot P)$

The state transitions of the Resource Lifecycle Manager module
described by APTC are as follows.

$M=\sum_{d_{I_{M}}\in\Delta}(r_{I_{UM}}(d_{I_M})\cdot M_{2})$

$M_{2}=MF_1\cdot M_{3}$

$M_{3}=\sum_{d_{I_{P}}\in\Delta}(s_{I_{MP}}(d_{I_{P}})\cdot M_4)$

$M_4=\sum_{d_{O_{P}}\in\Delta}(r_{O_{MP}}(d_{O_{P}})\cdot M_{5})$

$M_{5}=MF_2\cdot M_{6}$

$M_{6}=\sum_{d_{O_{M}}\in\Delta}(s_{O_{UM}}(d_{O_{M}})\cdot M)$

The state transitions of the Resource module
described by APTC are as follows.

$R=\sum_{d_{I_{R}}\in\Delta}(r_{I_{UR}}(d_{I_R})\cdot R_{2})$

$R_{2}=RF\cdot R_{3}$

$R_{3}=\sum_{d_{O_{R}}\in\Delta}(s_{O_{UR}}(d_{O_{R}})\cdot R)$

The sending action and the reading action of the same data through the same channel can communicate with each other, otherwise, will cause a deadlock $\delta$. We define the following
communication functions between the Resource User and the Resource Lifecycle Manager.

$$\gamma(r_{I_{UM}}(d_{I_{M}}),s_{I_{UM}}(d_{I_{M}}))\triangleq c_{I_{UM}}(d_{I_{M}})$$

$$\gamma(r_{O_{UM}}(d_{O_M}),s_{O_{UM}}(d_{O_M}))\triangleq c_{O_{UM}}(d_{O_M})$$

There are two communication functions between the Resource Provider and the Resource Lifecycle Manager as follows.

$$\gamma(r_{I_{MP}}(d_{I_{P}}),s_{I_{MP}}(d_{I_{P}}))\triangleq c_{I_{MP}}(d_{I_{P}})$$

$$\gamma(r_{O_{MP}}(d_{O_{P}}),s_{O_{MP}}(d_{O_{P}}))\triangleq c_{O_{MP}}(d_{O_{P}})$$

There are two communication functions between the Resource User and the Resource as follows.

$$\gamma(r_{I_{UR}}(d_{I_{R}}),s_{I_{UR}}(d_{I_{R}}))\triangleq c_{I_{UR}}(d_{I_{R}})$$

$$\gamma(r_{O_{UR}}(d_{O_{R}}),s_{O_{UR}}(d_{O_{R}}))\triangleq c_{O_{UR}}(d_{O_{R}})$$

Let all modules be in parallel, then the Resource Lifecycle Manager pattern $U\quad M \quad P\quad R$ can be presented by the following process term.

$\tau_I(\partial_H(\Theta(U\between M\between P\between R)))=\tau_I(\partial_H(U\between M\between P\between R))$

where $H=\{r_{I_{UM}}(d_{I_{M}}),s_{I_{UM}}(d_{I_{M}}),r_{O_{UM}}(d_{O_M}),s_{O_{UM}}(d_{O_M}),r_{I_{MP}}(d_{I_{P}}),s_{I_{MP}}(d_{I_{P}}),\\
r_{O_{MP}}(d_{O_{P}}),s_{O_{MP}}(d_{O_{P}}),r_{I_{UR}}(d_{I_{R}}),s_{I_{UR}}(d_{I_{R}}),r_{O_{UR}}(d_{O_R}),s_{O_{UR}}(d_{O_R})\\
|d_{I}, d_{I_{P}}, d_{I_{M}}, d_{I_R}, d_{O_{P}}, d_{O_M}, d_{O_R}, d_{O}\in\Delta\}$,

$I=\{c_{I_{UM}}(d_{I_{M}}),c_{O_{UM}}(d_{O_M}),c_{I_{MP}}(d_{I_{P}}),c_{O_{MP}}(d_{O_{P}}),c_{I_{UR}}(d_{I_{R}}),c_{O_{UR}}(d_{O_{R}}),\\
UF_1,UF_2,UF_3,MF_1,MF_2,PF,RF
|d_{I}, d_{I_{P}}, d_{I_{M}}, d_{I_R}, d_{O_{P}}, d_{O_M}, d_{O_R}, d_{O}\in\Delta\}$.

Then we get the following conclusion on the Resource Lifecycle Manager pattern.

\begin{theorem}[Correctness of the Resource Lifecycle Manager pattern]
The Resource Lifecycle Manager pattern $\tau_I(\partial_H(U\between M\between P\between R))$ can exhibit desired external behaviors.
\end{theorem}

\begin{proof}
Based on the above state transitions of the above modules, by use of the algebraic laws of APTC, we can prove that

$\tau_I(\partial_H(U\between M\between P\between R))=\sum_{d_{I},d_{O}\in\Delta}(r_{I}(d_{I})\cdot s_{O}(d_{O}))\cdot
\tau_I(\partial_H(U\between M\between P\between R))$,

that is, the Resource Lifecycle Manager pattern $\tau_I(\partial_H(U\between RP\between P\between R))$ can exhibit desired external behaviors.

For the details of proof, please refer to section \ref{app}, and we omit it.
\end{proof}

\subsection{Resource Release}\label{RR7}

In this subsection, we verify patterns for resource release, including the Leasing pattern, and the Evictor pattern.

\subsubsection{Verification of the Leasing Pattern}

The Leasing pattern uses a mediating lookup service to find and access resources.
There are four modules in the Leasing pattern: the Resource User, the Resource Provider, the Lease,
and the Resource. The Resource User interacts with the outside through
the channels $I$ and $O$; with the Resource Provider through the channel $I_{UP}$ and $O_{UP}$; with the Resource through the channels $I_{UR}$ and $O_{UR}$;
with the Lease through the channels $I_{UL}$ and $O_{UL}$.
As illustrates in Figure \ref{Le7}.

\begin{figure}
    \centering
    \includegraphics{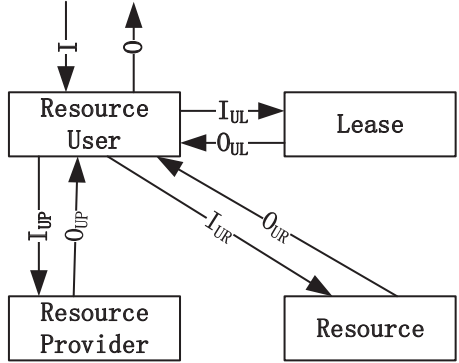}
    \caption{Leasing pattern}
    \label{Le7}
\end{figure}

The typical process of the Leasing pattern is shown in Figure \ref{Le7P} and as follows.

\begin{enumerate}
  \item The Resource User receives the input $d_{I}$ from the outside through the channel $I$ (the corresponding reading action is denoted $r_{I}(d_{I})$), then processes the input $d_{I}$ through a processing
  function $UF_1$,  and sends the input $d_{I_P}$ to the Resource Provider through the channel $I_{UP}$
  (the corresponding sending action is denoted $s_{I_{UP}}(d_{I_P})$);
  \item The Resource Provider receives the input from the Resource User through the channel $I_{UP}$ (the corresponding reading action is denoted $r_{I_{UP}}(d_{I_P})$), then
  processes the input and generate the output $d_{O_P}$ to the Resource User through a processing function $PF$, and sends the output to the Resource User
   through the channel $O_{UP}$ (the corresponding sending action is denoted $s_{O_{UP}}(d_{O_{P}})$);
  \item The Resource User receives the output $d_{O_{P}}$ from the Resource Provider through the channel $O_{UP}$ (the corresponding reading action is denoted $r_{O_{UP}}(d_{O_{P}})$), then processes
  the output through a processing function $UF_2$, generates and sends the input $d_{I_{R}}$ to the Resource through the channel $I_{UR}$ (the corresponding sending action is denoted $s_{I_{UR}}(d_{I_{R}})$);
  \item The Resource receives the input $d_{I_R}$ from the Resource User through the channel $I_{UR}$ (the corresponding reading action is denoted $r_{I_{UR}}(d_{I_R})$), then processes
  the input through a processing function $RF$, generates and sends the response $d_{O_R}$ (the corresponding sending action is denoted $s_{O_{UR}}(d_{O_R})$);
  \item The Resource User receives the response $d_{O_{R}}$ from the Resource through the channel $O_{UR}$ (the corresponding reading action is denoted $r_{O_{UR}}(d_{O_{R}})$), then
  processes the response and generates the response $d_{O}$ through a processing function $UF_3$, and sends the processed input $d_{I_{L}}$ to the Lease through the channel $I_{UL}$ (the corresponding sending action is denoted $s_{I_{UL}}(d_{I_{L}})$);
  \item The Lease receives $d_{I_{L}}$ from the Resource User through the channel $I_{UL}$ (the corresponding reading action is denoted $r_{I_{UL}}(d_{I_{L}})$), then processes the request
  through a processing function $LF$, generates and sends the processed output $d_{O_{L}}$ to the Resource User through the channel $O_{UL}$ (the corresponding sending action is denoted
  $s_{O_{UL}}(d_{O_L})$);
  \item The Resource User receives the output $d_{O_L}$ from the Lease through the channel $O_{UL}$ (the corresponding reading action is denoted $r_{O_{UL}}(d_{O_L})$),
  then processes the output and generates $d_{O}$ through a processing function $UF_4$, and sends the response to the outside through the channel $O$ (the corresponding sending action is denoted
  $s_{O}(d_{O})$).
\end{enumerate}

\begin{figure}
    \centering
    \includegraphics{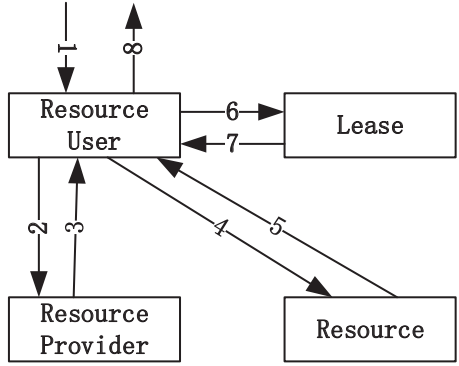}
    \caption{Typical process of Leasing pattern}
    \label{Le7P}
\end{figure}

In the following, we verify the Leasing pattern. We assume all data elements $d_{I}$, $d_{I_{L}}$, $d_{I_{P}}$, $d_{I_R}$, $d_{O_{L}}$, $d_{O_P}$, $d_{O_R}$, $d_{O}$ are from a finite set
$\Delta$.

The state transitions of the Resource User module
described by APTC are as follows.

$U=\sum_{d_{I}\in\Delta}(r_{I}(d_{I})\cdot U_{2})$

$U_{2}=UF_1\cdot U_{3}$

$U_{3}=\sum_{d_{I_{P}}\in\Delta}(s_{I_{UP}}(d_{I_{P}})\cdot U_4)$

$U_4=\sum_{d_{O_P}\in\Delta}(r_{O_{UP}}(d_{O_P})\cdot U_{5})$

$U_{5}=UF_2\cdot U_{6}$

$U_{6}=\sum_{d_{I_{R}}\in\Delta}(s_{I_{UR}}(d_{I_{R}})\cdot U_{7})$

$U_{7}=\sum_{d_{O_R}\in\Delta}(r_{O_{UR}}(d_{O_R})\cdot U_{8})$

$U_{8}=UF_3\cdot U_{9}$

$U_{9}=\sum_{d_{I_{L}}\in\Delta}(s_{I_{UL}}(d_{I_{L}})\cdot U_{10})$

$U_{10}=\sum_{d_{O_L}\in\Delta}(r_{O_{UL}}(d_{O_L})\cdot U_{11})$

$U_{11}=UF_4\cdot U_{12}$

$U_{12}=\sum_{d_{O}\in\Delta}(s_{O}(d_{O})\cdot U)$

The state transitions of the Resource Provider module
described by APTC are as follows.

$P=\sum_{d_{I_{P}}\in\Delta}(r_{I_{UP}}(d_{I_P})\cdot P_{2})$

$P_{2}=PF\cdot P_{3}$

$P_{3}=\sum_{d_{O_{P}}\in\Delta}(s_{O_{UP}}(d_{O_{P}})\cdot P)$

The state transitions of the Lease module
described by APTC are as follows.

$L=\sum_{d_{I_{L}}\in\Delta}(r_{I_{UL}}(d_{I_L})\cdot L_{2})$

$L_{2}=LF\cdot L_{3}$

$L_{3}=\sum_{d_{O_{L}}\in\Delta}(s_{O_{UL}}(d_{O_{L}})\cdot L)$

The state transitions of the Resource module
described by APTC are as follows.

$R=\sum_{d_{I_{R}}\in\Delta}(r_{I_{UR}}(d_{I_R})\cdot R_{2})$

$R_{2}=RF\cdot R_{3}$

$R_{3}=\sum_{d_{O_{R}}\in\Delta}(s_{O_{UR}}(d_{O_{R}})\cdot R)$

The sending action and the reading action of the same data through the same channel can communicate with each other, otherwise, will cause a deadlock $\delta$. We define the following
communication functions between the Resource User and the Resource Provider Proxy.

$$\gamma(r_{I_{UP}}(d_{I_{P}}),s_{I_{UP}}(d_{I_{P}}))\triangleq c_{I_{UP}}(d_{I_{P}})$$

$$\gamma(r_{O_{UP}}(d_{O_P}),s_{O_{UP}}(d_{O_P}))\triangleq c_{O_{UP}}(d_{O_P})$$

There are two communication functions between the Resource User and the Lease as follows.

$$\gamma(r_{I_{UL}}(d_{I_{L}}),s_{I_{UL}}(d_{I_{L}}))\triangleq c_{I_{UL}}(d_{I_{L}})$$

$$\gamma(r_{O_{UL}}(d_{O_{L}}),s_{O_{UL}}(d_{O_{L}}))\triangleq c_{O_{UL}}(d_{O_{L}})$$

There are two communication functions between the Resource User and the Resource as follows.

$$\gamma(r_{I_{UR}}(d_{I_{R}}),s_{I_{UR}}(d_{I_{R}}))\triangleq c_{I_{UR}}(d_{I_{R}})$$

$$\gamma(r_{O_{UR}}(d_{O_{R}}),s_{O_{UR}}(d_{O_{R}}))\triangleq c_{O_{UR}}(d_{O_{R}})$$

Let all modules be in parallel, then the Leasing pattern $U\quad L \quad P\quad R$ can be presented by the following process term.

$\tau_I(\partial_H(\Theta(U\between L\between P\between R)))=\tau_I(\partial_H(U\between L\between P\between R))$

where $H=\{r_{I_{UP}}(d_{I_{P}}),s_{I_{UP}}(d_{I_{P}}),r_{O_{UP}}(d_{O_P}),s_{O_{UP}}(d_{O_P}),r_{I_{UL}}(d_{I_{L}}),s_{I_{UL}}(d_{I_{L}}),\\
r_{O_{UL}}(d_{O_{L}}),s_{O_{UL}}(d_{O_{L}}),r_{I_{UR}}(d_{I_{R}}),s_{I_{UR}}(d_{I_{R}}),r_{O_{UR}}(d_{O_R}),s_{O_{UR}}(d_{O_R})\\
|d_{I}, d_{I_{P}}, d_{I_{L}}, d_{I_R}, d_{O_{P}}, d_{O_L}, d_{O_R}, d_{O}\in\Delta\}$,

$I=\{c_{I_{UP}}(d_{I_{P}}),c_{O_{UP}}(d_{O_P}),c_{I_{UL}}(d_{I_{L}}),c_{O_{UL}}(d_{O_{L}}),c_{I_{UR}}(d_{I_{R}}),c_{O_{UR}}(d_{O_{R}}),\\
UF_1,UF_2,UF_3,UF_4,PF,LF,RF
|d_{I}, d_{I_{P}}, d_{I_{L}}, d_{I_R}, d_{O_{P}}, d_{O_L}, d_{O_R}, d_{O}\in\Delta\}$.

Then we get the following conclusion on the Leasing pattern.

\begin{theorem}[Correctness of the Leasing pattern]
The Leasing pattern $\tau_I(\partial_H(U\between L\between P\between R))$ can exhibit desired external behaviors.
\end{theorem}

\begin{proof}
Based on the above state transitions of the above modules, by use of the algebraic laws of APTC, we can prove that

$\tau_I(\partial_H(U\between L\between P\between R))=\sum_{d_{I},d_{O}\in\Delta}(r_{I}(d_{I})\cdot s_{O}(d_{O}))\cdot
\tau_I(\partial_H(U\between L\between P\between R))$,

that is, the Leasing pattern $\tau_I(\partial_H(U\between L\between P\between R))$ can exhibit desired external behaviors.

For the details of proof, please refer to section \ref{app}, and we omit it.
\end{proof}

\subsubsection{Verification of the Evictor Pattern}

The Evictor pattern allows different strategies to release the resources.
There are three modules in the Evictor pattern: the Resource User, the Evictor,
and the Resource. The Resource User interacts with the outside through
the channels $I$ and $O$; with the Evictor through the channels $I_{UE}$ and $O_{UE}$; with the Resource through the channels $I_{UR}$ and $O_{UR}$. The Evictor interacts with the
Resource through the channels $I_{ER}$ and $O_{ER}$.
As illustrates in Figure \ref{Ev7}.

\begin{figure}
    \centering
    \includegraphics{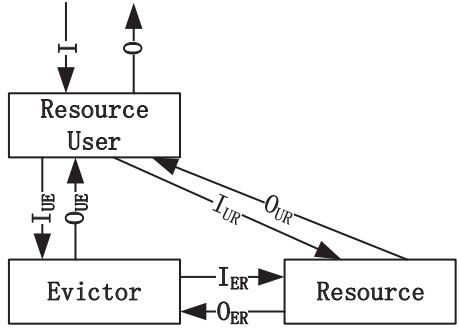}
    \caption{Evictor pattern}
    \label{Ev7}
\end{figure}

The typical process of the Evictor pattern is shown in Figure \ref{Ev7P} and as follows.

\begin{enumerate}
  \item The Resource User receives the input $d_{I}$ from the outside through the channel $I$ (the corresponding reading action is denoted $r_{I}(d_{I})$), then processes the input $d_{I}$ through a processing
  function $UF_1$, and generates the input $d_{I_R}$, and sends the input to the Resource through the channel $I_{UR}$
  (the corresponding sending action is denoted $s_{I_{UR}}(d_{I_R})$);
  \item The Resource receives the input from the Resource User through the channel $I_{UR}$ (the corresponding reading action is denoted $r_{I_{UR}}(d_{I_R})$), then
  processes the input and generate the output $d_{O_R}$ to the Resource User through a processing function $RF_1$, and sends the output to the Resource User
   through the channel $O_{UR}$ (the corresponding sending action is denoted $s_{O_{UR}}(d_{O_{R}})$);
  \item The Resource User receives the output $d_{O_{R}}$ from the Resource through the channel $O_{UR}$ (the corresponding reading action is denoted $r_{O_{UR}}(d_{O_{R}})$), then processes
  the output through a processing function $UF_2$, generates and sends the input $d_{I_{E}}$ to the Evictor through the channel $I_{UE}$ (the corresponding sending action is denoted $s_{I_{UE}}(d_{I_{E}})$);
  \item The Evictor receives the input $d_{I_E}$ from the Resource User through the channel $I_{UE}$ (the corresponding reading action is denoted $r_{I_{UE}}(d_{I_E})$), then processes
  the input through a processing function $EF_1$, generates and sends the input $d_{I_{R'}}$ (the corresponding sending action is denoted $s_{I_{ER}}(d_{I_{R'}})$);
  \item The Resource receives the input from the Evictor through the channel $I_{ER}$ (the corresponding reading action is denoted $r_{I_{ER}}(d_{I_{R'}})$), then
  processes the input and generate the output $d_{O_{R'}}$ to the Evictor through a processing function $RF_2$, and sends the output to the Evictor
   through the channel $O_{ER}$ (the corresponding sending action is denoted $s_{O_{ER}}(d_{O_{R'}})$);
  \item The Evictor receives $d_{O_{R'}}$ from the Resource through the channel $O_{ER}$ (the corresponding reading action is denoted $r_{O_{ER}}(d_{O_{R'}})$), then processes
  the input through a processing function $EF_2$, generates and sends the output $d_{O_{E}}$ (the corresponding sending action is denoted $s_{O_{UE}}(d_{O_{E}})$);
  \item The Resource User receives the response $d_{O_{E}}$ from the Evictor through the channel $O_{UE}$ (the corresponding reading action is denoted $r_{O_{UE}}(d_{O_{E}})$), then
  processes the response and generates the response $d_{O}$ through a processing function $UF_3$, and sends the response to the outside through the channel $O$ (the corresponding sending action is denoted
  $s_{O}(d_{O})$).
\end{enumerate}

\begin{figure}
    \centering
    \includegraphics{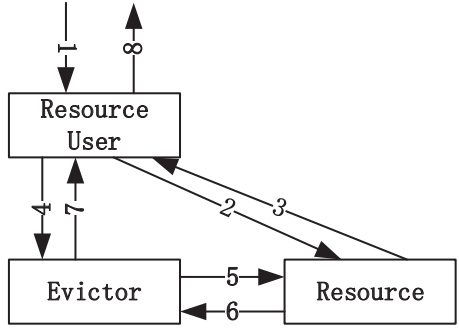}
    \caption{Typical process of Evictor pattern}
    \label{Ev7P}
\end{figure}

In the following, we verify the Evictor pattern. We assume all data elements $d_{I}$, $d_{I_{E}}$, $d_{I_R}$, $d_{I_{R'}}$, $d_{O_E}$, $d_{O_R}$, $d_{O_{R'}}$, $d_{O}$ are from a finite set
$\Delta$.

The state transitions of the Resource User module
described by APTC are as follows.

$U=\sum_{d_{I}\in\Delta}(r_{I}(d_{I})\cdot U_{2})$

$U_{2}=UF_1\cdot U_{3}$

$U_{3}=\sum_{d_{I_{R}}\in\Delta}(s_{I_{UR}}(d_{I_{R}})\cdot U_4)$

$U_4=\sum_{d_{O_R}\in\Delta}(r_{O_{UR}}(d_{O_R})\cdot U_{5})$

$U_{5}=UF_2\cdot U_{6}$

$U_{6}=\sum_{d_{I_{E}}\in\Delta}(s_{I_{UE}}(d_{I_{E}})\cdot U_{7})$

$U_{7}=\sum_{d_{O_E}\in\Delta}(r_{O_{UE}}(d_{O_E})\cdot U_{8})$

$U_{8}=UF_3\cdot U_{9}$

$U_{9}=\sum_{d_{O}\in\Delta}(s_{O}(d_{O})\cdot U)$

The state transitions of the Evictor module
described by APTC are as follows.

$E=\sum_{d_{I_{E}}\in\Delta}(r_{I_{UE}}(d_{I_E})\cdot E_{2})$

$E_{2}=EF_1\cdot E_{3}$

$E_{3}=\sum_{d_{I_{R'}}\in\Delta}(s_{I_{ER}}(d_{I_{R'}})\cdot E_4)$

$E_4=\sum_{d_{O_{R'}}\in\Delta}(r_{O_{ER}}(d_{O_{R'}})\cdot E_{5})$

$E_{5}=EF_2\cdot E_{6}$

$E_{6}=\sum_{d_{O_{E}}\in\Delta}(s_{O_{UE}}(d_{O_{E}})\cdot E)$

The state transitions of the Resource module
described by APTC are as follows.

$R=\sum_{d_{I_{R}}\in\Delta}(r_{I_{UR}}(d_{I_R})\cdot R_{2})$

$R_{2}=RF\cdot R_{3}$

$R_{3}=\sum_{d_{O_{R}}\in\Delta}(s_{O_{UR}}(d_{O_{R}})\cdot R)$

The sending action and the reading action of the same data through the same channel can communicate with each other, otherwise, will cause a deadlock $\delta$. We define the following
communication functions between the Resource User and the Evictor.

$$\gamma(r_{I_{UE}}(d_{I_{E}}),s_{I_{UE}}(d_{I_{E}}))\triangleq c_{I_{UE}}(d_{I_{E}})$$

$$\gamma(r_{O_{UE}}(d_{O_E}),s_{O_{UE}}(d_{O_E}))\triangleq c_{O_{UE}}(d_{O_E})$$

There are two communication functions between the Resource User and the Resource as follows.

$$\gamma(r_{I_{UR}}(d_{I_{R}}),s_{I_{UR}}(d_{I_{R}}))\triangleq c_{I_{UR}}(d_{I_{R}})$$

$$\gamma(r_{O_{UR}}(d_{O_{R}}),s_{O_{UR}}(d_{O_{R}}))\triangleq c_{O_{UR}}(d_{O_{R}})$$

There are two communication functions between the Evictor and the Resource as follows.

$$\gamma(r_{I_{ER}}(d_{I_{R'}}),s_{I_{ER}}(d_{I_{R'}}))\triangleq c_{I_{ER}}(d_{I_{R'}})$$

$$\gamma(r_{O_{ER}}(d_{O_{R'}}),s_{O_{ER}}(d_{O_{R'}}))\triangleq c_{O_{ER}}(d_{O_{R'}})$$

Let all modules be in parallel, then the Evictor pattern $U\quad E\quad R$ can be presented by the following process term.

$\tau_I(\partial_H(\Theta(U\between E\between R)))=\tau_I(\partial_H(U\between E\between R))$

where $H=\{r_{I_{UE}}(d_{I_{E}}),s_{I_{UE}}(d_{I_{E}}),r_{O_{UE}}(d_{O_E}),s_{O_{UE}}(d_{O_E}),r_{I_{UR}}(d_{I_{R}}),s_{I_{UR}}(d_{I_{R}}),\\
r_{O_{UR}}(d_{O_R}),s_{O_{UR}}(d_{O_R}),r_{I_{ER}}(d_{I_{R'}}),s_{I_{ER}}(d_{I_{R'}}),r_{O_{ER}}(d_{O_{R'}}),s_{O_{ER}}(d_{O_{R'}})\\
|d_{I}, d_{I_{E}}, d_{I_R}, d_{I_{R'}}, d_{O_{R'}}, d_{O_{E}}, d_{O_R}, d_{O}\in\Delta\}$,

$I=\{c_{I_{UE}}(d_{I_{E}}),c_{O_{UE}}(d_{O_E}),c_{I_{UR}}(d_{I_{R}}),c_{O_{UR}}(d_{O_{R}}),c_{I_{ER}}(d_{I_{R'}}),c_{O_{ER}}(d_{O_{R'}}),\\
UF_1,UF_2,UF_3,EF_1,EF_2,RF
|d_{I}, d_{I_{E}}, d_{I_R}, d_{I_{R'}}, d_{O_{R'}}, d_{O_{E}}, d_{O_R}, d_{O}\in\Delta\}$.

Then we get the following conclusion on the Evictor pattern.

\begin{theorem}[Correctness of the Evictor pattern]
The Evictor pattern $\tau_I(\partial_H(U\between E\between R))$ can exhibit desired external behaviors.
\end{theorem}

\begin{proof}
Based on the above state transitions of the above modules, by use of the algebraic laws of APTC, we can prove that

$\tau_I(\partial_H(U\between E\between R))=\sum_{d_{I},d_{O}\in\Delta}(r_{I}(d_{I})\cdot s_{O}(d_{O}))\cdot
\tau_I(\partial_H(U\between E\between R))$,

that is, the Evictor pattern $\tau_I(\partial_H(U\between E\between R))$ can exhibit desired external behaviors.

For the details of proof, please refer to section \ref{app}, and we omit it.
\end{proof}

\newpage\section{Composition of Patterns}

Patterns can be composed to satisfy the actual requirements freely, once the syntax and semantics of the output of one pattern just can be plugged into the syntax and semantics of the input of
another pattern.

In this chapter, we show the composition of patterns. In section \ref{CL8}, we verify the composition of the Layers patterns.
In section \ref{CPAC8}, we show the composition of Presentation-Abstraction-Control (PAC) patterns. We compose patterns for resource
management in section \ref{CRM8}.

\subsection{Composition of the Layers Patterns}\label{CL8}

In this subsection, we show the composition of the Layers patterns, and verify the correctness of the composition. We have already verified the correctness of the Layers pattern and its
composition in section \ref{Layers3}, here we verify the correctness of the composition of the Layers patterns based on the correctness result of the Layers pattern.

The composition of two layers peers is illustrated in Figure \ref{Layers8}. Each layers peer is abstracted as a module, and the composition of two layers peers is also abstracted as a
new module, as the dotted rectangles illustrate in Figure \ref{Layers8}.
\begin{figure}
    \centering
    \includegraphics{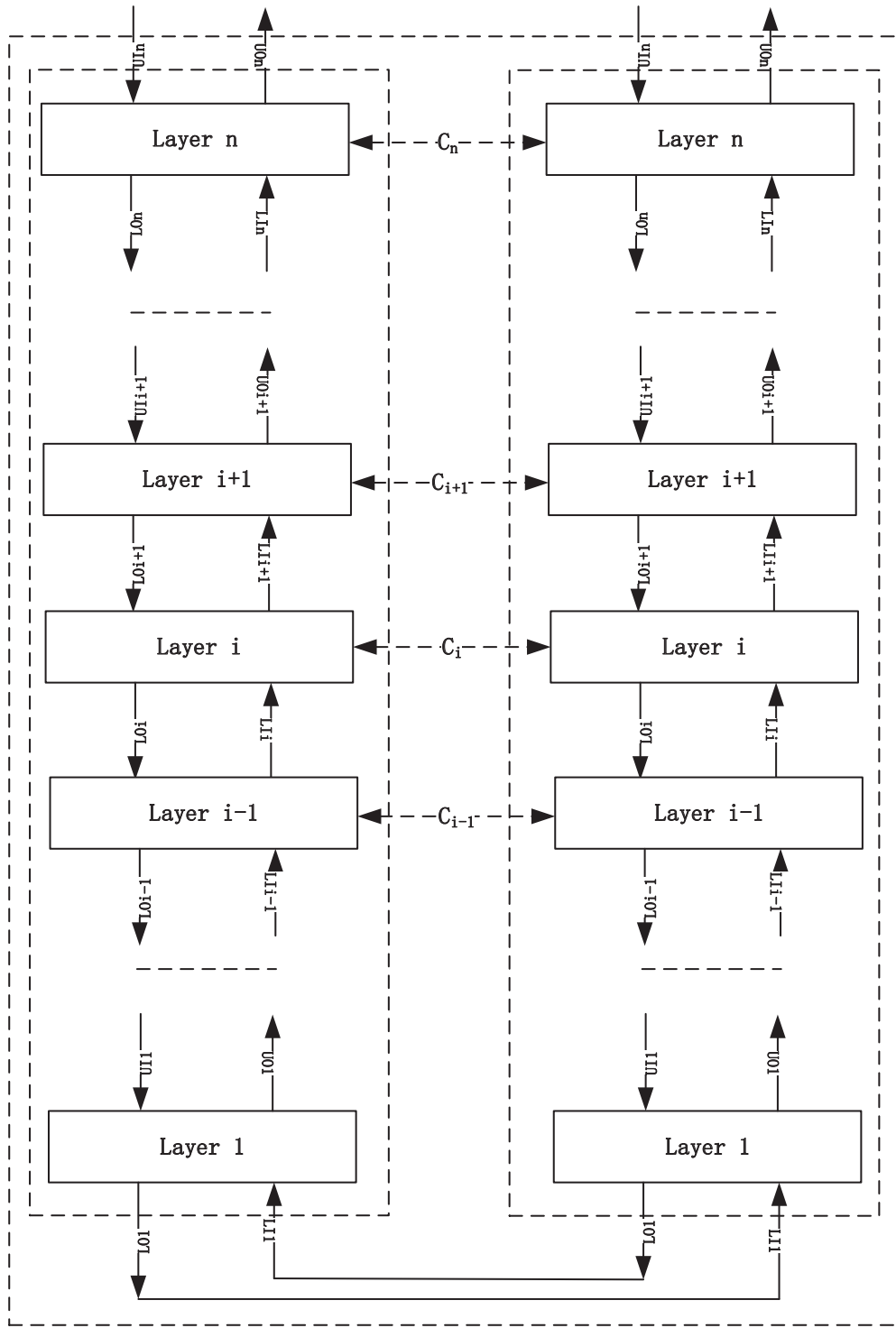}
    \caption{Composition of the Layers patterns}
    \label{Layers8}
\end{figure}

There are two typical processes in the composition of two layers peers: one is the direction from peer $P$ to peer $P'$, the other is the direction from $P'$ two $P$. We omit them, please
refer to section \ref{Layers3} for details.

In the following, we verify the correctness of the plugging of two layers peers. We assume all data elements $d_{L_1}$, $d_{L_{1'}}$, $d_{U_n}$, $d_{U_{n'}}$ are from a finite set
$\Delta$. Note that, the channel $LO_1$ and the channel $LI_{1'}$ are the same one channel; the channel $LO_{1'}$ and the channel $LI_1$ are the same one channel. And the data $d_{L_{1'}}$
and the data $PUF(d_{U_n})$ are the same data; the data $d_{L_1}$ and the data $P'UF(d_{U_{n'}})$ are the same data.

The state transitions of the $P$
described by APTC are as follows.

$P=\sum_{d_{U_n},d_{L_1}\in\Delta}(r_{UI_n}(d_{U_n})\cdot P_2\between r_{LI_1}(d_{L_1})\cdot P_3)$

$P_2=PUF\cdot P_4$

$P_3=PLF\cdot P_5$

$P_4=\sum_{d_{U_n}\in\Delta}(s_{LO_1}(PUF(d_{U_n}))\cdot P)$

$P_5=\sum_{d_{L_1}\in\Delta}(s_{UO_n}(PLF(d_{L_1}))\cdot P)$

The state transitions of the $P'$
described by APTC are as follows.

$P'=\sum_{d_{U_{n'}},d_{L_{1'}}\in\Delta}(r_{UI_{n'}}(d_{U_{n'}})\cdot P'_2\between r_{LI_{1'}}(d_{L_{1'}})\cdot P'_3)$

$P'_2=P'UF\cdot P'_4$

$P'_3=P'LF\cdot P'_5$

$P'_4=\sum_{d_{U_{n'}}\in\Delta}(s_{LO_{1'}}(PUF(d_{U_{n'}}))\cdot P')$

$P'_5=\sum_{d_{L_{1'}}\in\Delta}(s_{UO_{n'}}(PLF(d_{L_{1'}}))\cdot P')$

The sending action and the reading action of the same data through the same channel can communicate with each other, otherwise, will cause a deadlock $\delta$. We define the following
communication functions.

$$\gamma(r_{LI_1}(d_{L_1}),s_{LO_{1'}}(P'UF(d_{U_{n'}})))\triangleq c_{LI_1}(d_{L_1})$$

$$\gamma(r_{LI_{1'}}(d_{L_{1'}}),s_{LO_{1}}(PUF(d_{U_{n}})))\triangleq c_{LI_{1'}}(d_{L_{1'}})$$

Let all modules be in parallel, then the two layers peers $P\quad P'$ can be presented by the following process term.

$\tau_I(\partial_H(\Theta(\tau_{I_1}(\partial_{H_1}(P))\between \tau_{I_2}(\partial_{H_2}(P')))))=\tau_I(\partial_H(\tau_{I_1}(\partial_{H_1}(P))\between \tau_{I_2}(\partial_{H_2}(P'))))$

where $H=\{r_{LI_1}(d_{L_1}),s_{LO_{1'}}(P'UF(d_{U_{n'}})),r_{LI_{1'}}(d_{L_{1'}}),s_{LO_{1}}(PUF(d_{U_{n}}))\\
|d_{L_1}, d_{L_{1'}}, d_{U_n}, d_{U_{n'}}\in\Delta\}$,

$I=\{c_{LI_1}(d_{L_1}),c_{LI_{1'}}(d_{L_{1'}}),PUF,PLF,P'UF,P'LF
|d_{L_1}, d_{L_{1'}}, d_{U_n}, d_{U_{n'}}\in\Delta\}$.

And about the definitions of $H_1$ and $I_1$, $H_2$ and $I_2$, please see in section \ref{Layers3}.

Then we get the following conclusion on the plugging of two layers peers.

\begin{theorem}[Correctness of the plugging of two layers peers]
The plugging of two layers peers

$\tau_I(\partial_H(\tau_{I_1}(\partial_{H_1}(P))\between \tau_{I_2}(\partial_{H_2}(P'))))$

can exhibit desired external behaviors.
\end{theorem}

\begin{proof}
Based on the above state transitions of the above modules, by use of the algebraic laws of APTC, we can prove that

$\tau_I(\partial_H(\tau_{I_1}(\partial_{H_1}(P))\between \tau_{I_2}(\partial_{H_2}(P'))))=
\sum_{d_{U_n},d_{U_{n'}}\in\Delta}((r_{UI_n}(d_{U_n})\parallel r_{UI_{n'}}(d_{U_{n'}}))\\
\cdot(s_{UO_n}(PLF(P'UF(d_{U_{n'}})))\parallel s_{UO_{n'}}(P'LF(PUF(d_{U_n})))))\cdot
\tau_I(\partial_H(\tau_{I_1}(\partial_{H_1}(P))\between \tau_{I_2}(\partial_{H_2}(P'))))$,

that is, the plugging of two layers peers $\tau_I(\partial_H(\tau_{I_1}(\partial_{H_1}(P))\between \tau_{I_2}(\partial_{H_2}(P'))))$ can exhibit desired external behaviors.

For the details of proof, please refer to section \ref{app}, and we omit it.
\end{proof}

\subsection{Composition of the PAC Patterns}\label{CPAC8}

In this subsection, we show the composition of Presentation-Abstraction-Control (PAC) patterns (we already verified its correctness in section \ref{PAC3}) and verify the correctness of the
composition.

The PAC patterns can be composed into levels of PACs, as illustrated in Figure \ref{PACs8}.

\begin{figure}
    \centering
    \includegraphics{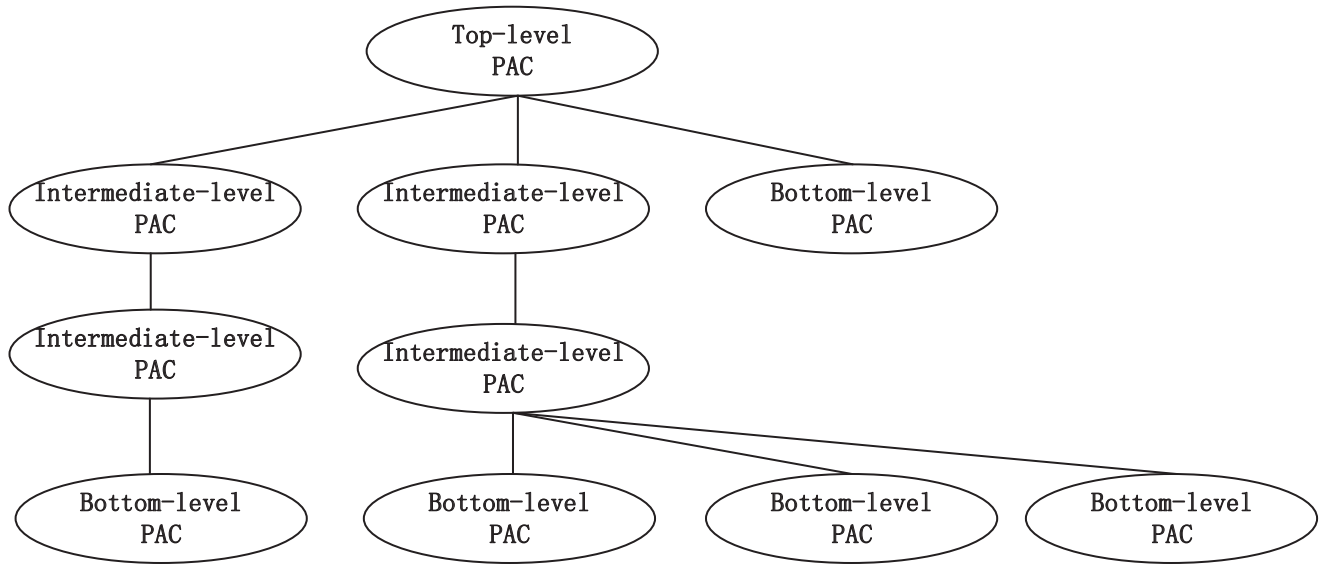}
    \caption{Levels of the PAC pattern}
    \label{PACs8}
\end{figure}

If the syntax and semantics of the output of one PAC match the syntax and semantics of the input of another PAC, then they can be composed. For the simplicity and without loss of generality,
we show the plugging of only two PACs, as illustrated in Figure \ref{PACs28}. Each PAC is abstracted as a module, and the composition of two PACs is also abstracted as a
new module, as the dotted rectangles illustrate in Figure \ref{PACs28}.

\begin{figure}
    \centering
    \includegraphics{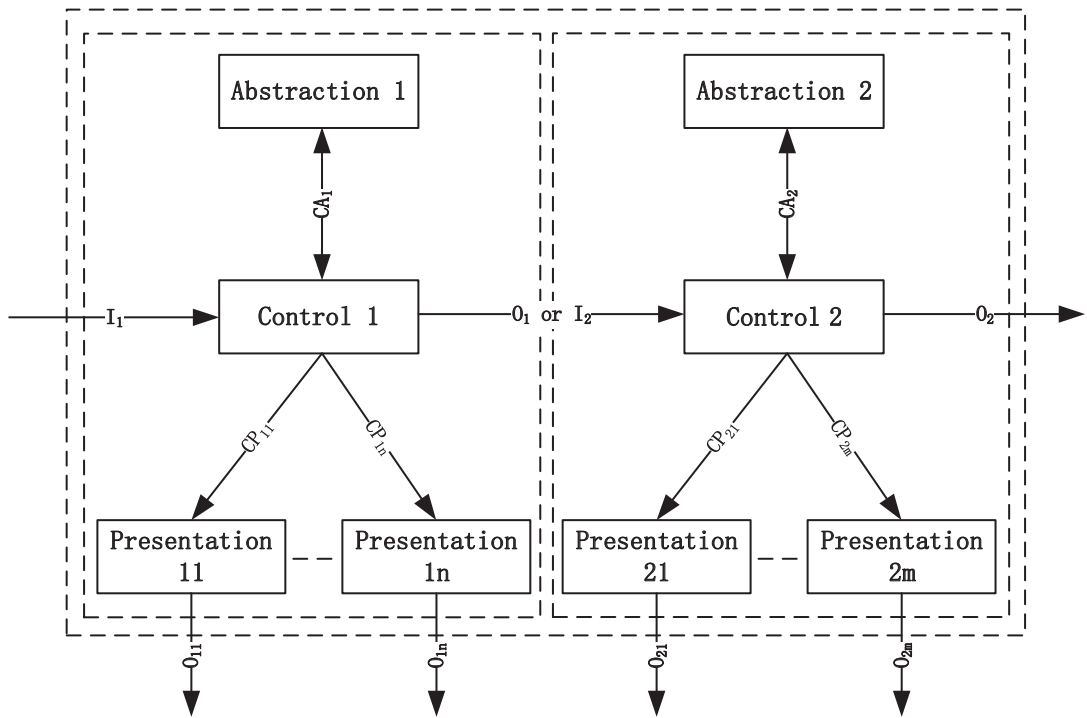}
    \caption{Plugging of two PACs}
    \label{PACs28}
\end{figure}

The typical process of plugging of two PACs is composed by the process of one PAC (the typical process is described in section \ref{PAC3}) follows the process of the other PAC, and we
omit it.

In the following, we verify the correctness of the plugging of two PACs. We assume all data elements $d_{I_1}$, $d_{O_1}$, $d_{I_2}$, $d_{O_2}$, $d_I$, $d_O$, $d_{O_{1i}}$ (for $1\leq i\leq n$), $d_{O_{2j}}$ (for $1\leq j\leq m$) are from a finite set
$\Delta$. Note that, the channel $I$ and the channel $I_1$ are the same one channel; the channel $O_1$ and the channel $I_2$ are the same one channel; and the channel $O_2$ and the channel
$O$ are the same channel. And the data $d_{I_1}$ and the data $d_I$ are the same data; the data $d_{O_1}$ and the data $d_{I_2}$ are the same data; and the data $d_{O_2}$ and the data
$d_O$ are the same data.

The state transitions of the $PAC_1$
described by APTC are as follows.

$PAC_{1}=\sum_{d_{I_1}\in\Delta}(r_{I_1}(d_{I_1})\cdot PAC_{1_2})$

$PAC_{1_2}=PAC_1F\cdot PAC_{1_3}$

$PAC_{1_3}=\sum_{d_{O_1}\in\Delta}(s_{O_1}(d_{O_1})\cdot PAC_{1_4})$

$PAC_{1_4}=\sum_{d_{O_{1_1},\cdots,d_{O_{1_n}}}\in\Delta}(s_{O_{1_1}}(d_{O_{1_1}})\between\cdots\between s_{O_{1_n}}(d_{O_{1_n}})\cdot PAC_{1})$

The state transitions of the $PAC_2$
described by APTC are as follows.

$PAC_{2}=\sum_{d_{I_2}\in\Delta}(r_{I_2}(d_{I_2})\cdot PAC_{2_2})$

$PAC_{2_2}=PAC_2F\cdot PAC_{2_3}$

$PAC_{2_3}=\sum_{d_{O_2}\in\Delta}(s_{O_2}(d_{O_2})\cdot PAC_{2_4})$

$PAC_{2_4}=\sum_{d_{O_{2_1},\cdots,d_{O_{2_m}}}\in\Delta}(s_{O_{2_1}}(d_{O_{2_1}})\between\cdots\between s_{O_{2_m}}(d_{O_{2_m}})\cdot PAC_{2})$

The sending action and the reading action of the same data through the same channel can communicate with each other, otherwise, will cause a deadlock $\delta$. We define the following
communication functions.

$$\gamma(r_{I_2}(d_{I_{2}}),s_{O_1}(d_{O_{1}}))\triangleq c_{I_2}(d_{I_{2}})$$

Let all modules be in parallel, then the two PACs $PAC_1PAC_2$ can be presented by the following process term.

$\tau_I(\partial_H(\Theta(\tau_{I_1}(\partial_{H_1}(PAC_1))\between \tau_{I_2}(\partial_{H_2}((PAC_2))))))=\tau_I(\partial_H(\tau_{I_1}(\partial_{H_1}(PAC_1))\between \tau_{I_2}(\partial_{H_2}((PAC_2)))))$

where $H=\{r_{I_2}(d_{I_{2}}),s_{O_1}(d_{O_{1}})
|d_{I_1}, d_{O_1}, d_{I_2}, d_{O_2}, d_I, d_O, d_{O_{1i}}, d_{O_{2j}}\in\Delta\}$ for $1\leq i\leq n$ and $1\leq j\leq m$,

$I=\{c_{I_2}(d_{I_{2}}),PAC_1F,PAC_2F
|d_{I_1}, d_{O_1}, d_{I_2}, d_{O_2}, d_I, d_O, d_{O_{1i}}, d_{O_{2j}}\in\Delta\}$ for $1\leq i\leq n$ and $1\leq j\leq m$.

And about the definitions of $H_1$ and $I_1$, $H_2$ and $I_2$, please see in section \ref{PAC3}.

Then we get the following conclusion on the plugging of two PACs.

\begin{theorem}[Correctness of the plugging of two PACs]
The plugging of two PACs

$\tau_I(\partial_H(\tau_{I_1}(\partial_{H_1}(PAC_1))\between \tau_{I_2}(\partial_{H_2}(PAC_2))))$

can exhibit desired external behaviors.
\end{theorem}

\begin{proof}
Based on the above state transitions of the above modules, by use of the algebraic laws of APTC, we can prove that

$\tau_I(\partial_H(\tau_{I_1}(\partial_{H_1}(PAC_1))\between \tau_{I_2}(\partial_{H_2}(PAC_2))))=\sum_{d_{I},d_O,d_{O_{1_1}},\cdots,d_{O_{1_n}},d_{O_{2_1}},\cdots,d_{O_{2_m}}\in\Delta}(r_{I}(d_{I})\cdot s_O(d_O)\cdot s_{O_{1_1}}(d_{O_{1_1}})\parallel\cdots\parallel s_{O_{1_i}}(d_{O_{1_i}})\parallel\cdots\parallel s_{O_{1_n}}(d_{O_{1_n}})\cdot
s_{O_{2_1}}(d_{O_{2_1}})\parallel\cdots\parallel s_{O_{2_j}}(d_{O_{2_j}})\parallel\cdots\parallel s_{O_{2_m}}(d_{O_{2_m}}))\cdot
\tau_I(\partial_H(\tau_{I_1}(\partial_{H_1}(PAC_1))\between \tau_{I_2}(\partial_{H_2}(PAC_2))))$,

that is, the plugging of two PACs $\tau_I(\partial_H(\tau_{I_1}(\partial_{H_1}(PAC_1))\between \tau_{I_2}(\partial_{H_2}(PAC_2))))$ can exhibit desired external behaviors.

For the details of proof, please refer to section \ref{app}, and we omit it.
\end{proof}

\subsection{Composition of Resource Management Patterns}\label{CRM8}

In this subsection, we show the composition of resource management patterns (we have already verified the correctness of patterns for resource management in section \ref{RMP}),
and verify the correctness of the composition.

The whole process of resource management involves resource acquisition firstly, resource utilization and lifecycle management secondly, and resource release lastly,
as Figure \ref{CRM81} illustrates. For resource acquisition, we take an example of the Lookup pattern, and the Lifecycle Manager pattern for resource lifecycle management, and the Leasing
pattern for resource release. The whole process of resource management is composed of the typical processes of the Lookup pattern, the Lifecycle Manager pattern and the Leasing pattern,
we do not repeat any more, please refer to the details of these three patterns in section \ref{RMP}. And we can verify the correctness of the whole resource management system shown in
Figure \ref{CRM81}, just like the work we doing many times for concrete patterns in the above sections.

\begin{figure}
    \centering
    \includegraphics{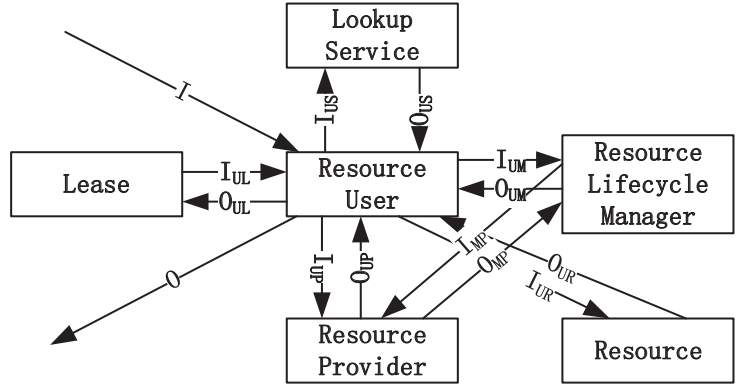}
    \caption{The whole resource management process}
    \label{CRM81}
\end{figure}

But, we do not verify the correctness of the whole resource management system like the previous work. The whole resource management system in Figure \ref{CRM81} contains the full functions
of the Lookup pattern, the Lifecycle manager pattern and the Leasing pattern, and actually can be implemented by the composition of these three patterns, as Figure \ref{CRM82} illustrates.
For the whole process of resource management, firstly the Lookup pattern works, then the Lifecycle Manager pattern, and lastly the Leasing pattern. That is, the output of the Lookup pattern
is plugged into the input of the Lifecycle Manager pattern, and the output of the Lifecycle Manager pattern is plugged into the Leasing pattern. Each pattern is abstracted as a module,
and the composition of these three patterns is also abstracted as a new module, as the dotted rectangles illustrate in Figure \ref{CRM82}.

\begin{figure}
    \centering
    \includegraphics{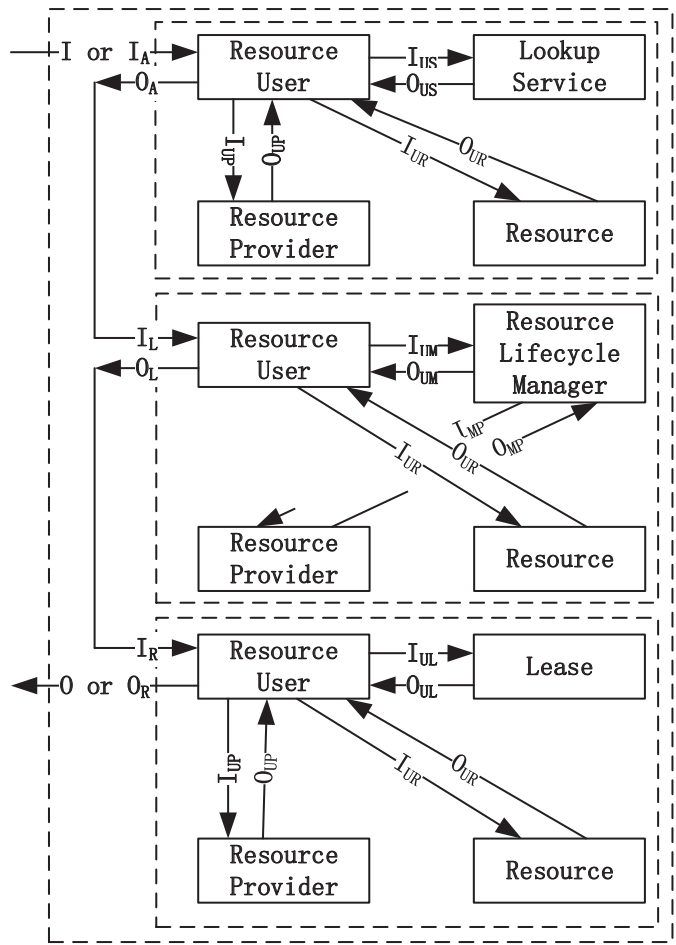}
    \caption{Plugging of resource management patterns}
    \label{CRM82}
\end{figure}

In the following, we verify the correctness of the plugging of resource management patterns. We assume all data elements $d_I$, $d_O$, $d_{I_A}$, $d_{O_A}$, $d_{I_L}$, $d_{O_L}$, $d_{I_R}$, $d_{O_R}$ are from a finite set
$\Delta$. Note that, the channel $I$ and the channel $I_A$ are the same one channel; the channel $O_A$ and the channel $I_L$ are the same one channel; and the channel $O_L$ and the channel
$I_R$ are the same channel; the channel $O_R$ and the channel $O$ are the same channel. And the data $d_{I_A}$ and the data $d_I$ are the same data; the data $d_{O_A}$ and the data $d_{I_L}$ are the same data; and the data $d_{O_L}$ and the data
$d_{I_R}$ are the same data; the data $d_{O_R}$ and the data $d_O$ are the same data.

The state transitions of the Lookup pattern $A$
described by APTC are as follows.

$A=\sum_{d_{I_A}\in\Delta}(r_{I_A}(d_{I_A})\cdot A_2)$

$A_2=AF\cdot A_3$

$A_3=\sum_{d_{O_A}\in\Delta}(s_{O_A}(d_{O_A})\cdot A)$

The state transitions of the Lifecycle Manager pattern $L$
described by APTC are as follows.

$L=\sum_{d_{I_L}\in\Delta}(r_{I_L}(d_{I_L})\cdot L_2)$

$L_2=LF\cdot L_3$

$L_3=\sum_{d_{O_L}\in\Delta}(s_{O_L}(d_{O_L})\cdot L)$

The state transitions of the Leasing pattern $R$
described by APTC are as follows.

$R=\sum_{d_{I_R}\in\Delta}(r_{I_R}(d_{I_R})\cdot R_2)$

$R_2=RF\cdot R_3$

$R_3=\sum_{d_{O_R}\in\Delta}(s_{O_R}(d_{O_R})\cdot R)$

The sending action and the reading action of the same data through the same channel can communicate with each other, otherwise, will cause a deadlock $\delta$. We define the following
communication functions between the Lookup pattern and the Lifecyle Manager pattern.

$$\gamma(r_{I_L}(d_{I_{L}}),s_{O_A}(d_{O_{A}}))\triangleq c_{I_L}(d_{I_{L}})$$

We define the following
communication functions between the Lifecyle Manager pattern and the Leasing pattern.

$$\gamma(r_{I_R}(d_{I_{R}}),s_{O_L}(d_{O_{L}}))\triangleq c_{I_R}(d_{I_{R}})$$

Let all modules be in parallel, then the resource management patterns $A\quad L\quad R$ can be presented by the following process term.

$\tau_I(\partial_H(\Theta(\tau_{I_1}(\partial_{H_1}(A))\between \tau_{I_2}(\partial_{H_2}(L))\between \tau_{I_3}(\partial_{H_3}(R)))))=\tau_I(\partial_H(\tau_{I_1}(\partial_{H_1}(A))\between \tau_{I_2}(\partial_{H_2}(L))\between \tau_{I_3}(\partial_{H_3}(R))))$

where $H=\{r_{I_L}(d_{I_{L}}),s_{O_A}(d_{O_{A}}),r_{I_R}(d_{I_{R}}),s_{O_L}(d_{O_{L}})\\
|d_I, d_O, d_{I_A}, d_{O_A}, d_{I_L}, d_{O_L}, d_{I_R}, d_{O_R}\in\Delta\}$,

$I=\{c_{I_L}(d_{I_{L}}),c_{I_R}(d_{I_{R}}),AF,LF,RF
|d_I, d_O, d_{I_A}, d_{O_A}, d_{I_L}, d_{O_L}, d_{I_R}, d_{O_R}\in\Delta\}$.

And about the definitions of $H_1$ and $I_1$, $H_2$ and $I_2$, $H_3$ and $I_3$, please see in section \ref{RMP}.

Then we get the following conclusion on the plugging of resource management patterns.

\begin{theorem}[Correctness of the plugging of resource management patterns]
The plugging of resource management patterns

$\tau_I(\partial_H(\tau_{I_1}(\partial_{H_1}(A))\between \tau_{I_2}(\partial_{H_2}(L))\between \tau_{I_3}(\partial_{H_3}(R))))$

can exhibit desired external behaviors.
\end{theorem}

\begin{proof}
Based on the above state transitions of the above modules, by use of the algebraic laws of APTC, we can prove that

$\tau_I(\partial_H(\tau_{I_1}(\partial_{H_1}(A))\between \tau_{I_2}(\partial_{H_2}(L))\between \tau_{I_3}(\partial_{H_3}(R))))=\sum_{d_{I},d_O\in\Delta}(r_{I}(d_{I})\cdot s_O(d_O))\cdot\\
\tau_I(\partial_H(\tau_{I_1}(\partial_{H_1}(A))\between \tau_{I_2}(\partial_{H_2}(L))\between \tau_{I_3}(\partial_{H_3}(R))))$,

that is, the plugging of resource management patterns $\tau_I(\partial_H(\tau_{I_1}(\partial_{H_1}(A))\between \tau_{I_2}(\partial_{H_2}(L))\between \tau_{I_3}(\partial_{H_3}(R))))$ can exhibit desired external behaviors.

For the details of proof, please refer to section \ref{app}, and we omit it.
\end{proof}

\end{document}